\documentclass[10.5pt]{ociamthesis}  

\usepackage[b5paper]{geometry}
\usepackage[dvipsnames]{xcolor}
\usepackage[colorlinks=true,linkcolor=blue,citecolor=Brown,urlcolor=blue,filecolor=black,hypertexnames=true]{hyperref}

\usepackage[utf8]{inputenc}

\makeatletter
\renewcommand\part{%
  \if@openright
    \cleardoublepage
  \else
    \clearpage
  \fi
  \thispagestyle{empty}
  \if@twocolumn
    \onecolumn
    \@tempswatrue
  \else
    \@tempswafalse
  \fi
  \null\vfil
  \secdef\@part\@spart}
\makeatother

\usepackage{tensor}

\usepackage[normalem]{ulem}

\usepackage{amsmath}
\usepackage{amsfonts,amstext}
\usepackage{amssymb,amsthm}
\usepackage{mathtools}
\usepackage{textcomp}
\usepackage{mathrsfs}
\usepackage{multirow}
\usepackage{array}
\usepackage{calc}
\usepackage{youngtab}
\usepackage[sort]{cite}

\usepackage{enumerate}

\usepackage[usestackEOL]{stackengine}
\usepackage{scalerel}

\def\Lbar#1{\ThisStyle{%
  \setbox0=\hbox{$\SavedStyle#1$}%
  \stackengine{1.2\LMpt}{$\SavedStyle#1$}{\rule{\wd0}{.4\LMpt}}{0}{c}{F}{F}{S}%
}}

\usepackage{scrextend}

\usepackage{mdwlist}

\theoremstyle{definition}

\newtheorem{theorem}{Theorem}[chapter]
\newtheorem{corollary}{Corollary}[chapter]
\newtheorem{lemma}[theorem]{Lemma}
\theoremstyle{proposition}
\newtheorem{proposition}{Proposition}[chapter]

\newcommand{\f}[2]{\frac{#1}{#2}}

\deffootnote{1em}{0em}{\thefootnotemark\quad\!\!\!}

\newcommand{\enquote}[1]{``#1''}
\providecommand{\url}[1]{\texttt{#1}}

\providecommand{\eprint}[2][]{\url{#2}}

\def\a{\alpha}

\def\d{\delta}

\def\h{\eta}

\def\Th{\Theta}

\newcommand{\half}{\frac{1}{2}}
\renewcommand{\d}{\partial}

\newcommand{\ffrac}[2]{\raisebox{.5pt}%
  {\footnotesize$\displaystyle\frac{#1}{#2}$}\kern1pt}

\newcommand{\ddl}[2]{\ffrac{\d #1}{\d #2}}
\newcommand{\ddll}[2]{\ffrac{\d^L #1}{\d #2}}

\newcommand{\vddr}[2]{\ffrac{\delta^R #1}{\delta #2}}
\newcommand{\vddl}[2]{{\ffrac{\delta #1}{\delta #2}}}

\def\cA{{\cal A}}

\def\cI{{\cal I}}

\def\cL{{\cal L}}

\def\cR{{\cal R}}

\def\cV{{\cal V}}

\def\be{\begin{equation}}
\def\ee{\end{equation}}
\def\bea{\begin{eqnarray}}
\def\eea{\end{eqnarray}}
\def\ba{\begin{array}}
\def\ea{\end{array}}

\def\nn{\nonumber}

\newcommand{\R}{\mathbb{R}}

\def\12{\frac{1}{2}}

\cleardoublepage
\title{String Dualities \\[1ex]     
        and \\[1ex]
       Gaugings of Supergravity }   

\author{Arash Ranjbar}             
\department{Physique Math\'ematique des Interactions Fondamentales}  

\degree{Docteur en Sciences}     
\advisor{Supervisor: Prof. Marc Henneaux\\ \hspace{-1.3cm} Co-Supervisor: Prof. Jorge Zanelli}
\degreedate{Bruxelles 2018}         

\begin{document}


\setcounter{secnumdepth}{3}
\setcounter{tocdepth}{3}

\setlength{\extrarowheight}{5pt}

\maketitle                  


\newpage
\thispagestyle{empty}
The thesis has been discussed privately on June 6, 2018 in front of the jury including the following members:

\begin{itemize}

\item Professor Riccardo Argurio (ULB)

\item Professor Glenn Barnich, President of jury (ULB)

\item Professor Andr\`es Collinucci, Secretary of jury (ULB)

\item Professor Joaquim Gomis (University of Barcelona)

\item Professor Marc Henneaux, Supervisor (ULB)

\item Professor Jorge Zanelli (Centro de Estudios Cient\'ificos (CECs), Chile)

\end{itemize}

The public presentation has been held on July 11, 2018 in front of the aforementioned jury and public audience.

\vspace{3cm}

The research of the author during the academic years 2013-2016 was partially supported by Centro de Estudios Cient\'ificos (CECs). The Centro de Estudios Cient\'ificos (CECs) is funded by the Chilean Government through the Centers of Excellence Base Financing Program of CONICYT-Chile.

\vspace{1cm}

The author's work during the academic years 2014-2018 was partially supported by the ERC Advanced Grant ``High-Spin-Grav" and by FNRS-Belgium (convention FRFC PDR T.1025.14 and convention IIS N 4.4503.15).

\vspace{1cm}

This thesis has been written using the style \textit{\href{https://www.maths.ox.ac.uk/members/it/faqs/latex/thesis-class}{ociamthesis}} which has been modified by the author in order to fit the specific feature required for its presentation.

\mbox{}

\newgeometry{textwidth=15cm,textheight=24cm}
\begin{acknowledgements}

I would like to thank, first and foremost, my supervisor Marc Henneaux who helped me at every stage during my PhD study in the past four years. The discussions with Marc have been always inspiring for me and I like to think it made me a better physicist. At the same level of appreciation, I would like to thank Jorge Zanelli who is in principle my co-supervisor and with whom I have had long and fruitful discussions for uncountable number of times. I am indebted to both for valuable lessons they have thought me and for being there always despite their overwhelming work schedule. 

I would also like to thank the members of my jury
for first accepting to be in my examination and second for all the inspiring and challenging questions that they have asked (whether or not I knew the answer). For sure, it will open my eyes to other points of view and new ideas to address in the future.

I am grateful to Glenn Barnich, Nicolas Boulanger, Gaston Giribet, Marc Henneaux, Bernard Julia, Victor Lekeu and Jorge Zanelli with whom I collaborated in the projects part of which turned into this thesis. I have learned a lot from each and every one of them and I am grateful to them. 

In the past few years and at different stages of my work, I have enjoyed discussions with many of my colleagues which I am thankful to all of them. Most specially to Andrea Campoleoni, Simone Giacomelli, Paolo Gregori, Adolfo Guarino, Patricia Ritter, Patricio Salgado, Javier Tarrio and Cedric Troessart.

I would also like to thank all members of Theoretical and Mathematical Physics department at ULB and Centro de Estudios Cient\'ificos (CECs) in Chile for creating a friendly and nice atmosphere and for all motivating discussions over past four years.

Finishing my PhD study would not have been possible without support that I have received during the course of past few years. I am grateful to my amazing friends 
 for being there always for me and to endure me in my most unpleasant moments.



Last but not least, I would like to thank my family for all the support that I have received in every second of my life and specially during my studies. It is hard to imagine that without their help I could have been here at this point of my life. I cannot think of a way to thank them enough for all love and support that I have been showered with all my life.


\begin{flushright}
Arash Ranjbar
\end{flushright}

\textit{Brussels}

\textit{July 2018}

\end{acknowledgements}   
\restoregeometry

\cleardoublepage
\pagenumbering{roman}

\begin{originality}
The current thesis, apart from the Chapter 2 and the introductory parts of each chapter which have been included for pedagogical reasons, is based on original materials of the following publications:

\begin{enumerate}[I)]

 \item {\footnotesize \cite{Henneaux:2017afd} M.~Henneaux and A.~Ranjbar,
  ``Timelike duality, $M'$-theory and an exotic form of the Englert solution,''
  JHEP {\bf 1708} (2017) 012 
  [arXiv: 1706.06948 [hep-th]].}
 
 \item  {\footnotesize \cite{Henneaux:2017kbx} M. Henneaux, B. Julia, V. Lekeu and A. Ranjbar, ``A note on gaugings in four spacetime dimensions and electric-magnetic duality," Class. Quant. Grav. {\bf 35}
(2018) no.3, 037001 [arXiv: 1709.06014 [hep-th]].}

 \item {\footnotesize \cite{Barnich:2017nty} G.~Barnich, N.~Boulanger, M.~Henneaux, B.~Julia, V.~Lekeu and A.~Ranjbar,
  ``Deformations of vector-scalar models,''
  JHEP {\bf 1802} (2018) 064
  [arXiv: 1712.08126 [hep-th]].}
  
  \item {\footnotesize \cite{Ranjbar2018} A. Ranjbar and J. Zanelli, ``Parallelizable hyperbolic manifolds in three and seven dimensions,'' in preparation.}

\suspend{enumerate}

Besides above publications, the following are the papers published during the author's PhD studies but have not been included in the thesis:

\resume{enumerate}[{[I)]}]

 \item {\footnotesize J.~P.~Babaro, G.~Giribet and A.~Ranjbar,
  ``Conformal field theories from deformations of theories with $W_n$ symmetry,''
  Phys.\ Rev.\ D {\bf 94} (2016) no.8,  086001
  [arXiv:1605.01933 [hep-th]].}
  
 \item {\footnotesize G.~Giribet and A.~Ranjbar,
  ``Screening Stringy Horizons,''
  Eur.\ Phys.\ J.\ C {\bf 75} (2015) no.10,  490
  [arXiv:1504.05044 [hep-th]].}

\end{enumerate}

\end{originality}


\begin{abstract}

This thesis is devoted to various questions connected with duality. It is composed of two parts.

The first part discusses some aspects of timelike T-duality. We explore the possibility of compactification of supergravity theories with various signatures (low energy limit of $M$-theories which are dual under timelike T-dualities) on parallelizable internal seven dimensional (pseudo-)spheres. We show that, beside the standard theory, only one of the dual theories known as $M'$-theory can admit such a solution. The effective four dimensional theory is non-supersymmetric and due to the presence of torsion the symmetry of seven dimensional (pseudo-)sphere breaks down to $Spin(3,4)$.

In the second part, in an attempt to have a systematic discussion of gaugings in supergravity, we show the isomorphism between the space of local deformations of the appropriate zero coupling limit of the embedding tensor Lagrangian and that of the second-order scalar-vector Lagrangian, describing the bosonic sector of supergravity ignoring gravity, in a chosen duality frame determined by embedding tensors. We analyze the BV-BRST deformation of a class of scalar-vector coupled Lagrangians, which contains supergravity Lagrangians as examples, and find a set of constraints that guarantee the consistency of the deformations of the Lagrangians. We show in principle that for a large class of theories considered in this thesis, the only deformations are those of the Yang-Mills type associated with a subgroup of the rigid symmetries.

\end{abstract}          


\cleardoublepage
\setcounter{page}{1}

\begin{romanpages}          
\tableofcontents            
\end{romanpages}            




\cleardoublepage
\pagenumbering{arabic}
\setcounter{page}{1}

\chapter{Introduction}

This thesis is devoted to various questions connected with duality. It is composed of two parts. The first part explores
some aspects of timelike T-duality in the context of comapctifications of the maximal supergravity in eleven-dimensions and its timelike T-duals. 
The second part studies the consistent deformations (``gaugings") of extended supergravities and sheds new light on the role played by the electric-magnetic duality in that context.

The question of why spacetime has Lorentzian signature $(-,+,+,...,+)$ with one and only one time has been repeatedly raised by many authors in the literature (see \cite{Bars:2000qm} and references therein). Duality provides a tool to address this question since the timelike T-dualities considered by Hull \cite{Hull:1998ym,Hull:1998vg} change the spacetime signature. In the case of $M$-theory, one can generate through timelike dualities various ``exotic" duals ($M^*$-theory and $M'$-theory as well as the reversed signature theories), which possess a different number of ``times". Part \ref{part1} of the thesis studies a special class of solutions of the corresponding exotic supergravities. More preciesly, we investigate compactifications of these theories to four dimensions on internal seven-manifolds that are parallelizable.

The condition that allows a Riemannian geometry
to admit an absolute parallelism is a classic question which was answered by Cartan and Schouten \cite{Cartan1,Cartan2}
in the case of Euclidean signature, concluding that the only possibilities for sphere topologies are $S^3$ and $S^7$. They also described curvature-free connections on
Lie groups and as a result exhibited parallelisms on group manifolds. They then generalized it to give a local description of 
parallelism (in terms of curvature-free metric-compatible connections in the presence of torsion) on Riemannian manifolds which
are products of flat manifolds, compact simple group manifolds, and 7-spheres. Moreover, the parallelism of the $3$-sphere, which was before understood
as an isolated phenomenon, was provided as an example.


There is a close relation between parallelizability of spheres and division algebras. It is particularly interesting since it explains the reason behind the specific form \cite{Cartan1,Cartan2} of the solutions that was found by Cartan and Schouten. It also allows one to go one step forward and to extend the result of Cartan and Schouten to pseudo-Riemannian manifolds \cite{wolf1972I,wolf1972II}. This will lead to the result that there exist a parallelism on the pseudo-Riemannian manifolds which are of the form of the product of a flat manifold, irreducible symmetric spaces and the symmetric coset spaces $S^7=SO(8)/SO(7)$, $S^{3,4}=SO(4,4)/SO(3,4)$ and $SO(8,\mathbb{C})/SO(7,\mathbb{C})$. Here ``irreducible" means that the symmetric spaces are not locally product spaces. Equivalently, it means that the infinitesimal holonomy group acts irreducibly on the tangent space and therefore the holonomy group is given by the Berger classification of irreducible holonomy groups \cite{Berger,10.2307/1970273}.  

We are interested in the case of seven dimensions, which stands out also as a result of Adams theorem \cite{Adams} that relates $S^0$, $S^1$, $S^3$ and $S^7$ to normed division algebras. We will briefly discuss the three-dimensional case to set the stage for a more difficult case of seven dimensions. The generalization of Cartan-Schouten theorem
to pseudo-Riemannian manifold will
teach us that in three and seven dimensions (beside $S^3$ and $S^7$ already noted by Cartan and Schouten), respectively AdS$_3$ manifold and $S^{3,4}=SO(4,4)/SO(3,4)$ (or $H^{4,3}=SO(4,4)/SO(4,3)$) are also parallelizable.

Seven dimensions have been shown to play an important role in supergravity or $M$-theory, where one finds an effective theory of supergravity in four dimensions from compactification of the eleven-dimensional supergravity on a seven-manifold. If this seven-manifold is toroidal then the compactification will result in an ungauged theory in four-dimensional flat background \cite{Cremmer:1978ds,Cremmer:1979up}. The Freund-Rubin compactification \cite{Freund:1980xh} on $S^7$ will lead to a four-dimensional $SO(8)$ gauged supergravity in AdS$_4$ spacetime \cite{deWit:1982bul}. The Freund-Rubin compactification is the only flux compactification of supergravity which has the maximum number of Killing spinors or, in other words, it preserves all supercharges. There is another type of flux compactification where one turns on other fluxes and therefore changes the internal geometry, which breaks part of the supersymmetry of the resulting theory. For example, $G_2$ compactification of eleven-dimensional supergravity \cite{Atiyah:2001qf,Duff:2002rw} preserves one Killing spinor, which leads to at most $\mathcal{N}=1$ supergravity in four dimensions. The Englert solution \cite{Englert:1982vs} is a specific solution of $G_2$ compactification where supersymmetry is completely broken. This solution is obtained by compactifying on parallelizable $S^7$.

Since $S^7$ plays an important role in this scenario, it might seem interesting to see if any role is played by $S^{3,4}$ for example in the compactification. At first sight, this looks impossible since we cannot compactify a theory in $10+1$ dimensions on a manifold with more than one time direction. However, as we recalled above it turns out that there is a signature-changing duality in string theory \cite{Hull:1998ym,Hull:1998vg}. This duality is called timelike T-duality as a result of which one finds an orbit of equivalent theories under this duality. The orbit involves three theories formulated in the same dimension but with different signatures. The theories in reversed signatures also belong to this orbit.  

We will show that the compactification on parallelizable $S^{3,4}$ is possible if one works in the picture described by one of these dual theories, known as $M'$-theory \cite{Henneaux:2017afd}. This solution, similarly to the Englert solution, will result in a non-supersymmetric theory in AdS$_4$ background. This can be understood as a spontaneous compactification on $S^{3,4}$ with $G_{2,2}$ invariant fluxes. The compactification on parallelizable seven-manifolds for the other supergravity theories in eleven dimensions (the low energy limit of $M$- and $M^*$-theories) is not possible.

In Part \ref{part2} we focus on electric-magnetic duality of supergravity in four dimensions \cite{deWit:2001pz}. Electric-magnetic duality is particularly interesting since the gauge group of a gauged supergravity in its second-order formulation is a subgroup of electric-magnetic duality symmetry group and moreover the available gaugings depend on the duality frame in which the theory is formulated\cite{deWit:1981sst}. However, the electric-magnetic duality symmetry of a supergravity (or in general a scalar-vector coupled model) in the second-order formulation is not manifest off-shell since part of this symmetry acts non-locally (in space) on fields \cite{Deser:1976iy,Deser:1981fr,Bunster:2011aw,Gaillard:1981rj}. As a result, the gauge group is normally a subgroup of an electric subgroup of electric-magnetic duality which acts locally on fields. On the other hand, the first-order formulation is manifestly duality invariant. The main purpose of Part \ref{part2} is to have a
systematic discussion of the gaugings of a supergravity theory, or in other words to find the consistent deformations of an ungauged supergravity. We
will show that even though the first-order formulation is manifestly duality invariant but due to the rigidity of the first-order formulation \cite{Henneaux:2017kbx} it is not suited for the discussion of the deformations. Therefore, in order to discuss the consistent deformations of supergravity, it is inevitable to work in the second-order formulation.

The embedding tensor formalism introduced by \cite{deWit:2002vt,deWit:2005ub,deWit:2007kvg,Samtleben:2008pe,Trigiante:2016mnt}, is a procedure of gauging a supergravity theory by promoting  abelian vector fields to non-abelian ones, replacing the ordinary derivatives by covariant ones and adding extra interactions in order to restore the symmetry. In this formalism the embedding of the gauge group in the global symmetry group is described through constant parameters given by the so-called embedding tensor; hence the gauged Lagrangian depends on the embedding tensor which must satisfy a set of constraints. The embedding tensor formalism will provide a gauged Lagrangian for the supergravity theories where one can see that the possible deformations are those which require changing the gauge algebra. However, there are in general other types of deformations where even though the gauge symmetry is deformed but the gauge algebra remains unchanged. It is then an interesting question to explore the possibility of gauging the symmetries corresponding to these type of deformations \cite{Barnich:2000zw}. In order to do that we use the BV-BRST deformation formalism \cite{Batalin:1981jr,Batalin:1984jr,Barnich:1993vg,Henneaux:1997bm} which provides precise tools to answer this question in a systematic manner. 

We will first show that \cite{Henneaux:2017kbx} the space of deformations of embedding tensor formalism is isomorphic to the space of deformation of a scalar-vector Lagrangian (describing the bosonic action of supergravity without spin-2 field) in a chosen duality frame determined by the embedding tensor. The infinitesimal local deformations are controlled by the cohomology of the BRST differential at ghost number zero in the space of local functionals. The BRST cohomology at ghost number one is related to the possible anomalies in the theory. In fact, we show that the BRST cohomology at all ghost numbers is equivalent for the embedding tensor Lagrangian and the one of scalar-vector coupled model. Then, performing a thorough analysis of BV-BRST deformation \cite{Barnich:2017nty} of the scalar-vector Lagrangian, we will obtain all possible symmetry deformations controlled by the BRST cohomology at ghost number minus one. It includes the symmetries of both types; those which deform the gauge algebra and the types which do not. The consistent deformations then is implemented by a set of constraints (similarly to the embedding tensor formalism) and one finds that the only symmetries which satisfy these conditions are of the type already considered in the embedding tensor formalism. The other types of symmetries even though are part of the electric symmetry group but they cannot be gauged. We will demonstrate this explicitly in different examples including the one of $\mathcal{N}=4$ supergravity in four dimensions.

The thesis is organized as follows:

In Chapter \ref{ch:parallelizability} we provide definitions regarding the notion of parallelism and the relation between parallelizability and division algebras. In Chapter \ref{ch:Lorentz_flatness} we use the division algebras to show the parallelizability on specific seven-dimensional topologies rediscovering the results of Cartan-Schouten and Wolf. 

In Chapter \ref{ch:exotic}, we start by reviewing the effect of time-like T-duality and the structure of dual $M$-theories. Then, we show all possible solutions of compactification of 
$M$-, $M'$- and $M^*$-theories on seven-dimensional (pseudo-)sphere or (pseudo-)hyperbolic manifolds in both cases of Freund-Rubin compactification, described by Hull \cite{Hull:1998fh}, and the Englert compactification where the internal manifolds are the parallelizable ones of Chapter \ref{ch:Lorentz_flatness}. We will show that the reduction of the $6+5$- dimensional supergravity, low energy limit of $M'$-theory, on the seven- dimensional parallelizable pseudo-sphere $S^{3,4}$, similarly to the reduction of standard supergravity on the parallelizable $7$-sphere, will result in well defined theories in $AdS_4$ background. The difference with the reduction on a `round' pseudo-sphere is that in the latter the dimensionally reduced four- dimensional theories are maximally supersymmetric with $SO(8)$ gauge symmetry while the former solutions are non-supersymmetric with $Spin(3,4)$ gauge symmetry. The presence of fluxes due to the flattening torsions on a parallelizable manifold breaks the supersymmetry.

In Chapter \ref{ch:EM_duality}, we review the electric-magnetic duality and explain the enhancement of the symmetry in the presence of scalars. In Chapter \ref{ch:Isom-Emb-VS-models}, we give a quick review of embedding tensor formalism and then we describe how the scalar-vector Lagrangian describing the maximal supergravity in four dimensions is captured as the undeformed limit of the embedding tensor gauged Lagrangian in a frame parameterized by the embedding tensor coefficients.

In Chapter \ref{ch:BV_def}, after explaining the BV-BRST deformation formalism we go on to prove the isomorphism between the space of local deformations of embedding tensor Lagrangian and the described scalar-vector coupled Lagrangian and therefore complete the analysis of previous chapter. In Section \ref{sec:BV-def-Embed}, we explicitly show the extra fields present in the embedding tensor formalism will disappear since whether they are auxiliary fields (like two-forms) or of pure gauge type. 

In Chapter \ref{ch:BRST-cohom-SV-Lagr}, we present a detailed computation of BRST cohomology in different ghost degrees analyzing the global symmetries, infinitesimal deformations and anomalies described respectively by  $H^{-1}(s\vert d)$, $H^{0}(s\vert d)$ and $H^{1}(s\vert d)$. We provide then various examples to demonstrate the solutions to the obstruction equations for the consistent deformations and to show which symmetries can be gauged. We specifically show this for $\mathcal{N}=4$ supergravity in four dimensions where one can in fact solve the constraints in practice. In principle, the result obtained in this chapter can be used for the gauging of supergravitis with $\mathcal{N}>2$ in four dimensions.

Finally, we will give a summary of all results in Chapter \ref{ch:conclusion}. The future perspective of the works along the line of thoughts of this thesis is also discussed.

Four appendices were included at the end. Appendix \ref{app-notation} contains the list of all notations used in the thesis. Appendix \ref{app:quaternion_octonion} provides a quick review of Cayley-Dickson construction of division algebras and definitions of the algebra of octonions and split octonions relevant to Chapter \ref{ch:parallelizability}. Appendix \ref{App:pseudo} includes the definitions of (pseudo-)sphere and (pseudo-)hyperbolic spaces and a presentation of two infinite families of parallelism (similarly to the ones of Cartan and Schouten) of the pseudo-sphere $S^{3,4}$. We also explained the relation between different groups $G_{2,2}$, $G^*_{2,2}$, $Spin(3,4)$ and $Spin^+(3,4)$ which are relevant in describing the parallelized $S^{3,4}$.
The final appendix, Appendix \ref{app:derivation} contains a proof of a critical relation being used in Chapter \ref{ch:BRST-cohom-SV-Lagr}.


\cleardoublepage

\part{Parallelizable manifolds}\label{part1}

\cleardoublepage





\chapter{Parallelizability}\label{ch:parallelizability}

\section{Absolute Parallelism}

In this chapter we provide some definitions and theorems necessary to understand the mathematics behind the next chapter.

We first start with explaining the meaning of absolute parallelism where we use the definition given by Wolf \cite{wolf1972I}.
An \textbf{absolute parallelism} on a differentiable manifold $\mathcal{M}$, is an isomorphism map $\phi$ between tangent spaces at any two points $x,y\in \mathcal{M}$ which does not depend on additional choices.
More precisely, 
\be 
\phi=\{\phi_{yx}\},\quad \phi_{yx}: T_x\mathcal{M} \rightarrow T_y\mathcal{M} \quad \forall x,y \in \mathcal{M}
\ee
which satisfies the consistency condition
\be
\phi_{zy}\phi_{yx}=\phi_{zx}, \qquad \phi_{xx}=id.\quad \forall x,y,z \in \mathcal{M}
\ee
and the regularity condition that guarantees the smoothness of the isomorphism $\phi$. Given $\phi$, one says that the tangent vector $\xi_x \in T_x \mathcal{M}$ and the tangent vector  $\xi_y \in T_y \mathcal{M}$ are parallel iff $\xi_{y} = \phi_{yx} \xi_x$. 

On a simply connected differentiable manifold $\mathcal{M}$ one can show that there is a one-to-one correspondence between an absolute parallelism on $\mathcal{M}$, smooth trivialization of the frame bundle on $\mathcal{M}$ and the existence of flat connections (zero curvature) on the frame bundle, see \cite{wolf1972I} for the proof.

Cartan and Schouten \cite{Cartan1,Cartan2} described curvature-free connections on Lie groups and as a result exhibited parallelisms on group manifolds. 
They then
generalized it to provide a description of parallelism (in terms of curvature-free metric-compatible connections with torsion) on Riemannian manifolds which
are products of flat manifolds, compact simple group manifolds and 7-spheres.

Later, Adams \cite{Adams} using a metric-free topological description of parallelism (i.e. existence of a global section on the frame bundle) showed that the only parallelizable spheres are $S^0$, $S^1$, $S^3$ and $S^7$.
The $2$-sphere is the simplest example
where the parallelism fails to hold and one can still use the geometrical tools to see this. It is known that one cannot define a coordinate patch that
covers the $2$-sphere thoroughly. In fact, using the sterographic map it
can be seen that one needs at least two different patches in order to cover the whole $2$-sphere. 

In the following, we shall focus on the relation between parallelizable spheres and normed division algebras. In the next section we define the normed division algebras and explain what will be necessary to use later on. We state a number of theorems without including the proofs since it would be a digression from our presentation.

\section{Division Algebras}

Consider an algebra $A$ over the field $F$, that is $A$ is a vector space  equipped with a bilinear map $l: A \times A \rightarrow A$ and a nonzero unit element $e \in A$. 
Consider any two arbitrary elements $a$ and $b$ of the algebra $A$. In addition, we assume that $a$ is a non-zero element of $A$. If linear equations
$ax = b$ and $ya = b$ are uniquely solved for $x, y \in A$, then the algebra $A$ is called a \textbf{division algebra} over the field $F$.

In other words, $A$ is a division algebra if given $a, b \in A$ with $ab = 0$, then either $a = 0$ or $b = 0$.
Equivalently, $A$ is a division algebra if the operations of left and right multiplication by any nonzero element are invertible.

Whenever $F$ is real, then $A$ is a real division algebra. For example we will see later on in this section that 
real quaternion and octonion algebras are real division algebras.

A \textbf{normed} division algebra is a division algebra which is also a Banach algebra, i.e. it is a normed vector space with a multiplicative property $||ab||=||a||\,||b||$.

\subsection*{Associativity and Alternative Algebra}
Given an algebra $A$, it need not be necessarily associative. One can define a trilinear map, called associator,
\be
(x,y,z)= x(yz) - (xy)z \qquad x,y,z\in A,
\ee
which measures the non-associativity of the algebra. For an associative algebra, the associator vanishes. 
If the associator of an algebra is totally antisymmetric, i.e.
$(x,y,z)=-(x,z,y)=-(y,x,z)=-(z,y,x)$, then the algebra is called an  \textbf{alternative} algebra. An alternative algebra is equivalently defined \cite{schafer2017,schafer1966} as
\begin{align}
x^2y&=x(xy),\\
yx^2&=(yx)x,
\end{align}
for any two elements $x,y$ of the algebra.
One very important feature of any alternative algebra is that even though the algebra itself is non-associative for elements $x,y,z\in A$,
any subalgebra of an alternative algebra generated from any two elements $x,y\in A$ is associative \cite{schafer2017,schafer1966}. 
For any alternative algebra, the following identities always hold for any three elements $x,y,z\in A$:
\begin{align}
    (xyx)z&=x(y(xz)),\label{Moufang1}\\
    z(xyx)&=((zx)y)x,\label{Moufang2}\\
    (xy)(zx)&=x(yz)x.\label{Moufang3}
\end{align}
These relations are known as \textbf{Moufang identities} and those algebras that satisfy these identities are known as \textbf{Malcev algebras} \cite{schafer2017,schafer1966}.

\subsection*{Dimension of Normed Division Algebras}

\begin{proposition}
Let $A$ be a non-associative algebra and $x$ a fixed element of $A$. Then the left multiplication
$L_x$ and the right multiplication $R_x$ are linear operators in the vector space $A$, i.e., $L_x,R_x : A\rightarrow A$,
defined by
\be
L_{x}y = xy, \quad R_{x}y =yx,
\ee
for all $y \in A$.
\end{proposition}

\begin{lemma}\label{lemma-iden-left-mult}
For any $x\in A$, we have $L_x L_{\bar{x}} =(x,x) I$ where $I$ is the identity multiplication operator defined by $Ix=x$, $\bar{x}$ is the complex conjugate of $x$ defined as $\bar{x}=2 (x,e)e - x$ and $(.\,,\,.)$ is the scalar product.
\end{lemma}
\begin{proof}
For the proof of the lemma see Section 3.1 and Proposition 1 in Chapter one of \cite{okubo1995}.
\end{proof}

\begin{theorem}{(Bott-Milnor, Kervaire \cite{Hurwitz, Bott-Milnor, Kervaire})}\label{Theorem:Horwitz}
The dimension of a normed division algebra is restricted to $1, 2, 4,$ or $8$.
\end{theorem}
\begin{proof}
For the proof of theorem see Chapter three of \cite{okubo1995}. For completeness, a sketch of the proof is given in Appendix \ref{app:quaternion_octonion}.
\end{proof}

This theorem is particularly interesting since it relates the dimension of irreducible representations of Clifford algebra on $N-1$ dimensional vector space to $N$ the dimension of a normed division algebra, see Appendix \ref{app:quaternion_octonion}.  

\subsection*{Quaternions and Octonions}

The alternative property of an algebra is important. In fact, the only alternative division
algebras are $\mathbb{R}, \mathbb{C}, \mathbb{H}$ and $\mathbb{O}$ \cite{schafer2017,schafer1966, Zorn1930}. This together with 
Theorem \ref{Theorem:Horwitz} will result in the following theorem:
\begin{theorem}\label{th:normed_div_alg}
$\mathbb{R}, \mathbb{C}, \mathbb{H}$ and $\mathbb{O}$ are the only normed divison algebras.
\end{theorem}
\begin{proof}
The proof of this theorem requires identities for normed divison algebras which we do not repeat here. See Section 3.1 of \cite{okubo1995} for the complete proof. 
\end{proof}

$\mathbb{R}$ and $\mathbb{C}$ are real and complex vector spaces respectively, $\mathbb{H}$ is 
the quaternion algebra and $\mathbb{O}$ is the octonion algebra. These division algebras have peculiar properties; $\mathbb{R}$ is an ordered,
commutative and associative algebra, $\mathbb{C}$ is a commutative and associative algebra, $\mathbb{H}$ is an associative
but non-commutative algebra and $\mathbb{O}$ is a non-commutative, non-associative
algebra. However, they have one common feature that makes them distinct and that all these algebras are alternative. This is indeed the reason that the
sequence of normed division algebras stops at octonions \cite{Baez:2002}. 

In Appendix \ref{app:quaternion_octonion}, we give a short review of Cayley-Dickson construction of division algebras, in particular we provide the example of the construction of octonions. 

As the final remark, we mention that it has been shown by Adams that for an $n$ dimensional 
normed division algebra, $S^{n-1}$ is parallelizable \cite{Adams}, for an illustrative proof see \cite{Baez:2002}. This is an important statement that we will use repeatedly in the next chapter.



\chapter{Lorentz Flatness}\label{ch:Lorentz_flatness}

In the previous chapter, we pointed out that parallelism is in one-to-one correspondence with the existence of a flat connection. This can be interpreted as path-independence of parallel transport. It means that transporting a vector from the point $x\in \mathcal{M}$ to another point $z\in \mathcal{M}$ always leads to the same vector at $z$ no matter what path is chosen. When there is no torsion then the condition that the parallel transport is path-independent is satisfied if the Riemann curvature -- defined in \eqref{calR} below as the curvature of the Levi-Civita connection --  vanishes. Therefore, in the absence of torsion a manifold is parallelizable if its Riemann curvature vanishes. Flat (simply connected) space is the simplest example of a parallelizable manifold where the curvature vanishes and the Levi-Civita connection is the flat affine connection that provides a path-independent parallel transport.

However, in the presence of torsion, the parallel transport is not defined by the Levi-Civita connection and the vanishing of its curvature is not necessary for parallelism to hold.  Torsion introduces new possibilities.  This can be most easily discussed in local Lorentz frames where the metric takes the Minkowskian form. Parallel transport is captured by the differential equation
\begin{equation} \label{Du=0}
 Du^a \equiv dx^\mu[\partial_\mu u^a + \omega\indices{^a_{b\mu}}(x) u^b]=0 \, ,
\end{equation}
for a tangent vector $u^a(x)$ at $x$. This equation is a consequence of the fact that $u^a$ transforms as a vector under the Lorentz group, i.e. $\delta u^a(x) = dx^\mu \omega\indices{^a_{b\mu}}(x) u^b$.



This can be seen as $u^a(x+dx)$ obtained by parallel transport from $u^a(x)$ where $Du^a$ is the covariant derivative of $u^a$.
The $1$-form $\omega\indices{^a_b} = \omega\indices{^a_{b\mu}} dx^\mu$ is called the Lorentz connection because we consider local orthonormal frames, and the 
condition that the covariant derivative
of the Lorentz metric vanishes identically ($D\eta^{ab} = 0$) implies that $\omega_{ab}=- \omega_{ba}$. If the expression \eqref{Du=0} is to have
the same meaning in any Lorentz frame, $Du^a$ should transform under the Lorentz group in the same representation as $u^a$ (see, e.g., \cite{schutz1980geometrical,gockeler1989differential}).

The integrability condition for \eqref{Du=0} or in other words the path-independence of the parallel transport is governed by
\begin{equation} \label{DDu=0}
 DDu^a \equiv R\indices{^a_b} u^b=0 \, ,
 \end{equation}
where the curvature two-form $R\indices{^a_b}$ is defined as\footnote{From now on wedge products ($\wedge$) of exterior forms will be implicitly understood.}
\begin{equation}\label{R}
R\indices{^a_b} =d\omega\indices{^a_b} +\omega\indices{^a_c}\omega\indices{^c_b}\, .
\end{equation}
If the condition \eqref{DDu=0} is to be fulfilled for an arbitrary $u^a$ everywhere in an open region $\mathcal{U}$, the curvature $R\indices{^a_b}$ must vanish identically in that region.
Therefore a parallelizable manifold requires its Lorentz curvature $R\indices{^a_b}$ to vanish. It is equivalent to 
the existence of a flat Lorentzian connection on the manifold.

\section{Flatness of a (Pseudo-)Riemannian Manifold}\label{sec:state-problem}

For any $p,q \geq 0$, consider the manifold $R^{(p,q)}$ which is equipped with a symmetric non-degenerate bilinear form $b_{(p,q)}$ defined as
\be
b_{p,q} = + \sum_{i=1}^{p} dx_i^2 - \sum_{j=1}^{q} dx_j^2 \, .
\ee
This bilinear form is invariant under the adjoint action of the $SO(p,q)$ group. Now consider an $n$ dimensional ($p+q-1=n$)
manifold $\mathcal{M}$ defined by $SO(p,q)$ invariant equation
\be
\sum_{i=1}^p x_i^2 -\sum_{j=1}^q x_j^2 =\pm 1.
\ee
The bilinear form on $R^{(p,q)}$ induces a generalized Lorentz invariant
metric $\eta_{ab}$ on the tangent space $T_x\mathcal{M}$ at any point $x\in\mathcal{M}$. The induced metric is diagonal $\eta_{ab}= diag (+,...,+;-,...,-)$ with the signature 
\begin{itemize}
\item $(p-1,q)$  and $\mathcal{M} = \f{SO(p,q)}{SO(p-1,q)}$ with a constant positive curvature,
\item $(p,q-1)$  and $\mathcal{M} = \f{SO(p,q)}{SO(p,q-1)}$ with a constant negative curvature.
\end{itemize}
$\mathcal{M}$ is also equipped with a generalized Lorentz connection one-form $\omega^a{}_b$ (with $\omega_{ab}=-\omega_{ba}$) and
a local orthonormal basis $e^a=e^a_\mu dx^\mu$ that provides the metric, $g_{\mu \nu}=\eta_{ab}e^a_\mu e^b_\nu$. The indices $a,b,...$ are tangent space indices and run over $1,2,...,n$, while $\mu,\nu,...$ are coordinate indices. 

The geometry of a curved (pseudo-)Riemannian manifold with metric $g_{\mu \nu}$ is characterized by the Riemann curvature two-form,
\begin{equation}\label{calR}
\mathcal{R}\indices{^\alpha_\beta} =d \Gamma\indices{^\alpha_\beta} + \Gamma\indices{^\alpha_\gamma} \wedge \Gamma\indices{^\gamma_\beta}\, ,
\end{equation}
where $\Gamma\indices{^\alpha_\beta} = \Gamma\indices{^\alpha_{\beta \mu}} dx^\mu$ is the Levi-Civita connection one-form.
This Riemannian curvature is similar to --but not necessarily equivalent to-- the curvature (\ref{R}) defined for the Lorentz
connection $\omega\indices{^a_b}$. The difference stems from the fact that $\Gamma\indices{^\alpha_\beta}$ is defined by the metric only while
$\omega\indices{^a_b}$ may involve a non-trivial torsion.
The Levi-Civita connection is determined by the torsion-free part of the Lorentz connection, denoted by  $\bar{\omega}\indices{^a_b}$, and is defined by the condition
\begin{equation}\label{Torsionless-Cartan}
de^a + \bar{\omega}^a{}_b e^b \equiv 0 \Longrightarrow \bar{\omega}^a{}_b = e^a_\nu [\delta ^\nu_\mu d + \Gamma\indices{^\nu_\mu}] E^\mu_b\, ,
\end{equation}
where $E^\mu_a$ is the inverse vielbein. The two connections $\omega^a{}_b$ and $\Gamma\indices{^\alpha_\beta}$ define different notions of parallelism that give rise to different curvature forms, $R^a{}_b$ and $\mathcal{R}\indices{^\alpha_\beta}$, respectively. 

In Minkowski space, there is a torsion-free  Lorentz connection $\bar{\omega}\indices{^a_b}$. In this case,
the connection is trivial, in the sense that the holonomy group of the
connection is trivial in an obvious manner since the tangent bundle is naturally trivialized on Minkowski space. It is equivalent to say that the Riemann curvature vanishes.
The question we examine here is what (pseudo-)Riemannian manifolds have flat connection. As we have already seen the Adams' theorem \cite{Adams} 
states that for Euclidean signature, the spheres $S^1$, $S^3$ and $S^7$ are parallelizable, that is, they admit 
globally defined orthonormal frames. The connection for the corresponding $SO(n)$ rotation groups have trivial holonomy group and the corresponding
curvatures vanish. We would like to explore the extension of this result for other metric manifolds of arbitrary signature.

One way to address this question is to look for a necessary condition for flatness, which is equivalent to the statement 
\begin{equation} \label{Ru}
R^a{}_b u^b_{(\mu)}=0, \;\;\;\;  \mu = 1, \cdots, n
\end{equation}
where $\{u^b_{(\mu)}\}$ is a family of linearly independent vectors. In a smooth manifold, a weaker but necessary condition is
\begin{equation} \label{Re}
R^a{}_b\wedge e^b=0 ,
\end{equation}
which would be the case if the torsion two-form is covariantly constant,
\begin{equation}\label{DT}
T^a \equiv D e^a = d e^a + \omega^a{}_b e^b  \Longrightarrow DT^a=R^a{}_b e^b=0.
\end{equation}
Hence, in a Lorentz-flat manifold, the torsion may not be zero but it must be covariantly constant, $DT^a=0$. The torsion two-form can be expressed in a basis of local orthonormal frames as
\begin{equation}\label{T2}
T^a = \tau f\indices{^a_{bc}} e^b e^c + \rho e^a\, ,
\end{equation}
where $\tau$ is a scalar zero-form, $f^a{}_{bc}$ is a zero-form antisymmetric in the lower indices, and $\rho$ is a scalar one-form. Then, the necessary condition for Lorentz flatness (\ref{DT}) reads
\be\label{Gen-sol}
(D f\indices{^a_{bc}}) \tau e^b e^c + 2\tau f^a{}_{bc}f^b{}_{df} e^c e^d e^f + (d\tau + \tau \rho ) f\indices{^a_{bc}} e^b e^c  + d \rho\, e^a =0\, ,
\ee
A sufficient condition for solving this equation is to assume for the functions multiplying $e^a$, $e^a e^b$, and $e^a e^b e^c$ to be independent and therefore
\begin{align} \label{Suff-cond-a}
d\rho=0, \\ \label{Suff-cond-b}
D f^a{}_{bc} +(d\log[\tau] + \rho ) f^a{}_{bc} =0, \\ \label{Suff-cond-c}
f^a{}_{b[c}f^b{}_{df]}=0.
\end{align}
By Poincar\'e's lemma, (\ref{Suff-cond-a}) means that in an open region $\rho =d\lambda$ and (\ref{Suff-cond-b}) becomes $D\tilde{f}^a{}_{bc}=0$, where $\tilde{f}^a{}_{bc} \equiv \tau e^{\lambda}f^a{}_{bc}$. In other words, the coefficients $\{f^a{}_{bc}\}$ are, up to a conformal factor, covariantly constant. Moreover, since the connection is flat it can always be trivialized in an open region by a local Lorentz transformation and without loss of generality one can take $f^a{}_{bc}=\Omega(x) \bar{f}^a{}_{bc}$, where $d\bar{f}^a{}_{bc}=0$.  As we shall see in the next section, it is very easy to achieve these conditions in three dimensions, while the situation in seven dimensions is less straightforward but as it is discussed in Section \eqref{sec:Lorentz-flat-7d} non-trivial solutions still exist.

\subsection{Three Dimensional Example}\label{sec:3d-example}

As explained before, conditions \eqref{Suff-cond-a} and \eqref{Suff-cond-b} imply that $f^a{}_{bc}=\Omega(x) \tilde{f}^a{}_{bc}$, where $D\tilde{f}^a{}_{bc}=0$, which is automatically satisfied if $\tilde{f}^a{}_{bc}$ are the structure constants of the Lorentz algebra. This means that $\tilde{f}$ is an invariant tensor of the Lorentz group, in which case (\ref{Suff-cond-c}) is the Jacobi identity. For $n=3$ this means that modulo local rescalings, $f^a{}_{bc} \equiv \epsilon^a{}_{bc}$ where $\epsilon_{abc}$ is the Levi-Civita tensor defined with the convention\footnote{The Lorentz indices of $a,b,c,...$ are raised and lowered with $\eta^{ab}$ and $\eta_{ab}$.} $\epsilon_{012}=+1$. 

Consequently, in three-dimensions the torsion two-form,
\begin{equation}\label{T'}
T^a = \tau \epsilon\indices{^a_{bc}} e^b e^c - d[\log \tau]e^a\, ,
\end{equation}
is covariantly constant by construction. By rescaling the frame as $\tilde{e}^a \equiv \lambda \tau e^a$ with $\lambda$ an arbitrary constant, the torsion can be written (after dropping the tildes) as 
\begin{equation}\label{T= epsilon ee}
T^a= \lambda \epsilon\indices{^a_{bc}} e^b e^c \,.
 \end{equation}
The equation $D e^a = \lambda \epsilon\indices{^a_{bc}} e^b e^c$ is a condition for the metric structure of the manifold. But, what does this imply for the Riemannian tensor?

In order to address this question, let us separate the torsion-free part of the Lorentz connection $\bar{\omega}\indices{^a_b}$ and the \textit{contorsion} $\kappa\indices{^a_b}$, 
\begin{equation}\label{w=w+k}
 \omega\indices{^a_b} =  \bar{\omega}\indices{^a_b}  + \kappa\indices{^a_b} \, ,
\end{equation} 
where $ \bar{\omega}\indices{^a_b}$ is defined by the torsion-free condition, 
\begin{equation} \label{de+we=0}
\bar{D} e^a = de^a +  \bar{\omega}\indices{^a_b} e^b \equiv 0 \, ,
\end{equation}
while the contorsion contains the information about the torsion, $T^a = \kappa\indices{^a_b} e^b$. Equation \eqref{de+we=0} determines $\bar{\omega}$ as a function of the metric structure of the manifold, and corresponds to the projection onto the tangent space basis of the Levi-Civita connection. The Lorentzian curvature splits as
\begin{equation} \label{R=R+Dk+kk}
R\indices{^a_b} = \bar{R}\indices{^a_b} + \bar{D}\kappa\indices{^a_b} + \kappa\indices{^a_c} \kappa\indices{^c_b} \, ,
\end{equation}
where $\bar{R}\indices{^a_b} = d\bar{\omega}\indices{^a_b} + \bar{\omega}\indices{^a_c} \bar{\omega}\indices{^c_b}$ is the projection onto the tangent space of the Riemann curvature tensor.

In three dimensions, \eqref{T= epsilon ee} means that $\kappa\indices{^a_b} = -\lambda \epsilon\indices{^a_{bc}} e^c$, which in turn implies that $\bar{D}\kappa\indices{^a_b}=0$ and
\be
\kappa\indices{^a_c} \kappa\indices{^c_b} = \lambda^2\, det(\eta)\, \delta^a_f\, \eta_{bd}\, e^d e^f,
\ee
where we have used 
\be
\epsilon\indices{^a_{cd}}\,\epsilon\indices{^c_{bf}} = det(\eta)\, \left[\delta^a_f\, \eta_{bd} - \delta^a_b\, \eta_{df}\right].
\ee

Then, \eqref{R=R+Dk+kk} reads
\begin{equation}\label{R=R+l ee}
R\indices{^a_b} = \bar{R}\indices{^a_b} - \lambda^2\, det(\eta)\, \eta_{bd}\, e^a e^d \, .
\end{equation}
Therefore, the Riemann curvature of a three-dimensional Lorentz-flat manifold is
\be
\bar{R}\indices{^a_b} = \lambda^2\, det(\eta)\, \eta_{bd}\, e^a e^d \, .
\ee
This relation implies that for $\lambda \neq 0$ there are two distinct possibilities\footnote{There are two other options that one obtains from the solutions shown here just by a metric transformation which flips the sign of the metric and that too changes the sign of the curvature.} for a three-dimensional manifold admitting a Lorentz-flat connection:
\begin{enumerate}
\item The metric $\eta_{ab}$ is Riemannian, $ds^2=\sum_{i=0}^{2} dx_i^2$. Then $det(\eta)=+1$ and the three-dimensional manifold is a constant positive curvature manifold, i.e. $3$-sphere (of finite radius) $S^3$, in agreement with \cite{Adams}.

\item The metric $\eta_{ab}$ is pseudo-Riemannian, $ds^2=-dx_0^2 + \sum_{i=1}^{2} dx_i^2$\, then $det(\eta)=-1$ and the three-dimensional manifold has constant negative curvature, i.e. AdS$_3$. This is the case considered in \cite{Alvarez:2014uda}.
\end{enumerate}

There is also the degenerate solution $\lambda = 0$: flat Riemannian manifolds $R^3$
and pseudo-Riemannian flat manifolds $R^{2,1}$ and $R^{1,2}$.
 
\subsection{A General Discussion}\label{sec:General-setting}

As we discussed before, an $n$ dimensional manifold $\mathcal{M}$ is called parallelizable (i.e. it admits framing) if one can find $n$ globally-defined linearly independent tangent vector fields on the manifold. This is always possible if there exists a flat connection in the frame bundle of the manifold. One can make sense of it as integrability of $Du^a =0$ where $u^a$ is a tangent vector defined at each point on the manifold.

Given a manifold with a connection in the fiber bundle $\omega^a{}_b$ and a local frame $e^a$ at each point (vielbein)\footnote{In fact, $e^a$ is a 
solder form and defines a canonical isomorphism between the tangent space of $\mathcal{M}$ and the fibre $E$.}, the equations  \eqref{R} and \eqref{DT} can be 
viewed as structure equations for the connection \cite{chern1966, Milnor1958, kobayashi1965}
\be\label{1ststructureeq}
de^a=-\omega^a{}_b e^b+T^a \, , \quad d\omega^a{}_b =-\omega^a{}_c  \omega^c{}_b + R^a{}_b \,.
\ee
Assuming there is a torsion-free connection, $\bar{\omega}^{a}_{~b}$, then $de^a=-\bar{\omega}^a_{~b} e^b$. From here, we can define the torsion $2$-form $T^a$ as
\be\label{Torsion}
T^a=f^a_{~bc}e^b e^c,
\ee
where $f^a_{~bc}e^c$ in fact plays the role of the contorsion $1$-form. Here the coefficient $f^a_{~bc}$ is a zero-form Lorentz tensor whose main function is to map two vectors into a third one. In other words, it defines an antisymmetric bilinear map of the tangent space $T{\mathcal{M}}$ into itself,
\be \label{crossproduct}
T{\mathcal{M}} \otimes T{\mathcal{M}} \rightarrow T{\mathcal{M}}\, .
\ee
In (\ref{T= epsilon ee}) $f^a_{~bc}$ was identified (up to a constant) with
$\epsilon^a_{~bc}$, which defines the vector cross product in three dimensions.
The cross product is defined as a bilinear map which maps two elements of a
vector space $V$ to another element of $V$ ($V \otimes V \rightarrow V$). In three dimensions, $f\indices{^a_{bc}}$ is the structure constants of $\mathfrak{su}(2)$ Lie algebra and since $f\indices{^a_{bc}}\equiv \epsilon\indices{^a_{bc}}$ then $f\indices{^a_{bc}}$ are invariant under $SO(3)$ Lorentz transformations.
Now recalling the definition of the quaternion algebra, see Appendix \eqref{app:quaternion_octonion}, if we only consider the imaginary elements of quaternions, i.e. $i,j,k$ with $i^2=j^2=k^2=-1$, they form an $\mathfrak{su}(2)$ Lie algebra. Then one can define an isomorphism between the three dimensional vector space $V$ and the imaginary quaternion vector space $\textrm{Im}\,\mathbb{H}$. 

A feature which is of great importance is the existence of the cross product in seven dimensions. In the same way that the three dimensional cross product is related to quaternions ($V \cong \textrm{Im}\,\mathbb{H}$), the seven dimensional one is related to 
octonions ($V \cong \textrm{Im}\,\mathbb{O}$). One defines this cross product as
\be\label{cross-product}
v_a \times v_b =f^c_{~ab} v_c,       
\ee
where $v_a \in \textrm{Im}\, \mathbb{O}$ and $f_{abc}$ is a totally antisymmetric tensor, with\footnote{Note that this is only one choice of octonion multiplication table and it is not unique.}
\begin{equation} \label{f}
f_{abc}=+1\,, \mbox{for    } abc=123, 145, 176, 246, 257, 347, 365.
\end{equation}
However, unlike the situation in three dimensions that the structure constant $f\indices{^a_{bc}}$ was invariant under $SO(3)$ transformations, $f\indices{^a_{bc}}$ is not invariant of $SO(7)$ transformations. The seven dimensional cross product defines a Moufang-Lie algebra (Malcev algebra) with $f\indices{^c_{ab}}$ as corresponding structure constants \cite{Nagy:1993SophusLie}

The exceptional Lie group $G_2$ is the automorphism
group of the octonions \cite{humphreys2012introduction}. Its Lie algebra $\mathfrak{g}_2$ is $\mathfrak{der}(\mathbb{O})$, the derivations of the octonions \cite{schafer1966}.
This can be considered as the definition of $G_2$ and its Lie algebra. The group $G_2$ is the isotropy group of $Spin(7)$\footnote{$Spin(7)$ is the double cover of $SO(7)$ and can be built with the use of spinor representations of $\mathfrak{so}(7)$.} which fixes the identity and since it is the automorphism group of octonions, it thus preserves the space orthogonal to the identity. This seven dimensional orthogonal space is the space of imaginary octonions $\textrm{Im}\,\mathbb{O}$ and there is an inclusion $G_2 \xhookrightarrow{} SO(\textrm{Im}\,\mathbb{O})$. Now, consider the adjoint representation of the algebra of imaginary octonions defined as $\textrm{ad}_{\textrm{Im}\,\mathbb{O}}=L_{\textrm{Im}\,\mathbb{O}}-R_{\textrm{Im}\,\mathbb{O}}$  where $L_{\textrm{Im}\,\mathbb{O}}$, $R_{\textrm{Im}\,\mathbb{O}}$ are spaces of linear transformations
of $\mathbb{O}$ given by left and right multiplications by imaginary octonions. At the Lie algebra level, $\mathfrak{so}(\textrm{Im}\,\mathbb{O})$ can then be written as
\be
\mathfrak{so}(\textrm{Im}\,\mathbb{O}) = \mathfrak{der}(\mathbb{O})\oplus \textrm{ad}_{\textrm{Im}\,\mathbb{O}},
\ee
where $\mathfrak{g}_2 \cong \mathfrak{der}(\mathbb{O})$ and $\textrm{ad}_{\textrm{Im}\,\mathbb{O}}$ is the annihilator of the identity element of $\mathbb{O}$. 
Therefore, $G_2$ has a 7-dimensional representation $\textrm{Im}\,\mathbb{O}$.

There is another definition for the group $G_2$. The exceptional group $G_2$ which was introduced as the automorphism group of octonions can also be thought of as a subgroup of $SO(7)$ which preserves a non-degenerate $3$-form \cite{Bryant1987, Agricola2008}
\be
\varphi= dx^{123}+dx^{145}+dx^{176}+dx^{246}+dx^{257}+dx^{347}+dx^{365},
\ee
where $dx^{abc}= dx^a \wedge dx^b \wedge dx^c$. Also, $G_2$ is the group of all real linear transformations of $\textrm{Im}\,\mathbb{O}$ preserving the cross product.

\section{Flat Connections in Seven Dimensions}\label{sec:Lorentz-flat-7d}

\subsection{Euclidean Signature} 

Let's consider a seven-dimensional simply-connected symmetric Riemannian manifold with the metric $g$ and the connection $\omega$. By imposing the flatness of $\omega$, the corresponding curvature \eqref{R} vanishes. Then, it has been shown by Cartan-Schouten \cite{Cartan1,Cartan2} and Wolf \cite{wolf1972I,wolf1972II} that such a manifold is the $7$-sphere given by the symmetric coset space $SO(8)/SO(7)$. The metric $g$ is invariant under $SO(8)$ isometry group of $S^7$. 

The group $SO(8)$ acts transitively on unit octonions $\mathbb{O}$ ($\mathbb{R}^8\simeq \mathbb{O}$) with the isotropy group $SO(7)=\{h \in SO(8) : he_0 = e_0\}$ where $e_0=1$ is the identity element of $\mathbb{O}$.  
We can also write $T^{AB}$ generators of $\mathfrak{so}(8)$ in terms of $J^{ab}$ generators of the subalgebra $\mathfrak{so}(7)$ and $J^{a}$ the remaining seven generators of $\mathfrak{so}(8)$,
\be
 T^{AB} = \{J^{ab}, J^{a8}=J^a\}, \quad A,B=1,...,8\quad \textrm{and}\quad a,b=1,...,7.
\ee

In the following we show that one can define another parallelism with respect to the metric $g$ with the connection $\bar{\omega}$ in the presence of torsion. It turns out that the metric $\tilde{g}$ is $SO(8)$ invariant metric and the torsion, given by the structure constants of imaginary octonions, breaks the isometry group down to $Spin(7)$.  

$Spin(7)$ acts on unit spinors in seven dimensions with $G_2$ a subgroup of which fixing a unit spinor.

We can define the flat $SO(8)$ connection as $\omega\indices{^A_B} = \theta T\indices{^A_B}$, with $\theta$ a closed $1$-form, (or more generally $\omega\indices{^A_B} = \theta\indices{^A_C} T\indices{^C_B}$ for any $1$-form $\theta\indices{^A_C}$). Given the decomposition of $T^{AB}$ to $J^{ab}$ and $J^a$, it is implied that the $Spin(7)$ connection $\bar{\omega}\indices{^a_b}=\alpha\indices{^a_c} J\indices{^c_b}$ (with $\alpha^a_c$ a $1$-form) cannot be a flat connection. However, one can construct a flat connection by adding a piece to $\bar{\omega}$, roughly speaking, in order to compensate for those terms in $\omega$ which are not in $\bar{\omega}$.  

Precisely, it means the flat connection $\omega\indices{^a_b}$ is written as
\be
\omega\indices{^a_b}= \alpha\indices{^a_c} J\indices{^c_b} + \beta^c (J_b)\indices{^a_c},
\ee
where $\alpha^a_c$ and $\beta^c$ are $1$-forms, defined as $\alpha^a_c=\alpha\indices{^a_{c\mu}} dx^\mu$ and $\beta^c=\beta e^c$. We can choose a basis where $(J_b)\indices{^a_{c}}=f\indices{^a_{bc}}$ and we get
\be\label{eq:flat_con_spin7_omega}
\omega\indices{^a_b}= \bar{\omega}\indices{^a_b} + \beta f\indices{^a_{bc}} e^c.
\ee
Now, the connection $\omega$ is a flat connection and the corresponding curvature vanishes.

As we discussed, the coefficients $f_{abc}$ define an invariant tensor under $G_2 \subset SO(7)$, which means that for any $G_2$ group element $g\indices{^a_b}$, the structure coefficients\footnote{We assume Euclidean signature so that $f^a{}_{bc}$ takes the same values as $f_{abc}$, but we still keep track of upper and lower indices to facilitate the transition to other signatures.} $f^a{}_{bc}$ transform as the components of an invariant tensor, i.e. $g:f \rightarrow f$, or
\be
gf:= g\indices{^a_d} f\indices{^d _{eh}} (g^{-1})\indices{^e_b}(g^{-1})\indices{^h_c}=f\indices{^a_{bc}}.
\ee
Therefore one gets
\begin{equation}\label{Inv-f}
g\indices{^a_d} f\indices{^d _{eh}} (g^{-1})\indices{^e_b}(g^{-1})\indices{^h_c}=f\indices{^a_{bc}}.
\end{equation}
Note, however, that $f_{abc}$ does not define an invariant tensor under $SO(7)$ , but only under its subgroup $G_2$. Let $g\in G_2$ be an element near the identity, 
\begin{equation}
g^a{}_b= \delta^a_b + \theta^s (\mathbb{J}_s)^a{}_b, 
\end{equation}
where the generators $\mathbb{J}_s$ belong to the Lie algebra $\mathfrak{g}_2$ and satisfy
\begin{equation}\label{g2}
[\mathbb{J}_r, \mathbb{J}_s] = C^t_{rs}\,  \mathbb{J}_t\; ,
\end{equation}
with $C^t_{rs}$ the structure constants of $\mathfrak{g}_2$. The invariance of the cross product under the action of $G_2$ implies
\begin{equation}\label{Jf-Jf-Jf=0}
(\mathbb{J}_s)^a{}_d f^d{}_{bc} - (\mathbb{J}_s)^d{}_b f^a{}_{dc} - (\mathbb{J}_s)^d{}_c f^a{}_{bd} = 0.
\end{equation} 
An important property of $G_2$ which was taken as a definition of $G_2$ in Section \eqref{sec:General-setting} is stated in the following theorem:

\begin{theorem}{(Bryant \cite{Bryant1987})}\label{th:Bryant}
The subgroup of $GL(\textrm{Im}\, \mathbb{O})$ that leaves invariant the 3-form 
\begin{equation}
\varphi = \frac{1}{2} f_{abc} e^a e^b e^c ,
\end{equation}
where $e^a$ is a local basis of vector 1-forms in seven dimensions, is a compact, simply connected Lie group of dimension 14, isomorphic to $G_2 = \textrm{Aut}~ \mathbb{O}$. Moreover, $G_2$ is also the holonomy group of squashed 7-sphere.
\end{theorem}

The $S^7$ manifold with the new structure, the metric $g$ which is invariant under $SO(8)$ and the torsion given by $f\indices{^a_{bc}}$ which is invariant under $G_2$, is defined by the coset $Spin(7)/G_2$ and is a squashed\footnote{The $SO(8)$ invariant metric $g$ induces a $Spin(7)$ invariant metric $\bar{g}$ on squashed $S^7$. Since the metric $g$ is $SO(8)$ invariant, the described squashed sphere is known in physics literature as the (round) seven-sphere with torsion.} $S^7$. 


We can then compute the curvature corresponding to the connection $\bar{\omega}$. From \eqref{eq:flat_con_spin7_omega}, we obtain
\be\label{eq:curv_flat_spin7_con}
R\indices{^a_b}= \bar{R}\indices{^a_b} + \beta \bar{D}f\indices{^a_{bc}} e^c + \beta^2  f\indices{^a_{cd}}f\indices{^c_{bg}} e^d e^g.
\ee
Since $\omega$ is the flat connection, then $R\indices{^a_b}=0$. The next job is to compute $\bar{D}f\indices{^a_{bc}}$.

We can consider different notions of parallel transport
associated to $G_2$, $Spin(7)$, and torsion-free $Spin(7)$ connections\footnote{Note that $Spin(7)$ is 
the universal cover of $SO(7)$ and hence simply connected and acts transitively on $S^7$. In the following, we will use $SO(7)$ instead of $Spin(7)$
wherever we refer to the group representation or $\mathfrak{so}(7)$ for the Lie algebra of $Spin(7)$.}, denoted by $\hat{\omega}$, $\omega$, and $\bar{\omega}$ respectively.
Each of these connections has a corresponding covariant derivative, which we denote by $\hat{D}$, $D$, and $\bar{D}$. Since $G_2\subset SO(7)$,
the connection $\hat{\omega}$ can also be viewed as a piece of the connection of $Spin(7)$.
These connections are enjoying independent properties,
\begin{itemize}
    \item
    $Spin(7)$ connection $\omega$ is torsionful, i.e. $\omega=\bar{\omega}+k$.
    \item $Spin(7)$ connection $\bar{\omega}$ is a torsion-free connection, i.e. $\bar{D} e^a  =0$.
    \item $G_2$ is the automorphism group of octonions, hence $G_2$ connection $\hat{\omega}$ satisfies $\hat{D} f\indices{^a_{bc}}=0$.
\end{itemize}


The $G_2$ connection $\hat{\omega}(x)$ can be expressed as
\begin{equation}
\hat{\omega}^a{}_b = \hat{\omega}^s(x) \left(\mathbb{J}_s\right)^a{}_b, \qquad  s=1, \cdots, 14
\end{equation}
where $\{\mathbb{J}_s\}$ are the 14 generators of $G_2$ in a $7\times7$ antisymmetric representation of the Lie algebra $\mathfrak{g}_2$ (for concrete examples of this $7 \times 7$ representation, see \cite{Gunaydin:1973rs}).
We can define the torsion $Spin(7)$ connection $\bar{\omega}$ by splitting to $G_2$ connection and another $1$-form which then has to be a multiple of $f\indices{^a_{bc}} e^c$,
\begin{align}\label{omega-omegahat}
\bar{\omega}\indices{^a_b} 
=\hat{\omega}\indices{^a_b} + \mu f\indices{^a_{bc}} e^c,
\end{align}
with $\mu$ a constant parameter. This is a consequence of \eqref{eq:flat_con_spin7_omega} together with splitting torsionful $Spin(7)$ connection $\omega$ to a $G_2$ connection and an extra piece proportional to $f\indices{^a_{bc}} e^c$, i.e. $\omega\indices{^a_b} = \hat{\omega}\indices{^a_b} + (\mu+\beta) f\indices{^a_{bc}} e^c$.

This can be made more precise in the following way. Consider the generators of $\mathfrak{so}(7)$ algebra $J^{ab}$ and the split $J^{ab} = \{\mathbb{J}_s, \mathbb{X}_c\}$ where $\mathbb{J}_s$ are generators of $\mathfrak{g}_2$ subalgebra and $\mathbb{X}_c$ are the remaining seven generators of $\mathfrak{so}(7)$. As we mentioned before $\mathbb{X}_c$ are the generators of the adjoint representation of the algebra of imaginary octonions. Then $Spin(7)$ connection is defined as
\begin{equation}
\bar{\omega}^a{}_b = \bar{\omega}^s(x) \left(\mathbb{J}_s\right)^a{}_b + \bar{\omega}^c(x) \left(\mathbb{X}_c\right)^a{}_b,
\end{equation}
with 1-form functions $\bar{\omega}^s(x)$ and $\bar{\omega}^c(x)$. Writing $\bar{\omega}^s(x)=\hat{\omega}^s(x)+ \bar{\theta}^s(x)$ for $\bar{\theta}^s(x)$ a $1$-form function, we have
\begin{equation}
\bar{\omega}^a{}_b = \hat{\omega}^s(x) \left(\mathbb{J}_s\right)^a{}_b + \bar{\theta}^s(x) \left(\mathbb{J}_s\right)^a{}_b + \bar{\omega}^c(x) \left(\mathbb{X}_c\right)^a{}_b,
\end{equation}
which simplifies to \eqref{omega-omegahat} where we write $\mu f\indices{^a_{bc}} e^c = \bar{\theta}^s(x) \left(\mathbb{J}_s\right)^a{}_b + \bar{\omega}^c(x) \left(\mathbb{X}_c\right)^a{}_b$. By evaluating the left hand side of this relation (i.e., $\mu f\indices{^a_{bc}} e^c$), one can determine $\bar{\theta}^s(x)$ and $\bar{\omega}^c(x)$ in terms of $\mu$ and $e^c(x)$.

Now from \eqref{omega-omegahat}, we can simply find
\be
\bar{D}f\indices{^a_{bc}} e^c = \hat{D}f\indices{^a_{bc}} e^c - \mu\, (f\indices{^a_{cd}} f\indices{^c_{bg}}+f\indices{^a_{cg}}f\indices{^c_{db}} +  f\indices{^a_{cb}} f\indices{^c_{gd}}) e^g e^d.
\ee
Plugging this back in \eqref{eq:curv_flat_spin7_con}, we get the $Spin(7)$ curvature
\be
\bar{R}\indices{^a_b}= -\beta^2 f\indices{^a_{cd}} f\indices{^c_{bg}} e^d e^g +\beta\mu\, (f\indices{^a_{cd}} f\indices{^c_{bg}} + f\indices{^a_{cg}} f\indices{^c_{db}}  + f\indices{^a_{cb}} f\indices{^c_{gd}}) e^g e^d.
\ee
As we discussed in previous section, the covariant constancy of torsion, $DT^a = 0$, is a requirement for flatness. It can be written as
\be
DT^a = \beta (3\mu + \beta) f\indices{^a_{dc}} f\indices{^d_{bg}} e^b e^g e^c,
\ee
which imposes a relation between $\beta$ and $\mu$
\be\label{mu-beta-rel}
\mu = -\f{\beta}{3}.
\ee
Direct computation using the explicit form of $f\indices{^a_{bc}}$ given in \eqref{f} for the Euclidean signature metric shows that for $\mu = -\beta/3 $ the right hand side is proportional to $e^a e_b$, and finally we obtain
\begin{equation} \label{Rbar-R'}
\bar{R}\indices{^a_b} = \beta^2 e^a e_b .
\end{equation}
This is the curvature of $Spin(7)$ connection $\bar{\omega}$ which has the form of the curvature of a symmetric space.   

From \eqref{Rbar-R'}, we find that
\be
\bar{R}\indices{^a_{bcd}} = \beta^2 (\delta^a_c \eta_{bd} -\delta^a_d \eta_{bc}),
\ee
which by comparing to the curvature of a seven dimensional symmetric manifold, one can determine $\beta$ in terms of $\bar{R}$ the scalar curvature of the manifold or $\gamma$ the radius of
$7$-sphere $\beta^2=\f{\bar{R}}{42}=\f{1}{\gamma^2}$. Therefore, there are two flat connections written as \cite{Ranjbar2018}
\be
\omega\indices{^a_b} = \bar{\omega}\indices{^a_b} \pm \f{1}{\gamma} f\indices{^a_{bc}} e^c.
\ee
The occurrence of the prefactor $\f{1}{\gamma}$ is not a surprise. In fact, in the theory of interacting quantum fields in the presence of torsion, new currents must be considered due to symmetries corresponding to the torsion which results in the appearance of new interacting terms in the action. Specifically, there is an interaction $\bar{\psi}\Gamma_{mnp}\psi$ which couples to the current $J^{mnp}$. This interaction term is proportional to the torsion $S_{mnp}$ with the coupling $\f{1}{\gamma}$ \cite{Biran:1982eg, Casher:1984ym}. In Chapter \ref{ch:exotic}, we will see how the torsion $S_{mnp}$ is related to $f\indices{^a_{bc}}$.

As the final remark, we emphasize that the coset space $Spin(7)/G_2$ which describes the parallelized $S^7$ is a non-symmetric space \cite{Englert:1982vs,Biran:1982eg,DAuria:1982chj}. It is shown explicitly in \cite{Gunaydin:1995as}.  

\subsection{Pseudo-Riemannian Signature}
The discussion of pseudo-Riemannian manifolds is the same as previous section. Let's consider a seven-dimensional simply-connected symmetric pseudo-Riemannian manifold with the metric $g$ and the connection $\omega$. By imposing the flatness of $\omega$, the corresponding pseudo-Riemannian curvature vanishes. Then, as shown by Wolf \cite{wolf1972I,wolf1972II}, such a manifold is the pseudo-sphere $S^{3,4}$ given by the  coset space $SO(4,4)/SO(3,4)$ or the pseudo-hyperbolic space $H^{4,3}$ given by the coset space $SO(4,4)/SO(4,3)$. The metric $g$ is invariant under $SO(4,4)$ isometry group of $S^{3,4}$ or $S^{4,3}$. In fact, Wolf showed that these are the only parallelizable pseudo-sphere or pseudo-hyperbolic spaces, see the proof of theorem 8.13 of \cite{wolf1972II}.

The group $SO(4,4)$ acts on split-octonions $\mathbb{O}_s$ with the isotropy group $SO(3,4)$ or $SO(4,3)$ depending on the signature of the induced metric on the space of imaginary split-octonion which is orthogonal to the identity element of $\mathbb{O}_s$. More precisely, we should consider $SO^+(4,4)$, $SO^+(4,3)$, and $SO^+(3,4)$ which are the connected components of the corresponding $SO(p,q)$ group. In fact, $SO^+(4,4)$ is the group which acts transitively on $S^{3,4}$ or $S^{4,3}$.

Similarly to previous section, we show that one can define another parallelism with respect to the metric $g$ with the connection $\bar{\omega}$ in the presence of torsion. In this case, the metric $g$ is $SO(4,4)$ invariant and the torsion, given by the structure constants of imaginary split-octonions, breaks the isometry group down to $Spin(3,4)$.

The complex Lie algebra $\mathfrak{g}_2$ has, besides the compact form, a split real form $\mathfrak{g}_{2,2}$ to which corresponds the Lie group split $G_2$, referred to as $G_{2,2}$ in the following. This group is the group of automorphism of split octonions $\mathbb{O}_s$, see \cite{Agricola2008, Baez:2014}. The split octonion vector space defines a product with a bilinear form of the signature $(4,4)$. This induces a metric of the signature $(3,4)$ (or $(4,3)$) on the tangent space which is spanned by the basis of imaginary split octonions $\textrm{Im}\, \mathbb{O}_s$, see appendix \ref{app:split_oct}. This defines a cross product on $R^{3,4}$ (or $R^{4,3}$) with the coefficients of cross product $f_{abc}$ being components of a totally antisymmetric tensor and taking the values $\pm 1, 0$ with $f_{abc}=+1$ for\footnote{The ordering written here is once again only a choice, consistent with our chosen convention for the octonionic multiplication table.} $abc=123,145,167,246,275,347,356$. Therefore, there is a $3$-form,
\be\label{3_form_split}
\varkappa= dx^{123}-dx^{154}+dx^{167}-dx^{264}+dx^{275}-dx^{374}+dx^{356},
\ee
which is invariant under $G_{2,2}\subset SO(3,4)$.

There is a definition for the group $G_{2,2}$ along the same line as Theorem \ref{th:Bryant}:
\begin{theorem}\label{th:Split_G2}
The subgroup of $GL(\textrm{Im}\, \mathbb{O}_s)$ that leaves invariant the 3-form 
\begin{equation}
\varkappa = \frac{1}{2} f_{abc} e^a e^b e^c ,
\end{equation}
where $e^a$ is a local basis of vector 1-forms in 7 dimensions, is a non compact Lie group of dimension 14, isomorphic to $G_{2,2} = \textrm{Aut}~ \mathbb{O}_s$.
\end{theorem}

The $S^{3,4}$ manifold with the new structure, the metric $g$ which is invariant under $Spin(3,4)$ and the torsion given by $f\indices{^a_{bc}}$ which is invariant under $G_{2,2}$, is written as the coset $Spin(3,4)/G_{2,2}$ and can be called a squashed $S^{3,4}$. We can write the flat connection, in a chosen basis, similar to the previous section as
\be\label{eq:flat_con_spin34_omega}
\omega\indices{^a_b}= \bar{\omega}\indices{^a_b} + \tilde{\beta} f\indices{^a_{bc}} e^c,
\ee
with the corresponding curvature given by
\be\label{eq:curv_flat_spin34_con}
R\indices{^a_b}= \bar{R}\indices{^a_b} + \tilde{\beta} \bar{D}f\indices{^a_{bc}} e^c + {\tilde{\beta}}^2  f\indices{^a_{cd}}f\indices{^c_{bg}} e^d e^g.
\ee
Similar to the Euclidean discussion of previous section, we can take into account 
different notions of parallel transport
associated to $G_{2,2}$ and $Spin(3,4)$\footnote{Note that $Spin(3,4)$ is 
the universal cover of $SO(3,4)$ and it is not connected (it has two components). Its connected component $Spin^+(3,4)$ acts transitively on $S^{3,4}$. In the following, we will use $SO(3,4)$ instead of $Spin^+(3,4)$ whenever we refer to the group representation or $\mathfrak{so}(3,4)$ for the Lie algebra of $Spin^+(3,4)$, but keeping in mind we are always referring to connected component of each group.}, and torsion-free $Spin(3,4)$ denoted by $\hat{\omega}$ and $\omega$, and $\bar{\omega}$ respectively.
Each of these connections has a corresponding covariant derivative, which we denote by $\hat{D}$, $D$, and $\bar{D}$. Since $G_{2,2}\subset SO(3,4)$,
the connection $\hat{\omega}$ can also be viewed\footnote{We should emphasize an important point that one has to in fact consider the group $G^*_{2,2}$ which is a centerless group to be obtained from $G_2$ and is the stabilizer of $Spin^+(3,4)$. See Appendix \ref{app:sp-and-g2} for details.} as a piece of the connection of $Spin(3,4)$.
These connections are enjoying independent properties,
\begin{itemize}
    \item $spin(3,40$ connection $\omega$ is torsionful, i.e. $\omega=\bar{\omega}+k$.
    \item $Spin(3,4)$ connection $\bar{\omega}$ is a torsion-free connection, $\bar{D} e^a  =0$.
    \item $G_{2,2}$ is the automorphism group of split-octonions, $\hat{D} f\indices{^a_{bc}}=0$.
\end{itemize}

The structure constants of split octonions are invariant with respect to the $G_{2,2}$ connection $\hat{\omega}$,
\be\label{eq:f-inv-split-O}
\hat{D} f\indices{^a_{bc}} = \hat{\omega}\indices{^a_k} f\indices{^k_{bc}}-\hat{\omega}\indices{^k_b} f\indices{^a_{kc}}-\hat{\omega}\indices{^k_c} f\indices{^a_{bk}}= 0.
\ee
The $G_{2,2}$ connection $\hat{\omega}(x)$ can be expressed as
\begin{equation}
\hat{\omega}^a{}_b = \hat{\omega}^s(x) (\hat{\mathbb{J}}_s)\indices{^a_b}, \qquad  s=1, \cdots, 14
\end{equation}
where $\{\hat{\mathbb{J}}_s\}$ are the 14 generators of $G_{2,2}$ in a $7\times7$ antisymmetric representation of the Lie algebra $\mathfrak{g}_{2,2}$.

Then, one can write a similar equation as \eqref{omega-omegahat}, $\bar{\omega}\indices{^a_b}=\hat{\omega}\indices{^a_b}+\tilde{\mu} f\indices{^a_{bc}}e^c$, in order to find a relation between $\bar{D}f\indices{^a_{bc}}$ and $\hat{D}f\indices{^a_{bc}}$ which leads us to get
the curvature of $Spin(3,4)$ connection
\be
\bar{R}\indices{^a_b}= -\tilde{\beta}^2 f\indices{^a_{cd}} f\indices{^c_{bg}} e^d e^g + \tilde{\beta}\tilde{\mu}\, (f\indices{^a_{cd}} f\indices{^c_{bg}} + f\indices{^a_{cg}} f\indices{^c_{db}}  + f\indices{^a_{cb}} f\indices{^c_{gd}}) e^g e^d.
\ee
The relation \eqref{mu-beta-rel} was obtained with no reference to the actual value of $f\indices{^a_{bc}}$ and as a result it holds in the case of split-octonion where it imposes the condition $\tilde{\mu}=-\tilde{\beta}/3$.

Direct computation using the explicit form of $f\indices{^a_{bc}}$ given in \eqref{3_form_split} for the pseudo-Riemannian signature metric and $\tilde{\mu}=-\tilde{\beta}/3$, and in a chosen basis, shows that if the tangent space is spanned by the imaginary split-octonion basis, the induced metric on the tangent space is a pseudo-Riemannian metric of signature $(3,4)$ and the curvature of the manifold is
\be
\bar{R}^{ab}= \tilde{\beta}^2  det(\xi) e^a e^b.
\ee
Here $det(\xi)$ is the determinant of the metric $\xi_{ab}$ on the quaternionic subspace of $\textrm{Im}\,\mathbb{O}_s$ spanned by the basis $e^a$, $e^b$ and $e^c$ where $c$ is the index for which $f_{abc}$ is non-zero \cite{Ranjbar2018}. This result holds for the case where the metric has the signature $(4,3)$. It is straightforward to see since the only change under metric reversal is $f_{abc} \rightarrow f_{abc}$ which does not affect the result as it appears quadratically in the expression for the curvature.


We can summarize the result in seven dimensions as follows:
\begin{itemize}
\item If $T\mathcal{M}\cong \textrm{Im}\,\mathbb{O}$, then the induced metric $\eta_{ab}$ on the tangent space is Euclidean. The manifold is a seven dimensional constant positive curvature symmetric manifold, that is $S^7=\f{SO(8)}{SO(7)}$ or constant positive curvature non-symmetric coset space $S^7=Spin(7)/G_2$ in the presence of torsion.
\item If $T\mathcal{M}\cong \textrm{Im}\,\mathbb{O}_s$, then the induced metric $\eta_{ab}$ on the tangent space is a diagonal metric
\begin{enumerate}
\item with the signature $(4,3)$, $ds^2=-\sum_{i=1}^{3} dx_i^2 + \sum_{i=4}^{7} dx_i^2$, and consequently $det(\xi)=-1$. The manifold $\mathcal{M}$ is a constant negative curvature symmetric manifold $S^{4,3}$ written as the coset $S^{4,3}=\f{SO(4,4)}{SO(4,3)}$ if there is no torsion or a constant negative curvature non-symmetric coset space $S^{4,3}=\f{Spin(4,3)}{G_{2,2}}$ in the presence of torsion.
\item with the signature $(3,4)$, $ds^2=-\sum_{i=1}^{4} dx_i^2 + \sum_{i=5}^{7} dx_i^2$, and consequently $det(\xi)=+1$. The manifold $\mathcal{M}$ is a constant positive curvature symmetric manifold $S^{3,4}$ written as the coset $S^{3,4}=\f{SO(4,4)}{SO(3,4)}$ if there is no torsion or a constant positive curvature non-symmetric coset space as $S^{3,4}=\f{Spin(3,4)}{G_{2,2}}$ in the presence of torsion.
\end{enumerate}
\end{itemize}

This is consistent with the discussion of parallelizability for both $S^7$ and $S^{3,4}$ in seminal papers by Cartan-Schouten \cite{Cartan1,Cartan2} and Wolf \cite{wolf1972I, wolf1972II}.
\chapter{$S^{3,4}$ Reduction and Exotic Solutions}\label{ch:exotic}

\section{Duality and Unorthodox Spacetime Signatures}\label{sec:duality}

In general, dualities are nontrivial isomorphisms relating two models that are seemingly completely unrelated. Dualities have been a cornerstone in different branches of 
mathematics. As examples of dualities in mathematics we can mention the isomorphism between a finite dimensional vector space and its dual vector space, the relation which maps points to lines and lines to points in Desarguesian geometry and many other examples in different areas of mathematics. A very well known example of duality in physics is the electric-magnetic duality where the Maxwell equations are invariant under the transformation $(F^{\mu\nu}, \star F^{\mu\nu})\rightarrow (\star F^{\mu\nu}, -F^{\mu\nu})$, where $\star$ is the Hodge dual operator.
Since the early 1980's, with the discovery of Mirror symmetry in string theory, it has been understood that dualities can relate theories 
with different gauge groups, different spacetime dimensions,
different amounts of supersymmetry, and even relate theories of gravity to gauge theories. Moreover, there are dualities suggesting that even the spacetime signature
is a relative concept that depends on the description \cite{Hull:1998ym, Hull:1998vg}.

In any duality prescription between two dual models, each model has a moduli space of physical quantities and duality is the mean to relate these moduli spaces in a one-to-one correspondence. In a sense, 
observable quantities depend on a ``frame of reference" and duality is a map between these frames. As an example in string theory S-duality is a relation between a theory in strong coupling
regime and another theory in weak coupling regime where the coupling constant bears the relative concept of the ``frame of reference". The other example is the T-duality where a string theory with strings wrapped around a circle of radius $R$ is proved to be equivalent to another string theory with strings wrapped around a circle of radius $1/R$. 
In both S-duality and T-duality the dual theory can be in a different scale of coupling or in different spacetime dimensions but always the signature of spacetime is Lorentzian.
In \cite{Hull:1998ym, Hull:1998vg}, Hull extended this to include the relativeness of the notion of signature. T-duality maps type IIB superstring theory to type IIA while Heterotic theory is mapped to itself. $M$-theory in $10+1$ dimensions compactified on a 2-torus $T^{2,0}$ in the limit where the torus shrinks to zero size, gives a type IIB theory in $9+1$ dimensions, while compactifying the $10+1$ dimensional $M$-theory on a 3-torus $T^{3,0}$ in the zero size limit gives again $M$-theory. In all these cases, the compactification was done on a spacelike circle. In order to change the signature of the dual theory one needs to compactify the theory on a timelike circle. This can in principle produce ghosts or tachyons due to presence of closed timelike curves and brings about instabilities. Therefore a dual theory obtained by compactifying on a timelike circle can be pathological as a theory itself, however there are few reasons to consider them. First, it is true that while doing the dimensional reduction, truncating the dependence on the internal dimensions leaves a dimensionally reduced theory
in which some of the fields have kinetic terms of the wrong sign, but in the full compactified theory such ghosts can be gauged away \cite{Hull:1998br}. This is similar to the example of dimensional reduction of (super) Yang-Mills theory from $D+d$ dimensions to $D$ dimensions by reducing on a Lorentzian torus $T^{d-1,1}$ which gives Yang-Mills plus $d$ scalars in Euclidean space, acted on by an $SO(d-1, 1)$ R-symmetry. The scalar field coming from $A_0$ has a `wrong' sign for the kinetic term but in the full compactified theory $A_0$ can be gauged away leaving the theory with no ghost.

Secondly, one notices that even though the dual theories themselves can be pathological but the theories obtained through timelike duality can be written on flat space independently without resorting to the timelike duality and avoiding possible pathologies with which will come. Moreover, since all these theories are related through the chain of dualities to an $M$-theory which is known to be free of the mentioned pathologies, the dual theories are as `good' as (or as `bad' as) the $10+1$ dimensional standard $M$-theory.    

Given the possibility of timelike reduction, the $M$-theory in $10+1$ dimensions will be linked to a theory in $9+2$ dimensions, which will be 
referred to as $M^*$-theory, and a theory in $6+5$ dimensions, $M'$-theory. Timelike T-duality maps Heterotic string theory to itself but type IIA and type IIB theories are mapped to type IIB$^*$ and type IIA$^*$ theories respectively. The strong coupling limit of the IIA$^*$ theory is the $M^*$-theory in $9+2$ dimensions. Compactifying $M$-theory on $T^{2,1}$ on the zero size limit of 3-torus gives $M^*$-theory in $9+2$ dimensions. Then compactifying $M^*$-theory on a Euclidean 3-torus $T^{3,0}$ in the limit that torus shrinks to zero size gives the $M'$-theory in $6+5$ dimensions. One can understand the appearance of extra three dimensions by looking at the strong coupling limit of type IIA and IIA$^*$  theories in the following way: in the case of $M$-theory, the type IIA theory is T-dual to $M$-thoery. Type IIA theory has D$0$-branes, which in the compactification limit on $3$-torus, are wrapped around the compact $3$-torus. At strong coupling limit, i.e. $g_s\rightarrow \infty$ where $g_s$ is the string coupling, three spacelike $2$-cycles of $T^{3,0}$ become light and signals opening of the three extra spacelike dimensions. In the case of $M^*$-theory, one observes that the type IIA$^*$ theory and $M^*$-theory are T-dual. The  IIA$^*$ theory has E$1$-branes wrapped around one spacelike and two timelike $2$-cycles of compact $T^{2,1}$ geometry which becomes light at strong coupling limit and can be considered as the whole tower of Kaluza-Klein modes from the decompactification of extra dimensions. It hints towards opening one spacelike and two timelike dimensions.
The complete list of string theories and their strong coupling limit and the net of dualities which relate them can be seen in \cite{Hull:1998ym}. It has been reviewed more recently in \cite{Dijkgraaf:2016lym}, however the notation is different than what we use here.

The change of signature induced by duality transformations can be given a group theoretical interpretation in terms of the infinite-dimensional Kac-Moody algebras $E_{11}$ or $E_{10}$ that have been conjectured to be ``hidden symmetries" of 11-dimensional supergravity or an appropriate extension \cite{Julia:1980gr,Julia:1982gx,Julia:1981wc,Nicolai:1991kx,Julia:1997cy,West:2001as,Damour:2000hv,Damour:2002cu}. Weyl reflections in the exceptional root may change the spacetime signature. 
In the orbit of the standard signature $(10,1, +)$ of $M$-theory appear also the signatures $(9,2,-)$ $M^*$-theory, $(6,5,+)$ $M'$-theory, and the theories with reversed signatures $(1,10, -)$, $(2,9, +)$ and $(5,6, -)$ \cite{Englert:2003py,Keurentjes:2004bv,Keurentjes:2004xx,Englert:2004ph}.  In $(s,t, \pm)$, $s$ is the number of space directions, $t$ the number of time directions, and  $\pm$ refers to the sign of the kinetic term of the $3$-form of eleven-dimensional supergravity where $-$ means the kinetic term has a `wrong' sign. 

For all these reasons, it is of interest to explore solutions with exotic spacetime signatures in the context of $M$-theory. This is the the path we follow in this chapter. The analogues of  Freund-Rubin solutions \cite{Freund:1980xh} have been already studied in \cite{Hull:1998fh}. In this chapter, we derive the analogue of Englert solution \cite{Englert:1982vs}.

As we discussed in Chapter \ref{ch:Lorentz_flatness}, the seven-sphere $S^7$ is one of the four spheres that are parallelizable. If one considers pseudo-Riemannian signatures, in seven dimensions the pseudo-sphere $S^{3,4}$,  which is the space of constant curvature $SO(4,4)/SO(3,4)$ or space of unit split octonions\footnote{The negative curvature spaces obtained by multiplying the metric by an overall minus sign, which are the pseudo-hyperbolic spaces $H^{0,7}$ and $H^{4,3}$  (see appendix \ref{app-notation} for conventions) also enjoy this property.}, turns out to be parallelizable. The $7$-sphere plays an important role in eleven-dimensional supergravity \cite{Cremmer:1978km}, either through Freund-Rubin compactification where it leads to the $SO(8)$ gauged supergravity  \cite{deWit:1982bul} or through Englert solution which breaks the $SO(8)$ symmetry down to $Spin(7)$ \cite{Englert:1983jc} due to the presence of internal fluxes. We shall see that there is an analogue of the Englert solution in $M'$-theory, which takes the form $AdS_4 \times S^{3,4}$.  Similar to the original proposal by Englert where the solution characterized by internal fluxes provided by a torsion that parallelizes $S^7$, this exotic form of the Englert solution is characterized by internal fluxes provided by a torsion that flatten a connection on the tangent bundle of $S^{3,4}$.  These internal fluxes break the $SO(4,4)$ symmetry of the pseudo-sphere $S^{3,4} \equiv SO(4,4)/SO(3,4)$ down to\footnote{For more information on the split octonions and the relevant associated groups see Appendix \ref{app:quaternion_octonion} and see Appendix \ref{App:pseudo} for information on the pseudo-sphere $S^{3,4}$.} $Spin^+(3,4)$.

\section{Action Principle and General Ansatz}

\subsection*{Action}

Our starting point is the bosonic action of eleven-dimensional supergravity with mixed signature \cite{Cremmer:1978km,Hull:1998ym, Hull:1998vg}, written as
\be
S=\f{1}{2\kappa^2} \int dx^{s+t} \mathcal{L},
\ee
where the Lagrangian $\mathcal{L}$ is
\begin{align}
\mathcal{L}=\sqrt{|g^{s+t}|} (R&-\f{\eta}{24}F_{MNPQ}F^{MNPQ}\nonumber\\
&-\f{2\sqrt{2}}{12^4}\f{1}{\sqrt{|g^{s+t}|}}\varepsilon^{M_1...M_{11}}F_{M_1...M_4}F_{M_5...M_8}A_{M_9...M_{11}}).\label{11d_Lagrangian}
\end{align}
Here, $s+t$-dimensions ($s+t=11$) refers to $s$ spacelike directions and $t$ timelike directions and $F_{MNPQ}=4\partial_{[M}A_{NPQ]}$. The sign of $\eta$ is determined as\footnote{From the point of view of fermionic content of the theory, whenever $\eta=+1$ there exists a purely real representation of Clifford algebra and as a result spinors are Majorana. However, if $\eta=-1$ then there exists a purely imaginary representation of Clifford algebra and the spinors are pseudo-Majorana \cite{Kugo:1982bn}.}
\begin{align}
\eta &= +1\quad \textrm{if}\quad (s-t)\, \textrm{mod}\,8=1\quad\textrm{i.e. in dimensions}\quad 10+1, 6+5, 2+9,\\
\eta &= -1\quad \textrm{if}\quad (s-t)\, \textrm{mod}\,8=7\quad\textrm{i.e. in dimensions}\quad 1+10, 5+6, 9+2.
\end{align}
Therefore, in dimensions $10+1$, $6+5$ and $2+9$ the kinetic term of the $3$-from has the correct sign ($\eta=+1$) while in 
dimensions $9+2$, $1+10$ and $5+6$ the kinetic term of $3$-form comes with the `wrong' sign ($\eta=-1$). As mentioned in Section \ref{sec:duality}, $M$-theory corresponds to $\eta = 1$, $(s,t)= (10,1)$, whereas $M^*$-theory corresponds to $\eta = -1$, $(s,t) = (9,2)$ and $M'$-theory corresponds to $\eta = 1$, $(s,t) = (6,5)$.

The equations of motions can be obtained by varying the Lagrangian \eqref{11d_Lagrangian} with respect to the metric $g_{MN}$ and the $3$-from $A_{MNP}$,
\begin{align}
R_{MN}-\f{1}{2}g_{MN}R &= -\f{\eta}{48}(g_{MN}F^2-8 F_{MPQR}F\indices{_N^{PQR}}),\label{EOM:var-wt-metric}\\
F\indices{^{MNPQ}_{;M}} &= \f{1}{\eta} \f{18\sqrt{2}}{12^4}\f{1}{\sqrt{|g^{11}|}} \varepsilon^{M_1...M_8 NPQ}F_{M_1...M_4}F_{M_5...M_8}\label{EOM:var-wt-3form}.
\end{align}

The equations of motion are invariant under the simultaneous transformations
\be
g_{MN} \rightarrow - g_{MN},\qquad
A_{MNP} \rightarrow  A_{MNP},
\ee
accompanied with the sign change $\eta \rightarrow -\eta$. We recall that with our conventions,  $\epsilon^{01\cdots 10}$ changes sign 
since in the transformation $g_{MN} \rightarrow - g_{MN}$,  the parity of the number of time directions changes. Because of this invariance, 
we concentrate on $M$-theory, $M^*$-theory and $M'$-theory, knowing that the solutions for the reversed 
signature theories corresponding to the other theories in the same duality orbit as ordinary supergravity can be readily obtained using the invariance.

\subsection*{Product of Constant Curvature Spaces}

We assume that the eleven-dimensional spacetime manifold (referred to as the ``\textbf{ambient space}" in the sequel) splits as the product of a four-dimensional manifold (``\textbf{background spacetime}") and a seven-dimensional manifold (``\textbf{internal space}")\footnote{We stress that the ``internal space" is thus in our terminology always a seven-dimensional manifold, independently of its curvature or signature.}.  The metric itself is a direct sum and splits into two parts, the metric on the internal space, and the one on the background spacetime. Capital Latin indices $M,N,...$ run over $0,...,10$. Small Latin indices $m,n,...$ are internal space indices and run over $4,...,10$.  Greek indices $\mu,\nu,...$ are background spacetime indices running over $0,1,2,3$. Notice that here $|g^{11}|$ is the absolute value of the determinant of the metric of the total space. We use the absolute value everywhere when needed for the determinant of the ambient, internal and background metrics since, a priori, there is no restriction on the signature of the aforementioned metrics. 

We also assume that spacetime and the internal space are spaces of constant curvature, i.e., either a pseudo-sphere or a pseudo-hyperbolic space, with curvatures respectively given by 
\begin{align}
 R_{mnrs} &= \f{a}{6} \left(g_{mr} g_{ns} - g_{ms} g_{nr} \right) \label{CC1},\\ 
R_{\mu \nu \rho \sigma} &= \f{b}{3} \left(g_{\mu \rho} g_{\nu \sigma} - g_{\mu \sigma} g_{\nu \rho} \right) \label{CC2},
\end{align}
where $a$ and $b$ are constants, which are positive for pseudo-spheres and negative for pseudo-hyperbolic spaces\footnote{The conventions and notations used in this chapter are collected in Appendix \ref{app-notation}.}.  This implies in particular
\begin{align}
R_{mn} &=a\, g_{mn},\label{Ricci-internal-space}\\
R_{\mu\nu} &=b \, g_{\mu\nu},\label{Ricci-spacetime}\\
R_{m\mu} &=0.\label{Ricci-mixed}
\end{align}

Assuming that some of the internal or background fluxes are turned on, we can find different type of solutions for the bosonic system of supergravity given the above ansatz for the internal and the spacetime geometry.

\section{Freund-Rubin Type Solutions}

We first derive the analogues of Freund-Rubin solutions.  These were already considered in \cite{Hull:1998fh}.

In addition to conditions \eqref{Ricci-internal-space}-\eqref{Ricci-mixed} on the metric, Freund-Rubin ansatz assumes that the only $4$-form flux is in spacetime and reads explicitly as 
\begin{align}
F^{\mu\nu\rho\sigma} &=\f{1}{\sqrt{|g^4|}}\f{f}{\sqrt{4!}} \varepsilon^{\mu\nu\rho\sigma},\label{flux-background}\\
F^{mPQR} &=0,
\end{align}
where $f$ is a constant. There is no internal flux.  

Now, from \eqref{EOM:var-wt-metric} and  the ansatz,  one finds that
\begin{align}
R_{mn} &= (-\f{1}{3}) \f{\eta}{24}f^2 (-1)^{T'}  g_{mn},\\
R_{\mu\nu} &= (+\f{2}{3}) \f{\eta}{24}f^2 (-1)^{T'}  g_{\mu\nu},
\end{align}
where $T'$ is the number of timelike directions in the four-dimensional background spacetime metric.  We shall also introduce $T$ as the number of timelike directions in the internal space, so that $t$  (the number of timelike directions of the eleven-dimensional ambient space) is equal to $t = T + T'$. Comparing with \eqref{Ricci-internal-space} and \eqref{Ricci-spacetime}, we find
\be\label{FR-solution}
a=-\f{1}{3}\f{\eta}{24}f^2 (-1)^{T'}, \qquad b=-2a,
\ee
as in \cite{Hull:1998fh}.
For example, the standard Freund-Rubin solution is obtained by setting $\eta=+1$, $T'=1$.

The Freund-Rubin solutions constitute a one-parameter family of solutions. One can take $f$ as the free parameter characterizing the solutions.  The curvature of the background spacetime and of the internal space are then determined by $f$ and are of opposite signs. Note that the dependence of both curvatures on $f$ is always through $f^2$ in \eqref{FR-solution} and therefore they are invariant under $f \rightarrow -f$. 

The sign of $a$ is easily seen to be always equal to  $(-1)^T$. For $M$-theory and $M'$-theory, $\eta = 1$ and $t=T+T'$ is an odd number and therefore one has $(-1)^{T'} =-(-1)^T$. This is also true for the reversed $M^*$-theory with signature $(2,9)$ and $\eta = 1$. For $M^*$-theory, one has $\eta = -1$ but $t$ is now an even number resulting in $(-1)^{T'} =(-1)^T$.  The same is true for the reversed signature $M$- and $M'$-theories. Therefore, the sign of $a$ is always equal to  $(-1)^T$. In all cases, the curvatures in the four-dimensional spacetime and in the internal manifold have opposite signs.

All Freund-Rubin solutions that can be obtained in this way are listed in Table \ref{tab:FRsolutions}. All these solution are maximally supersymmetric and their corresponding brane construction can be found in \cite{Hull:1998fh}.

\begin{table}[h!]
\begin{center}
\begin{tabular}{ c|c|c|c| }
\cline{2-4} 
       & $t=T+T'$ & $T$ & solution \\
\cline{2-4}
       \multirow{2}{2cm}{$M$-Theory ($\eta=+1$)} &  \multirow{2}{*}{1} & 0  & $S^7\times AdS_4$  \\ &  & 1 & $AdS_7 \times S^4$\\
\cline{2-4}
\cline{2-4}
       \multirow{3}{2.2cm}{$M^*$-Theory ($\eta=-1$)} & \multirow{3}{*}{2} & 0  & $S^7\times H^{2,2} = S^7 \times \f{SO(2,3)}{SO(2,2)}$ \\ &  & 1 & $AdS_7 \times dS_4$ \\ &  & 2 & $S^{5,2}\times H^4 = \f{SO(6,2)}{SO(5,2)}\times \f{SO(4,1)}{SO(4)}$\\
\cline{2-4}
       \multirow{5}{2.2cm}{$M'$-Theory ($\eta=+1$)} & \multirow{5}{*}{5} & 1 & $AdS_7\times S^{0,4} = AdS_7 \times \f{SO(1,4)}{SO(0,4)}$ \\ &  & 2 & $S^{5,2}\times H^{1,3} = \f{SO(6,2)}{SO(5,2)}\times \f{SO(1,4)}{SO(1,3)}$\\ &  & 3 & $H^{4,3}\times S^{2,2} = \f{SO(4,4)}{SO(4,3)}\times \f{SO(3,2)}{SO(2,2)}$ \\ &  & 4 & $S^{3,4}\times AdS_4 = \f{SO(4,4)}{SO(3,4)}\times AdS_4$ \\ &  & 5 & $H^{2,5}\times S^4 = \f{SO(2,6)}{SO(2,5)}\times S^4$\\ 
\cline{2-4}
\end{tabular}
\end{center}
\caption{Freund-Rubin type solutions of different timelike T-dual $M$-theories. Here, $t=T+T'$ is the number of timelike direction of the corresponding $M$-theory, $T$ is the number of time direction of the internal space, and the background geometry is always considered to be four dimensional. The notation for different pseudo-spheres $S^{p,q}$ and pseudo-hyperbolic spaces $H^{p,q}$ is explained in the appendix \ref{App:pseudo}.}\label{tab:FRsolutions}
\end{table}

\section{Englert Type Solutions}

The Englert type solutions \cite{Englert:1982vs} can be considered as spontaneous breakings of the Freund-Rubin solutions through non-vanishing expectation values for some of the  internal components of the $4$-form.  In other words there is now an internal flux which breaks completely the supersymmetry. 

The Englert construction corresponds to the compactification on a parallelized internal manifold. This occurs when the internal  space of constant curvature is either the sphere $S^{7,0}$, corresponding to the original Englert solution \cite{Englert:1982vs}, or the pseudo-sphere $S^{3,4}$ (as well as the cases trivially obtained from these ones by an overall change of sign of the metric, corresponding to the pseudo-hyperbolic spaces $H^{0,7}$ and $H^{4,3}$).  All the other cases ($S^{6,1}$, $S^{5,2}$,  $H^{6,1}$ etc) are not parallelizable.  

By matching the signatures, one easily sees that there are a priori four cases where the parallelizable seven-(pseudo-)spheres might appear for the theories in the time-like duality orbit of $M$-theory.  These are
\begin{itemize}
\item Case 1: $M$-theory with spacetime manifold  $M^{3,1} \times S^{7,0}$,
\item Case 2: $M^*$-theory with spacetime manifold $M^{2,2} \times S^{7,0}$,
\item Case 3: $M'$-theory with spacetime manifold $M^{3,1} \times S^{3,4}$,
\item Case 4: reversed $M'$-theory with spacetime manifold $M^{2,2} \times S^{3,4}$,
\end{itemize}
as well as the reversed signature solutions.  Here, $M^{p,q}$ is a priori either $S^{p,q}$ or $H^{p,q}$, i.e., we leave the sign of the curvature of four-dimensional spacetime undetermined. It will be determined by solving the equations of motion.  We shall show that only cases 1 and 3 are actually compatible with the equations of motion and furthermore, that $M^{3,1}$ is then the anti-de Sitter space $H^{3,1}$.

It may be worth mentioning that in all above cases, the product $\eta (-1)^{T'}$, where $T'$ is the number of time directions in the four-dimensional background spacetime, is equal to $-1$,
\be 
\eta (-1)^{T'} = -1.
\ee
Indeed, one has $\eta = 1$ for cases 1 and 3 with $T'$ odd, and $\eta = -1$ for cases 2 and 4 but $T'$ is now even. 

Furthermore, we point out that in all cases the curvature $a$ of the internal space ($S^{7,0}$ or $S^{3,4}$) is positive, $a >0$, 
and $T$ the number of time directions in the internal space is even. For the reversed signature solution, the curvature $a$ would be negative and $T$ odd.

The Englert ansatz \cite{Englert:1982vs} consists in imposing the conditions
\begin{align}
 R_{mnrs} &= \f{a}{6} \left(g_{mr} g_{ns} - g_{ms} g_{nr} \right) \label{Eng-CC1},\\ 
R_{\mu \nu \rho \sigma} &= \f{b}{3} \left(g_{\mu \rho} g_{\nu \sigma} - g_{\mu \sigma} g_{\nu \rho} \right) \label{Eng-CC2},\\
F^{\mu\nu\rho\sigma} &=\f{1}{\sqrt{|g^4|}}\f{f}{\sqrt{4!}} \varepsilon^{\mu\nu\rho\sigma}\label{Eng-flux-background1},\\
F^{m\mu QR} &=0. \label{flux-mixed}
\end{align}
which are the conditions \eqref{CC1}, \eqref{CC2} and \eqref{flux-background} with an extra condition \eqref{flux-mixed} on the fluxes.
The condition that the internal flux should vanish is, however, relaxed, i.e., one allows $F^{mnpq}  \neq 0$. 
In fact, one postulates instead 
\be
F_{mnpq}=\lambda S_{[npq,m]} = \lambda S_{npq;m}. \label{flux-internal}
\ee
where $S_{npq}$ is one of the (metric compatible) torsions that parallelizes the internal space, i.e., such that when added
to the Levi-Civita torsion-free connection, the resulting connection is 
flat (zero curvature)\footnote{The relation between the torsion coefficients $f\indices{^a_{bc}}$ in the equation \eqref{Torsion} and $S_{mnp}$ in \eqref{flux-internal} is given by $f\indices{^a_{bc}}= e^a_m S\indices{^m_{np}}E^n_b E^p_c$.}.   In (\ref{flux-internal}),   $\lambda$ is a constant parameter.

As Cartan and Schouten showed \cite{Cartan1, Cartan2} (see also \cite{Englert:1982vs}), the torsions satisfy the following identities:
\begin{align}
S\indices{^{tr}_m} S_{trn} &=a\, g_{mn}, \label{Torsion101}\\
S\indices{_{tm}^r} S\indices{_{rn}^s} S\indices{_{sp}^t} &=\f{1}{2}a\, S_{mnp}, \label{Torsion102}\\
S_{npq;m} = S_{t[np}S\indices{^t_{q]m}} &= S_{[npq,m]}.  \label{Torsion103}
\end{align}

In order to prove these identities, recall the equation \eqref{R=R+Dk+kk} with $\kappa\indices{^a_b}=- e^a_{m} S\indices{^m_n} E^n_b$ where $S\indices{^m_n}=S\indices{^m_{np}}dx^p$. Now, given that for a flat connection $R\indices{^a_b}=0$, if one rewrites \eqref{R=R+Dk+kk} in the coordinates then the Riemann curvature is written as
\be\label{Riemann_with_S}
\bar{R}_{mnpq}= S_{rmn}S\indices{^r_{pq}} - S_{mnp;q},
\ee
where $S_{mnp;q}= \bar{\nabla}_{q} S_{mnp}$. Using the symmetries of curvature tensor and Bianchi identities, with a bit of work, we can show \eqref{Torsion103}. To prove the equations \eqref{Torsion101} and \eqref{Torsion102}, we can substitute the left hand side of \eqref{Riemann_with_S} with \eqref{CC1} and multiply both sides with $g^{mp}$ and $S\indices{^{mp}_r}$ to get \eqref{Torsion101} and \eqref{Torsion102} respectively. It is noted that the torsion tensor is totally antisymmetric, i.e., $S_{mnp}=S_{[mnp]}$.    

The parallelisms are respectively related to the octonion and split octonion algebras.  Because these algebras are non-associative, there are in each case two infinite classes, denoted $+$ and $-$, which are related to left and right multiplications.  We refer to \cite{Rooman:1984} for useful information concerning  the $S^{7,0}$ sphere, easily extendable to $S^{3,4}$. The corresponding torsions can be expressed in terms of the structure constants of the octonion or split-octonion algebras. The invariance group (at any given point) of these $3$-forms structure constants is isomorphic to $G_2$ or $G^*_{2,2}$ in the Riemannian 
or pseudo-Riemannian case respectively (see Appendix \ref{App:pseudo}). Once again, the fact that there are only two possible internal manifolds $S^{7,0}$ or $S^{3,4}$, and that $S^{6,1}$, $S^{5,2}$ etc are excluded, is connected with the fact that there are only two real forms of the octonion algebra, and two real forms of the Lie algebra $\mathfrak{g}_2$.

Furthermore, one finds from (\ref{Torsion101})-(\ref{Torsion103}) that the
torsion $3$-forms are eigenfunctions of the (pseudo-)Laplace operator in seven dimensions,  
\be
S\indices{_{npq;m}^{;m}}=-\f{2}{3} a \, S_{npq}. \label{Laplace}
\ee
One can also prove that the torsions are dual or anti-dual to their curvature \cite{Englert:1982vs},
\begin{align}
S^{mnp} &= \pm \sqrt{\frac{6}{a}}\frac{1}{4!}\frac{1}{\sqrt{|g^7|}}\varepsilon^{mnpqrst}S_{[rst,q]},\\
S_{mnp} &= \pm \sqrt{\frac{6}{a}}\frac{1}{4!}\sqrt{|g^7|}\varepsilon_{mnpqrst}S^{[rst,q]}.
\end{align}
where the $\pm$ sign depends on the $\pm$-class to which the torsion belongs.

With the above assumptions, the equation \eqref{EOM:var-wt-3form} for $F_{MNPQ}$ reduces to
\be\label{EOM-for-F-Englert}
F\indices{_{mnpq}^{;m}}= \f{(-1)^{T'}}{\eta} \f{f}{2(12)^{3/2}}\sqrt{|g^{7}|}\,\varepsilon_{rstunpq}F^{rstu},
\ee
 and is therefore solved in view of (\ref{flux-internal}) and (\ref{Laplace}) if we impose
\be\label{alpha-for-Englert-sol}
a = -\f{3}{4}\f{\eta}{24}f^2 (-1)^{T'} = \f{1}{32}f^2.
\ee
Given $a$, the two choices $f=\pm |f|= \pm \sqrt{32 a}$  correspond to the $+$ or $-$ parallelism, respectively.

Now, the equation \eqref{EOM:var-wt-metric} is equivalent to 
\be
R_{MN} = -\f{\eta}{48}\left(\f23 g_{MN}F^2-8 F_{MPQR}F\indices{_N^{PQR}}\right),\label{EOM:var-wt-metric2}
\ee
Using the form of the fields obtained so far, the $(mn)$-components of \eqref{EOM:var-wt-metric2} read
\be
a = \frac49 a + \frac{1}{24}\frac{10}{9} \eta \lambda^2 a^2,
\ee
yielding
\be
\lambda^2 = \eta \frac{12 }{a}.
\ee
Since $a>0$, this equation has no (real) solution when $\eta = -1$, i.e., for cases 2 and 4 above.   When $\eta = 1$, this equation determines $\lambda$ in terms of $a$,
\be
\lambda = \pm \sqrt{\frac{12}{a}}.
\ee
The $(\mu \nu)$-components of \eqref{EOM:var-wt-metric2} determine then the radius of curvature $b$ of the four-dimensional background spacetime,
\be
b = - \frac53 a,
\ee
which is thus negative.  The four-dimensional background spacetime is therefore anti-de Sitter space $AdS_4 \equiv H^{3,1}$.

\subsection*{Summary}
We thus see that there are only two cases among the cases listed above that can actually be realized (and the cases obtained by reversing the signature). These are \cite{Henneaux:2017afd}
\begin{enumerate}
\item The Englert solution $ H^{3,1} \times S^{7,0}$ in $M$-theory.
\item The ``exotic" Englert solution  $H^{3,1} \times S^{3,4}$  in $M'$-theory.
\end{enumerate}
For all these solutions $\lambda^2=\f{12}{a}$ and $b=-\f{5}{3} a$, in consistency with \cite{Englert:1982vs}.

The parallelizable seven-sphere or seven-pseudo-sphere with internal torsion (fluxes) turned on is therefore a possible solution only in $M$-theory and $M'$-theory. $M^*$-theory does not accommodate such solutions.  Furthermore, the four-dimensional background spacetime is in both cases anti-de Sitter, with a single time direction.  There is no solution with two time directions in the four-dimensional background spacetime. 

\subsection{Breaking of Symmetry}

All Freund-Rubin solutions listed table \eqref{tab:FRsolutions} are maximally supersymmetric. However, the Englert solution of $M$-theory not only breaks supersymmetry but also breaks the $SO(8)$-symmetry of the seven-sphere down to $Spin(7)$ \cite{Englert:1983jc}. This is what we explained in Section \ref{sec:Lorentz-flat-7d} that the due to presence of 
torsion not all $28$ Killing vectors of $SO(8)$ isometry group of $S^7$ are preserved. As we discussed,
in the presence of torsion the $SO(4,4)$ isometry group of $S^{3,4}$ breaks down to $Spin(3,4)$.

A similar situation occurs in $M'$-theory. The internal fluxes of the ``exotic" Englert solution break the $SO(4,4)$ symmetry of the pseudo-sphere $S^{3,4}$.  This is because the $4$-form $F_{mnpq}$ at any given point is not invariant under the full isotropy subgroup $SO(3,4)$.  In fact, the $4$-form field strength $F_{mnpq}$ is only invariant under the subgroup $G^*_{2,2}$ of $SO(3,4)$. The reason is that the torsion $S_{mnp}$ is determined by the structure constants of the algebra of split octonions which is invariant under $G^*_{2,2}$. The group $G^*_{2,2}$ is the non compact group with the trivial center corresponding to the split real form of the complex Lie algebra $\mathfrak{g}_2$ and is the automorphism group of split octonions.  Now, the pseudo-sphere is also equal to $Spin^+(3,4)/G^*_{2,2}$, where the group $Spin^+(3,4)$ is the connected component of $Spin(3,4)$ (see \cite{Kath:1996gm} and Appendix \ref{App:pseudo}).  This implies that the symmetry group of the $M'$-theory ``exotic" Englert solution is $Spin(3,4)$.  This solution, similarly to the Englert solution of $M$-theory, breaks the supersymmetry.

\newpage
\thispagestyle{empty}
\mbox{}

\part{BRST Cohomology}\label{part2}




\chapter{Electric-Magnetic Duality}\label{ch:EM_duality}

\section{Electric-magnetic duality in Maxwell theory}

It has long been known that the vacuum Maxwell equations
\begin{align}
\nabla . \textbf{E} &= 0,\qquad \nabla . \textbf{B}=0,\\
\nabla \times \textbf{E} &= -\partial_t \textbf{B},\qquad \nabla \times \textbf{B} = \partial_t \textbf{E},
\end{align}
are invariant under the $SO(2)$ transformations
\be
 \begin{pmatrix}
 \textbf{E} \\ \textbf{B}
 \end{pmatrix}
\rightarrow 
 \begin{pmatrix}
 \cos{\alpha} & -\sin{\alpha}\\ \sin{\alpha} & \cos{\alpha}
 \end{pmatrix}
 \begin{pmatrix}
 \textbf{E} \\ \textbf{B}
 \end{pmatrix}
\ee
in the two dimensional plane $(\textbf{E},\textbf{B})$ of electric and magnetic fields. From the covariant form of vacuum Maxwell equations
\be
\partial_{\mu} F^{\mu\nu}=0,\qquad \partial_{\mu} \star\!F^{\mu\nu}=0,
\ee 
one can see that the $SO(2)$ transformations rotate the Maxwell field strength $F^{\mu\nu}$ and its dual $\star F^{\mu\nu}$ to each other,
\begin{align}\label{U1-duality-trans}
    F^{\mu\nu} &= \cos{\alpha}\, F^{\mu\nu} -\sin{\alpha}\, \star\!F^{\mu\nu},\\
    \star F^{\mu\nu} &= \sin{\alpha}\, F^{\mu\nu} +\cos{\alpha}\, \star\!F^{\mu\nu},
\end{align}
hence they are called electric-magnetic duality transformations. Even though it is not the concern of our thesis but it is worth to mention that this duality can also be generalized to realize as the symmetry of Maxwell equations in the presence of sources. In that case, one requires to turn on not only an electric source $j^\mu$ but also a magnetic source $k^\mu$,
\be
\partial_{\mu} F^{\mu\nu}=j^\nu,\qquad \partial_{\mu} \star\!F^{\mu\nu}=k^\nu,
\ee
such that these sources transform under duality symmetry as
\begin{align}
    j^{\mu} &= \cos{\alpha}\, j^{\mu} -\sin{\alpha}\, k^{\mu},\\
    k^{\mu} &= \sin{\alpha}\, j^{\mu} +\cos{\alpha}\, k^{\mu}.\label{eq:EM-dual_trans_current}
\end{align}
However, one has to be careful since $\partial_{\mu} \star\!F^{\mu\nu}\neq 0$, one cannot solve the homogeneous wave equation in order to find the electromagnetic field $A_{\mu}$ everywhere. In other words, only if the magnetic source is spatially localized, then we are able to find $A_{\mu}$ where $k^\mu$ vanishes. In the rest of this thesis we focus on the free theory without sources.

Even though we found the electric magnetic duality as the symmetry of the equation of motion, it is important to notice that the electric-magnetic duality transformations are in fact an off-shell symmetry which leaves the action invariant. The action for a free Maxwell vector field in four dimensions
\be
S=-\f{1}{4}\int d^4x F^{\mu\nu}F_{\mu\nu},
\ee
is invariant under the duality transformations, although not manifestly \cite{Deser:1976iy,Deser:1981fr,Bunster:2011aw}. This means that in order to realize the invariance of this action under (infinitesimal) duality transformations, the electromagnetic field $A_k$ needs to transform (infinitesimally) non-locally in space (but local in time) such that transformations of $E^k$ and $B^k$ on-shell coincide with the transformations \eqref{U1-duality-trans} written infinitesimally.

However, if one considers the first-order Maxwell action, which puts electric and magnetic potentials on the same footing and that is called the ``first-order action", one loses the manifest Lorentz invariance but instead the duality symmetry of the action can be realized manifestly. In order to do so, one goes to the Hamiltonian formulation by introducing the conjugate momenta $\pi^i$ (conjugate to $A_i$) as independent variables in the Hamiltonian action. Then one solves Gauss's law, $\partial_i \pi^i=0$, for new variables $\pi^i$ introducing a second vector potential $Z_i$, $\pi^i=-\epsilon^{ijk}\partial_j Z_k$, and at the same time getting rid of $A_0$ component of the electromagnetic field which appears as Lagrange multiplier in the Hamiltonian action. Putting back the new fields in the Hamiltonian action one obtains
\be
S[A^M_i] = \f{1}{2}\int dx^0\,d^3x\, (\epsilon_{MN}\, \textbf{B}^M . \dot{\textbf{A}}\!^N - \delta_{MN}\, \textbf{B}^M . \textbf{B}^N), \qquad M,N=1,2
\ee
where $\textbf{A}^M= (\textbf{A},\textbf{Z})$,  $\textbf{B}^M = \nabla \times \textbf{A}^M$ and $\epsilon_{MN}$ and $\delta_{MN}$ are the Levi-Civita tensor and the Kronecker delta respectively. Now, it is obvious that this action is invariant under the duality transformations $\textbf{A}\rightarrow R_2 \textbf{A}$, with $R_2$ the $SO(2)$ transformations, since both $\epsilon_{MN}$ and $\delta_{MN}$ are $SO(2)$ invariant tensors, i.e. $R_2^T \epsilon R_2 =\epsilon$ and $R_2^T I R_2 = I$.

As we saw the first order action is manifestly duality invariant at the expense of losing manifest Lorentz invariance, while in the manifestly Lorentz invariant second order action the duality symmetry is not locally manifest in the sense that part of duality symmetry acts non-locally (in space) on the fields. Later, we will see that for example in supergravity there is a subgroup of the duality symmetry group that acts locally on fields in the second order action. 

Now, for $n_v$ free abelian Maxwell vector fields $\textbf{A}\!^I$, $I=1,...,n_v$, the duality invariant first order action takes the form
\be\label{1st-order-action-free}
S[A^M_i] = \f{1}{2}\int dx^0\,d^3x\, (\sigma_{MN}\, \textbf{B}^M . \dot{\textbf{A}}\!^N - \delta_{MN}\, \textbf{B}^M . \textbf{B}^N), \qquad M,N=1,2,...,2n_v
\ee
where now $\sigma_{MN}$ is the antisymmetric canonical symplectic form 
and $\textbf{A}^M=(\textbf{A}^I,\textbf{Z}^I)$. Therefore, in order to keep both the kinetic term and the Hamiltonian invariant under duality transformations, the duality group needs to be a subgroup of $GL(2n_v,\mathbb{R})$ of linear transformations of potentials $A^M$ which preserve both the symplectic form and the scalar product. It turns out that this group is $U(n_v)$, the maximal compact subgroup of the group of symplectic transformations $Sp(2n_v,\mathbb{R})$ in $2n_v$ dimensions \cite{Ferrara:1976iq,Gaillard:1981rj,deWit:2001pz,Bunster:2010wv}. 


\section{Non-compact duality group}\label{sec:noncompact_duality}
The duality group of Maxwell theory of $n_v$ free abelian vector field is the unitary group $U(n)$ which is compact. The non-compact duality symmetry group is possible only in the presence of non-minimally coupled scalar fields where the non-linear transformations of scalars play an important role. Let's consider a general case of $n_v$ abelian electromagnetic fields coupled to $n_s$ scalar fields,
\begin{align}
\cL &= \cL_S + \cL_V,\label{eq:lag-V+S}\\
\cL_S &= -\frac{1}{2} g_{ij}(\phi) \d_\mu \phi^i \d^\mu \phi^j,\label{eq:lag-S}\\
\cL_V &= - \frac{1}{4} \mathcal{I}_{IJ}(\phi) F^I_{\mu\nu} F^{J\mu\nu} + \frac{1}{8} \mathcal{R}_{IJ}(\phi)\, \varepsilon^{\mu\nu\rho\sigma} F^I_{\mu\nu} F^J_{\rho\sigma},\label{eq:lag}
\end{align}
with $ F^I_{\mu\nu} = \partial_\mu A^I_\nu -  \partial_\nu A^I_\mu $ and the matrices $\cI$ and $\cR$ give the non-minimal couplings between the scalars and the abelian vectors. 

In order to realize a group $G$ as the symmetry of this action, scalar fields have to transform in a specific way both to compensate for the transformations of vector fields under the symmetry group and to keep the scalar Lagrangian invariant under $G$. This requires that $g_{ij}$ and matrices $\cI$ and $\cR$ transform in a specific form. This can be achieved if for example scalar fields parameterize the coset $G/H$ where $H$ is the maximal compact subgroup of $G$. This is the case in the ungauged supergravity theories. In fact,
neglecting gravity, the Lagrangian \eqref{eq:lag-V+S} is the generic bosonic sector of ungauged supergravity. We discuss the symmetries of this system in detail later in Sections \ref{sec:embedding} and \ref{sec:second order}.

We want to write the first-order Lagrangian obtained from \eqref{eq:lag-V+S}-\eqref{eq:lag} following the discussion of previous section and to find the duality symmetry group of this system. In what comes next in the thesis we only consider theories in four spacetime dimensions unless otherwise explicitly stated.  

We start by writing the canonical momenta conjugate to the $A^I$ which are given by
\begin{equation}
\pi^i_I = \frac{\d \mathcal{L}}{\d \dot{A}^I_i} = \mathcal{I}_{IJ}\, F\indices{^J_0^i} -\frac{1}{2} \mathcal{R}_{IJ}\,\varepsilon^{ijk}F^J_{jk}, 
\end{equation}
along with the constraint $\pi^0_I=0$. This relation can be inverted to get
\begin{equation}
\dot{A}^{Ii} = (\mathcal{I}^{-1})^{IJ} \,\pi^i_J + \partial^i A^I_0 + \frac{1}{2} (\mathcal{I}^{-1} \mathcal{R})\indices{^I_J} \,\varepsilon^{ijk} F^J_{jk} ,
\end{equation}
from which we can compute the first-order Hamiltonian action
\begin{equation}
S_H =  \int\!d^4x\, \left( \pi^i_I \dot{A}^I_i - \mathcal{H} -  A^I_0 \, \mathcal{G}_I \right),
\end{equation}
where
\begin{align}
\mathcal{H} &= \frac{1}{2} (\mathcal{I}^{-1})^{IJ} \pi^i_I \pi_{J i} + \frac{1}{4} (\mathcal{I} + \mathcal{R}\mathcal{I}^{-1}\mathcal{R})_{IJ} F^I_{ij} F^{Jij} + \frac{1}{2} (\mathcal{I}^{-1} \mathcal{R})\indices{^I_J} \,\varepsilon^{ijk} \pi_{Ii} F^J_{jk} \\
\mathcal{G}_I &= - \partial_i \pi^i_I .
\end{align}
The time components $A^I_0$ appear in the action as Lagrange multipliers for the constraints
\begin{equation}
\partial_i \pi^i_I = 0 .
\end{equation}
These constraints can be solved by introducing new (dual) potentials $Z_{I i}$ through the equation
\begin{equation}
\pi^i_I = -  \varepsilon^{ijk} \partial_j Z_{I k},
\end{equation}
which determines $Z_{I}$ up to a gauge transformation $Z_{I i} \rightarrow Z_{I i} + \partial_i \tilde{\epsilon}_I$. Note that the introduction of these potentials is non-local but permitted in (flat) contractible space. 
Putting this back in the action gives
\begin{equation} \label{eq:symlag}
S = \frac12 \int \!d^4x \left(  \Omega_{MN} \mathcal{B}^{Mi} \dot{\mathcal{A}}^N_i -   \mathcal{M}_{MN}(\phi) \mathcal{B}^M_i \mathcal{B}^{Ni} \right),
\end{equation}
where the doubled potentials are packed into a vector
\begin{equation}
(\mathcal{A}^M) = \begin{pmatrix} A^I \\ Z_I \end{pmatrix}, \quad M= 1, \dots, 2n_v,
\end{equation}
and their curls  $\mathcal{B}^{Mi}$ are
\begin{equation}
\mathcal{B}^{Mi} =  \varepsilon^{ijk} \partial_j \mathcal{A}^M_k .
\end{equation}
The matrices $\Omega$ and $\mathcal{M}(\phi)$ are the $2n_v \times 2n_v$ matrices
\begin{equation} \label{eq:OMdef}
\Omega = \begin{pmatrix}
0 & I \\ -I & 0
\end{pmatrix}, \qquad
\mathcal{M} = \begin{pmatrix}
\mathcal{I} + \mathcal{R}\mathcal{I}^{-1}\mathcal{R} & - \mathcal{R} \mathcal{I}^{-1} \\
- \mathcal{I}^{-1} \mathcal{R} & \mathcal{I}^{-1}
\end{pmatrix},
\end{equation}
each block being $n_v\times n_v$.  The dual vector potentials $Z_I$ are canonically conjugate to the magnetic fields $B^I$.  

The kinetic term for the potentials in the first-order action (\ref{eq:symlag}) can be rewritten 
\be
\frac12\int dt \, d^3x \, d^3y \, \sigma_{M N}^{ij}(\vec{x}, \vec{y}) \,  \mathcal{A}^{M}_i(\vec{x}) \, \dot{\mathcal{A}}^{N}_j(\vec{y}), \label{eq:presymp} \ee
where $\sigma_{M N}^{ij}(\vec{x}, \vec{y})$ defines a (pre-)symplectic structure in the infinite-dimensional space of the $\mathcal{A}^M_i(\vec{x})$ and is antisymmetric for the simultaneous exchange of $(M,i,\vec{x})$ with $(N,j,\vec{y})$,
\be
\sigma_{M N}^{ij}(\vec{x}, \vec{y}) = - \sigma_{N M}^{ji}(\vec{y}, \vec{x}).
\ee
Explicitly, $\sigma_{M N}^{ij}(\vec{x}, \vec{y})$ is given by the product 
\be
\sigma_{M N}^{ij}(\vec{x}, \vec{y})=  -\Omega_{MN} \varepsilon^{ijk} \partial_k \delta(\vec{x} - \vec{y})
\ee
The internal part $\Omega_{MN}$ is antisymmetric under the exchange of $M$ with $N$, while the spatial part $\varepsilon^{ijk} \partial_k \delta(\vec{x} -\vec{y})$ is symmetric under the simultaneous exchange of $(i,\vec{x})$ with $(j,\vec{y})$. A necessary
condition for a linear transformation in the internal space of the potentials to be a symmetry of the action is that it should leave
 the kinetic term invariant, which in turn implies that $\Omega_{MN}$ must remain invariant under these transformations, i.e., they must belong to the symplectic group $Sp(2n_v, \mathbb{R})$.

At this point, it is worth mentioning that the situation will be completely different if one considers a system of $2$-forms interacting with scalars in six spacetime dimensions. 

In that case, for $2$-forms $\mathcal{A}^M_{ij}(\vec{x})=  -\mathcal{A}^M_{ji}(\vec{x})$ in six spacetime dimensions, the expression analogous to (\ref{eq:presymp}) is
 \be
\frac12\int dt \, d^5x \, d^5y \, \sigma_{M N}^{ij \, pq}(\vec{x}, \vec{y}) \,  \mathcal{A}^{M}_{ij}(\vec{x}) \, \dot{\mathcal{A}}^{N}_{pq}(\vec{y}), \label{eq:presymp2} \ee
where the (pre-)symplectic form  $\sigma_{M N}^{ij \, pq}(\vec{x}, \vec{y})$ is given by the product of a matrix in internal space and a matrix involving the spatial indices (including ($\vec{x}$)),
\be
\sigma_{M N}^{ij \, pq}(\vec{x}, \vec{y})=  -\eta_{MN} \varepsilon^{ijpqk} \partial_k \delta(\vec{x}- \vec{y}).
\ee
As a (pre-)symplectic form,  $\sigma_{M N}^{ij \, pq}(\vec{x}, \vec{y})$ must be antisymmetric for the simultaneous exchange of $(M,ij,\vec{x})$ with $(N,pq,\vec{y})$.  But now, the spatial part $\varepsilon^{ijpqk} \partial_k \delta(\vec{x} - \vec{y})$ is {\it antisymmetric},
$$ \varepsilon^{ijpqk} \partial_k \delta(\vec{x}- \vec{y}) = - \varepsilon^{pqijk} \partial_k \delta(\vec{y}- \vec{x}) $$ 
and therefore the internal part $\eta_{MN}$ must be {\it symmetric}.  And indeed, this is what one finds \cite{Henneaux:1988gg,Bekaert:1998yp}.  The signature of $\eta_{MN}$ is $(n_v,n_v)$ and its diagonalization amounts to splitting the $2$-forms into chiral (self-dual) and anti-chiral (anti-self-dual) parts. A necessary condition for a linear transformation in the internal space of the potentials  
to be a symmetry of the action is again that it should leave the kinetic term invariant, but this implies now that $\eta_{MN}$ must be invariant under these transformations, i.e., they must belong to the orthogonal group $O(n_v,n_v)$.

In general, for a $p$-form in $D= 2p + 2$ spacetime dimensions, the electric and magnetic fields are exterior forms of the same rank $p+1$ and one may consider electric-magnetic duality transformations that mix them while leaving the theory invariant. The above analysis proceeds in exactly the same way in higher (even) dimensions.  The (pre-)symplectic form $\sigma_{M N}^{i_1 i_2 \cdots i_p \, j_1 j_2 \cdots j_p}(\vec{x}, \vec{y})$ for the $p$-form potentials $\mathcal{A}^M_{i_1i_2 \cdots i_p}(\vec{x})$ is equal to the product $-\rho_{MN} \varepsilon^{i_1 i_2 \cdots i_p  j_1 j_2 \cdots j_p k} \partial_k \delta(\vec{x}- \vec{y}) $.  The spatial part $ \varepsilon^{i_1 i_2 \cdots i_p  j_1 j_2 \cdots j_p k} \partial_k \delta(\vec{x}- \vec{y}) $ is symmetric in spacetime dimensions $D$ equal to $4m$ ($p$ odd) and antisymmetric in spacetime dimensions $D$ equal to $4m + 2$ ($p$ even).  Accordingly the matrix $\rho_{MN}$ in electric-magnetic internal space must be antisymmetric for $D=4m$ ($p= 2m - 1$) and symmetric for $D=4m+2$ ($p=2m$), implying that the relevant electric-magnetic group is a subgroup of $Sp(2n_v, \mathbb{R})$  or of $O(n_v,n_v)$ in $D=4m$ and $D=4m+2$ dimensions respectively. These results were established in \cite{Deser:1997mz,Deser:1997se,Julia:1997cy,Julia:1980gr} which were further developed in \cite{Cremmer:1997ct,Bremer:1997qb,Deser:1998vc,Julia:2005wg}. As it was pointed out in \cite{Julia:1997cy}, they do not depend on the spacetime signature.

As we have just discussed in $4m$ spacetime dimensions scalars parameterize a coset $Sp(2n_v,\mathbb{R})/U(n_v)$. If there is a smaller set of scalars in the theory, then the scalars parametrize a coset $G/H$ where now $G \subset Sp(2n_v,\mathbb{R})$ is the duality symmetry group of the action and $H$ is its maximal compact subgroup. In the case of maximal $\mathcal{N}=8$ supergrvaity in four dimensions, the duality group is $E_{7,7}$ which is a subgroup of $Sp(56,\mathbb{R})$ and the $70$ scalars are parameterizing the coset $E_{7,7}/SU(8)$.

\section{A toy model inspired by $\mathcal{N}=4$ supergravity in four dimensions}\label{sec:EM-dual-N4_sug}
Here we consider as an example a model that can be easily extended to capture the gravity multiplet of the bosonic sector of $\mathcal{N}=4$ supergravity in four dimensions. Consider the Lagrangian \eqref{eq:lag-V+S} for a system constitutes of one vector field and two scalars $(\phi,\chi)$ (in fact one scalar $\phi$ and one pseudo-scalar $\chi$), i.e. $n_v=1$ and $n_s=2$, with the scalar and vector Lagrangians
\begin{align}
\mathcal{L}_S &=-\f{1}{2}\partial_{\mu}\phi\, \partial^{\mu}\phi -\f{1}{2} e^{2\phi}\partial_{\mu}\chi\, \partial^{\mu} \chi,\\
\mathcal{L}_V &=- \frac{1}{4} e^{-\phi} F_{\mu\nu} F^{\mu\nu} + \frac{1}{8} \chi\, \varepsilon^{\mu\nu\rho\sigma} F_{\mu\nu} F_{\rho\sigma}.
\end{align}
The full Lagrangian is invariant under
\begin{enumerate}[I)]
    \item constant shift of $\chi$ ($\chi \rightarrow \chi + c$),
    
    \item simultaneous transformations 
    \begin{align}
          A_\mu &\rightarrow \lambda A_\mu,\\
          \phi &\rightarrow \phi + 2 \log \lambda,\\
          \chi &\rightarrow \lambda^{-2} \chi.
    \end{align}
\end{enumerate}

The non-linear scalar transformations are those of the coset $SL(2,\mathbb{R})/SO(2)$ generated by the upper-triangular matrices, $B^+$, in the matrix representation of $SL(2,\mathbb{R})$,
\be\label{eq:mat_rep_sl2}
 X_+ =\begin{pmatrix} 0 & 1 \\ 0 & 0 \end{pmatrix},\qquad X_0 =\begin{pmatrix} 1 & 0 \\ 0 & -1 \end{pmatrix}.
\ee
The scalar Lagrangian is also invariant under another transformation
\be
\delta_- \phi =2 \chi \epsilon, \qquad \delta_- \chi = (e^{-2\phi}-\chi^2) \epsilon,
\ee
which acts locally on scalars and enlarges the group of transformations to full $SL(2,\mathbb{R})$. However, this transformation acts non-locally on vector fields in the second-order formalism. In order to see that $SL(2,\mathbb{R})$ is in fact the group of electric-magnetic duality of this model, one can consider the first-order form of the action. The first-order action is written as \eqref{eq:symlag} with $M,N=1,2$ and
\begin{align}
   \Omega_{MN}=\sigma_{MN}=\begin{pmatrix} 0 & 1 \\ -1 & 0 \end{pmatrix}, \qquad  
   \mathcal{M}_{MN}=\begin{pmatrix} \chi^2 e^\phi+ e^{-\phi} & -\chi e^\phi \\ -\chi e^\phi & e^\phi \end{pmatrix}.
\end{align}
Now, all $SL(2,\mathbb{R})$ transformations act linearly on vector fields as

\be 
\delta A^M = \varepsilon^\alpha (X_\alpha)\indices{^M_N} A^N
\ee
and $\mathcal{M}$ transforms as
\be
\delta \mathcal{M}= - (\varepsilon X)^T \mathcal{M} - \mathcal{M} (\varepsilon X),
\ee
where $X_{\alpha}$ are generators of $SL(2,\mathbb{R})$, \eqref{eq:mat_rep_sl2} together with
\be
 X_- =\begin{pmatrix} 0 & 0 \\ 1 & 0 \end{pmatrix}.
\ee
Therefore, in the first-order formulation the group $SL(2,\mathbb{R})$ is manifest as the group of electric-magnetic duality of the system. 

The $SO(6)$ (or $SU(4)$) invariant-formulation \cite{Cremmer:1979up,Cremmer:1977zt,Cremmer:1977tt,Chamseddine:1980cp} of $\mathcal{N}=4$ supergravity in four dimensions is obtained by considering $n_v=6$ vector fields in the gravity multiplet of theory, once again enlarging the duality group to $SO(6)\times SL(2,\mathbb{R})$. We will come back to this point later when we discuss examples of gauging of scalar-vector models in Section \ref{sec:N=4}. 

\chapter{Equivalent description in terms of the Embedding tensor formalism}\label{ch:Isom-Emb-VS-models}

\section{Gauged Supergravity}\label{sec:gauged_sugra}
As we discussed before in Section \ref{sec:duality} the $M$-theory, or the eleven dimensional supergravity as its low energy limit, is unique (up to T-duality invariant theories) with supersymmetry fixing its field content to that of the gravity multiplet \cite{Cremmer:1978km}. The toroidal compactification of the $M$-theories or superstring theories results in the ungauged supergravity theories in lower dimensions in flat backgrounds. The ungauged supergravities, however, do not appear to fit phenomenology at least at the classical perturbative level. Also, since the moduli space of vacua of scalars are not determined dynamically, the appearance of scalars in the interaction terms in the Lagrangian destroys the predictive power of the theory. For example, the maximal supergravity in four dimensions contains $70$ scalars part of which can be interpreted as the moduli of the internal manifold $T^7$ and the other part comes from the $3$-form field of the eleven dimensional supergravity during the compactification \cite{Cremmer:1977zt,Cremmer:1978ds}. There is no scalar potential in the classical theory which could fix the values of these scalars on the vacuum. One attempt to overcome this issue is to gauge the ungauged theories by $28$ vector fields of gravity supermultiplet through minimally coupling of the vector fields and other fields. In this way, one can obtain an $SO(8)$ gauged supergravity in four dimensions \cite{deWit:1982bul}. This in turn introduces a scalar potential in the Lagrangian which fixes the problem of moduli degeneracy already at the classical level. 

Moreover, if the $M$-theory is the correct UV-completion of the effective quantum field theories in four dimensions then the standard model of fundamental interactions needs to be realized as a solution of the supergravity theory. In order to describe the standard model of particle physics one needs ultimately to gauge the ungauged supergravities obtaining a low energy effective action describing the known interactions between the gauge fields of standard model. This is an ongoing program which despite promising progress has not yet been successful. It should be however mentioned that the most interesting supergravity theories are in Anti-de Sitter background. As of today, a supergravity theory with a non-supersymmetric de Sitter (quasi Minkowski) stable vacuum, a residual gauge symmetry $SU(3)\times SU(2) \times U(1)$, and with the matter content of the standard model has not been discovered.

One can build different gauged supergravities by gauging different subgroups of global symmetry of the ungauged theories by introducing the minimal coupling with the vector fields. In many cases, the non-perturbative strong coupling limit of these gauged supergravities is described by superstring/$M$-theories. The compactification of the superstring/$M$-theories on some non-toroidal internal manifolds leads to finding the gauged supergravities in lower dimensions, however, it is not always possible to uplift an effective lower dimensional gauged supergravity theory to its UV complete superstring/$M$-theory.

In order to obtain gauged supergravities from compactification, one has to turn on internal form fluxes, namely  non-vanishing vacuum expectation values  
of field strengths. These fluxes will appear in the  Einstein  equation through
their  energy-momentum  tensors and  may  affect  in  a  non-trivial way the geometry of the background, including the internal manifold.
They change the structure and topology of the internal manifolds.

In  the  presence  of  fluxes  one  is  then  led  to  consider  spontaneous  compactifications of superstring/$M$-theory on more general backgrounds which have the form of a
warped product $M_D \times M_{int}$ of a non-compact $D$-dimensional spacetime and an internal manifold which in general is no longer Ricci-flat, at least in the sense of Ricci flatness of Levi-Civita connection, see Chapters \ref{ch:Lorentz_flatness} and \ref{ch:exotic}. The problem of finding backgrounds of this kind which preserve a minimal amount of supersymmetries requires considering more general geometries for
internal manifold and hence leads to the so called $G$-structure manifolds. This class of
internal manifolds comprises geometries characterized by background quantities, called
geometric fluxes, which define topology-deformation of traditional Ricci-flat
manifolds like tori or Calabi-Yau spaces. 

In Chapter \ref{ch:Lorentz_flatness} we discussed the $G$-structure manifolds where we focused on manifolds with $G_2$-structure. In the $M$-theory compactification, the global existence of Killing spinors restricts the structure group allowing the internal manifold with $G_2$-structure. In Chapter \ref{ch:exotic} we explained two ways of compactification in the presence of internal fluxes; Freund-Rubin compactification which preserves all supersymmetries and Englert spontaneous compactification which prompts a non-supersymmetric low energy effective theory. Also, we have seen that in the case of Englert(-like) compactification, the torsion on the -parallelizable- internal manifold $S^7$ ($S^{3,4}$) plays the role of geometric fluxes.  

The first examples of gauged extended supergravities \cite{Freedman:1976aw,Freedman:1978ra,Fradkin:1976xz,Zachos:1978iw,Zachos:1979uh} exhibited some of the
characteristic features of these theories such as fermion masses, spontaneous supersymmetry breaking and a scalar potential, or a cosmological constant when scalars were not
present. 

The gauged $\mathcal{N}=8$ maximal supergravity with $SO(8)$ gauge group is obtained from Kaluza-Klein reduction of eleven dimensional supergravity on round $S^7$ where the SO(8) local symmetry is related to the  isometry  group  of $S^7$ \cite{deWit:1983vq,deWit:1986oxb,deWit:2013ija}. This was later generalized to maximal gauged supergravities with non-compact gauge group in \cite{Hull:1984yy,Hull:1984vg,Hull:1984qz}.

As we pointed out before, at the classical level supergravities in lower dimensions are consistently defined independently of their string or $M$-theory origin and are characterized by the amount of supersymmetry which they preserve, their field content and the local internal symmetry group gauged by the vector fields (for the gauged supergravities). They are also invariant under $U(1)^{n_v}$ abelian gauge symmetry. At the classical level, gauged supergravities are obtained from ungauged ones, with the
same amount of supersymmetry, field content and $U(1)^{n_v}$ abelian gauge symmetry, through the well defined
gauging procedure. It involves promoting a suitable subgroup of the global symmetry group
to local symmetry of the theory and modifying the action by the introduction of
fermion mass terms and of a scalar potential in order to restore the same number of
supersymmetries as the original ungauged theory.  This uniquely defines the gauged supergravity, provided a set of consistency conditions to be satisfied by the gauge group.
The gauging procedure is the only known way for introducing fermion mass terms or a scalar potential in an extended
supergravity without explicitly breaking supersymmetry.  In the $M$-theory compactification, these additional terms in the action as well as the gauge group
are determined by the background fluxes and therefore by the structure of internal manifold. The background fluxes typically induce, in the lower dimensional effective theory,
minimal couplings as well as mass terms  and  a  scalar  potential.

It is important to emphasize again that not all gaugings can be obtained from compactification of supergravity theories. There are known examples of gauged supergravities in lower dimensions
for which it is not possible to find an uplift to the $M$-theory \cite{Shelton:2005cf,Dabholkar:2002sy,Kachru:2002sk}. Even in those cases where the gauged supergravities can be obtained from the dimensional reduction, it is not always an easy task to find a consistent reduction ansatz. Therefore, it is tempting to seek a formulation that can provide all possible gaugings in a systematic way regardless of whether it can or cannot be obtained through dimensional reduction. 

The BV-BRST formulation \cite{Batalin:1981jr,Batalin:1983wj,Batalin:1984jr,Batalin:1984ss,Batalin:1985qj}, which is described in Chapter \ref{ch:BV_def}, will provide a powerful tool in order to deal with the question of finding all possible consistent gaugings of supergravity theories. This is the path we pursue in this thesis. 

\section{Electric-Magnetic duality and Gauging}
 
As we have seen in Chapter \ref{ch:EM_duality} in the second-order Lagrangian formulation of supergravity, only the electric vector fields appear in the Lagrangian. The minimal coupling then, involving  only  the  electric  vector  fields, manifestly breaks the electric-magnetic duality symmetry of the original ungauged theory.  

The electric-magnetic duality is therefore particularly relevant in the context of gaugings of supergravities in their second-order formulation since the available gaugings depend on the 
Darboux frame in which the theory is formulated. The Darboux frame is also known as the ``duality frame" or the ``symplectic frame". The use of symplectic frame is referring to the fact that the real symplectic group preserves the symplectic tensor, written in Darboux form, on the $2n_v$ electric and magnetic field directions. We should distinguish the action of the full symplectic group acting on the frame from the ``duality" symmetry subgroup $G$ that acts also on the scalar fields and may be strictly smaller.  

To be specific, the second-order Lagrangian, describing the maximal supergravity toroidally reduced to four spacetime dimensions \cite{Cremmer:1979up}, is invariant under the non-compact rigid symmetry $E_{7,7}$ (a subgroup of the symplectic group $Sp(56,\mathbb{R})$) which contains standard electric-magnetic duality rotations among the electric and magnetic fields. The invariance of the equations of motion under duality symmetry group was already demonstrated in \cite{Cremmer:1979up}. However, as we explained in Section \ref{sec:noncompact_duality}, the invariance of the action is more easily verified in the first-order formalism.  This was explicitly achieved in \cite{Hillmann:2009zf}, following the crucial pioneering work of \cite{Deser:1976iy} further developed in  \cite{Henneaux:1988gg,Deser:1997mz}.

Now, while the full $E_{7,7}$ group is always present as an off-shell symmetry group of the Lagrangian, only a subgroup of which acts locally on all fields in the second-order formulation.  The other transformations act non-locally (in space).  The explicit non-local form of the duality transformations for standard electromagnetism may be found in \cite{Deser:1981fr,Bunster:2011aw}. 

The subgroup of  symmetries that acts locally on the fields of the second-order formalism depends on the  symplectic frame and is called the ``electric symmetry subgroup" in that symplectic frame. A subgroup of this electric symmetry subgroup is the seed for the local deformations of the abelian gauge symmetry and of the action of the ungauged theory usually considered in the literature\footnote{We should point out an abuse of the word ``ungauged" in the jargon of gauged supergravity. The original ungauged theory is really already a gauge theory, since it has abelian vector fields, each with its own $U(1)$ gauge invariance.  The terminology ``gauging" might for that reason not be very precise, since one is actually not gauging an initial theory without gauge invariance, but deforming the abelian gauge transformations of an original theory that is already gauged. In the process, one may charge not only the vector fields but also the matter fields through minimal coupling. For that reason, it might be more appropriate to use the terminology ``charging" or ``deforming" though we adopt the terminology ``gauging" commonly used in the literature.}.
For instance, the original gaugings of maximal supergravity were performed in \cite{deWit:1981sst} starting from the Lagrangian formulation of \cite{Cremmer:1979up} with a specific choice of 
28 ``electric" vector potentials out of the 56 potentials that one may introduce to render dualities manifest. The electric subgroup is $SL(8, \mathbb{R})$ and the subgroup into which the abelian gauge symmetry is deformed is $SO(8)$. Changing the duality frame, as in \cite{Andrianopoli:2002mf,Hull:2002cv}, opens new inequivalent possibilities, the most recent examples of which were proposed in \cite{DallAgata:2011aa,DallAgata:2012mfj,DallAgata:2012plb}. A systematic investigation of the space of inequivalent deformations of the Yang-Mills type has been undertaken in \cite{DallAgata:2014tph,Inverso:2015viq}.

Due to the dependence of the available deformations on the duality frame, it is of interest to keep the duality frame unspecified in the gaugings in order to be able to contemplate simultaneously all possibilities.  There are at least three different ways to do so.  One of them is simply to try to deform directly the manifestly duality invariant action of \cite{Hillmann:2009zf}.  In that first-order formulation, all duality transformations are local and no choice of duality frame is needed. However, as we show in Section \ref{sec:NO_YM_def}, generalizing the argument of \cite{Bunster:2010wv} by introduction of scalar fields one finds that 
all standard non-abelian Yang-Mills deformations are obstructed for this Lagrangian, so that one needs to turn to a second-order Lagrangian in order to have access to non-abelian 
deformations including those of the Yang-Mills type. As the change of variables from the second-order formulation to the first-order one is non-local (in space), 
the concept of local deformations changes. 

The transition from a Hamiltonian formulation to the usual second-order formulation involves a choice of symplectic frame. Explicitly, one needs to determine what are the $q$'s and $p$'s  where a point on a symplectic manifold is given by $(q,p)$. We always refer to those potentials which are kept, during the transition from first-order to second-order formulation, as $q$'s and to those which are eliminated as $p$'s. Since the electric and magnetic potentials of the first-order formulation play the respective roles of coordinates and momenta, this amounts to choosing which potentials are electric and which ones are magnetic. It is the reason why this step breaks manifest duality.  One can nevertheless keep track of the symplectic transformation that relates the chosen symplectic frame to a fixed reference symplectic frame.  The symplectic matrix coefficients appear then as constant parameters that define the second-order Lagrangian.   So, one is really dealing with a family of Lagrangians depending on constant parameters.   In order to treat all symplectic frames on the same footing, one should study the gaugings in a uniform manner independently of the values of these constants -- although the number and features of the solutions will depend on their actual values.   There are moreover redundancies, because many choices of the symplectic matrix coefficients, i.e.  of symplectic frame, lead to locally equivalent second-order formulations. We will discuss this approach in Section \ref{sec:transition}. 

A third way to keep track of the symplectic frame is through the ``embedding tensor" formalism \cite{deWit:2002vt,deWit:2005ub,deWit:2007kvg,Samtleben:2008pe,Trigiante:2016mnt},  where both electric and magnetic vector potentials are introduced from the beginning without choosing a priori which one is electric and which one is magnetic. Furthermore, extra $2$-forms are added in order to prevent doubling of the number of physical degrees of freedom. In addition to these fields, the Lagrangian contains also constants which form the components of the embedding tensor $\Theta$ and indicate how the subgroup to be gauged is embedded in the duality group $G$. These constants are not to be varied in the action so that, up to some equivalences, one has a family of Lagrangians depending on external constants and any choice destroys the manifest symmetry between electricity and magnetism. In Section \ref{sec:Emb_Tensor} we review the important features of this formalism and the gauging procedure using the embedding tensors.

The last two formulations are classically equivalent since their corresponding second-order Lagrangians yield equivalent equations of motion \cite{deWit:2002vt,deWit:2005ub,deWit:2007kvg,Samtleben:2008pe,Trigiante:2016mnt}.  One may, however, wonder whether the spaces of available local deformations are isomorphic. In Section \ref{sec:BV-def-Embed} we prove that this is indeed the case and the space of local deformations of both pictures are isomorphic.

\section{Embedding Tensor}\label{sec:Emb_Tensor}

In the embedding tensor formulation of the gauging construction \cite{deWit:2002vt,deWit:2005ub,deWit:2007kvg,Samtleben:2008pe,Trigiante:2016mnt}, all the deformations of 
the original ungauged model are expressed in terms of a single embedding tensor $\Theta$ which is covariant with respect to the global symmetries of
the theory and defines the embedding of the gauge algebra into the  global symmetry algebra. The choice of gauge group inside the global symmetry group is constrained
by consistency conditions which depend on the original ungauged model. They hint to the point that the dimension of gauge group cannot exceed the number $n_v$ of electric 
vector fields. The consistency conditions appear in the form of a set of linear and quadratic $G$-covariant (purely group-theoretical) constraints on the components of embedding tensor $\Theta$.

We shall review a general formulation of the gauging procedure in four dimensions which was developed in \cite{deWit:2005ub,deWit:2007kvg} and does not depend on the symplectic matrices $E$. The matrices $E$ relate the chosen symplectic frame to a fixed reference symplectic frame, so that the kinetic terms are not written in terms of the electric vector fields.

We start by introducing a symplectic-invariant gauge connection of the form
\begin{equation}
 \Omega_g = \Omega_{g\mu} dx^{\mu},\qquad \Omega_{g\mu}\equiv g A^M_\mu\,X_M=g (A^I_\mu\,X_I+\tilde{A}_{I\mu}\,\tilde{X}^I) = g A^M_\mu\,\Theta_M{}^\alpha\,t_\alpha\,,\label{newcon}
\end{equation}
where $A^M= (A^I, \tilde{A}_I)$ are the symplectic vector fields and $X_M = (X_I, \tilde{X}^I)$ are the generators of the gauge group $G_g$. Here the index $I=1,...,n_v$ and one should notice that even though there are $2n_v$ generators $X_M$ but there are at most $n_v$ independent ones which generate the gauge group. The embedding tensors describe the embedding of the gauge group in the global symmetry group $G$ as
\be
X_M = \Theta\indices{_M^\alpha} t_{\alpha},
\ee
with $t_\alpha$ generators of the algebra of global symmetry transformation, therefore $\alpha=1,...,dim(G)$. One can, for the sake of simplicity, reabsorb the gauge coupling constant $g$ into $\Theta$. 

There is a redundancy in the choice of the dynamical fields between $A^I_\mu$ and $\tilde{A}_{I\mu}$; half of them play the role of auxiliary fields and at most $n_v$ linear combinations $A^{\check{I}}_\mu$ of the $2n_v$ vectors $A^I_\mu$ and $\tilde{A}_{I\mu}$ effectively enter the gauge connection
 \begin{equation}
 A^M_\mu\,X_M=A^{\check{I}}_\mu\,X_{\check{I}}\,,\label{ATHcomb}
 \end{equation}
where the index $\check{I}=1,...,n_v$ is referring to the $n_v$ physical vector fields\footnote{This could be easily seen bearing in mind that there is a symplectic matrix $E$ which takes the frame to the electric one where $\tilde{X}^{\check{I}}$ vanishes.}. 

The consistency of gauging requires that the embedding tensor satisfies a set of linear and quadratic algebraic $G$-covariant constraints \cite{Trigiante:2016mnt}. The constraints are written as
   \begin{align}
   X_{(MNP)}&=0\,,\label{linear2}\\
   \Omega^{MN}\Theta_M{}^\alpha\Theta_N{}^\beta &=0\,,\label{quadratic1}\\
   [X_M,\,X_N] &=-X_{MN}{}^P\,X_P\,,\label{quadratic2}
   \end{align}
where $X\indices{_{MN}^P}=\Theta\indices{_M^\alpha}(t_{\alpha})\indices{_N^P}$ and $X_{MNP}=X\indices{_{MN}^Q} \Omega_{QP}$.

The linear constraint (\ref{linear2}) is also known as the representation constraint 
\cite{Trigiante:2016mnt}. 
The first quadratic constraint \eqref{quadratic1} guarantees that there exists a symplectic matrix $E$  which rotates the embedding tensor $\Theta_M{}^\alpha$ to an electric frame in which the magnetic components $\Theta^{\hat{I}\,\alpha}$ vanish. 

The second quadratic constraint \eqref{quadratic2} is the condition that the gauge algebra close within the global symmetry algebra $\mathfrak{g}$ and implies that $\Theta$ is a singlet with respect to $G_g$. To see this one can rewrite \eqref{quadratic2} as
\be
\Theta\indices{_M^\alpha}\Theta\indices{_N^\beta} f_{\alpha\beta}^\gamma + \Theta\indices{_M^\alpha} (t_{\alpha})\indices{_N^P} \Theta\indices{_P^\gamma}=0,
\ee
where $f_{\alpha\beta}^\gamma$ are generators of the adjoint representation of the global symmetry group $G$. 

The locality constraint (\ref{quadratic1}) is not in general independent of other constraints. As it was proved in \cite{Trigiante:2016mnt}, it will pose as an independent constraint only in theories featuring scalar isometries with no duality action, namely those theories in which the symplectic duality representation $\mathcal{R}_v$ of the isometry algebra $\mathfrak{g}$ is not faithful. This is the case of the quaternionic isometries in $\mathcal{N}=2$ theories. In general, in theories in which all scalar fields sit in the same supermultiplets as the vector fields, as it is the case of supergravities with $\mathcal{N}>2$ or $\mathcal{N}=2$ with no hypermultiplets, the locality condition (\ref{quadratic1}) is not independent but follows from the other constraints.


Following the known prescription of gauging a theory through minimal coupling, the covariant derivatives are then defined in terms of (\ref{newcon}) and replaces ordinary derivatives everywhere in the action; therefore the gauge fields $A^M$ are transformed under the gauge symmetry as
\begin{equation}
\delta A^M_\mu=D_\mu\zeta^M\equiv \partial_\mu\zeta^M+\,A^N_\mu X_{NP}{}^M\,\zeta^P\,,\label{deltaA}
\end{equation}
where $X_{MP}{}^R\equiv {\mathcal{R}}_{v*}[X_{M}]_{P}{}^R$. Here, $\mathcal{R}_{v*}$ is referring to the symplectic matrix representation of the generators $X_M$ acting on covariant vectors. 
The symplectic covariant non-Abelian field strengths $\mathcal{F}^M$ is then defined as
 \begin{equation}
\mathcal{F}^M\equiv dA^M+\frac{1}{2}\,X_{NP}{}^M\,A^N\wedge A^P\,.\label{FMdef}
 \end{equation}

Although in the present formulation both vectors $A^I$ and $\tilde{A}_I$ appear in the kinetic term but in a symplectic frame which is not the electric one, $A^I$ fields are not well defined since the Bianchi identity (Gauss constraint) is not satisfied for their corresponding curvature. This is due to the non-vanishing components $\Theta^{I\alpha}$ of the embedding tensor in such a frame; these are nothing but magnetic charges. It can be indeed seen from
\begin{equation}
D{\mathcal{F}}^M\equiv d{\mathcal{F}}^M+\,X_{NP}{}^M\,A^N\wedge {\mathcal{F}}^P=\,X_{(PQ)}{}^M\,A^P\wedge\left( dA^Q+\frac{1}{3}\,X_{RS}{}^QA^R\wedge A^S\right),\label{Bianchifail}
\end{equation}
where in particular one finds $D \mathcal{F}^I\neq 0$ since $X_{(PQ)}{}^I=-\frac{1}{2}\,\Theta^{I\alpha}\,(t_{\alpha})\indices{_P^N}\Omega_{NQ}$ vanishes only in the electric frame  where $\Theta^{I\alpha}= 0$. This can be also interpreted as the failure to satisfy Jacobi identity for the structure constants $X_{MP}{}^R$ of gauge algebra in
the ${\mathcal{R}}_{v*}$-representation \cite{Trigiante:2016mnt}
\begin{equation}
X_{[MP]}{}^R X_{[NR]}{}^Q+X_{[PN]}{}^R X_{[MR]}{}^Q+X_{[NM]}{}^R X_{[PR]}{}^Q=-(X_{NM}{}^R\,X_{(PR)}{}^Q)_{[MNP]}\,.\label{nojacobi}
\end{equation}

This problem can be resolved by adding extra $2$-form gauge fields $B_{\alpha\mu\nu}$ and a set of gauge symmetries to preserve the number of degrees of freedom. Then by combining the non-abelian vector field strengths $\mathcal{F}^M_{\mu\nu}$ with a set of massless antisymmetric tensor fields $B_{\alpha\,\mu\nu}$ one is able to construct gauge-covariant field strength $\mathcal{H}^M_{\mu\nu}$ as
\begin{align}
\mathcal{H}^M_{\mu\nu}\equiv \mathcal{F}^M_{\mu\nu}+ \mathcal{Y}^{M\,\alpha}\,B_{\alpha\,\mu\nu}\;:\;\;
\begin{cases}
\mathcal{H}^I=\mathcal{F}^I+\frac{1}{2}\,\Theta^{I\alpha}\,B_\alpha\,,\cr
\tilde{\mathcal{H}}_I=\tilde{\mathcal{F}}_I-\frac{1}{2}\,\Theta_{I}{}^{\alpha}\,B_\alpha\,.
\end{cases}
\label{HZB}
\end{align}
where
\begin{equation}
\mathcal{Y}^{M\,\alpha}\equiv \frac{1}{2}\,\Omega^{MN}\,\Theta_N{}^\alpha\,.\label{defZ}
\end{equation}

In order for the field strength $\mathcal{H}^M_{\mu\nu}$ to be gauge covariant, the $2$-form field $B_{\alpha\mu\nu}$ should transform under the gauge symmetry  
such that it cancels the non-covariant terms in the variation of $\mathcal{F}^M_{\mu\nu}$. We have thus introduced, together with the new fields $\tilde{A}_{I\mu}$ and $B_{\alpha\,\mu\nu}$, extra gauge symmetries parametrized by $\tilde{\zeta}_I$ and $\Xi_{\alpha\,\mu}$. This ensures the correct number of propagating degrees of freedom. The fields are transformed under these symmetries as
\begin{align}
    \delta A^M &= D \zeta^M - \mathcal{Y}^{M\alpha}\, \Xi_{\alpha},\\
    \delta B_{\alpha} &= D\, \Xi_{\alpha} + t_{\alpha NP} A^N \wedge \mathcal{Y}^{M\alpha}\, \Xi_{\alpha} - t_{\alpha NP} (2 \zeta^N \mathcal{H}^P - A^N \wedge D \zeta^N),
\end{align}
where the covariant derivatives are defined as
\begin{align}
    D \zeta^M &= d \zeta^M + X\indices{_{NP}^M} A^N \wedge \zeta^P,\\
    D\, \Xi_{\alpha} &= d\, \Xi_{\alpha} + \Theta\indices{_M^\beta} f\indices{_{\beta\alpha}^\gamma} A^M \wedge \Xi_{\gamma}.
\end{align}
Finally the vector-tensor part of the gauged bosonic Lagrangian invariant under the above symmetries is written as \cite{deWit:2002vt,deWit:2005ub},
\begin{align}
\mathcal{L}_{VT} = -\frac{1}{4} \, {\cal I}_{IJ}\,\mathcal{H}_{\mu\nu}^{I}
\mathcal{H}^{J\mu\nu} +\frac{1}{8} {\cal
R}_{IJ}\;\varepsilon^{\mu\nu\rho\sigma}
\mathcal{H}_{\mu\nu}^{I}
\mathcal{H}_{\rho\sigma}^{J}+\mathcal{L}_{top,\,B}+\mathcal{L}_{GCS}\,,\label{boslag2}
\end{align}
where 
\begin{align}
\mathcal{L}_{top,\,B} &= -\frac{1}{8}g\, \varepsilon^{\mu\nu\rho\sigma}\,
\Theta^{I\alpha}\,B_{\alpha\mu\nu} \, \Big(
2\,\partial_{\rho} \tilde{A}_{I\sigma} + X_{MN\,I}
\,A_\rho^M A_\sigma^N
-\frac{1}{4}g\,\Theta_{I}{}^{\beta}B_{\beta\rho\sigma}
\Big)\,,\label{topB}\\
\mathcal{L}_{GCS} &= -\frac{1}{3}g\,
\varepsilon^{\mu\nu\rho\sigma}X_{MN\,I}\, A_{\mu}^{M}
A_{\nu}^{N} \Big(\partial_{\rho} A_{\sigma}^{I}
+\frac{1}{4}g  X_{PQ}{}^{I}
A_{\rho}^{P}A_{\sigma}^{Q}\Big)
\nonumber\\[.9ex]
&{} -\frac{1}{6}g\,
\varepsilon^{\mu\nu\rho\sigma}X_{MN}{}^{I}\, A_{\mu}^{M}
A_{\nu}^{N} \Big(\partial_{\rho} \tilde{A}_{I\sigma}
+\frac{1}{4}g\, X_{PQI}
A_{\rho}^{P}A_{\sigma}^{Q}\Big)\,,\label{GCS}
\end{align}
where we restored the coupling $g$ and wrote them down explicitly.

The $\mathcal{L}_{top,\,B}$ is the topological term which is added in order to guarantee the invariance under the Peccei-Quinn transformations.
Notice that if the magnetic charges $\Theta^{I\alpha}$ vanish, which will happen in the electric frame, $B_\alpha$ disappears from the action since \eqref{topB} vanishes as well as the $B$-dependent Stueckelberg term in $\mathcal{H}^I$.

One can add the Einstein-Hilbert and the scalar Lagrangians to this action to obtain the complete bosonic action of gauged supergravity in four dimensions and $\mathcal{N}>2$. One can do the same analysis in the presence of fermions since the only constraint required by supersymmetry is the linear constraint \eqref{linear2}. In fact, the consistency of gauging in the bosonic Lagrangian already imposes this condition and one just needs to covariantize the fermionic sector, add the mass shift term and introduce the scalar potential in order to restore the supersymmetry of the original ungauged theory. 

After this short review of embedding tensor formalism, it might be good to look back at the second way of gauging, mentioned in Section \ref{sec:gauged_sugra}, where one uses the action \eqref{eq:lag-V+S}-\eqref{eq:lag} with $n_v$ vector fields. Then, one can formulate the theory in any arbitrary frame by means of a symplectic transformation which relates the arbitrary frame to the reference frame. But first, we want to show that the first-order formulation is rigid even in the presence of scalars.

\section{Yang-Mills deformations of first-order action}\label{sec:NO_YM_def}
One may think that it is more tractable to consider the gauging procedure in the first-order formulation since the duality symmetry is manifest in this formulation.
However, it was shown in \cite{Bunster:2010wv} that the first-order action (\ref{eq:LV1}) without scalar fields,
\be
S[A_i^M] = \frac12\int dt\, d^3x \left(\Omega_{MN} B^{Mi} \dot{A}^N_i -  \delta_{MN} B^M_i B^{Ni}\right) \label{eq:LV1bis}
\ee
does not admit non-abelian deformations of the Yang-Mills type.

For the future reference, it is good to make the abelian gauge invariance of the action \eqref{eq:LV1bis} manifest by introducing the temporal component $A^M_0$ of the potentials and rewrite the action (\ref{eq:LV1bis}) in terms of the abelian curvatures $F_{\mu \nu}^M = \partial_\mu  A^M_\nu - \partial_\nu  A^M_\mu$.  One can replace $\dot{A}^N_i$ by $F_{0i}^M$ in (\ref{eq:LV1bis}) because the magnetic fields are identically transverse and one finds
\be
S[A_{\mu}^M] = \frac14\int dt\, d^3x \left(\Omega_{MN} \epsilon^{ijk} F^M_{ij} F^N_{0k} -  \delta_{MN} F^M_{ij} F^{Nij}\right). \label{eq:LV2bis}
\ee

Then, due to the result of \cite{Bunster:2010wv}, there is no subgroup of the duality group $U(n_v)$ that can be deformed into a non-abelian one.  A much stronger result was actually derived earlier in \cite{Bekaert:2001wa} through the BRST formalism, namely, that the action (\ref{eq:LV1bis}) admits no local deformation that deforms the gauge algebra at all.

In fact, this obstruction as described in \cite{Bunster:2010wv} does not depend on the scalar sector and obstructs Yang-Mills deformations even when scalar fields are present. The obstruction prevails due to the incompatibility of an adjoint action, required by the Yang-Mills construction, and the symplectic condition as required by the invariance of the scalar-independent kinetic term. It was explained in \cite{Bunster:2010wv}, that the obstruction persists to remain even in the presence of scalar fields. This is because even in the presence of scalars, the kinetic term of vector fields is independent of scalars. Following \cite{Henneaux:2017kbx}, we give a detailed proof of the rigidity of Yang-Mills deformation in the presence of scalar fields. 

Let's consider the Yang-Mills deformation 
\be
\delta A^M_i = \partial_i \zeta^M + g C\indices{^M_{NP}} A^N_i \zeta^P 
\ee
of the original abelian gauge symmetry of \eqref{eq:LV2bis},
\be
\delta A^M_i = \partial_i \zeta^M
\ee
with gauge parameters $\zeta^M$.   Here, the $C\indices{^M_{NP}}$ are the structure constants of the gauge group $G_g$ of dimension $2n_v$ into which the original abelian gauge group is deformed, and $g$ is the deformation parameter. Since the $C\indices{^M_{NP}}$ are the structure constants of the gauge group, they satisfy $C\indices{^M_{NP}}= - C\indices{^M_{PN}}$.

Under the Yang-Mills deformation, the abelian curvatures are replaced by the non-abelian ones,
\be
\mathcal{F}_{\mu \nu}^M = F_{\mu \nu}^M + g C\indices{^M_{NP}} A^N_\mu A^P_\nu,
\ee
which transform in the adjoint representation as
\be
\delta \mathcal{F}_{\mu \nu}^M = g C\indices{^M_{NP}} \mathcal{F}^N_{\mu \nu} \zeta^P, \label{eq:adjoint}
\ee
and the ordinary derivatives $\partial_\mu \phi^i$ of the scalar fields are replaced by the covariant derivatives $D_\mu \phi^i$.  These contain linearly the undifferentiated vector potentials $A^N_\mu$.  

Therefore the deformed action is written as
\begin{align}
S[\mathcal{F},\phi] &= \frac{1}{2} \int d^4 x \, g_{ij}(\phi) D_\mu \phi^i D^\mu \phi^j \nonumber\\
&+ \frac14 \int d^4 x \left(\Omega_{MN} \epsilon^{ijk}\mathcal{F}^M_{ij} \mathcal{F}^N_{0k} - \mathcal{M}_{MN}(\phi) \mathcal{F}^M_{ij} \mathcal{F}^{Nij} \right).
\end{align}
The important point is that the kinetic term for the vector potentials is the same as in the absence of scalar fields.  Invariance of the kinetic term under the Yang-Mills gauge transformations yields therefore the same conditions as when there is no scalar $\phi^i$.  The kinetic term must be invariant by itself since the scalar Lagrangian is invariant and there can be no compensation with the energy density $\mathcal{M}_{MN}(\phi) \mathcal{F}^M_{ij} \mathcal{F}^{Nij}$ since it does not contain any time derivative.

As in the case with no scalars \cite{Bunster:2010wv}, the compatibility conditions are that the adjoint representation (\ref{eq:adjoint}) of the gauge group $G_g$ should preserve the symplectic form $\Omega_{MN}$ in internal electric-magnetic space, i.e. writing $C\indices{^M_{NP}}\equiv (C_N)\indices{^M_P}$ then the matrices $C_N$ must be symplectic hence to satisfy $(C_N)^T \,\Omega \,C_N = \Omega$ or explicitly
\be
C_{NPM} = C_{MPN} \label{eq:InvSymp}
\ee
with $C_{MNP} \equiv \Omega_{MQ}C^Q_{\; \; NP}$. The symmetry of $C_{MNP}$ under the exchange of its first and last indices, and its antisymmetry in its last two indices, force it to vanish,
\be
C_{MNP}= - C_{MPN} = - C_{NPM} = C_{NMP} = C_{PMN} = - C_{PNM} = -C_{MNP}.
\ee
Hence, the structure constants $C^Q_{\; \; NP}$ must also vanish and there can be no non-abelian deformation of the gauge symmetries.  The Yang-Mills construction in which the potentials become non-abelian connections  is unavailable in the manifestly duality invariant first-order formulation \cite{Bunster:2010wv}. However, one should notice that the other types of deformations, charging the matter fields or adding functions of the gauge-invariant abelian curvatures, are not excluded by the argument since they preserve the abelian nature of the gauge group. 

In order to surpass the obstruction to the non-abelian deformations, one must avoid the condition (\ref{eq:InvSymp}) expressing the invariance of the symplectic form.  This is naturally attainable once one goes to the second-order formalism by choosing a Lagrangian submanifold, on which the pull-back of $\Omega_{MN}$ is by definition zero.  This is explained in the next section.

\section{Transition to the second-order formalism}\label{sec:transition}

While the first-order and second-order formalisms are equivalent in terms of symmetries in the sense that any symmetry in one formalism is also a symmetry of the other, they are not equivalent regarding the concept of locality, more specifically locality in space. A local function in one formulation may not be local in the other and consequently the spaces of consistent local deformations are not the same in both formalisms. This subtle difference may open a window of possibilities for the non-abelian deformations in second-order formalism which were not accessible in the first-order formalism. Below, we will shed light on the origin of this practically important difference.

\subsection{Choice of symplectic frame and locality restriction}\label{sec:symplecticchoice}

Let's consider the first-order Lagrangian discussed in Section \ref{sec:noncompact_duality} which is written as
\begin{align}
\mathcal{L} &= \mathcal{L}_S + \mathcal{L}_V,\\
\mathcal{L}_S &= \frac{1}{2} g_{ij}(\phi) \partial_\mu \phi^i \partial^\mu \phi^j,\\
\mathcal{L}_V &= \frac{1}{2} \Omega_{MN} B^{Mi} \dot{A}^N_i - \frac{1}{2} \mathcal{M}_{MN}(\phi) B^M_i B^{Ni}, \label{eq:LV1}
\end{align}
with $\Omega$ and $\mathcal{M}$ given by \eqref{eq:OMdef}.

Of special interest are the symmetries of  $\mathcal{L}$ for which the vector fields transform linearly,
\be 
\delta A^M_i(\vec{x}) = \epsilon^\alpha \left(t_\alpha\right)^M_{\; \; N} A^N_i(\vec{x}) \label{eq:trans1}
\ee
accompanied by a transformation of the scalars that may be nonlinear,
\be
\delta \phi^i = \epsilon^{\alpha} R_{\alpha}^{\; \; i}(\phi). \label{eq:trans2}
\ee
Here, the $\epsilon^\alpha$'s are infinitesimal parameters.  One might consider more general (nonlinear canonical) transformations of the vector fields, and the scalars might have their own independent symmetries, but the concern of the thesis is duality symmetries which are governed by the transformations of type \eqref{eq:trans1} and \eqref{eq:trans2}.
 
For the first-order action to be invariant, the transformation (\ref{eq:trans1}) should not only be symplectic, but also such that the accompanying transformation (\ref{eq:trans2}) of the scalar field is given by a Killing vector of the scalar metric $g_{ij}(\phi)$ (so as to leave the scalar Lagrangian $\cL_S$ invariant) which implies that
\be
\epsilon^{\alpha} \left(\frac{\partial  \mathcal{M}_{MN}}{\partial \phi^i} R_{\alpha}^{\; \; i} + \mathcal{M}_{MP} \left(t_\alpha\right)^P_{\; \; N}.
+ \mathcal{M}_{PN} \left(t_\alpha\right)^P_{\; \; M} \right)= 0 \label{eq:trans3}
\ee
It also guarantees to leave the energy density of the electromagnetic fields invariant.  In finite form, the symmetries are 
$A \rightarrow \bar{A}$, $\phi \rightarrow \bar{\phi}$ where the symplectic transformation
\begin{equation}
\bar{A} = S A, \label{eq:Trans1}
\end{equation}
 and the isometry $\phi \rightarrow \bar{\phi}$ of the scalar manifold are such that
\begin{equation}
\mathcal{M}(\bar{\phi}) = S^{-T} \mathcal{M}(\phi) S^{-1} .\label{eq:Trans3}
\end{equation}

The algebra of transformations of the type (\ref{eq:trans1}) and (\ref{eq:trans2}) that leave the action invariant is called the duality algebra $\mathfrak{g}$, while the corresponding group is the duality group $G$. 
As it was indicated in Section \ref{sec:noncompact_duality}, the duality algebra is a subalgebra of $\mathfrak{sp}(2n_v, \mathbb{R})$.  This symplectic condition holds irrespective of the scalar sector and its couplings to the vectors, since it comes from the invariance of the vector kinetic term which assumes always the same form.  Which transformations among those of $\mathfrak{sp}(2n_v, \mathbb{R})$ are actually symmetries depend, however, on the number of scalar fields, their internal manifold and their couplings to the vectors \cite{Ferrara:1976iq,Gaillard:1981rj,deWit:2001pz,Cremmer:1977zt,Cremmer:1977tt,Cremmer:1977tc}.

In order to go from the first-order formalism to the standard second-order formalism, one has to tell what are the ``$q$"'s (to be kept) and the ``$p$"'s (to be eliminated in favour of the velocities through the inverse Legendre transformation).  
We have actually made the choice to use Darboux coordinates in \eqref{eq:OMdef} and 
we call the corresponding frame a Darboux frame or symplectic frame.

Since we consider here only linear canonical transformations, a choice of canonical coordinates amounts to a choice of an element of the symplectic group $Sp(2n_v, \mathbb{R})$ relating that symplectic frame (i.e., symplectic basis) to  a reference symplectic frame.

Once a choice of symplectic frame has been made, one goes from (\ref{eq:LV1}) to the second-order formalism by following the  steps reverse to those that led to the first-order formalism:
\begin{itemize}
\item One keeps the first half of the vector potentials $A^M\vert_{M=1,...,n_v} = A^I$, the ``electric ones" in the new frame ($I = 1, \cdots, n_v$).
\item One replaces the second half of the vector potentials $A^M\vert_{M=n_v+1,...,2n_v}= Z_I$ (the ``magnetic ones") by the momenta
\be
\pi^i_I = - \varepsilon^{ijk} \partial_j Z_{I k}
\ee
subject to the constraints 
\begin{equation}
\partial_i \pi^i_I = 0 
\end{equation}
which one enforces by introducing the Lagrange multipliers $A_0^I$.  This is the non-local step (in space).
\item One finally eliminates the momenta $\pi^i_I$, which can be viewed as auxiliary fields, through their equations of motion which are nothing else but the inverse Legendre transformation expressing $\dot{A}^I_i$ in terms of $\pi^i_I$.
\end{itemize}

One can encode the choice of symplectic frame through the symplectic transformation that relates  the chosen symplectic frame to a  reference symplectic frame,
\begin{equation} \label{eq:Aprime}
A^M = E\indices{^M_N} A^{'N},
\end{equation}
where $E$ is a symplectic matrix in the fundamental representation of $ Sp(2n_v, \mathbb{R})$,  while $A$ and $A^{'}$ are respectively the potentials in the reference and new duality frames. The symplectic property $E^T \Omega E = \Omega$ ensures that the kinetic term in \eqref{eq:LV1} remains invariant. The first-order action therefore takes the same form, but with matrices $\mathcal{M}$ and $\mathcal{M}^{'}$  related by 
\begin{equation}\label{eq:Mprime}
\mathcal{M}^{'} = E^T \mathcal{M} E .
\end{equation}

Using the fact that $E$ and $E^T$ are symplectic, a straightforward but not very illuminating computation shows that $\mathcal{M}^{'}$ determines matrices $\mathcal{I}^{'}$, $\mathcal{R}^{'}$ uniquely such that equation \eqref{eq:OMdef} with primes holds,  i.e, there exist unique $\mathcal{I}^{'}$ and $\mathcal{R}^{'}$ such that
\begin{equation} \label{eq:OMdefBis}
\mathcal{M}^{'} = \begin{pmatrix}
\mathcal{I}^{'} + \mathcal{R}^{'}\mathcal{I}^{'-1}\mathcal{R}^{'} & - \mathcal{R}^{'} \mathcal{I}^{'-1} \\
- \mathcal{I}^{'-1} \mathcal{R}^{'} & \mathcal{I}^{'-1}
\end{pmatrix}.
\end{equation}
The matrices $\mathcal{I}^{'}$ and $\mathcal{R}^{'}$ depend on the scalar fields, but also on the symplectic matrix $E$.
Performing the steps described above in reverse order,  one then gets the second-order Lagrangian  
\begin{equation} \label{eq:lprime}
\cL^{'}_V = - \frac{1}{4} \mathcal{I}^{'}_{IJ}(\phi) F^{'I}_{\mu\nu} F^{'J\mu\nu} + \frac{1}{8} \mathcal{R}_{IJ}^{'}(\phi)\, \varepsilon^{\mu\nu\rho\sigma} F^{'I}_{\mu\nu} F^{'J}_{\rho\sigma} .
\end{equation}

The new Lagrangian \eqref{eq:lprime} depends on the parameters of the symplectic transformation $E$ used in equation \eqref{eq:Mprime}.  So we really have a family of Lagrangians labelled by an $Sp(2n_v, \mathbb{R})$ element $E$.  By construction, these Lagrangians differ from one another by a change of variables that is in general non-local in space. 

There are of course redundancies in this description. The presence of these redundancies will be important in the study of deformations. There are in general two sets of redundancies:

\begin{enumerate}

\item The stability subgroup $T$ of a Lagrangian subspace consists of the block lower triangular symplectic transformations.

\item The duality symmetries.
\end{enumerate}

First, the different symplectic transformations can lead to the same final $n_v$ dimensional Lagrangian subspace of $q$'s upon elimination of the momenta. The choice of $q$'s is equivalent to a choice of Lagrangian linear subspace, or linear polarization, because $q$'s form  a complete set of commuting variables (in the Poisson bracket). The same is true for a choice of $p$'s. The stability subgroup $T$ of a Lagrangian subspace consists of the block lower triangular symplectic transformations yielding 
\begin{itemize}
    \item The canonical \textit{point} transformation where the second-order Lagrangians differ by a redefinition of the $q$'s,
    \item The canonical \textit{phase} transformation where the second-order Lagrangians modified by a total derivative.
\end{itemize}

Second, duality symmetries correspond also to redundancies in the transition to the second-order formalism since they do not modify the first-order Lagrangian (when the scalars are transformed appropriately) and hence clearly lead to the same final second-order Lagrangian.

One can characterize the redundancies as follows.  The symplectic transformations \eqref{eq:Aprime} with matrices $E \in Sp(2 n_v, \mathbb{R})$ and $g E t$ are equivalent. Here, $t$ belongs to the stability subgroup $T$ of the final Lagrangian subspace and $g$ belongs to the duality group $G$. Indeed, the first-order Lagrangian is left invariant under the transformation  $A \rightarrow g^{-1} A$ provided the scalars are transformed appropriately, and two sets of new symplectic coordinates $A^{'}$ and $ t A^{'}$ yield equivalent second-order Lagrangians after elimination of the momenta. One could also consider what happens when the coordinates $A^{'}$ defined by \eqref{eq:Aprime} are taken as the reference frame. Then, the changes of frame $A^{'} = U A^{''}$ and $A^{'} = g^{'} U t^{'} A^{''}$ are equivalent, where $g^{'} = E^{-1} g E$ is the matrix associated with the duality symmetry $g$ in the $A^{'}$--frame and $t^{'}$ is an element of the stability subgroup of the Lagrangian subspace in the $A^{''}$--frame. This gives an equivalent description of the redundancies.

The relevant space is thus the quotient $G \backslash Sp(2n_v, \mathbb{R})/ T$.  Here we considered the invariance of action that means the invariances up to a total derivative which does not matter classically. However, if one just considers the invariance of the Lagrangian instead of the action then the stability subgroup reduces to $T = GL(n_v,\mathbb{R})$.

\subsection{Electric group}\label{sec:elec-group}

The duality transformations
$A \rightarrow \bar{A} = S A$, $\phi \rightarrow \bar{\phi}$, where the symplectic matrix $A$  and the isometry $\phi \rightarrow \bar{\phi}$ of the scalar manifold are such that the condition (\ref{eq:Trans3}) holds, 
generically mix the electric and magnetic potentials.  Accordingly, when expressed in terms of the variables of the second-order formalism they will, in general, take a non-local form since the magnetic potentials become non-local functions of the electric potentials\footnote{Given a duality symmetry of the action $S[A^M_k, \phi^i]$, one gets the corresponding duality symmetry of the first-order action $S[A^I_k, \pi_I^k, A_0^I,\phi^i]$ by (i) expressing in the variation $\delta A^I_k$ the magnetic potentials $Z_{Ik}$ (if they occur) in terms of $\pi_I^k$, which is a non-local expression and determined up to a gauge transformation that can be absorbed in a gauge transformation of the electric variables;  (ii) computing the variation $\delta \pi^k_I$ from  $\pi^i_I = - \varepsilon^{ijk} \partial_j Z_{I k}$; and (iii) determining the variation of the  Lagrange multipliers $A_0^I$ so that the terms proportional to the constraints $\partial_i \pi^i_I$ cancel in $\delta S[A^I_k, \pi_I^k, A_0^I,\phi^i]$.  One gets the symmetry of the second-order action by expressing the auxiliary fields $\pi^i_I$ that are eliminated in terms of the retained variables.} and their time derivatives in the second-order formalism \cite{Deser:1976iy,Deser:1981fr,Bunster:2011aw}.

Thus, an electric symmetry transformation is characterized by the property that the matrix $S_e$ is lower-triangular, i.e., given that it must be symplectic
\begin{equation} \label{eq:Ge0}
 S_e = \begin{pmatrix} M & 0 \\ BM & M^{-T} \end{pmatrix} , \; M \in GL(n_v), \; B^T = B.
\end{equation}
An electric symmetry is therefore a transformation $\bar{A}^I_\mu = M\indices{^I_J} A^J_\mu$ of the electric potentials for which there is a symmetric matrix $B$ and an isometry $\phi \rightarrow \bar{\phi}$ of the scalar manifold such that
\begin{equation} \label{eq:Ge}
\mathcal{M}(\bar{\phi}) = {S_e}^{-T} \mathcal{M}(\phi) {S_e}^{-1}
\end{equation}
with $S_e$ given by (\ref{eq:Ge0}).  The terminology ``electric group" for the group of transformations characterized by (\ref{eq:Ge0}) is standard.  The transformations with $B \not=0$ involve transformations of Peccei-Quinn type describing the axion shift symmetry \cite{Peccei:1977hh,Peccei:1977ur}. 

The electric group $G_e$ in a given frame depends of course on the chosen duality frame. Indeed, going to another duality frame with the symplectic matrix $E$ as in \eqref{eq:Aprime} will replace the matrices $S_e$ by their conjugates $S_e^{'} = E^{-1} S_e E$, and these might not be lower-triangular.
The condition that $S_e^{'}$ has the lower-triangular form \eqref{eq:Ge0} therefore depends on the choice of $E$.
 
We have repeatedly emphasized that a local expression in the $Z_I$'s need not be local when expressed in terms of the $\pi_I$'s and thus of the Lagrangian variables.  This is because the magnetic potentials are not local functions of the conjugate momenta.  But the $\pi_I$'s are local functions of the $Z_I$'s given by $\pi^i_I = - \varepsilon^{ijk} \partial_j Z_{Ik}$ and so it would seem that any local function of the $\pi_I$'s, and thus of the Lagrangian variables, should also be a local function of the $Z_I$'s.  
This is true as long as the relation $\pi^i_I = - \varepsilon^{ijk} \partial_j Z_{I k}$ is not modified.  In deformations of the second-order Lagrangians one allows, however, deformations that could modify this relation. For instance, after Yang-Mills deformations, the abelian constraints $\partial_i\pi^i_I = 0$ get replaced by non-abelian ones and their local solution gets replaced by a non-local expression involving both kinds of potentials.

The interest of the electric group $G_e$ is that its transformations are local in the second-order formalism.  Therefore, among the duality transformations, it is the electric ones that are candidates for (local) gaugings.

\section{Deformation leading to the embedding tensor formalism}\label{sec:embedding}

In Section \ref{sec:gauged_sugra} we explained that another way to cover arbitrary choices of symplectic frames is given by the embedding tensor formalism of \cite{deWit:2005ub,deWit:2007kvg,Samtleben:2008pe,Trigiante:2016mnt}. 
As we have seen in Section \ref{sec:Emb_Tensor} in this formalism in addition to the $n_v$ electric vector fields $A^I_\mu$, the Lagrangian also contains $n_v$ magnetic vector fields $\tilde{A}_{I\mu}$ and $\dim(G)$ two-forms $B_{\alpha \mu \nu}$ with $\alpha = 1, \dots, \dim(G)$, where $G$ is the full duality group in which one can also include pure scalar symmetries. The Lagrangian \eqref{boslag2} also depends on the embedding tensor $\Theta\indices{_M^\alpha}=(\Theta\indices{_I^\alpha},\Theta\indices{^{I\alpha}})$, which are constants in spacetime and not to be varied in the action. They contain information about the embedding of the gauge group in the duality group through the $\mathfrak{g}$-valued connection 
\begin{equation}
\label{Theta}
\mathcal{T}_\mu^\alpha = A_\mu^M \Theta_M^\alpha
\end{equation}
where $\mathfrak{g}$ is the Lie algebra of $G$ and we wrote $A^M_{\mu} = (A^I_\mu, \tilde{A}_{I\mu})$. $\tilde{A}_{Ik}$ are the magnetic potentials called $Z_{Ik}$ in Section \ref{sec:transition} but for more harmonious notation used below, we use $\tilde{A}_{IK}$ in the following. The embedding tensor also contains information about the choice of symplectic frame; an important feature upon which our analysis in the following based. The consistency of gaugings in this formalism controls by three sets of constraints \eqref{linear2}-\eqref{quadratic2}. The locality constraint \eqref{quadratic1} guarantees the existence of a symplectic matrix which takes a chosen symplectic frame to the electric one and hence implies that the gauged group must be a subgroup of an electric subgroup $G_e$ of $G$, i.e. $dim(G_g) = dim(\Theta\indices{_{\check{I}}^{\alpha}}) \leq n_v$. 
 Here $\check{I}=1,...,n_v$ is again referring to the $n_v$ physical vector fields.

The gauged Lagrangian \eqref{boslag2}-\eqref{GCS}, after restoring the coupling constant $g$, can be reorganized as 
\begin{align}
\mathcal{L}^{\Theta, \text{int.}}_{VT} (A,\tilde{A},B) = &- \frac{1}{4} \mathcal{I}_{IJ}(\phi) \mathcal{H}^I_{\mu\nu} \mathcal{H}^{J\mu\nu} + \frac{1}{8} \mathcal{R}_{IJ}(\phi)\, \varepsilon^{\mu\nu\rho\sigma} \mathcal{H}^I_{\mu\nu} \mathcal{H}^J_{\rho\sigma} \nonumber \\
&- \frac{g}{8} \varepsilon^{\mu\nu\rho\sigma} \Theta^{I\alpha}B_{\alpha \mu\nu} \left( \tilde{\mathcal{F}}_{I \rho\sigma} - \frac{g}{4} \Theta\indices{_I^\beta} B_{\beta \rho\sigma}\right)\nonumber \\
&+ g \mathcal{L}^\text{extra} (A,\tilde{A}, gB ) .
\label{eq:gauged-Lag-Emb_tens}
\end{align}
where
\begin{equation}
\mathcal{H}^I_{\mu\nu} = \mathcal{F}^I_{\mu\nu} + \frac{g}{2} \Theta^{I \alpha} B_{\alpha \mu\nu} 
\end{equation}
and $\mathcal{F}^I_{\mu\nu}$, $\tilde{\mathcal{F}}_{I\mu\nu}$ are the Yang-Mills curvatures, differing from the abelian ones by $O(g)$-terms.
The term $\mathcal{L}^\text{extra}$ in \eqref{eq:gauged-Lag-Emb_tens} are terms necessary to secure gauge invariance and will vanish in the limit explained below; their explicit 
form can be extracted from \eqref{boslag2}-\eqref{GCS}. The complete form of supersymmetric action in the presence of matter couplings can be found in \cite{deWit:2005ub,Trigiante:2016mnt}, but we just focus on the bosonic action with no matter coupling. Note that the full duality group is in general not faithfully represented.

The  gauged Lagrangian \eqref{eq:gauged-Lag-Emb_tens} is an interacting Lagrangian.
To view it as a deformation of an ungauged Lagrangian with the same field content, so as to phrase the gauging problem as a deformation problem, we observe that the straightforward free limit $g\rightarrow 0$ of \eqref{eq:gauged-Lag-Emb_tens} is just the Lagrangian \eqref{eq:lag}, without any additional fields and without the embedding tensor. However, one notices that all terms in \eqref{eq:gauged-Lag-Emb_tens} depend on the two-forms $B_\alpha$ only through the combination $\hat{B}_\alpha = g B_\alpha$. The $\hat{B}_\alpha$'s are necessary to shift away some field strengths even for the undeformed limit $g=0$. If we redefine the two-forms appearing in the Lagrangian \eqref{eq:gauged-Lag-Emb_tens} as $\hat{B}_\alpha = g B_\alpha$ and then take the limit $g \rightarrow 0$, one gets
\begin{align} \label{lagtwoforms}
\mathcal{L}^\Theta = &- \frac{1}{4} \mathcal{I}_{IJ}(\phi) H^I_{\mu\nu} H^{J\mu\nu} + \frac{1}{8} \mathcal{R}_{IJ}(\phi)\, \varepsilon^{\mu\nu\rho\sigma} H^I_{\mu\nu} H^J_{\rho\sigma}\nonumber\\
 &- \frac{1}{8} \varepsilon^{\mu\nu\rho\sigma} \Theta^{I\alpha}\hat{B}_{\alpha \mu\nu} \left( \tilde{F}_{I \rho\sigma} - \frac{1}{4} \Theta\indices{_I^\beta} \hat{B}_{\beta \rho\sigma}\right),
\end{align}
with
\begin{equation}
H^I_{\mu\nu} = F^I_{\mu\nu} + \frac{1}{2} \Theta^{I \alpha} \hat{B}_{\alpha \mu\nu} 
\end{equation}
and $F^I_{\mu\nu}$, $\tilde{F}_{I \mu\nu}$ the abelian curvatures.  This Lagrangian contains the extra fields and has 
the ability to cover generic symplectic frames through  the embedding tensor components $\Theta^{I\alpha}$ and $\Theta\indices{_I^\alpha}$. These tensors satisfy the so-called ``locality" quadratic constraint \eqref{quadratic1},
\begin{equation}
\Theta^{I [\alpha} \Theta\indices{_I^{\beta]}} = 0 . \label{eq:ThetaCons}
\end{equation}
In addition to (\ref{eq:ThetaCons}), the embedding tensor fulfills other constraints \eqref{linear2} and \eqref{quadratic2}, however, we stress that the locality constraint is the only condition relevant to our analysis. The other constraints have to appear as the solution of obstruction equations in the process of deformation. We will come back to this in Chapter \ref{ch:BRST-cohom-SV-Lagr} while we discuss BV-BRST deformation of such a model.

The Lagrangian \eqref{lagtwoforms} is invariant under the $2n_v$ gauge transformations
\begin{align}
\delta A^I_\mu &= \partial_\mu \lambda^I - \frac{1}{2} \Theta^{I \alpha} \hat{\Xi}_{\alpha \mu}, \\
\delta \tilde{A}_{I\mu} &= \partial_\mu \tilde{\lambda}_I + \frac{1}{2} \Theta\indices{_I^\alpha} \hat{\Xi}_{\alpha \mu},
\end{align}
and the $\dim G$ generalized gauge transformations
\begin{align}
\delta \hat{B}_{\alpha\mu\nu} &= 2 \partial_{[\mu}\hat{\Xi}_{\alpha \nu]},
\end{align}
where we redefined the gauge parameter as $\hat{\Xi}_\alpha = g \Xi_\alpha$ with respect to the gauge parameter of the corresponding symmetry of \eqref{boslag2}. The parameters $\lambda^I$, $\tilde{\lambda}_I$ and $\hat{\Xi}_{\alpha \nu}$ are arbitrary functions.  The $\lambda$'s  and $\tilde{\lambda}$'s correspond to standard $U(1)$ gauge symmetries of the associated one-form potentials.
The $\hat{\Xi}$'s define ordinary abelian two-form gauge symmetries and also appear as shift transformations of the one-form potentials. This set of gauge symmetries is reducible since adding a gradient to $\hat{\Xi}$ and shifting simultaneously the $\lambda^I$ and $\tilde{\lambda}_I$, i.e.
\begin{equation}
\hat{\Xi}_{\alpha \mu} \rightarrow \hat{\Xi}_{\alpha \mu} + \partial_\mu \xi_{\alpha}, \qquad \lambda^I \rightarrow \lambda^I + \frac{1}{2} \Theta^{I\alpha} \xi_{\alpha}, \qquad \tilde{\lambda}_I \rightarrow \tilde{\lambda}_I - \frac{1}{2} \Theta\indices{_I^\alpha} \xi_\alpha,
\end{equation}
leads to no modification of the gauge transformations.

By construction, the non-abelian Lagrangian of \eqref{boslag2} can be viewed as a consistent local deformation of (\ref{lagtwoforms}),  since one can charge the fields in (\ref{lagtwoforms}) in a smooth way -- by adding $O(g)$ and $O(g^2)$ terms -- to get the non-abelian Lagrangian.  The abelian Lagrangian (\ref{lagtwoforms}) is thus a sensible starting point for the deformation procedure.  The space of consistent local deformations of (\ref{lagtwoforms}) will necessarily include the non-abelian Lagrangian \eqref{boslag2}, given in \cite{deWit:2005ub}. If the embedding tensor components $\Theta^{I\alpha}$ and $\Theta\indices{_I^\alpha}$ do not fulfill the extra constraints \eqref{linear2} and \eqref{quadratic2} of the embedding formalism, one will simply find that there is no Yang-Mills type deformation, as in \cite{deWit:2005ub}, among the consistent deformations of (\ref{lagtwoforms}). 

\textit{A natural question then is whether the space of local deformations of (\ref{lagtwoforms}) is isomorphic to the space of local deformations of the conventional Lagrangian \eqref{eq:lprime} in an appropriate symplectic frame.  We will show that it is indeed the case.}  

To that end, we will now perform a sequence of field redefinitions and show that the Lagrangian (\ref{lagtwoforms}) with a given definite (arbitrary) choice of embedding tensor differs from the Lagrangian \eqref{eq:lprime} in a related symplectic frame by the presence of algebraic auxiliary fields that can be eliminated without changing the local BRST cohomology.  The other extra fields are pure gauge with gauge transformations that are pure shifts so that they do not appear in the Lagrangian and do not contribute to the local BRST cohomology either. We shall explain this later in Section \ref{sec:BV-def-Embed} but for now we recall that the local BRST cohomology at ghost number zero is what controls the space of non-trivial deformations \cite{Barnich:1993vg}. 

The steps performed below for the limit $g=0$  follow closely the steps given in \cite{deWit:2005ub} to prove that the extra fields appearing in the embedding tensor formalism do not add new degrees of freedom with respect to conventional gaugings.  As shown by the examples of the first-order and second-order formulations discussed in the previous sections, however, equivalent actions might admit inequivalent spaces of local deformations.  Another example is given in \cite{Brandt:1996au,Brandt:1997uq,Brandt:1997ny}. Therefore it seems important to relate the local BRST cohomologies; particularly we monitor how locality goes through each field redefinition. Furthermore, we do \textit{not} fix the gauge at any stage in order to make it clear that the analysis is gauge-independent.  This is also important since gauge fixing might have an impact on local cohomology.

The magnetic components $\Theta^{I\alpha}$ of the embedding tensor form a rectangular $n_v \times (\dim G)$ matrix. Let $r \leq n_v$ be its rank.
Then, there exists an $n_v \times n_v$ invertible matrix $Y$ and a $(\dim G) \times (\dim G)$ invertible matrix $Z$ such that
\begin{equation} \label{eltheta}
\Theta^{' I \alpha} = Y\indices{^I_J} \Theta^{J \beta} Z\indices{_\beta^\alpha}
\end{equation}
is of the form
\begin{equation} \label{formthetaprime}
(\Theta^{' I \alpha}) = \left(\begin{array}{c|c}
\theta & 0  \\ \hline
0 & 0
\end{array}\right),
\end{equation}
where $\theta$ is an $r\times r$ invertible matrix\footnote{The matrices $Y$ and $Z$ can be constructed as a product of the matrices that implement the familiar elementary operations on the rows and columns of $\Theta^{I \alpha}$.}. 

We also define the new electric components by
\begin{equation} \label{magtheta}
\Theta\indices{^'_I^\alpha} = (Y^{-1})\indices{^J_I} \Theta\indices{_J^\beta} Z\indices{_\beta^\alpha},
\end{equation}
so that the new components still satisfy the locality constraint (\ref{eq:ThetaCons})
\begin{equation} \label{primedconstraint}
\Theta^{' I [\alpha} \Theta\indices{^'_I^{\beta]}} = 0 .
\end{equation}
Similarly to what has been done in \cite{deWit:2005ub}, one splits the indices as $I = (\hat{I}, \hat{U})$ and $\alpha = (i, m)$, where $\hat{I}, i = 1, \dots, r$, $\hat{U} = r+1, \dots, n_v$ and $m = r+1, \dots, \dim G$. Moreover, we define
\begin{equation}
\tilde{\theta}\indices{_{\hat{I}}^i} = \Theta\indices{^'_{\hat{I}}^i} .
\end{equation}
With this split, equations \eqref{formthetaprime} and \eqref{primedconstraint} become
\begin{align} \label{thetasplit}
\Theta^{'\hat{I}i} &= \theta^{\hat{I}i} \quad \text{(invertible)} \nn \\
\Theta^{'\hat{I}m} &= \Theta^{'\hat{U}i} = \Theta^{'\hat{U}m} = \Theta\indices{^'_{\hat{I}}^m} = 0 \nn \\
\theta^{\hat{I}i} \tilde{\theta}\indices{_{\hat{I}}^j} &= \theta^{\hat{I}j} \tilde{\theta}\indices{_{\hat{I}}^i} .
\end{align}

Now, we make the field redefinitions
\begin{equation} \label{eq:change1}
\hat{B}_\alpha = Z\indices{_\alpha^\beta} B^{'}_\beta, \qquad A^I = (Y^{-1})\indices{^I_J} A^{'J}, \qquad \tilde{A}_I = Y\indices{^J_I} \tilde{A}^{'}_J.
\end{equation}
In the new primed-variables, the Lagrangian takes the same form \eqref{lagtwoforms} but with  primed quantities everywhere,
\begin{align} \label{lagtwoforms_primed}
\mathcal{L}^\Theta = &- \frac{1}{4} \mathcal{I}^{'}_{IJ}(\phi) H^{'I}_{\mu\nu} H^{'J\mu\nu} + \frac{1}{8} \mathcal{R}^{'}_{IJ}(\phi)\, \varepsilon^{\mu\nu\rho\sigma} H^{'I}_{\mu\nu} H^{'J}_{\rho\sigma}\nonumber\\
 &- \frac{1}{8} \varepsilon^{\mu\nu\rho\sigma} \Theta\indices{^{'I\alpha}} B^{'}_{\alpha \mu\nu} \left( \tilde{F}^{'}_{I \rho\sigma} - \frac{1}{4} \Theta\indices{^'_I^\beta} B^{'}_{\beta \rho\sigma}\right),
\end{align}
with
\begin{equation}
H^{'I}_{\mu\nu} = F^{'I}_{\mu\nu} + \frac{1}{2} \Theta\indices{^{'I \alpha}} B^{'}_{\alpha \mu\nu}. 
\end{equation}

The new matrices $\mathcal{I}^{'}$ and $\mathcal{R}^{'}$ are given by
\begin{equation}
\mathcal{I}^{'} = Y^{-T} \mathcal{I} Y^{-1}, \qquad
\mathcal{R}^{'} = Y^{-T} \mathcal{R} Y^{-1}.
\end{equation}
These field  redefinitions are local, so any local function of the old set of variables is also a local function of the new set of variables. 

Using equation \eqref{thetasplit}, it can be seen that the magnetic vector fields $\tilde{A}^{'}_{\hat{U}}$ and the two-forms $B^{'}_m$ do not appear in the Lagrangian. This means that their gauge symmetries are in fact pure shift symmetries. In particular, a complete description of the gauge symmetries of the two-forms $B^{'}_m$ is actually given by $\delta B^{'}_m = \epsilon^{'}_m$ rather than $\delta B^{'}_m = d\, \Xi^{'}_m$, implying that the above set of gauge transformations is not complete. 

Let us also redefine the gauge parameters as $\Xi^{'}_\alpha = (Z^{-1})\indices{_\alpha^\beta}\, \hat{\Xi}_\beta$, $\lambda^{'I} = Y\indices{^I_J} \lambda^J$ and $\tilde{\lambda}^{'}_I = (Y^{-1})\indices{^J_I} \tilde{\lambda}_J$. The gauge variations of $A^{'\hat{U}}$, $A^{'\hat{I}}$, $\tilde{A}^{'}_{\hat{I}}$ and $B^{'}_i$ are then
\begin{align}
\delta A^{'\hat{U}}_\mu &= \partial_\mu \lambda^{'\hat{U}},\\
\delta A^{'\hat{I}}_\mu &= \partial_\mu \lambda^{'\hat{I}} - \frac{1}{2} \theta^{\hat{I}i} \Xi^{'}_{i\mu},\label{eq:gauge_sym_AI}\\
\delta \tilde{A}^{'}_{\hat{I}\mu} &= \partial_\mu \tilde{\lambda}^{'}_{\hat{I}} + \frac{1}{2} \tilde{\theta}\indices{_{\hat{I}}^i} \Xi^{'}_{i\mu},\label{eq:gauge_sym_A_tilde}\\
\delta B^{'}_{i\mu\nu} &= 2 \partial_{[\mu}\Xi^{'}_{i \nu]}\label{eq:gauge_sym_B}.
\end{align}
 They suggest the further changes of variables
\begin{align} \label{eq:change2}
\bar{A}^i_{\mu} &= \theta^{\hat{I}i} \tilde{A}^{'}_{\hat{I}\mu} + \tilde{\theta}\indices{_{\hat{I}}^i} A^{'\hat{I}}_{\mu} \\
\Delta_{i\mu\nu} &= B^{'}_{i\mu\nu} + 2 (\theta^{-1})_{i\hat{I}} F^{'\hat{I}}_{\mu\nu} ,
\end{align}
which we complete in the $A$-sector by taking other independent linear combinations of the $A^{'\hat{I}}$, $\tilde{A}^{'}_{\hat{I}}$, that can be taken to be for instance  $A^{'\hat{I}}_\mu$.   The change of variables is again such that any local function of the old set of variables is also a local function of the new set of variables.  Using the constraint $\theta^{\hat{I}i}\tilde{\theta}\indices{_{\hat{I}}^j} = \theta^{\hat{I}j}\tilde{\theta}\indices{_{\hat{I}}^i}$, one finds that the variation of the new variables simplifies to
\begin{equation}
\delta \bar{A}^i_\mu = \partial_\mu \eta^i, \quad \delta \Delta_{i\mu\nu} = 0, \quad \delta A^{'\hat{I}}_\mu = \epsilon^{\hat{I}}_\mu,
\end{equation}
with the gauge parameters $\eta^i$ and $\epsilon^{\hat{I}}_{\mu}$ defined as
\begin{align}
    \eta^i &= \theta^{\hat{I}i} \tilde{\lambda}^{'}_{\hat{I}} + \tilde{\theta}\indices{_{\hat{I}}^i} \lambda^{'\hat{I}},\\
    \epsilon^{\hat{I}}_\mu &= \partial_\mu \lambda^{'\hat{I}} - \frac{1}{2} \theta^{\hat{I}i} \Xi^{'}_{i\mu}.
\end{align}
We have used a different symbol $\Delta_{i\mu\nu}$ to emphasize that it does not transform anymore as a generalized gauge potential. Because $\theta^{\hat{I}i}$ is invertible, the gauge parameters  $\eta^i$ and $ \epsilon^{\hat{I}}_\mu$ provide an equivalent description of the gauge symmetries as \eqref{eq:gauge_sym_AI}-\eqref{eq:gauge_sym_B}. However, they provide an irreducible set of gauge symmetries contrary to description given by $\lambda^{'\hat{I}}$, $\tilde{\lambda}^{'}_{\hat{I}}$ and $\Xi^{'}_{i\mu}$.

Written in those variables, the Lagrangian only depends on $n_v$ vector fields $A^{'\hat{U}}$ and $\bar{A}^i$ and on $r$ two-forms $\Delta_i$. The variables $A^{'\hat{I}}_\mu$ drop out in agreement with the shift symmetry $\delta A^{'\hat{I}}_\mu = \epsilon^{\hat{I}}_\mu$.

 The Lagrangian is explicitly
\begin{align} 
\mathcal{L} = &- \frac{1}{4} \mathcal{I}^{'}_{IJ}(\phi) \bar{H}^I_{\mu\nu} \bar{H}^{J\mu\nu} + \frac{1}{8} \mathcal{R}^{'}_{IJ}(\phi)\, \varepsilon^{\mu\nu\rho\sigma} \bar{H}^I_{\mu\nu} \bar{H}^J_{\rho\sigma} \nonumber\\
&- \frac{1}{8} \varepsilon^{\mu\nu\rho\sigma} \Delta_{i \mu\nu} \left( \bar{F}^i_{\rho\sigma} - \frac{1}{4} \tilde{\theta}\indices{_{\hat{I}}^i} \theta^{\hat{I}j} \Delta_{j \rho\sigma}\right),\label{eq:lbarB}
\end{align}
where the $\bar{H}^I$ are
\begin{equation}
\bar{H}^{\hat{I}} = \frac{1}{2} \theta^{\hat{I}i} \Delta_i, \quad \bar{H}^{\hat{U}} = F^{'\hat{U}} .
\end{equation}
The gauge variations of the fields are
\begin{equation}
\delta A^{'\hat{U}}_\mu = \partial_\mu \lambda^{'\hat{U}}, \quad \delta \bar{A}^i_\mu = \partial_\mu \eta^i, \quad \delta \Delta_{i\mu\nu} = 0.
\end{equation}
The two-forms $\Delta_i$ are auxiliary fields and can be eliminated from the Lagrangian \eqref{eq:lbarB}, yielding a Lagrangian of the form \eqref{eq:lprime} in a definite symplectic frame. This is because the relevant quadratic form is invertible, see section 5.1 of \cite{deWit:2005ub} for details.

One can get this Lagrangian more directly as follows.  After the auxiliary fields are eliminated, the variables that remain are the scalar fields and the $n_v$ vector fields $A^{'\hat{U}}$ and $\bar{A}^i$.  We thus see that the embedding tensor determines a symplectic frame. The matrix $E \in Sp(2n_v, \mathbb{R})$ defining the symplectic frame can be viewed in this approach as the function $E(\Theta)$ of the embedding tensor obtained through the above successive steps that lead to the final Lagrangian where only half of the potentials remain. Of course, there are ambiguities in the derivation of $E$ from the embedding tensor, since choices were involved at various stages in the construction.  One gets the matrix $E$ up to a transformation of the stability subgroup of the Lagrangian subspace of the $n_v$ electric potentials  $A^{'\hat{U}}$ and $\bar{A}^i$.

More explicitly, one gets the matrix $E(\Theta)$ from the above construction as follows: the change of variables \eqref{eq:change1} can be written as
\begin{equation} \label{E1}
\begin{pmatrix}
A^I \\ \tilde{A}_I
\end{pmatrix} = \begin{pmatrix}
(Y^{-1})\indices{^I_J} & 0 \\ 0 & Y\indices{^J_I}
\end{pmatrix} \begin{pmatrix}
A^{'J} \\ \tilde{A}^{'}_J
\end{pmatrix}
\end{equation}
and the full change \eqref{eq:change2}, including the trivial redefinitions of the other vector fields, is
\begin{equation} \label{E2}
\begin{pmatrix}
A^{'\hat{I}} \\ A^{'\hat{U}} \\\tilde{A}^{'}_{\hat{I}} \\ \tilde{A}^{'}_{\hat{U}}
\end{pmatrix} =
\begin{pmatrix}
0 & 0 & \delta^{\hat{I}}_{\hat{J}} & 0 \\
0 & \delta^{\hat{U}}_{\hat{V}} & 0 & 0 \\
\delta^{\hat{J}}_{\hat{I}} & 0 & - (\theta^{-1})_{i \hat{I}}\, \tilde{\theta}\indices{_{\hat{J}}^i} & 0 \\
0 & 0 & 0 & \delta^{\hat{V}}_{\hat{U}}
\end{pmatrix}
\begin{pmatrix}
\bar{A}_{\hat{J}} \\ \bar{A}^{\hat{V}} \\ \bar{A}^{\hat{J}} \\ \bar{A}_{\hat{V}}
\end{pmatrix},
\end{equation}
where we defined
\begin{equation}
\bar{A}_{\hat{I}} = (\theta^{-1})_{i \hat{I}} \bar{A}^i
\end{equation}
with respect to \eqref{eq:change2} in order to have the same kind of indices. The matrix $E(\Theta)$ is therefore simply given by
\begin{equation}
E(\Theta) = E_1 E_2,
\end{equation}
where $E_1$ and $E_2$ are the matrices appearing in equations \eqref{E1} and \eqref{E2} respectively. Using the property
\begin{equation}
(\theta^{-1})_{i[\hat{I}} \tilde{\theta}\indices{_{\hat{J}]}^i} = 0,
\end{equation}
which follows from the last of \eqref{thetasplit} upon contractions with $\theta^{-1}$, one can show that $E(\Theta)$ defined in this way is indeed a symplectic matrix. Once $E(\Theta)$ is known, the final Lagrangian in this symplectic frame, i.e. the Lagrangian that follows from the elimination of $\Delta_i$ in \eqref{eq:lbarB}, is then simply the Lagrangian \eqref{eq:lprime} with $\mathcal{I}^{'}$ and $\mathcal{R}^{'}$ determined from \eqref{eq:OMdefBis} where $\mathcal{M}^{'}$ is given by 
\begin{equation}\label{eq:MprimeTheta}
\mathcal{M}^{'} = E(\Theta)^T \mathcal{M} E(\Theta) .
\end{equation}
In general, the matrix $E$ will have an upper triangular part, and hence, will not belong to the stability subgroup of the original electric frame.

Finally, we make a remark that, as we have seen, it is the rank of the magnetic part of the embedding tensor that controls how many original magnetic fields become ``electric".  In particular, if the rank is zero, all original electric fields remain electric, while if the rank is maximum, all original magnetic fields become electric.

We delay the detailed discussion of deformations in the language of BRST cohomology to the end of Chapter \ref{ch:BV_def} after we introduce BRST cohomology and the tools necessary to develop the discussion.

\chapter{BV-BRST Formalism and Deformation}\label{ch:BV_def}

\section{A quick look at BV-BRST formalism}

The nature as we know, including the phenomena we deal with in our every day life and those which happens in microscopic scales, is described by the gauge theory of fundamental interactions
between particles. It is common for the gauge theories to be formulated in a local and manifestly covariant manner as it may facilitate otherwise computational difficulties. However, the price to be paid is to introduce the extra (non-propagating) degrees of freedom to render the locality and the covariance of the theories manifest. The theories then, in the presence of non-dynamical
local degrees of freedom, show invariance under the symmetries which account for the freedom of field transformations while leaving the theories unchanged; these symmetries are
known as gauge symmetries. 

The presence of gauge symmetries specially for the non-abelian theories causes issues once one wants to quantize the classical gauge theories. Typically, 
a gauge-fixing procedure is used to remove the non-dynamical degrees of freedom. Ghost fields are used to compensate for the effects of the gauge degrees of freedom so as 
the final theory yet remains unitary. In the theory of electrodynamics in the linear gauges, ghosts decouple while in non-abelian gauge theories, convenient gauges generically involve
interacting ghosts. The Faddeev-Popov quantization procedure \cite{Faddeev:1967fc,DeWitt:1967yk,DeWitt:1967ub,DeWitt:1967uc}, particularly well-defined in the path integral formulation of field theories, was proposed in order to address this issue. In this procedure, the presence of ghost fields is essential to produce the correct measure for the functional integral by introducing quadratic terms in the ghost fields in the Lagrangian.
The gauge-fixed action then is invariant under a nilpotent global symmetry, the so-called BRST symmetry \cite{Becchi:1974xu,Becchi:1974md,Becchi:1975nq,Tyutin:1975qk}, involving transformations of both fields and ghosts.

There are though systems that cannot be treated by the Faddeev-Popov quantization procedure of which are the theories involving higher order ghost vertices and the gauge theories with reducible gauge group. An important example is supergravity where the algebra does not necessarily close off-shell, i.e. the commutator of two off-shell gauge symmetries will be written by terms proportional to the equations of motion. Moreover, in the presence of higher order $p$-forms, as in the supergravity, the gauge algebra is usually reducible.    

Since the presence of ghost fields are useful in the discussion of quantization and renormalization of gauge theories (most importantly of the non-abelian ones), it is of interest to have a formalism that not only includes the ghost fields from the beginning but also encompasses the remnant BRST symmetry of the theory. The Batalin-Vilkovisky (BV) formulation, originally proposed by Zinn-Justin in the context of renormalization of Yang-Mills theories \cite{ZinnJustin:1974mc,ZinnJustin:2002ru} and developed by Batalin and Vilkovisky to more generalized models \cite{Batalin:1981jr,Batalin:1983wj,Batalin:1984jr,Batalin:1984ss,Batalin:1985qj}, does indeed provide the appropriate means. It involves anti-fields which source the BRST symmetries and a canonical symplectic structure, given by an odd non-degenerate symplectic form on the space of fields and anti-fields known as the anti-bracket map, which controls the dynamics of the theory.

The BV formalism has many advantages; it enjoys the manifest gauge invariance, allows for a perturbative expansion
of the quantum theory and makes possible to study the quantum corrections to the symmetry structure of the theory. Furthermore, the formalism can treat the theories with an open algebra or reducible gauge symmetries by introducing the extra fields and antifields, the ghosts of ghosts and their corresponding anti-fields and so on. 

In the rest of this chapter, we introduce in Section \ref{sec:BRST-diff-cohom} the BRST cohomology and in Section \ref{sec:BV-def-formalism} show its importance in the BV deformation
procedure. Afterward, in Section \ref{sec:BV-def-Embed} we go back to the question asked in Chapter \ref{ch:Isom-Emb-VS-models} that if the spaces of local deformations of the undeformed Lagrangian \eqref{lagtwoforms} of embedding tensor formalism and the one of scalar-vector model \eqref{eq:lprime} are equivalent. We use the BV-BRST deformation techniques of Section \ref{sec:BV-def-formalism} to have a thorough analysis and to answer to this question.

\section{BRST cohomology}\label{sec:BRST-diff-cohom}

The BRST differential $s$ is a nilpotent differential defined on the space of all fields (including ghost fields) $\Phi^A=\left(A^{I}_{p\mu}, \psi^i, C^I,...\right)$ and their
anti-fields $\Phi^\ast_A=\left(A^\ast_{pI\mu},\psi^{\ast i}, C^\ast_I,...\right)$ where $A^{I}_{p\mu}$ and $\psi^i$ are $p$-forms and the matter fields respectively and ellipses stand for the ghosts of ghosts and their corresponding anti-fields. Therefore, one has 
\be
s^2 =0.
\ee
Then, the BRST cocycles are those objects $K$ that are BRST-closed, i.e. elements of the kernel of $s$,
\be
s K =0,
\ee
and the BRST coboundaries are those BRST-closed objects that are BRST-exact, i.e. elements of the image of $s$,
\be
s K=0 ,\qquad K= s B.
\ee
The BRST cohomolgy is formally defined as
\be
H(s)= \f{\textrm{Ker}(s)}{\textrm{Im}(s)}.
\ee
An element in $H(s)$ is an equivalence class of BRST-cocycles, where two BRST-cocycles are identified if they differ by a BRST-coboundary.

One can consider the BRST cohomology in the space of local functions where the cochain $K$ are local functions of the fields $\Phi^A$ and a finite number of their derivatives (and depend on spacetime coordinates) or in the space of local functionals where the elements of $K$ are the integral of forms built from the local functions.

In the standard field theoretic setting, one insists on spacetime
locality which implies that the cohomology is computed in the space of
local functionals in the fields and antifields. In turn, this can be
shown to be equivalent to the cohomology of $s$ in the space of local
functions up to total derivatives or, in form notation, to the
cohomology of $s$ in top form degree $n$, up to the horizontal
 exterior derivative of an $(n-1)$-form. The horizontal exterior derivative is not the de Rham
differential but is instead given by $d=dx^\mu\d_\mu$, where
$\d_\mu=\ddl{}{x^\mu}+\d_\mu z^\Sigma\ddl{}{z^\Sigma}+\dots$ is the
total derivative. Here $z^\Sigma$ is referring to fields and antifields collectively, $z^\Sigma=(\Phi^A,\Phi^*_A)$.

Therefore the cocycle and coboundary conditions are written as
\begin{align}
    s a + d m =0 \qquad &~\textrm{cocycle condition},\label{eq:cocycle_cond}\\
    a= s b + d n \qquad &~\textrm{coboundary condition}, 
\end{align}
where $a,b$ are $n$-form local functions and $m,n$ are some $(n-1)$-form local functions. In the following we use $H(s)$ for the BRST-cohomology in the sapce of local functions and $H(s\vert d)$ for the BRST-cohomology in the space of local functionals. One could consider a grading of these cohomologies with the ghost numbers, then one defines $H^{g,n}(s \vert d)$, isomorphic to $H^g(s)$, as the cohomology of local functions of form degree $n$ and ghost number $g$ modulo total derivatives
\begin{equation}
  \label{eq:9a}
  s a^{g,n} +d a^{g+1,n-1}=0,\quad a^{g,n}\sim a^{g,n}+s b^{g-1,n}+d
  b^{g,n-1}.
\end{equation}

We already mentioned the importance of BV-BRST formalism in the quantization of gauge systems at the quantum and classical level. In the language of BRST cohomology, important problems at the quantum level like the gauge anomalies and the counterterms in order to render the gauge theories renormalizable are equivalent to finding $H^{1,n}(s\vert d)$ and $H^{0,n}(s\vert d)$, see \cite{Henneaux:1992ig}.

At the classical level, the BRST cohomology is a useful tool to explain the generalized conservation laws and to build all consistent interactions that one can add to a free theory while preserving the number of gauge symmetries of starting theory. The former is captured by the characteristic cohomology $H^{n+g}_{\rm char}(d)$ and the latter is described by $H^{0,n}(s_0\vert d)$ and $H^{1,n}(s_0\vert d)$ of the undeformed theory \cite{Barnich:2000zw}. We will explain this in more detail in the next section.

\section{Batalin-Vilkovisky antifield formalism}\label{sec:BV-def-formalism}

In order to systematically construct consistent interactions in gauge
theories, it is useful to reformulate the problem in the context of
algebraic deformation theory
\cite{Gerstenhaber:1963zz,nijenhuis1966,Julia:1986gv,Julia:1986ha}. The
appropriate framework is provided by the Batalin-Vilkovisky antifield
formalism
\cite{Batalin:1981jr,Batalin:1984jr,Barnich:1993vg,Henneaux:1997bm}.

The structure of an irreducible gauge system, i.e., the Lagrangian
$\cL_0$ with field content $\varphi^{a}$, generating set of gauge
symmetries\footnote{We use the condensed De Witt notation.}
$\delta_\epsilon \varphi^{a}=R^{a}{}_\alpha [\varphi^{b}]
\,(\epsilon^\alpha)$ and their algebra, is captured by the
Batalin-Vilkovisky (BV) master action $S$ (see
e.g.~\cite{Henneaux:1992ig,Gomis:1994he} for reviews). The master
action is a ghost number $0$ functional
\begin{equation}
  \label{eq:masteraction}
  S=\int \!d^n\!x\, \cL=\int \!d^n\!x\, \left[ \cL_0 +\varphi^*_{a}
    R^{a}{}_\alpha \,  
    (C^\alpha) + \frac{1}{2}C^*_\alpha
    f\indices{^\alpha_{\beta\gamma}}( C^\beta, C^\gamma) + \dots  \right], 
\end{equation}
that satisfies what is called the master equation
\begin{equation}
\label{eq:masterequation}
\half (S,S)=0.
\end{equation}
In this equation, the BV antibracket is the odd graded Lie bracket
defined by
\begin{equation}
  \label{eq:antibracket}
  (X,Y)= \int\!d^n\!x\, \left[ \frac{\delta^R X}{\delta \Phi^A(x)}
\frac{\delta^L Y}{\delta \Phi^*_A(x)}-\frac{\delta^R X}{\delta 
\Phi^*_A(x)}\frac{\delta^L Y}{\delta \Phi^A(x)} \right]
\end{equation}
on the extended space $\Phi^A=(\varphi^{a},C^\alpha,\dots)$ of
original fields and ghosts (and ghosts for ghosts in the case of
reducible gauge theories) and their antifields $\Phi^*_A$. The ghost number and the Grassmann number of each fields and their anti-fields are summarized in Table \eqref{tab:gh+afl}.
\begin{table}[h!]
\begin{center}
\begin{tabular}{|c|c|c|c|c|}
\hline 
        $A$  &  puregh($A$)  &  antifld($A$)  & gh($A$) & $\epsilon(A)$ \\
\hline 
        $\varphi^a$  &  0  &  0 & 0 & 0  \\ 
\hline  
           $C^I$  &  1  &  0 & 1 & 1  \\ 
\hline
        $\varphi^\ast_a$  &  0  &  1 & -1 & 1  \\ 
\hline
        $C^\ast_a$  &  0  &  2 & -2 & 0  \\ 
\hline
\end{tabular}
\end{center}
\caption{The list of pure ghost number (puregh), anti-field (anti-ghost) number (antifld), ghost number (gh) and Grassmann number ($\epsilon$) of each of the fields and anti-fields. The ghost number is defined as $\textrm{gh}(A)=\textrm{puregh}(A)-\textrm{antifld}(A)$, and we have in general ${\rm gh}(\Phi^*_A)=-{\rm gh}(\Phi^A)-1$. }\label{tab:gh+afl}
\end{table}
The Lagrangian, gauge
variations and structure functions of the gauge algebra are contained
in the first, second and third term of the master action
\eqref{eq:masteraction} respectively.

For the deformation problem, one assumes the existence of an
undeformed theory described by $S^{(0)}$ satisfying the master
equation $\half (S^{(0)},S^{(0)})=0$ and one analyzes the conditions
coming from the requirement that, in a suitable expansion, the
deformed theory
\begin{equation}
  \label{eq:deformedmasteraction}
  S=S^{(0)}+ S^{(1)}+ S^{(2)}+\dots,
\end{equation}
satisfies the master equation \eqref{eq:masterequation}. This results in a chain of equations at different order in deformation parameters
\begin{align}
    &(S^{(0)},S^{(0)})=0,\label{eq:def_con_N0}\\
    &(S^{(0)},S^{(1)})=0,\label{eq:def_con_N1}\\
    &(S^{(0)},S^{(2)})+\half (S^{(1)},S^{(1)})=0,\label{eq:def_con_N2}\\
    ~\qquad &\vdots \nonumber
\end{align}
This can be easily proved using the properties of anti-bracket map, see for example section 4 of \cite{Gomis:1994he} for a review of all properties of anti-bracket map. The deformed
Lagrangian, gauge symmetries and structure functions can then be read
off from the deformed master action \eqref{eq:deformedmasteraction}.

The first condition \eqref{eq:def_con_N1} on the infinitesimal deformation $S^{(1)}$ is
\begin{equation}
  \label{eq:40a}
  (S^{(0)},S^{(1)})=0.
\end{equation}
This equation admits solutions $S^{(1)}=(S^{(0)},\Xi)$, for all $\Xi$ of
ghost number $-1$.  Such deformations can be shown to be trivial in
the sense that they can be absorbed by (anticanonical) field-antifield
redefinitions. Moreover, trivial deformations in that sense are always
of the form $S^{(1)}=(S^{(0)},\Xi)$ for some local $\Xi$. It thus follows
that equivalence classes of deformations up to trivial ones are
classified by $H^0(s)$, the ghost number zero cohomology of the
antifield dependent BRST differential $s=(S^{(0)},\cdot)$ of the
undeformed theory,
\begin{equation}
  \label{eq:42}
  [S^{(1)}]\in H^0(s). 
\end{equation}
For our problem of determining the most general deformation, we start
by computing $H^{0}(s)$ and couple its elements with independent
parameters to the starting point action to obtain
$S^{(0)}+S^{(1)}$. The parameters thus play the role of generalized
coupling constants. In a second step, we determine the constraints on
these coupling constants coming from the existence of a completion
such that \eqref{eq:masterequation} holds. The expansion is then in
terms of homogeneity in these generalized coupling constants and not,
as often done, in homogeneity of fields (in which case $S^{(0)}$
corresponds to an action quadratic in the fields). In particular, this
approach treats the different types of symmetries involved in the
determination of $H^{0}(s)$ on the same footing.

The BRST differential is defined on the undifferentiated
fields and antifields by $s\Phi^A=-\vddr{\cL}{\Phi^*_A}$,
$s\Phi^*_A=\vddr{\cL}{\Phi^A}$. It is extended to the
derivatives through $[s,\d_\mu]=0$ resulting in $\{s,d\}=0$. This
reformulation allows one to use systematic homological techniques
(descent equations) for the computation of these classes (see
e.g.~\cite{DuboisViolette:1985jb}).

At second order, the condition \eqref{eq:def_con_N2} on the infinitesimal deformation
$S^{(1)}$ is  
\begin{equation}
  \label{eq:43}
  \half (S^{(1)},S^{(1)})+(S^{(0)},S^{(2)})=0.
\end{equation}
The antibracket gives rise to a well defined map in cohomology\footnote{In the following, whenever we refer to the cohomology of local functional of a top form-degree, we drop the index $n$ and write it as $H^g(s\vert d)$.},
\begin{equation}
  \label{eq:31a}
  (\cdot,\cdot): H^{g_1}(s|d)\otimes H^{g_2}(s|d)\longrightarrow H^{g_1+g_2+1}(s|d).
\end{equation}
For cocycles $C_i$ with $[C_i]\in H^{g_i}(s|d)$, it is explicitly given
by
\begin{equation}
  \label{eq:44}
  ([C_1],[C_2])=[(C_1,C_2)]\in H^{g_1+g_2+1}(s|d). 
\end{equation}
Condition \eqref{eq:43} constrains the infinitesimal
deformation $S^{(1)}$ to satisfy
\begin{equation}
  \label{eq:45}
  \half([S^{(1)}],[S^{(1)}])=[0]\in H^1(s|d).
\end{equation}
If this is the case, $S^{(2)}$ in \eqref{eq:43} is defined up to a
cocycle in ghost number $0$. Higher order brackets and constraints
can be analyzed in a similar way, see
e.g.~\cite{RETAKH1993217,FUCHS2001215}.

Besides the group $H^0(s|d)$ that describes infinitesimal deformations,
and $H^1(s|d)$ that controls the obstructions to extending these to
finite deformations, one can furthermore show \cite{Barnich:1994db}
that $H^{g}(s|d)\simeq H^{n+g}_{\rm char}(d)$ for $g\leq -1\,$. The
 characteristic cohomology groups are defined by forms
$\omega$ in the original fields $\varphi^a$ such that
\begin{equation}
  \label{eq:10}
  d\omega^{n+g}\approx 0,\quad 
  \omega^{n+g}\sim\omega^{n+g}+d \eta^{n+g-1}+t^{n+g},
\end{equation}
with $t^{n+g}\approx 0$ and where $\approx 0$ denote terms that vanish
on all solutions to the Euler-Lagrange equations of motion. In
particular, these groups can be shown to vanish for $g\leq -3$ in
irreducible gauge theories \cite{Barnich:1994db,Barnich:2000zw}. The
group $H^{-2}(s|d)$ describes equivalence classes of ``global"
reducibility parameters, i.e., particular local functions $f^\alpha$
such that $R\indices{^a_\alpha}(f^\alpha)\approx 0$ where
$f^\alpha\sim f^\alpha+t^\alpha$ with $t^\alpha\approx 0$.
This terminology reflects the fact that this cohomology may be non 
trivial even for (locally) irreducible gauge systems, in other words in 
the absence of  $p$-form gauge fields with higher $p$.  This will become clear momentarily 
and is crucial throughout our discussion.
  These classes
correspond to global symmetries of the master action rather than of
the original action alone \cite{Brandt:1996uv,Brandt:1997cz}. The
associated characteristic cohomology $H^{n-2}_{\rm char}(d)$ captures
non-trivial (flux) conservation laws. More generally in the case of
free abelian $p$-form gauge symmetry it was shown in
\cite{Julia:1980gn} that one can generalize the first Noether theorem
($p=0$) and deduce by a similar formula a class of
$H^{n-p-1}_{\rm char}(d)$ generalizing the electric flux which corresponds 
to the case $p=1$, i.e., to ordinary gauge invariance.
The groups $H^{-1-p}(s|d)$ appear for $p$-form gauge theories
and vanish for $p\ge 2$ in the irreducible case \cite{Henneaux:1996ws}.
The group
$H^{-1}(s|d)$ describes and generates the inequivalent global
symmetries, with $H^{n-1}_{\rm char}(d)$ encoding the associated
inequivalent Noether currents. 
We mention these groups here since
they play an important role in the determination of $H^0(s|d)$ as it
will be seen in Section \ref{sec:gauging}.

When $g_1=-1=g_2$, $(\cdot,\cdot) : H^{-1}\otimes H^{-1}\to H^{-1}$; in
this case the antibracket map encodes the Lie algebra structure of the
inequivalent global symmetries \cite{Barnich:1996mr}. More generally,
it follows from $(\cdot,\cdot) : H^{-1}\otimes H^{g}\to H^{g}$ that,
for any ghost number $g$, the BRST cohomology classes form a
representation of the Lie algebra of inequivalent global symmetries.

From here on, for notational simplicity, we will drop the square brackets when
computing the antibracket map, but keep in mind that it involves
classes and not their representatives.

\subsection{Depth of an element}\label{sec:antibr-maps-desc}

Consider a BRST-cocycle $\omega^{g,k}$ at ghost number $g$ and form degree $k$. Using \eqref{eq:cocycle_cond} or \eqref{eq:9a}, any cocycle $\omega^{g,k}$ of the local BRST cohomology is associated with a $(s,d)$-descent
\begin{equation}
  \label{eq:B1}
  s\omega^{g,k}_l+d\omega^{g+1,k-1}_l=0,\quad
  s\omega^{g+1,k-1}_l+d\omega^{g+2,k-2}_l=0,\, \dots \,,\,
  s\omega^{g+l,k-l}_l=0, 
\end{equation}
that stops at some BRST cocycle $\omega^{g+l,k-l}_l$. The length $l$ of the shortest non trivial
descent is called the ``depth'' of $[\omega^{g,k}]\in H^{g,k}(s|d)$.  The last element $\omega^{g+l,k-l}_l$ is then non trivial in 
$H^{g+l,k-l}(s)$.  The usefulness of the depth in analyzing the BRST cohomology is particularly transparent in  \cite{DuboisViolette:1985jb,DuboisViolette:1985hc,DuboisViolette:1985cj}.

Local BRST cohomology classes $[\omega^{g,k}]\in H^{g,k}(s|d)$ are
thus characterized, besides \textbf{ghost number} $g$ and \textbf{form degree} $k$, by the
\textbf{depth} $l$. Here, we explain how the antibracket map
behaves with respect to the depth of its elements.

Both the covariantizable and non-covariantizable currents \footnote{By ``covariantizable", we mean that one can choose the ambiguities in the Noether currents so as to take them gauge invariant; otherwise the current is ``non-covariantizable".} as elements of
$H^{-1,n}(s|d)$ are distinguished by the property that for which the
depth is $1$. The associated infinitesimal deformations
as elements of $H^{0,n}(s|d)$ are distinguished by the property that the
depth is deeper than one. Therefore, the following will be
relevant when studying the obstruction to infinitesimal deformations.

\begin{proposition}
The depth of an image of the antibracket map is less than or equal to the depth of
  its most shallow argument.
\end{proposition}

\begin{proof}
Consider
$[\omega^{g_1,n}_{l_1}],[\omega^{g_2,n}_{l_2}]\in H^{*,n}(s|d)$, where
we can assume without loss of generality that $l_1\geqslant l_2$. For the
antibracket, let us not choose the expression with Euler-Lagrange
derivatives on the left and right that is graded antisymmetric without
boundary terms, but rather the one that satisfies a graded Leibniz
rule on the right
\begin{equation}
  \label{eq:B2}
  (\omega^{g,n},\cdot)_{\rm
    alt}=\d_{(\nu)}\vddr{(-\star\omega^{g,n})}{\phi^A}\ddll{\cdot}{\d_{(\nu)}\phi^*_A}
    -(\phi^A\leftrightarrow \phi^*_A), 
\end{equation}
and the following version of the graded Jacobi identity without
boundary terms,
\begin{align}
  (\omega^{g_1,n},(\omega^{g_2,n},\cdot)_{\rm alt})_{\rm
    alt} &= ((\omega^{g_1,n},\omega^{g_2,n})_{\rm alt},\cdot)_{\rm
    alt} \nonumber\\ 
    &+(-)^{(g_1+1)(g_2+1)}(\omega^{g_2,n},(\omega^{g_1,n},\cdot)_{\rm alt})_{\rm
    alt}\label{eq:B3}
\end{align}
(see appendix B of
\cite{Barnich:1996mr} for details and a proof). 
Furthermore, 
\begin{equation}
  (\omega^{g,n},d (\cdot))_{\rm
    alt}=(-)^{g+1}d((\omega^{g,n}, \cdot)_{\rm
    alt}),\quad (d\omega^{g+1,n-1},\cdot)_{\rm alt}=0. \label{eq:B5}
\end{equation}
Let $S=\int (-\star \cL)$ be the BV master action. We have
$s\cdot=(-\star \cL,\cdot)_{\rm
  alt}$. 
Using these properties, we get 
  \begin{equation}
    \label{eq:B4}
    s(\omega^{g_1,n}_{l_1},\omega^{g_2,n}_{l_2})_{\rm
      alt}+d((\omega^{g_1,n}_{l_1},\omega^{g_2+1,n-1}_{l_2})_{\rm
      alt})=0,
\dots, s (\omega^{g_1,n}_{l_1},\omega^{g_2+l_2,n-l_2}_{l_2})_{\rm
      alt}=0, 
  \end{equation}
which proves the proposition.
\end{proof}

By using $[(\omega^{-1,n},\omega^{g,n})]\in H^{g,n}(s|d)$, it follows
that :

(i) Characteristic cohomology in degree $n-1$ described by
$H^{-1,n}(s|d)$ is a Lie algebra. It is the Lie algebra of non trivial
global symmetries. It also describes the Dirac or Dickey bracket
algebra of non trivial conserved currents (up to constants or more
generally topological classes),

(ii) $H^{-2,n}(s|d)$ is a module thereof (module structure of
flux charges - Gauss or ADM type surface charges - under global symmetries). The proposition
gives rise for instance to the following refinements:

\begin{corollary}
Covariantizable characteristic cohomology in form
degree $n-1$ forms an ideal in the Lie algebra of characteristic
cohomology in form degree $n-1$. The module action of covariantizable
characteristic cohomology of degree $n-1$ on characteristic cohomology
in degree $n-2$ is trivial.
\end{corollary}
Similar results hold for the associated infinitesimal deformations.

\section{Properties of BRST differential}

The BRST differential $s$ is defined as $s=\gamma + \delta + \textrm{``more"}$ which acts on the extended phase space of fields and antifields. The differentials $\gamma$ and $\delta$ are horizontal differential and Koszul-Tate differential respectively. The horizontal differential $\gamma$ has a non-trivial action when acting on fields, while it vanishes on anti-fields. The Koszul-Tate differential $\delta$ has a non-trivial action on anti-fields but vanishes when acting on fields. The ``more" part in the definition of BRST differential is introduced in order to guarantee the nilpotency of $s$, i.e. $s^2=0$. When the gauge transformations form an abelian group or the gauge algebra closes off-shell, one simply has \cite{Henneaux:1989jq}
\be
s=\gamma+\delta,\qquad s^2=0,
\ee
with $\delta$ and $\gamma$ the graded commutating nilpotent differentials satisfying
\be
\gamma^2=0,\quad \delta^2=0,\quad \gamma\delta+\delta\gamma=0.
\ee
Also, $\delta$ and $\gamma$ have the properties 
\begin{align}
    \textrm{puregh}(\delta)&=0,\qquad \textrm{antifld}(\delta)=-1,\\
    \textrm{puregh}(\gamma)&=1,\qquad \textrm{antifld}(\gamma)=0,
\end{align}
such that the BRST differential $s$ defined to increase the ghost number by $1$.

In order to have a covariant formalism, the space of all observable quantities, the so-called phase space, has to be covariant. In the absence of gauge invariances, given the space of all possible smooth fields, a subspace of which is called the covariant phase space and is defined as the space $\Sigma$ of all smooth fields which are the solutions of equations of motion. When there is a gauge symmetry, then the covariant phase space is comprised of all functions in $\Sigma$ which are also gauge-invariant. Since the gauge transformations are integrable on-shell,
therefore they generate a congruence of gauge orbits on $\Sigma$ along which the gauge-invariant quantities are conserved. The Koszul-Tate differential $\delta$ and the horizontal differential $\gamma$ in fact take care of the first and the second steps in defining gauge-invariant observables respectively.

The first step is possible to achieve with the help of the Koszul-Tate resolution \cite{Henneaux:1992ig}, which states as
\begin{align}
H_g(\delta)&=0, \qquad g\neq 0,\nonumber\\
H_0(\delta)&=C^\infty(\Sigma).
\end{align}

Then since the only non-vanishing homology of $\delta$ is $H_0(\delta)$, one can easily implement the second step by defining a differential along the orbits of gauge groups acting on
the space $C^\infty(\Sigma)$ of smooth functions given by $H_0(\delta)$. It is shown that in general one has \cite{Henneaux:1994rb}
\begin{align}
H^g(s) &= H^g(\gamma, H_0(\delta)) \qquad g\geq 0,\nonumber\\
H^g(s) &=0 \qquad\qquad g < 0. 
\end{align}

We should emphasize that the cohomology of $s$, $\gamma$ and $\delta$ in the space of local functional are in general different and for example $H^{g,p}(s\vert d) \neq 0$ for $g<0$. This is an important feature that we will discuss in the next chapter. For more details and complete review of definitions, properties of different differentials and the homological and cohomological arguments regarding each,
see \cite{Henneaux:1992ig,Barnich:2000zw,Barnich:1994db}.

\section{Deformation of the Embedding tensor formalism}\label{sec:BV-def-Embed}

We discussed in Chapter \ref{ch:Isom-Emb-VS-models} that the systems described by the Lagrangian \eqref{lagtwoforms}, of the embedding tensor formalism, 
and the Lagrangian \eqref{eq:lprime} of ungauged scalar-vector model
are governed by the same physics and how one can obtain one Lagrangian from the other one by a series of symplectic transformations. However, the space of deformations of these models can in general be different. As we just discussed through this chapter, the BRST cohomology provides us a strong tool to obtain the consistent deformation of a theory. In this section, we show that \textit{the space of deformations of \eqref{lagtwoforms} and \eqref{eq:lprime} are equivalent.}

We recall that the space of equivalence classes of non-trivial, consistent deformations of a local action is isomorphic to the BRST cohomology $H^0(s \vert d)$ in field-antifield space \cite{Batalin:1981jr} at ghost number zero \cite{Barnich:1993vg}. The question, then, is to determine how the variables appearing in addition to the standard variables of the Lagrangian \eqref{eq:lprime}  could modify $H^0(s \vert d)$.

From the discussion of Section \ref{sec:embedding}, these extra variables are of two types:
\begin{itemize}
\item either they are of pure gauge type or equivalently invariant under arbitrary shifts, and therefore they drop out from the Lagrangian;
\item or they are auxiliary fields appearing quadratically and undifferentiated in the Lagrangian, consequently they can be eliminated algebraically through their own equations of motion.  \end{itemize}
We collectively denote by $W^\alpha$ the fields of pure gauge type.  So the $W^\alpha$'s stand for the $n_v$ vector potentials $\tilde{A}^{'}_{\hat{U}}$, $A^{'\hat{I}}$ and the (dim$\,G-r$) two-forms $B^{'}_m$.  The auxiliary fields are the $r$ two-forms $\Delta_i$.    For convenience, we absorb in (\ref{eq:lbarB}) the linear term in the auxiliary fields by a redefinition $\Delta_i \rightarrow \Delta_i^{'} = \Delta_i + b_i$ where the term $b_i$ is $\Delta_i$-independent.  Once this is done, the dependence of the Lagrangian on the auxiliary fields simply reads $\frac12 \kappa^{ij} \Delta_i^{'} \Delta_j^{'}$ with an invertible quadratic form $\kappa^{ij}$.

The BRST differential acting in the sector of the first type of variables read
\be
s W^\alpha = C^\alpha \, ,\; \; \;   s C^\alpha = 0 \, ,\; \; \; s C^*_\alpha = W^*_\alpha \, ,\; \; \; s W^*_\alpha =0
\ee
where $C^\alpha$ are the ghosts of the shift symmetry and $W^*_\alpha$, $C^*_\alpha$ the corresponding antifields, so that $(C^\alpha,  W^\alpha)$ and $(W^*_\alpha, C^*_\alpha)$  form ``contractible pairs" \cite{Henneaux:1992ig}. 

Similarly, the BRST differential acting in the sector of the second type of variables read
\be
s \Delta_i^{'} = 0, \; \; \; s \Delta^{'*i} = \kappa^{ij} \Delta_j^{'}
\ee
where $\Delta^{'*i}$ are the antifields conjugate to $\Delta_i^{'}$, showing that the auxiliary fields and their antifields form also contractible pairs since $\kappa^{ij}$ is non degenerate.

Now, the above BRST transformations involve no spacetime derivatives so that one can construct a ``contracting homotopy" \cite{Henneaux:1992ig} in the sector of the extra variables $(W^\alpha,  \Delta_i^{'})$, their ghosts and their antifields that commutes with the derivative operator $\partial_\mu$.  The algebraic setting
is in fact the same as for the variables of the ``non-minimal sector" of \cite{Batalin:1981jr}.  This implies that the extra variables neither contribute to $H^g(s)$ nor to $H^g(s \vert d)$  \cite{Barnich:1994db}.

To be more specific, let us focus on the contractible pair $(C^\alpha, W^\alpha)$.  The analysis proceeds in exactly the same way for the other pairs.  One can write the BRST differential  in that sector as
\be
s = \sum_{\{ \mu \}} \partial_{\mu_1 \cdots \mu_s} C^\alpha \frac{\partial}{\partial _{\mu_1 \cdots \mu_s} W^\alpha}
\ee
where the sum is over all derivatives of the fields.  Define the ``contracting homotopy" $\rho$ as
\be
\rho = \sum_{\{ \mu \}} \partial_{\mu_1 \cdots \mu_s} W^\alpha \frac{\partial}{\partial _{\mu_1 \cdots \mu_s} C^\alpha}  .
\ee
One has by construction 
\be
[ \rho, \partial_\mu] = 0,
\ee
just as 
\be
[ s, \partial_\mu] = 0.
\ee
This is equivalent to $\rho d + d \rho = 0$.  Furthermore, the counting operator $N$ defined by
\be
N = s \rho + \rho s
\ee
is explicitly given by
\be
N = \sum_{\{ \mu \}} \partial_{\mu_1 \cdots \mu_s} C^\alpha \frac{\partial}{\partial _{\mu_1 \cdots \mu_s} C^\alpha} + \sum_{\{ \mu \}} \partial_{\mu_1 \cdots \mu_s} W^\alpha \frac{\partial}{\partial _{\mu_1 \cdots \mu_s} W^\alpha}
\ee
and commutes with the BRST differential,
\be
[N, s] = 0.
\ee
The operator $N$ gives the homogeneity degree in $W^\alpha$, $C^\alpha$ and their derivatives.  So, for the polynomial $a$, we have $Na = k a $ with $k$ a non negative integer if and only if $a$ is of degree $k$ in $W^\alpha$, $C^\alpha$ and their derivatives (by Euler theorem for homogeneous functions). 

Because $N$ commutes with $s$, we can analyze the cocycle condition
\be
s a + d b = 0 \label{eq:cocycle}
\ee
at definite polynomial degree, i.e., assume that $N a = k a$, $Nb = kb$ where $k$ is a non-negative integer.  Our goal is to prove that the solutions of (\ref{eq:cocycle}) are trivial when $k \not=0$, i.e., of the form $a = se + df$, so that one can find, in any cohomological class of $H(s \vert d)$, a representative that does not depend on $W^\alpha$ or $C^\alpha$ -- that is, $W^\alpha$ and $C^\alpha$ ``drop from the cohomology".

To that end, we start from $ a = N \left(\frac{a}{k} \right)$ ($k \not=0$), which we rewrite as
$a = s \rho \left(\frac{a}{k} \right)+ \rho s  \left(\frac{a}{k} \right)$ using the definition of $N$.  The first term is equal to $s \left(\rho \left(\frac{a}{k} \right)\right)$ and hence is BRST-exact.   The second term is equal to $\rho \left(-d \left(\frac{b}{k} \right)\right)$ (using (\ref{eq:cocycle})), which is the same as $d \left(\rho \left(\frac{b}{k} \right)\right)$ since $\rho$ and $d$ anticommute.  So it is $d$-exact. We thus have shown that 
\be
a = se + df
\ee
with $e = \rho \left(\frac{a}{k} \right)$ and $f = \rho \left(\frac{b}{k} \right)$.   This is what we wanted to prove.

 We can thus conclude that the local deformations ($H^0(s \vert d)$) -- and in fact also $H^g(s \vert d)$ and in particular the candidate anomalies ($H^1(s \vert d)$) --  are the same whether or not one includes the extra fields $(W^\alpha, \Delta_i^{'})$.  This is what we wanted to prove:
\begin{center}
\textit{BRST cohomologies computed from the Lagrangians \eqref{lagtwoforms} and \eqref{eq:lprime} are isomorphic.}
\end{center}
As a final remark, we point out \cite{Boulanger:2008nd} where a dual formulation of three dimensional non-linear Einstein gravity with similar features have been studied.

\chapter{BRST cohomology of scalar-vector coupled models}\label{ch:BRST-cohom-SV-Lagr}

In Chapter \ref{ch:Isom-Emb-VS-models} we proved that not only the Lagrangian of embedding tensor formalism \eqref{lagtwoforms} and the one of ungauged scalar-vector model of the form \eqref{eq:lprime} (in a symplectic frame determined by the coefficients of embedding tensor) are equivalent\footnote{We recall that these Lagrangians are the same only after the elimination of auxiliary fields in \eqref{eq:lbarB} which leads to a Lagrangian of the form \eqref{eq:lprime} in a definite symplectic frame.} but also the space of deformations of both theories are in fact isomorphic. In Chapter \ref{ch:BV_def}, we discussed that the space of infinitesimal deformations of a theory is entirely determined by the ghost number zero cohomology of BRST differential. We explicitly showed that $H_0(s|d)$ for both theories are isomorphic.

In this chapter, following the discussion of Section \ref{sec:BV-def-formalism}, we seek for a systematic study of all consistent deformations of general scalar-vector models described by 
the gauge invariant actions of the form
\begin{equation}
  \label{eq:1bis}
  S_0[A_\mu^I,\phi^i]=\int \!d^4\!x\, \mathcal{L}_0,
\end{equation}
depending on $n_s$ uncharged scalar fields $\phi^i$ and $n_v$ abelian vector
fields $A_\mu^I$. In what comes next we consider that the Lagrangian has Poincar\'e invariance (which is guaranteed in the presence of gravity) even though it is not much used for the most part of our computations in this chapter. It is in practice possible to work out the result in the case of non-Poincar\'e invariant Lagrangians with a little bit of effort. We assume that the only gauge symmetries of
(\ref{eq:1bis}) are the standard $U(1)$ gauge transformations for each
vector field, so that the gauge algebra is abelian and given by $n_v$
copies of $\mathfrak{u}(1)$.  A generating set of gauge invariances
can be taken to be
\begin{equation}
  \label{eq:2bis}
  \delta A_\mu^I=\d_\mu\epsilon^I,\quad \delta\phi^i =0. 
\end{equation}
The Lagrangian takes the form
\begin{equation}
  \mathcal{L}_0= \mathcal{L}_S[\phi^i]+\mathcal{L}_V[A^I_\mu,\phi^i],  \label{eq:Starting}
\end{equation}
where $ \mathcal{L}_V$ is a function that depends on the vector fields
through the abelian curvatures
$F^I_{\mu\nu}=\d_\mu A_\nu^I-\d_\nu A^I_\mu$ only, and which can also
involve the scalar fields $\phi^i$.  Derivatives of these variables
are in principle allowed in the general analysis carried out below,
but actually do not occur in the explicit Lagrangians discussed in
more detail.  The scalar fields can occur non linearly, e.g. terms of
the form $\mathcal{I}_{IJ}(\phi) F^I_{\mu\nu} F^{J\mu\nu}$ as in the Lagrangian \eqref{eq:lag}. 
In addition, the scalar Lagrangian need not to be quadratic in general case. We will discuss this in more detail below in Section \ref{sec:model}.

The gauge transformations (\ref{eq:2bis}) are sometimes called ``free abelian gauge transformations'' to emphasize that the scalar fields are uncharged and do not transform under them.  This does not mean that the abelian vector fields themselves are free since non linear terms (non minimal couplings) are allowed in (\ref{eq:Starting}).

This class of models contains the vector-scalar sectors of
ungauged extended supergravities, of which ${\mathcal N}=4$
\cite{Das:1977uy,Cremmer:1977tc,Cremmer:1977tt} and ${\mathcal N}=8$
\cite{Cremmer:1978km,Cremmer:1979up} supergravities offer prime
examples. These will be considered in detail in Sections
\ref{sec:second order} and \ref{sec:applications}.  Born-Infeld type
generalizations \cite{Gibbons:1995cv} are also covered together
with first order manifestly duality invariant formulations
\cite{Deser:1976iy,Bunster:2010wv,Bunster:2011aw}, which fall into
this class when reformulated with suitable additional scalar fields
\cite{Barnich:2007uu}.

From previous chapter, we recall that consistent deformations of a gauge invariant action are deformations
that preserve the number (but not necessarily the form or the algebra)
of the gauge symmetries.  In the supergravity context, these are
called ``gaugings'', and the deformed theories are called gauged
supergravities, even though the undeformed theories possess already
a gauge freedom. We shall
consider only local deformations, i.e., deformations of the Lagrangian
by functions of the fields and their derivatives up to some finite
(but unspecified) order.

Gaugings in extended supergravities have a long history that goes back
to \cite{Freedman:1976aw,Freedman:1978ra,Fradkin:1976xz,Zachos:1978iw,Zachos:1979uh}.
For maximal supergravity, the first gauging
has been performed in \cite{deWit:1981sst} in the Lagrangian
formulation of \cite{Cremmer:1979up}, which involves a specific choice
of so-called duality frame. More recent gaugings
involving a change of the duality frame have been constructed in
\cite{DallAgata:2012mfj}.

These works consider from the very beginning deformations in which the
vector fields become Yang-Mills connections for a non-abelian
deformation of the original abelian gauge algebra.  The corresponding
couplings are induced through the replacement of the abelian
curvatures by non-abelian ones and the ordinary derivatives by
covariant ones, plus possible additional couplings necessary for
consistency.  Given the question that whether this
embraces all possible consistent deformations, it
was shown in \cite{Henneaux:2017kbx}, see Sections \ref{sec:embedding} and \ref{sec:BV-def-Embed}, that the space of consistent
deformations in the embedding formalism is isomorphic to the space
of consistent deformations for the action (\ref{eq:1bis}) written in
the duality frame picked by the choice of embedding tensor.  For that
reason, one can investigate the question of gaugings by taking
(\ref{eq:1bis}) as starting point of the deformation procedure,
provided one allows the scalar field dependence in the vector piece of
the Lagrangian to cover all possible choices of duality frame.
By doing so, one does not miss
any of the gaugings available in the embedding tensor formalism. This is what we are going to do in this chapter.

Since the BV-BRST formalism \cite{Barnich:1993vg} is a systematic way to explore deformations of theories with a gauge
freedom, we completely
characterize the BRST cohomology for the theories defined by (\ref{eq:1bis}), i.e., we completely characterize, in four spacetime dimensions, the  
deformations of  abelian vector fields coupled non-minimally to scalar 
chargeless fields with a possibly non polynomial dependence on the (undifferentiated) scalar fields.

In particular, we show that besides the obvious deformations that
consist in adding gauge invariant terms to the Lagrangian without
changing the gauge symmetries, the gaugings can be related to the
global symmetries of the action (\ref{eq:1bis}).  These gaugings
modify the form of the gauge transformations.

The global symmetries can be classified into two different types:
\begin{enumerate}[(i)]
\item
global symmetries with covariantizable Noether currents, where by
``covariantizable'', we mean that one can choose the ambiguities in
the Noether currents so as to take them gauge invariant ($V$-type
  symmetries),
  
\item
global symmetries with non-covariantizable
Noether currents which in turn can be subdivided into
two subtypes:
\begin{enumerate}[(a)]
\item
global symmetries with non-covariantizable Noether
currents that lead to a deformation that does not modify the gauge
algebra ($W$-type symmetries),

\item
global symmetries with
non-covariantizable Noether currents that lead to a deformation that
does modify also the gauge algebra ($U$-type symmetries).
\end{enumerate}

\end{enumerate}

Only the type (i) global symmetries directly gives rise to an
infinitesimal consistent deformation through minimal coupling of the
corresponding current to the vector potentials. The gaugings associated with the other types of global symmetries need to satisfy
additional constraints.
  
The Noether current corresponding to the global symmetries of type (a) contains 
non-gauge invariant Chern-Simons terms that cannot be removed by
suitably adjusting trivial contributions. The global symmetries of
type (b) are associated with
ordinary free abelian gauge symmetries with co-dimension 2
conservation laws (see e.g.~\cite{Julia:1980gn} for an early
discussion).  The divergence of a current of type (a) is itself gauge invariant, while the divergence of a current of type (b) is not. Yang-Mills gaugings are associated with currents of type (b) and are hence of $U$-type. Topological couplings \cite{deWit:1987ph} are associated with non-covariantizable Noether currents of either type (a) or (b). Charging deformations (if available), in which the scalar fields become charged but the gauge transformations of the vector fields  are not modified and remain therefore abelian, are of $V$- or $W$-type.

The BRST deformation procedure applies not only to the consistent
first order deformations, but also to higher orders where one might
encounter obstructions.  That procedure provides a natural
deformation-theoretic interpretation of quadratic constraints and
higher order constraints in terms of what is called the antibracket
map.

After establishing general theorems on the BRST cohomology valid without assuming a specific form of the Lagrangian or the rigid symmetries, including the above classification of the deformations and useful triangular properties of their algebra, we turn to various models that have been considered in the literature, for which we completely compute the deformations of $U$ and $W$-types. 

In Section \ref{sec:gauging}, we 
compute the local BRST cohomology of the
models described by the action (\ref{eq:1bis}).  This is done by
following the method of \cite{Barnich:1994db,Barnich:1994mt} where the
BRST cohomology was computed for arbitrary compact -- in fact
reductive -- gauge group.  The difficulty in the computation comes
from the free abelian factors by which we
mean abelian factors of the gauge algebra such that all matter fields
are uncharged, i.e., invariant under the associated gauge
transformations.  This is precisely the case relevant to the action
(\ref{eq:1bis}), which needs thus special care.  The method of
\cite{Barnich:1994db,Barnich:1994mt} is based on an expansion
according to the antifield number. It makes direct contact with
symmetries and conservation laws through the lowest antifield number
piece of the BRST differential, that is the Koszul-Tate
differential, which involves the equations of motion
\cite{Fisch:1989rp,Henneaux:1990rx}.  The Noether charges appear
through the characteristic cohomology, given by the local
cohomology of the Koszul-Tate differential \cite{Barnich:1994db}.

We then discuss in Section \ref{sec:AntiMap} the structure of the
antibracket map, which is relevant for the consistency of the
deformation at second order and the possible appearance of
obstructions, and provide information on the structure of the global
symmetry algebra.

The method of \cite{Barnich:1994db,Barnich:1994mt} provides the
general structure of the BRST cocycles in terms of conserved
currents. In order to reach more complete results, one must use
additional information specific to each model.  We therefore specify
further the models in Section \ref{sec:second order}, where we
concentrate on scalar-coupled second order Lagrangians that are
quadratic in the vector fields and their derivatives.  These
specialized models still cover the scalar-vector sectors of extended
supergravities.  Explicit examples are treated in detail to illustrate
the method in Section \ref{sec:applications}, where complete results
for the local BRST cohomology, up to the determination of $V$-type
symmetries, are worked out.  In Section \ref{sec:first-order-actions},
we then illustrate our techniques in the case of the manifestly
duality-symmetric first order action \eqref{eq:symlag} of \cite{Bunster:2011aw}, in the
formulation of \cite{Barnich:2007uu}, which is adapted to the direct
use of the methods developed here.

\section{Abelian vector-scalar models in 4 dimensions}\label{sec:gauging}

\subsection{Structure of the models}\label{sec:model}

We now apply the formalism discussed in Section \ref{sec:BV-def-formalism} to the scalar-vector models described by the action (\ref{eq:1bis}).
We write $ \mathcal{L}_0= \mathcal{L}_S[\phi^i]+\mathcal{L}_V[A^I_\mu,\phi^i]$.
  In four spacetime dimensions, there is no Chern-Simons term in the Lagrangian, which can be assumed to be strictly gauge invariant and not just invariant up to a total derivative.  
Gauge invariant functions are functions
that depend on $F^I_{\mu\nu}=\d_\mu A_\nu^I-\d_\nu A^I_\mu$, $\phi^i$
and their derivatives, but not on
$A_\mu^I,\d_{(\nu} A_{\mu)}^I,\d_{(\nu_1}\d_{\nu_2} A_{\mu)}^I$, etc.  Thus $\mathcal{L}_V[A^I_\mu,\phi^i]$ depends on the vector potentials $A^I_\mu$ only through $F^I_{\mu\nu}=\d_\mu A_\nu^I-\d_\nu A^I_\mu$
and their derivatives.

We define 
\begin{equation}
  \label{eq:47}
  \vddl{\mathcal L_V}{F^I_{\mu\nu}}=\half (\star G_I)^{\mu\nu}
\end{equation}
where the $(\star G_I)^{\mu\nu}$ are also manifestly gauge invariant functions.
The equations of motion
for the vector fields can be written as 
\begin{equation}
  \label{eq:4bis}
  \vddl{\mathcal L_0}{A_\mu^I}=\d_\nu(\star G_I)^{\mu\nu}
\end{equation}
and the Lagrangian can
be taken to be 
\begin{equation}
  \mathcal{L}_0= \mathcal{L}_S[\phi^i]+\mathcal{L}_V[A^I_\mu,\phi^i],
  \quad d^4x\, \mathcal{L}_V=\int_0^1\frac{dt}{t} [G_I F^I]
  [t A^I_\mu,\phi^i]. \label{4.4}
\end{equation}

The vector Lagrangian $\mathcal{L}_V$ has been written in the form of a homotopic integral \cite{Anderson92introductionto}. It simply counts the number of vector fields in the integrand in order to produce the correct prefactor in each term coming from $G_I$. This is a compact way to include all possible terms in $G_I$. As an example consider the case where $G_I$ is a series of homogeneous functions of order $n$ in the field strength $F^I$, i.e. $G_{I} = \sum_{n=1} (G_n)_{I}$ with
\begin{align}
(G_n)_{I\mu\nu} &= (\alpha_n)_{II_1...I_n} (F^n)^{I_1...I_n}_{\mu\nu} + (\beta_n)_{II_1...I_n} ((\star F)^n)^{I_1...I_n}_{\mu\nu}\nonumber\\
 &+(\zeta_n)_{II_1...I_n} (F^k \wedge (\star F)^{n-k})^{I_1...I_n}_{\mu\nu},
\end{align}
where $(\alpha_n)_{II_1...I_n}$, $(\beta_n)_{II_1...I_n}$ and $(\zeta_n)_{II_1...I_n}$ are in general functions of scalar fields $\phi^i$.
The Lagrangian is then written as
\begin{equation}
   d^4x\, \mathcal{L}_V=\sum_n \frac{1}{n+1} (G_n)_I(A^J_\mu,\phi^i) F^I(A^J_\mu,\phi^i). 
\end{equation}

The associated solution to the BV master 
equation is given by
\begin{equation}
  \label{eq:3bis}
  S^{(0)}=S_0+\int \!d^4\!x\, A^{*\mu}_I \d_\mu C^I.
\end{equation}
The ghost number, antifield number and pure ghost number of the various fields and antifields are given in the Table \eqref{tab:YM_gh}.
\begin{table}[h!]
\begin{center}
\begin{tabular}{|c|c|c|c|}
\hline 
        ~  &  gh  &  antifld  & puregh  \\
\hline 
        $\phi^i$  &  0  &  0 & 0   \\ 
\hline  
           $A^I_{\mu}$  &  0  &  0 & 0   \\ 
\hline
         $C^I$   & 1 & 0 & 1 \\
\hline
        $\phi^\ast_i$  &  -1  &  1 & 0   \\
\hline  
         $A^\ast_{I\mu}$  &  -1  &  1 & 0   \\ 
\hline
        $C^\ast_I$  &  -2  &  2 & 0  \\ 
\hline
\end{tabular}
\end{center}
\caption{The list of pure ghost number (puregh), anti-field number (antifld) and ghost number (gh) of each of the fields and anti-fields. The ghost number is defined as $\textrm{gh}(A)=\textrm{puregh}(A)-\textrm{antifld}(A)$, and we have in general ${\rm gh}(\Phi^*_A)=-{\rm gh}(\Phi^A)-1$. }\label{tab:YM_gh}
\end{table}

The BRST differential $s$ splits according to antifield number as
\be
s = \delta + \gamma 
\ee 
and the action of each differentials on all field and antifields are summarized in the Table \eqref{tab:YM_diff}.

\begin{table}[h!]
\begin{center}
\begin{tabular}{|c|c|c|c|}
\hline 
        ~  &  $s$  &  $\gamma$  & $\delta$  \\
\hline 
        $\phi^i$  &  0  &  0 & 0   \\ 
\hline  
           $A^I_{\mu}$  &  $\partial_{\mu} C^I$  &  $\d_\mu C^I$ & 0   \\ 
\hline
         $C^I$   & 0 & 0 & 0 \\
\hline
        $\phi^\ast_i$  &  $\frac{\delta \cL_0}{\delta \phi^i}$  &  0 & $\frac{\delta \cL_0}{\delta \phi^i}$   \\
\hline  
         $A^\ast_{I\mu}$  &  $\d_\nu(\star G_I)^{\mu\nu}$  &  0 & $\d_\nu(\star G_I)^{\mu\nu}$   \\ 
\hline
        $C^\ast_I$  &  $- \d_\mu A^{*\mu}$  &  0 & $- \d_\mu A^{*\mu}$  \\ 
\hline
\end{tabular}
\end{center}
\caption{The action of BRST differential $s$, the horizontal differential $\gamma$ and the Koszul-Tate differential $\delta$ on each field and antifield.}\label{tab:YM_diff}
\end{table}

The Koszul-Tate differential $\delta$ and the differential $\gamma$ have antifield number $-1$ and $0$ respectively and we have
\be 
\delta^2 = 0, \; \; \; \delta \gamma + \gamma
\delta = 0, \; \; \; \; \gamma^2 = 0.  
\ee 

In terms of the Koszul-Tate differential, the cocycle condition for
$m$ in
characteristic cohomology takes the form $ d m + \delta n = 0$.
This 
equation is the same as the (co)cycle condition for 
$n$ in the local
(co)homology of $\delta$, which is indeed $\delta n + dm = 0$.  Using
this observation, and vanishing theorems for $H(d)$ and $H(\delta)$ in
relevant degrees, one can establish isomorphisms between the
characteristic cohomology and $H(\delta \vert d)$
\cite{Barnich:1994db}. For example, the characteristic cohomology
  $H^{n-2}_{\rm char}(d)$ is given by the 2-forms $\mu^IG_I$, while
  $H^n_2(\delta \vert d)$  is given by the 4-forms
  $d^4\! x \, \mu^I C^*_I$. The isomorphism is realized through the
  $(\delta,d)$-descent
\begin{equation}
\delta \, d^4x\, C^*_I+d \star\! A^*_I=0,\quad \delta \star \!A^*_I+d
G_I=0,
\end{equation}
where $A^*_I=dx^\mu A^*_{I\mu}$.

\subsection{Consistent deformations}

One can characterize the BRST cohomological classes with non trivial
antifield dependence in terms of conserved currents and rigid
symmetries for all values of the ghost number.  For definiteness, we
illustrate explicitly the procedure for $H^0(s \vert d)$ in maximum
form degree, which defines the consistent local deformations.  We
consider next the case of general ghost number.

The main equation to be solved, see \eqref{eq:cocycle_cond}, for $a$ is 
\begin{equation}
\label{eq:cocycle1}
sa + db = 0,
\end{equation}
where $a$ has form degree
4 and ghost number $0$.  To solve it, we expand the cocycle $a$
according to the antifield number, \be a = a_0 + a_1 + a_2.  \ee 
Because $a$ has total ghost number zero, each term $a_n$ has antifield number $n$ and pure ghost number (degree in the ghosts) $n$ as well.
As shown in \cite{Barnich:1994mt}, the expansion stops at most at
antifield number $2$.  The term $a_0$ is the (first order) deformation
of the Lagrangian.  A non-vanishing $a_1$ corresponds to a deformation
of the gauge variations, while a non-vanishing $a_2$ corresponds to a
deformation of the gauge algebra. All three terms are related by the cocycle condition (\ref{eq:cocycle1}).

\subsubsection{Solutions of $U$-type ($a_2$ non trivial)}

The first case to consider is when $a_2$ is non-trivial.  This defines
``class I'' solutions in the terminology of \cite{Barnich:1994mt},
which we call here ``$U$-type'' solutions to comply with the general
terminology introduced below.  One has from the general theorems of
\cite{Barnich:1994db,Barnich:1994mt} on the invariant characteristic
cohomology that \be a_2 = d^4 x\,C^*_I \Theta^I \ee with \be
\Theta^I=\frac{1}{2!}{f^I}_{J_1J_{2}} C^{J_1} C^{J_{2}} .  \ee 
Here
${f^I}_{J_1J_{2}}$ are some constants, antisymmetric in $J_1$, $J_2$.
The reason why the coefficient $d^4 x\,C^*_I $ of the ghosts in $a_2$
is determined by the characteristic cohomology follows from the
equation $\delta a_2 + \gamma a_1 + db_1 =0$ that $a_2$ must fulfill
in order for $a$ to be a cocycle of $H(s \vert d)$.  Given that $a_2$
has antifield number equal to $2$, and the isomorphisms $H^{g,n}(s\vert d)\simeq H_{-g}^n(\delta\vert d) \simeq H^{n+g}_{\textrm{char}}(d\vert \delta)$ for negative $g$, it is the 
characteristic cohomology
in form degree $n-2=2$ that is relevant\footnote{The precise way to express the relation between the local cohomology of $\delta$
and the highest term of the equation obeyed by $a$ is given in
section 7 of \cite{Barnich:1994mt}: $a_2$ must be a non trivial representative 
of  invariant cohomology $H_{inv}(\delta \vert d)$, more precisely it must come from 
$H_{inv,2}^4(\delta \vert d)$ in ghost number zero. This relates the U-type deformations 
to the free abelian factors of the undeformed gauge group.}. The emergence
of the characteristic cohomology in the computation of $H(s \vert d)$
will be observed again for $a_1$ below, where it will be the conserved
currents that appear.  This central feature follows from the fact that
the Koszul-Tate differential, which encapsulates the equations of
motion, is an essential building block of the BRST differential. We must now find the 
lower terms $a_1+a_0$ and relate them as expected to Noether currents that 
correspond to $H_1^4(\delta \vert d)$.

By the argument of section 8 of \cite{Barnich:1994mt} suitably generalized in section 12, the term $a_1$ is then
found to be \be a_1 = \star A^*_IA^{K}\d_{K}\Theta^I + m_1 \ee
where $\gamma m_1 =0$ and $\d_K=\ddl{}{C^K}$.  
The term $m_1$ (to be determined by the next
equation) is linear in $C^I$ and can be taken to be linear in the
undifferentiated antifields $A^*_I$ and $\phi^*_i$ since derivatives
of these antifields, which can occur only linearly, can be redefined
away through trivial terms.  We thus write \be m_1 = \hat K=\star
A^*_I \hat g^I -\star \phi^*_i \hat \Phi^i \ee with \be \hat
g^I=dx^\mu g^{I}_{\mu K}C^{K}\;, \quad~ \hat \Phi^i=\Phi^{i}_{K}C^{K}.
\ee Here $g^{I}_{\mu K}$ and $\Phi^{i}_{K}$ are gauge invariant
functions which are arbitrary at this stage but will be constrained by
the requirement that $a_0$ exists.

We must now consider the equation $\delta a_1 + \gamma a_0 + d b_0 = 0$
that determines $a_0$ up to a solution of $\gamma a'_0 + db'_0 =0$.
This equation is equivalent to
 \begin{equation}
 \label{eq:00}
  \left(\vddl{\mathcal L_0}{A_\mu^I}\delta_{K} A_\mu^I
  +\vddl{\mathcal
    L_0}{\phi^i}\delta_{K}\phi^i \right) C^K+\gamma \alpha_0 + \d_\mu \beta_0^\mu=0, 
\end{equation}
where we have passed to dual notations ($a_0 = d^4 x \,\alpha_0$,
$db_0 = d^4 x \,\d_\mu \beta_0^\mu$) and where we have set
\begin{equation}
 \delta_{K} A_\mu^I=A^J_\mu {f^I}_{JK}+g_{\mu
    K}^I  ,\quad \delta_{K}
  \phi^i=\Phi^i_{K}.  \label{eq:TypeISymmb}
\end{equation}
Writing $\beta_0^\mu = j^\mu_K C^K + $ ``terms containing derivatives
of the ghosts'', we read from (\ref{eq:00}), by comparing the
coefficients of the undifferentiated ghosts, that \be \vddl{\mathcal
  L_0}{A_\mu^I}\delta_{K} A_\mu^I +\vddl{\mathcal
  L_0}{\phi^i}\delta_{K}\phi^i + \d_\mu j_K^\mu=0 .  \label{eq:3.19} \ee A necessary
condition for $a_0$ (and thus $a$) to exist is therefore that
$ \delta_{K} A_\mu^I$ and $ \delta_K \phi^i$ define symmetries.
 
 To proceed further and determine $a_0$, we observe that the non-gauge
 invariant term $ \vddl{\mathcal L_0}{A_\mu^I} A^J_\mu {f^I}_{JK}$ in
 $\vddl{\mathcal L_0}{A_\mu^I}\delta_{K} A_\mu^I$ can be written as
 $ \partial_\mu \left( \star G^{\nu\mu}_I A_\nu^J {f^I}_{JK}\right) $
 plus a gauge invariant term, so that
 $j^\mu_K - \star G^{\mu\nu}_I A_\nu^J {f^I}_{JK}$ has a
 gauge invariant divergence. Results on the invariant cohomology of
 $d$ \cite{Brandt:1989gy,DuboisViolette:1992ye} imply then that the
 non-gauge invariant part of such an object can only be a Chern-Simons
 form, i.e.
 $j^\mu_K - \star G^{\mu\nu}_I A_\nu^J {f^I}_{JK} = J^\mu_{K} + \half
 \epsilon^{\mu\nu\rho\sigma}A_\nu^I F_{\rho\sigma}^J h_{I|JK}$, or
\begin{equation}
  j^\mu_{K}
  =J^\mu_{K}+\star G^{\mu\nu}_I
  A_\nu^J {f^I}_{JK} + \half
  \epsilon^{\mu\nu\rho\sigma}A_\nu^I F_{\rho\sigma}^J 
  h_{I|JK}
\end{equation}  
where $J^\mu_{K}$ is gauge invariant and where the symmetries of the
constants $h_{I|JK}$ will be discussed in a moment.  It is useful to
point out that one can switch the indices $I$ and $J$ modulo a trivial
term.

The equation (\ref{eq:00}) becomes
$-(\partial_\mu j^\mu_K) \, C^K + \gamma \alpha_0 + \partial_\mu
\beta_0^\mu = 0$, i.e.,
$j^\mu_K \, (\gamma A_\mu^K) + \gamma \alpha_0 + \partial_\mu
{\beta'}_0^\mu = 0$.  The first two terms in the current yield
manifestly $\gamma$-exact terms,
\begin{equation}
J^\mu_K \,  (\gamma A_\mu^K) = \gamma(J^\mu_K \,   A_\mu^K), \; \;
\; \star G^{\mu\nu}_I 
A_\nu^J {f^I}_{JK} \, (\gamma A_\mu^K)= \frac12 \gamma (\star
G^{\mu\nu}_I A_\nu^J {f^I}_{JK} \, A_\mu^K)
\end{equation} 
and so $h_{I|JK}$ must be such that the term $A^I F^J dC^K h_{I|JK}$
is by itself $\gamma$-exact modulo $d$. This is a problem that has
been much studied in the literature through descent equations (see
e.g. \cite{DuboisViolette:1985cj}).  It has been shown that
$h_{I|JK}$ must be antisymmetric in $J$, $K$ and should have vanishing
totally antisymmetric part in order to be ``liftable" to $a_0$ and
non-trivial,
\begin{equation}
  h_{I|JK}=h_{I|[JK]},\quad  h_{[I|JK]}=0. \label{eq:Symmh}
\end{equation}

Putting things together, one finds for $a_0$
\be
a_0 = A^I\d_I\hat J + \half G_I
  A^{K}A^{L}\d_{L}\d_{K}\Theta^I + \half F^I
  A^KA^L\d_L\d_K\Theta'_I
\ee
where
\be
\hat J=\star dx^\mu J_{\mu K}
    C^{K}\;,\quad \Theta'_I=\frac{1}{2}h_{I|J_1 J_2}C^{J_1} C^{ J_{2}} .
\ee

A non-trivial $U$-solution
modifies the gauge
algebra. Deformations of the Yang-Mills type belong to this class. A
$U$-solution is characterized by constants ${f^I}_{J_1J_{2}}$ which
are antisymmetric in $J_1$, $J_2$. These constants must be such that
there exist gauge invariant functions $g^{I}_{\mu K}$ and
$\Phi^{i}_{K}$ such that $ \delta_{K} A_\mu^I$ and $ \delta_K \phi^i$
define symmetries of the undeformed Lagrangian.  Here
$ \delta_{K} A_\mu^I$ and $ \delta_K \phi^i$ are given by
(\ref{eq:TypeISymmb}).  Furthermore, the $h$-term in the corresponding
conserved current (if any) must fulfill (\ref{eq:Symmh}).  The
deformation $a_0$ of the Lagrangian takes the Noether-like form.

Given the ``head'' $a_2$ of a $U$-type solution, characterized by a set of ${f^I}_{J_1J_{2}}$'s,  the lower terms $a_1$ and $a_0$, and in particular the $h$-piece, are not uniquely determined.  One can always add solutions of $W$, $V$ or $I$-types described below, which have the property that they have no $a_2$-piece.  
Hence one may require that the completion of the ``head'' $a_2$ of a $U$-type solution should be chosen to vanish when $a_2$ itself vanishes. But this leaves some freedom in the completion of $a_2$, since for instance any $W$-type solution multiplied by a component of ${f^I}_{J_1J_{2}}$ will vanish when the ${f^I}_{J_1J_{2}}$'s are set to zero. The situation has a triangular nature since two $U$-type solutions with the same $a_2$ differ by solutions of ``lower'' types, for which there might not be a canonical choice.
 
 Note that further constraints on ${f^I}_{J_1J_{2}}$ (notably the
 Jacobi identity) arise at second order in the deformation parameter.

 \subsubsection{Solutions of $W$ and $V$-type (vanishing $a_2$ but
   $a_1$ non trivial)}

These solutions are called ``class II'' solutions in \cite{Barnich:1994mt}. Here we have \be a = a_0 + a_1 \ee and $a_1$ can be taken to be gauge
invariant, i.e., annihilated by $\gamma$ \cite{Barnich:1994mt}.  We
thus find \be a_1 = \hat K=\star A^*_I \hat g^I -\star \phi^*_i \hat
\Phi^i \ee with \be \hat g^I=dx^\mu g^{I}_{\mu K}C^{K}\;, \quad~ \hat
\Phi^i=\Phi^{i}_{K}C^{K} \, . \ee Here $g^{I}_{\mu K}$ and $\Phi^{i}_{K}$
are again gauge invariant functions, which we still denote by the same
letters as above, although they are independent from the similar
functions related to the constants ${f^I}_{J_1J_{2}}$.  We also set
\begin{equation}
 \delta_{K} A_\mu^I=g_{\mu
    K}^I  ,\quad \delta_{K}
  \phi^i=\Phi^i_{K}.  \label{eq:DefKbb} 
\end{equation}

The equation $\delta a_1 + \gamma a_0 + d b_0 = 0$ implies then, as
above,
\begin{equation} \vddl{\mathcal L_0}{A_\mu^I}\delta_{K} A_\mu^I
  +\vddl{\mathcal L_0}{\phi^i}\delta_{K}\phi^i + \d_\mu j_K^\mu=0\,. \label{eq:3.28}
\end{equation}
A necessary condition for $a_0$ (and thus $a$) to exist is
therefore that $ \delta_{K} A_\mu^I$ and $ \delta_K \phi^i$ given by
(\ref{eq:DefKbb}) define symmetries. Equation (\ref{eq:3.28}) take the same form as Eq. (\ref{eq:3.19}), but there is an important difference: the divergence of the current  $j^\mu_{K}$ is now gauge invariant, unlike in (\ref{eq:3.19}) that the divergence of current, due to the contribution coming from $a_2$, is not gauge invariant.

The current takes the form
 \begin{equation}
  j^\mu_{K}
  =J^\mu_{K}+ \half
  \epsilon^{\mu\nu\rho\sigma}A_\nu^I F_{\rho\sigma}^J 
  h_{I|JK},
\end{equation}
(with $h_{I|JK}$ fulfilling the above symmetry properties) yielding
\be a_0 = A^I\d_I\hat J + \half F^I A^KA^L\d_L\d_K\Theta'_I \ee where still
\be \hat J=\star dx^\mu J_{\mu K} C^{K}\;,\quad
\Theta'_I=\frac{1}{2}h_{I|J_1 J_2}C^{J_1} C^{ J_{2}} .  \ee We define
$W$-type solutions to have $h_{I|JK} \neq 0$, while $V$-type solutions have
$h_{I|JK}=0$.  Both these types deform the gauge transformations but
not their algebra (to first order in the deformation).  They are
determined by rigid symmetries of the undeformed Lagrangian 
with gauge invariant variations (\ref{eq:DefKbb}). The $V$-type have gauge invariant currents, while the
currents of the $W$-type contain a non-gauge invariant piece.

Note that again, the solutions of $W$- and $V$-types are determined up to a solution of lower type with no $a_1$-``head'', and that there might not be a canonical choice. In fact one may require similarly that $W$-type transformations become trivial when 
$h_{I|JK}$ tends to zero.

\subsubsection{Solutions of $I$-type (vanishing $a_2$ and $a_1$)}

In that case, 
\be
a = a_0
\ee
with $\gamma a_0 + db_0 =0$.

Since there is no Chern-Simons term in four dimensions, one can assume
that $b_0 =0$.  The deformation $b_0$ is therefore a gauge invariant
function, i.e., a function of the abelian curvatures $F_{\mu \nu}^I$,
the scalar fields, and their derivatives. The $I$-type deformations
neither deform the gauge transformations nor (a fortiori) the gauge
algebra.  Born-Infeld deformations belong to this type. They are
called ``class III'' solutions in \cite{Barnich:1994mt}.

\subsection{Local BRST cohomology at other ghost numbers}
\label{sec:locBRST}

\subsubsection{$h$-terms}

The previous discussion can be repeated straightforwardly at all ghost
numbers. The analysis proceeds as above.  The tools necessary to
handle the ``$h$-term'' in the non gauge invariant ``currents'' have
been generalized to higher ghost numbers through familiar means and
can be found in
\cite{DuboisViolette:1985hc,DuboisViolette:1985jb,DuboisViolette:1985cj}.
 
The $h$-terms belong to the ``small'' or ``universal'' algebra
involving only the $1$-forms $A^I$, the $2$-forms $F^I = dA^I$, the
ghosts $C^I$ and their exterior derivative.  The product is the
exterior product.  One describes the $h$-term through a $(\gamma,d)$-descent
equation and what is called the ``bottom'' of that descent, which is
annihilated by $\gamma$ and has form degree $<4$ in four dimensions.
The only possibilities in the free abelian case are the $2$-forms \be
\frac1m h_{I \vert J_1 \cdots J_m} F^I C^{J_1} \cdots C^{J_m} \ee
where \be h_{I \vert J_1 \cdots J_m} = h_{I \vert [J_1 \cdots J_m]} .
\ee One can assume $h_{[I \vert J_1 \cdots J_m]} = 0$ since the
totally antisymmetric part gives a trivial bottom. The lift of this
bottom goes two steps, up to the $4$-form \be h_{I \vert J_1 J_2
  \cdots J_m} F^I F^{J_1}C^{J_2} \cdots C^{J_m} \ee producing along
the way  a $3$-form \be h_{I \vert J_1 J_2 \cdots J_m} F^I
A^{J_1}C^{J_2} \cdots C^{J_m}   \ee
which has the property of not being gauge (BRST) invariant although its exterior derivative is invariant (modulo trivial terms). 

\subsubsection{Explicit description of cohomology}

By applying the above method, one finds that the local BRST cohomology
of the models of Section \ref{sec:model} can be described along
exactly the same lines as given below.  Note that the cohomology at
negative ghost numbers reflect general properties of the
characteristic cohomology that go beyond the mere models considered
here \cite{Barnich:1994db}.

\begin{enumerate}[(i)]
\item $H^g(s|d)$ is empty for $g\leqslant -3$.  

\item $H^{-2}(s|d)$ is represented by the $4$-forms 
  \begin{equation}
    U^{-2}=\mu^I d^4x\, C^*_I\label{eq:11a}. 
\end{equation}
If $A^*_I=dx^\mu A^*_{I\mu}$, the associated descent equations
are
  \begin{equation}
s\, d^4x\, C^*_I+d \star A^*_I=0,\quad s \star A^*_I+d
G_I=0,\quad sG_I=0.\label{eq:14}
\end{equation}
Characteristic cohomology $H^{n-2}_{\rm char}(d)$ is then
represented by the 2-forms $\mu^IG_I$.

\item Several types of cohomology classes in ghost numbers $g \geqslant
  -1$, which we call $U$, $V$ and $W$-type, can be described by
  constants ${f^I}_{JK_1\dots K_{g+1}}$ which are antisymmetric in the
  last $g+2$ indices,
\begin{equation}
    \label{sk}
   {f^I}_{JK_1\dots K_{g+1}}={f^I}_{[JK_1\dots K_{g+1}]}, 
\end{equation}
and constants $h_{I|JK_1\dots K_{g+1}}$ that are antisymmetric in the last $g+2$
indices but without any totally antisymmetric part\footnote{We write
  $h_{IJ}:= h_{I|J}$ for $g=-1$.}, 
\begin{equation}
    \label{Y1}
  h_{I|JK_1\dots K_{g+1}}=h_{I|[JK_1\dots K_{g+1}]},\quad  h_{[I|JK_1\dots K_{g+1}]}=0, 
\end{equation}
together with gauge invariant
functions $g^I_{\mu K_1\dots K_{g+1}},\Phi^i_{K_1\dots K_{g+1}}$ 
that are antisymmetric in the last $g+1$ indices. They are constrained by
the requirement that the transformations
\begin{equation}
  \label{eq:16}
 \delta_{K_1\dots K_{g+1}} A_\mu^I=A^J_\mu {f^I}_{JK_1\dots K_{g+1}}+g_{\mu
    K_1\dots K_{g+1}}^I  ,\quad \delta_{K_1\dots K_{g+1}}
  \phi^i=\Phi^i_{K_1\dots K_{g+1}}, 
\end{equation}
define symmetries of the action in the sense that 
 \begin{equation}
  \label{eq:47A}
  \vddl{\mathcal L_0}{A_\mu^I}\delta_{K_1\dots K_{g+1}} A_\mu^I
  +\vddl{\mathcal
    L_0}{\phi^i}\delta_{K_1\dots K_{g+1}}\phi^i+\d_\mu j^\mu_{K_1\dots K_{g+1}}=0,
\end{equation}
with currents $j^\mu_{K_1\dots K_{g+1}}$ that are antisymmetric in the
last $g+1$ indices. This can be made more precise by making the gauge
(non-)invariance properties of these currents manifest. One finds
\begin{equation}\label{eq:current}
  j^\mu_{K_1\dots K_{g+1}}
  =J^\mu_{K_1\dots K_{g+1}}+\star G^{\mu\nu}_I
  A_\nu^J {f^I}_{JK_1\dots K_{g+1}} + \half
  \epsilon^{\mu\nu\rho\sigma}A_\nu^I F_{\rho\sigma}^J 
  h_{I|JK_1\dots K_{g+1}},
\end{equation}
where $J^\mu_{K_1\dots K_{g+1}}$ is gauge invariant and
antisymmetric in the lower $g+1$ indices. When taking into account that
\begin{equation}
  \label{eq:15}
  G_I F^J= d(G_I A^J+\star A^*_I C^J)+s(\star A^*_IA^J+d^4x C^*_IC^J) ,
  \quad F^IF^J=d(A^I F^J),  
\end{equation}
and defining 
\begin{equation}
  \begin{split}
     \Theta^I &=\frac{1}{(g+2)!}{f^I}_{J_1\dots J_{g+2}}C^{J_1\dots
      J_{g+2}}\;, \\
       \Theta'_I &=\frac{1}{(g+2)!}h_{I|J_1\dots
        J_{g+2}}C^{J_1\dots J_{g+2}}\;,\\
    \hat J &=\star dx^\mu J_{\mu K_1\dots K_{g+1}}
    \frac{1}{(g+1)!}C^{K_1\dots K_{g+1}}\;, \\
    \hat K &=(\star A^*_I \hat g^I -\star \phi^*_i \hat \Phi^i)\;,\\ 
    \hat g^I &=\frac{1}{(g+1)!}dx^\mu g^{I}_{\mu K_1\dots
      K_{g+1}}C^{K_1\dots K_{g+1}}\;, \\
      \hat\Phi^i &=\frac{1}{(g+1)!}\Phi^{i}_{K_1\dots
      K_{g+1}}C^{K_1\dots K_{g+1}}\;,
\label{eq:17b}
  \end{split}
\end{equation} 
where $C^{{K_1}\dots K_g}=C^{K_1}\dots C^{K_g}$, the ``global symmetry" condition \eqref{eq:47A} is equivalent to a 
$(s,d)$-obstruction equation, 
\begin{equation}
  G_I
  F^J\d_J\Theta^I+F^IF^J\d_{J}\Theta'_{I}+s(\hat K +A^I\d_I \hat J) + d\hat J
 = 0,\label{eq:16b}
\end{equation}
with $\d_I=\ddl{}{C^I}$. Note that the last two terms combine into
\[d[\star dx^\mu J_{\mu K_1\dots K_{g+1}}]
  \frac{1}{(g+1)!}C^{K_1}\dots C^{K_{g+1}},\] so that this equation
involves gauge invariant quantities only. It is this form that arises
in a systematic analysis of the descent equations. One can now
distinguish the three types of solutions.

\begin{enumerate}[a)]
\item $U$-type corresponds to solutions with non vanishing
  ${f^I}_{JK_1\dots K_{g+1}}$ and particular
  ${}^U h_{I|JK_1\dots K_{g+1}}$, ${}^U g^I_{\mu K_1\dots K_{g+1}}$, ${}^U
  \Phi^i_{K_1\dots K_{g+1}}$, ${}^U J_{\mu K_1\dots K_{g+1}}$ that
  vanish when the $f$'s vanish (and that may be vanishing even when
  the $f$'s do not). As we explained above, different choices of the particular completion  ${}^U h_{I|JK_1\dots K_{g+1}}$, ${}^U g^I_{\mu K_1\dots K_{g+1}}$, ${}^U
  \Phi^i_{K_1\dots K_{g+1}}$, ${}^U J_{\mu K_1\dots K_{g+1}}$ of $a_2$ exist and there might not be a canonical one, but a completion exists if the $U$-type solution is indeed a solution. Similar ambiguity holds for the solutions of $W$ and $V$-types described below.  A $U$-type solution is trivial if and only
  if $f$'s vanish. Denoting by $\hat K_U,\hat J_U,(\Theta'_U)_I$, the
  expressions as in \eqref{eq:17b} but involving the particular
  solutions, the associated BRST cohomology classes are represented by
\begin{multline}
  \label{eq:18b}
  U=(d^4x C^*_I +\star A^*_IA^{K}\d_{K}+ \half G_I
  A^{K}A^{L}\d_{L}\d_{K})\Theta^I \\+\hat K_U+\half F^I
  A^KA^L\d_L\d_K(\Theta'_U)_I+A^I\d_I\hat J_U,
\end{multline}
with $sU+d(\star
A^*_I\Theta^I+G_IA^J\d_J\Theta^I+F^IA^J\d_J(\Theta'_U)_I+\hat J_U)=0\,$;

\item $W$-type corresponds to solutions with vanishing $f$'s but non
vanishing $h_{I|J K_1\dots K_{g+1}}$ and particular ${}^W g^I_{\mu
  K_1\dots K_{g+1}},{}^W 
\Phi^i_{K_1\dots K_{g+1}}, {}^W J_{\mu K_1\dots K_{g+1}}$ that may be chosen to vanish
when the $h$'s vanish. Such solutions are trivial when the $h$'s
vanish. With the obvious notation, the associated BRST
cohomology classes are represented by 
\begin{equation}
  \label{eq:19b}
  W=\hat K_W+\half F^IA^KA^L\d_L\d_K\Theta'_I+A^I\d_I\hat J_W,
\end{equation}
with $sW+d(F^IA^J\d_J\Theta'_I+\hat J_W)=0\,$; 

\item $V$-type corresponds to solutions with vanishing $f$'s and
  $h$'s. They are represented by
\begin{equation}
  \label{eq:20}
  V=\hat K_V+A^I\d_I \hat J_V,
\end{equation}
with $s V+d\hat J_V=0\,$ and $s \hat J_V=0\,$.
$V$ and its descent have depth $1$.

\end{enumerate}

\item Lastly, $I$-type cohomology classes exist in ghost numbers
  $g\geqslant 0$ and are described by
\begin{equation}
\hat I=d^4x\, \frac{1}{g!} I_{K_1\dots K_g}C^{K_1}\dots C^{K_g}\label{eq:12}
\end{equation}
with $s\hat I=0$, i.e., gauge invariant $I_{K_1\dots K_g}$ that are
completely antisymmetric in the $K$ indices. Such classes are to be
considered trivial if the $I_{K_1\dots K_s}$ vanish on-shell up to a
total derivative. This can again be made more precise by making the
gauge (non-)invariance properties manifest: an element of $I$-type class is
trivial if and only if
\begin{equation}
  \label{eq:13}
  d^4x\, I_{K_1\dots K_g}\approx dJ_{K_1\dots
    K_g}+{m^I}_{JK_1\dots K_g}G_IF^J+\half F^IF^J m'_{IJK_1\dots K_g},
\end{equation}
where $J_{K_1\dots K_g}$ are gauge invariant $3$ forms that are
completely antisymmetric in the $K$ indices, while
${m^I}_{JK_1\dots K_g},m'_{IJK_1\dots K_g}$ are constants that are
completely antisymmetric in the last $g+1$ indices. Note also that the
on-shell vanishing terms in \eqref{eq:13} need to be gauge invariant.
When there are suitable restrictions on the space of gauge invariant
functions (such as for instance $x^\mu$ independent, Lorentz invariant
polynomials with power counting restrictions) one may sometimes
construct an explicit basis of non-trivial gauge invariant $4$ forms,
in the sense that if $d^4x I\approx \rho^{\cA}I_{\cA}+d\omega^{3}$ and
$\rho^{\cA}I_{\cA}\approx d\omega^{3}$, then $\rho^{\cA}=0$. The
associated BRST cohomology classes are then parametrized by constants
${\rho^{\cA}}_{K_1\dots K_g}$.

\end{enumerate}

At a given ghost number $g\geqslant -1$, the cohomology is the direct sum
of elements of type $U,W,V$ and also $I$ when $g\geqslant 0$.

This completes our general discussion of the local BRST cohomology.
Reference \cite{Barnich:1994mt} also considered simple factors in
addition to the abelian factors, as well as any spacetime dimension
$\geq 3$. One can extend the above results to cover these cases. The computation of the local
BRST cohomology $H^{*,*}(s \vert d)$ for gauge models involving general 
reductive gauge algebras has been done in \cite{Barnich:2018nqa} by following the different route adopted in
\cite{Barnich:2000zw}, which did not consider free abelian factors in
full generality. As requested by the analysis of the deformations of
the action (\ref{eq:1bis}), reference \cite{Barnich:2018nqa} generalizes
Theorem 11.1 of \cite{Barnich:2000zw} to arbitrary reductive Lie
algebras that include also (free) abelian factors (and in any
spacetime dimension $\geq 3$).

\subsubsection{Depth of solutions}

The depth of the various BRST cocycles plays a key role in the analysis
of the higher-order consistency condition. The $U$-type and $W$-type solutions have depth $2$ because they
involve $A_\mu j^\mu$ with a non-gauge invariant current.  The
$V$-type solutions have depth $1$ because the Noether term
$A_\mu j^\mu$ of them involves a gauge invariant current.  Finally,
$I$-type solutions clearly have depth $0$.

\section{Antibracket map and structure of symmetries}
\label{sec:AntiMap}

\subsection{Antibracket map in cohomology}
We now investigate the antibracket map
$H^g \otimes H^{g'} \to H^{g+g'+1}$ for the different types of
cohomology classes described above. It follows from the detailed
discussion of the cohomology in Section \ref{sec:locBRST} that the
shortest non trivial length of descents, the ``depth'', of elements of
type $U,W,V,I$ is $2$, $2$, $1$, $0$. In particular, the antibracket map is
sensitive to the depth of its arguments: as we have seen in Section \ref{sec:antibr-maps-desc} the depth of the map is
less than or equal to the depth of its most shallow element.

The antibracket map involving $U^{-2}=\mu^I d^4x\, C^*_I$ in $H^{-2}$ is given by
\begin{equation}
(\cdot,U^{-2}) : H^g \to H^{g-1}, \quad \omega^{g,n}\mapsto
\vddl{{}^R\omega^{g,n}}{C^I}\mu^I.
\label{eq:53A}
\end{equation}
More explicitly, it is trivial for $g=-2$. It is also trivial for
$g=-1$ except for $U$-type where it is described by
$f\indices{^I_J}\mapsto f\indices{^I_J}\mu^J$. For $g> 0$, it is
described by
$\rho^{\cA}_{K_1\dots K_g}\mapsto \rho^{\cA}_{K_1\dots K_g}\mu^{K_g}$
for $I$-type,
$k\indices{^{v_1}_{K_1\dots K_{g+1}}}\mapsto
k\indices{^{v_1}_{K_1\dots K_{g+1}}}\mu^{K_{g+1}}$ for $V$-type, 
$h_{IJK_1\dots K_{g+1}}\mapsto h_{IJK_1\dots K_{g+1}}\mu^{K_{g+1}}$
and
$f\indices{^I_{JK_1\dots K_{g+1}}} \mapsto f\indices{^I_{JK_1\dots
    K_{g+1}}}\mu^{K_{g+1}}$ for $U$- and $W$-type.

The antibracket map for $g, g' \geqslant -1$ has the following
triangular structure:
\begin{equation} \label{table}
\begin{array}{c|c|c|c|c}
  (\cdot,\cdot)  & U & W & V & I \\
  \hline
  U & U\oplus W\oplus V\oplus I & W\oplus V\oplus I  & V\oplus I & I
  \\
  \hline
  W & W\oplus V \oplus I & W\oplus V\oplus I & V\oplus I & I 
  \\
  \hline
  V & V\oplus I & V\oplus I & V\oplus I & I
\\
  \hline
  I & I & I & I & 0
\end{array}
\end{equation}
Indeed, $(\hat I,\hat I')=0$ because $I$-type cocycles can be chosen to be
antifield independent. For all other brackets involving $I$-type
cocycles, it follows from discussion of Section \ref{sec:antibr-maps-desc} that the
result must have depth $0$ and the only such
classes are of $I$ type. Alternatively, since all cocycles can be
chosen to be at most linear in antifields, the result will be a
cocycle that is antifield independent and only classes of $I$-type
have trivial antifield dependence. It thus follows that $I$-type
cohomology forms an abelian ideal.

Again given the discussion of Section \ref{sec:antibr-maps-desc}, the depth of the
antibracket map of $V$-type cohomology with $V,W,U$-type is
less or equal to $1$, so it must be of $V$- or $I$-type. 

Finally, the remaining structure follows from the fact that only
brackets of $U$-type cocycles with themselves may give rise to terms
that involve $C^*_I$'s.  

\subsection{Structure of the global symmetry
  algebra} \label{sec:globalsymmetries}

Let us now concentrate on brackets 
between two elements that
have both ghost number $-1$, i.e., on the detailed structure of the
Lie algebra of inequivalent global symmetries when taking into account
their different types.

In this case, one may use the table above supplemented by the fact
that $I^{-1}=0$.  Let then
\begin{equation}
U_{u}, \quad W_{w},\quad V_{v},
\end{equation} 
be bases of symmetries of $U,W,V$-type\footnote{These are bases in the
  cohomological sense, i.e., $\sum_u \lambda^u [U_u] = [0]$
  $\Rightarrow$ $\lambda^u = 0$ (and similarly for $W_w$ and
  $V_v$). In terms of the representatives, this becomes
  $\sum_{u} \lambda^u U_u = s a + db$ $\Rightarrow$ $\lambda^u =
  0$.}. At ghost number $g=-1$, equations \eqref{eq:18b},
\eqref{eq:19b}, \eqref{eq:20} give
\begin{equation}
  \label{eq:5bis}
  V_{v}=K_{v},\quad W_{w}=K_{w},\quad
  U_{u}= (f_{u})\indices{^I_J}[d^4x\, C^*_I C^J+\star
  A^*_I A^J]+K_{u}. 
\end{equation}

It follows from \eqref{table} that $V$-type symmetries and the direct
sum of $V$ and $W$-type symmetries form ideals in the Lie algebra of
inequivalent global symmetries.

The symmetry algebra $\mathfrak{g}_U$ is defined as the quotient of
all inequivalent global symmetries by the ideal of
$V\oplus W$-type symmetries. In particular, if $U$-type symmetries
form a sub-algebra, it is isomorphic to $\mathfrak g_U$.

First, $V$-type symmetries are parametrized by constants $k^v$,
$V^{-1}=k^v V_{v}$. The gauge invariant symmetry transformation on the
original fields then are
\begin{equation}
  \delta_{v}
  A^I_\mu=-(V_{v},A^I_\mu)=g\indices{_{{v}\mu}^I},\quad 
  \delta_{v} \phi^i=
  -(V_{v},\phi^i)=\Phi^i_{v}\label{eq:50}. 
\end{equation}
Furthermore, there
exist constants ${C^{v_3}}_{{v_1}{v_2}}$ such that
\begin{equation}
  \label{eq:25}
  ([V_{v_1}],[V_{v_2}]) = -{C^{v_3}}_{{v_1}{v_2}} [V_{v_3}]
\end{equation}
holds for the cohomology classes.
We choose the minus sign because 
\begin{equation}
  \label{eq:28a}
  (V_{v_1},V_{v_2}) = - d^4x (A^{*\mu}_I
  [\delta_{v_1},\delta_{v_2}] A_\mu^I+\phi^*_i
  [\delta_{v_1},\delta_{v_2}]\phi^i),
\end{equation}
so that the ${C^{v_3}}_{{v_1}{v_2}}$ are the structure constants of
the commutator algebra of the $V$-type symmetries,
$[\delta_{v_1},\delta_{v_2}] = {C^{v_3}}_{{v_1}{v_2}}
\delta_{v_3}$. For the functions $g\indices{_{v\mu}^I}$ and
$\Phi^i_v$, this gives
\begin{equation} \label{eq:commsymV}
  \begin{split}
    \delta_{v_1} g\indices{_{v_2\mu}^I} - \delta_{v_2}
    g\indices{_{v_1\mu}^I} &
    = {C^{v_3}}_{{v_1}{v_2}} g\indices{_{v_3\mu}^I} + (\text{trivial}) \\
    \delta_{v_1} \Phi^i_{v_2} - \delta_{v_2} \Phi^i_{v_1} &=
    {C^{v_3}}_{{v_1}{v_2}} \Phi^i_{v_3} + (\text{trivial}).
  \end{split}
\end{equation}
The ``trivial" terms on the right hand side take the form ``(gauge
transformation) $+$ (antisymmetric combination of the equations of
motion)" which is the usual ambiguity in the form of global
symmetries, see e.g. section 6 of \cite{Barnich:2000zw}. They come
from the fact that equation \eqref{eq:25} holds for classes: for the
representatives $V_v$ themselves, \eqref{eq:25} is
$(V_{v_1},V_{v_2}) = -{C^{v_3}}_{{v_1}{v_2}} V_{v_3} + sa + db$. The
trivial terms in \eqref{eq:commsymV} are then the symmetries generated
by the extra term $sa + db$, which is zero in cohomology.  The graded
Jacobi identity for the antibracket map implies the ordinary Jacobi
identity for these structure constants,
\begin{equation}
  \label{eq:26}
  {C^{v_1}}_{v_2[v_3}{C^{v_2}}_{v_4v_5]}=0. 
\end{equation}

Next, $W$-type symmetries are parametrized by constants $k^w$,
$W^{-1}=k^w W_{w}$ and encode the gauge invariant symmetry
transformations
\begin{equation}
  \delta_{w}
  A^I_\mu=-(W_{w},A^I_\mu)=g\indices{_{{w}\mu}^I},\quad 
  \delta_{w} \phi^i=
  -(W_{w},\phi^i)=\Phi^i_{w}\label{eq:49a} 
\end{equation}
with associated Noether 3 forms
$j_W=k^w (h_w)_{IJ}F^{(I}A^{J)} + k^w J_{Ww}$. There then exist
${C^{v_2}}_{w v_1}$, $C\indices{^{w_3}_{w_1 w_2}}$, $C\indices{^{v}_{w_1 w_2}}$
such that
\begin{equation}
  \label{eq:27a}
  \begin{split}
([W_w],[V_{v}]) &= -{C^{v_2}}_{w v} [V_{v_2}],\\
([W_{w_1}], [W_{w_2}]) &= - C\indices{^{w_3}_{w_1 w_2}} [W_{w_3}] -
C\indices{^{v}_{w_1 w_2}} [V_v],
\end{split}
\end{equation}
with associated Jacobi identities that we do not spell out. For the functions $g\indices{_{w\mu}^I}$ and $\Phi^i_w$, this implies
\begin{align}
    \delta_{w} g\indices{_{v\mu}^I} - \delta_{v} g\indices{_{w\mu}^I} &= {C^{v_2}}_{{w}{v}} g\indices{_{v_2\mu}^I}, \quad \delta_{w} \Phi^i_{v} - \delta_{v} \Phi^i_{w} = {C^{v_2}}_{{w}{v}} \Phi^i_{v_2}, \\
    \delta_{w_1} g\indices{_{w_2\mu}^I} - \delta_{w_2}
    g\indices{_{w_1\mu}^I} &= {C^{w_3}}_{{w_1}{w_2}}
    g\indices{_{w_3\mu}^I} + {C^{v}}_{{w_1}{w_2}}
    g\indices{_{v\mu}^I}, \label{eq:435}
    \\
    \delta_{w_1} \Phi^i_{w_2} - \delta_{w_2} \Phi^i_{w_1} &=
    {C^{w_3}}_{{w_1}{w_2}} \Phi^i_{w_3} + {C^{v}}_{{w_1}{w_2}}
    \Phi^i_{v}\,,
    \label{eq:436}
  \end{align}
up to trivial terms, see the discussion below \eqref{eq:commsymV}.

Finally, $U$-type symmetries are parametrized by $k^{u}$,
$U^{-1}=k^{u} U_{u}$ and encode the symmetry
transformations
\begin{equation}
  \begin{split}
  \label{eq:51a}
  \delta_{u} A^I_\mu &= -(U_{u},A^I_\mu)=(f_{u})\indices{^I_J}A^J_\mu
  + g\indices{_{{u}\mu}^I},\quad \delta_{u} \phi^i=
  -(U_{u},\phi^i)=\Phi^i_{u}, \\
  \delta_{u} A^{*\mu}_I &= -(f_{u})\indices{^K_I} A^{*\mu}_K
  - \frac{\delta}{\delta A^I_\mu} ( A^{*\nu}_K g\indices{_{u\nu}^K} +
  \phi^*_i \Phi^i_u ), \\
  \delta_{u} \phi^*_i &= - \frac{\delta}{\delta \phi^i}
  ( A^{*\nu}_K g\indices{_{u\nu}^K} + \phi^*_j \Phi^j_u ),  \\
  \delta_{u} C^I &= (f_{u})\indices{^I_J} C^J,\quad \delta_{u} C^*_I =
  - (f_{u})\indices{^K_I} C^*_K.
\end{split}
\end{equation}
Again, there exist constants $C$ with various types of indices such that
\begin{align}
    ([U_u],[V_{v}]) &= -C\indices{^{v_2}_{u v}} [V_{v_2}],
    \label{eq:UVbracket}
    \\
([U_u],[W_{w}]) &= -C\indices{^{w_2}_{u w}} [W_{w_2}] -C\indices{^{v}_{u w}} [V_{v}],
\label{eq:UWbracket}
\\
([U_{u_1}], [U_{u_2}]) &= -C\indices{^{v}_{u_1 u_2}} [V_{v}]
-C\indices{^{w}_{u_1 u_2}} [W_{w}] -C\indices{^{u_3}_{u_1 u_2}} [U_{u_3}],
                     \label{eq:29b}
\end{align}
with associated Jacobi identities. Working
out the term proportional to $C^*_I$ in
$(U_{u_1},U_{u_2})$ gives the commutation relations for
the $(f_u)\indices{^I_J}$ matrices,
\begin{equation}\label{eq:commU}
[f_{u_1},f_{u_2}] = -C\indices{^{u_3}_{u_1u_2}}
f_{u_3} .
\end{equation}
In turn, this implies Jacobi identities for this type of structure
constants alone:
\begin{equation}
  \label{eq:26a}
  {C^{u_1}}_{u_2[u_3}{C^{u_2}}_{u_4u_5]}=0. 
\end{equation}
The $C\indices{^{u_3}_{u_1u_2}}$ are the structure constants of
  $\mathfrak{g}_U$. 

From equation \eqref{eq:UVbracket}, we get the identities
\begin{align}\label{eq:422}
  \delta_{u} g\indices{_{v\mu}^I} - \delta_{v} g\indices{_{u\mu}^I}
  - (f_u)\indices{^I_J} g\indices{_{v\mu}^J} &= {C^{v_2}}_{{u}{v}}
  g\indices{_{v_2\mu}^I}, \quad \delta_{u} \Phi^i_{v}
 - \delta_{v} \Phi^i_{u} = {C^{v_2}}_{{u}{v}} \Phi^i_{v_2} .
\end{align}
Equation \eqref{eq:UWbracket} gives the same identities with the
right-hand side replaced by the appropriate sum, i.e.
\begin{align}
  \delta_{u} g\indices{_{w\mu}^I} - \delta_{w} g\indices{_{u\mu}^I}
  - (f_u)\indices{^I_J} g\indices{_{w\mu}^J} &= {C^{w_2}}_{{u}{w}}
  g\indices{_{w_2\mu}^I} + {C^{v}}_{{u}{w}}
  g\indices{_{v\mu}^I},\nonumber\\
  \delta_{u} \Phi^i_{w}
 - \delta_{w} \Phi^i_{u} &= {C^{w_2}}_{{u}{w}} \Phi^i_{w_2} + {C^{v}}_{{u}{w}} \Phi^i_{v}.
\end{align}

Considering the terms proportional to antifields $A^{*\mu}_I$ and $\phi^*_i$, the equation \eqref{eq:29b} gives
\begin{equation}\label{eq:423}
  \begin{split}
  \delta_{u_1} g\indices{_{u_2\mu}^I} - (f_{u_1})\indices{^I_J}
  g\indices{_{u_2\mu}^J}
  - (u_1 \leftrightarrow u_2) &= {C^{u_3}}_{{u_1}{u_2}}
  g\indices{_{u_3\mu}^I} \\
  &+ {C^{w}}_{{u_1}{u_2}} g\indices{_{w\mu}^I} + {C^{v}}_{{u_1}{u_2}}
  g\indices{_{v\mu}^I}, \\
  \delta_{u_1} \Phi^i_{u_2} - \delta_{u_2} \Phi^i_{u_1} &=
  {C^{u_3}}_{{u_1}{u_2}} \Phi^i_{u_3}
  + {C^{w}}_{{u_1}{u_2}} \Phi^i_{w} + {C^{v}}_{{u_1}{u_2}} \Phi^i_{v} .
\end{split}
\end{equation}
Equations \eqref{eq:422}-\eqref{eq:423} are again valid only
  up to trivial symmetries.

Let us now concentrate on identities containing the $h_{IJ}$, which
appear in the currents of $U$ and $W$-type. We first consider
$(U_u,W_w)$ projected to $W$-type. Here, we use the alternative form of anti-bracket defined in Section 
\ref{sec:antibr-maps-desc} since it makes the computation more easier.
Then, we have
\begin{align}
    s(U_u,W_w)_{{\rm alt}}&=-d(U_u,(h_{w})_{IJ}F^{I}A^{J}+J_{w})_{{\rm
    alt}}\nonumber\\
    &=d\{(h_{w})_{IJ} [(f_u)\indices{^I_K} F^K A^J +F^I
(f_u)\indices{^J_K} A^K ]+{\rm invariant}\}.
\end{align}

When comparing this to $s$ applied to the right hand side of \eqref{eq:UWbracket}, using
$s W_{w} = -d j_{Ww}$,  $s V_{v} = -d J_{Vv}$ and 
the fact that $W$-type cohomology is characterized by the Chern-Simons
term in its Noether current, we get
\begin{equation} \label{eq:hwfu}
(h_{w})_{IN}(f_u)\indices{^I_M}+(h_{w})_{MI}(f_u)\indices{^I_N}=
C\indices{^{w_2}_{u w}}(h_{w_2})_{MN}. 
\end{equation}
This computation amounts to identifying the Chern-Simons term in
  the $U$-variation $\delta_u j_{Ww}$ of a current of $W$-type.  The same
  computation applied to $(W_{w_1}, W_{w_2})$ shows that
  $C\indices{^{w_3}_{w_1 w_2}}(h_{w_3})_{MN} = 0$, which implies
\begin{equation} \label{eq:hwfw}
  C\indices{^{w_3}_{w_1 w_2}} = 0
\end{equation}
since the matrices $h_w$ are linearly independent (otherwise, the
$W_w$ would not form a basis). In other words, the $W$-variation
$\delta_{w_1} j_{W w_2}$ of a current of $W$-type is gauge invariant up
to trivial terms, i.e., is of $V$-type. 

In order to work out
$(U_{u_1},U_{u_2})$ projected to $W$-type, a slightly involved
reasoning gives 
\begin{equation}
\label{eq:Utrasf}
\delta_u G_I + (f_u)\indices{^J_I} G_J \approx  - 2 (h_u)_{IJ} F^J + \lambda^w_u (h_w)_{IJ} F^J + d(\text{invariant})
\end{equation}
for some constants $\lambda^w_u$. This is proved in Appendix \ref{app:derivation}  in the case where $G_I$ does not depend on derivatives of $F^I$ (but can have otherwise arbitrary dependence of $F^I$). We were not able to find the analog of \eqref{eq:Utrasf} in the higher derivative case.

Applying then $(U_{u_1},\cdot)_{\rm alt}$
to the chain of descent equations for $U_{u_2}$ and adding the chain
of descent equations for $C^{u_3}_{u_1u_2} U_{u_3}$ yields
\begin{align}
  (h_{u_2})_{IN}(f_{u_1})\indices{^I_M} &+(h_{u_2})_{MI}(f_{u_1})\indices{^I_N}
  -
    (h_{u_1})_{IN}(f_{u_2})\indices{^I_M}-(h_{u_1})_{MI}(f_{u_2})\indices{^I_N}
    \nonumber \\ &+ \frac{1}{2} \left[ (h_w)_{IN} (f_{u_2})\indices{^I_M} +  (h_w)_{IM} (f_{u_2})\indices{^I_N}  \right]\lambda^w_{u_1} \nonumber \\
&= C\indices{^{u_3}_{u_1 u_2}}(h_{u_3})_{MN} + C\indices{^{w}_{u_1
                   u_2}}(h_{w})_{MN}.
\end{align}
Again, this amounts to identifying the Chern-Simons terms in the
$U$-variation $\delta_{u_1} j_{u_2}$ of a $U$-type current. Equation
\eqref{eq:Utrasf} is crucial for this computation since $U$-type
currents contain $G_I$.
Using \eqref{eq:hwfu}, this becomes
\begin{align}\label{eq:hufu}
  (h_{u_2})_{IN}(f_{u_1})\indices{^I_M} &+(h_{u_2})_{MI}(f_{u_1})\indices{^I_N}
  -
    (h_{u_1})_{IN}(f_{u_2})\indices{^I_M}-(h_{u_1})_{MI}(f_{u_2})\indices{^I_N}
    \nonumber \\ &= C\indices{^{u_3}_{u_1 u_2}}(h_{u_3})_{MN} + \left[ C\indices{^{w}_{u_1 u_2}} - \frac{1}{2} C\indices{^w_{u_2 w_2}} \lambda^{w_2}_{u_1} \right] (h_{w})_{MN}.
\end{align}
We see that the effect of the $\lambda^w_u$ is to shift the structure constants of type $C\indices{^{w}_{u_1 u_2}}$. The constants $\lambda^w_u$ vanish for the explicit models considered below; it would be interesting to find an explicit example where this is not the case. As a last comment, we note that antisymmetry of equation \eqref{eq:hufu} in $u_1$ and $u_2$ imposes the constraint
\begin{equation}
C\indices{^w_{u_2 w_2}} \lambda^{w_2}_{u_1} + C\indices{^w_{u_1 w_2}} \lambda^{w_2}_{u_2} = 0
\end{equation}
on the constants $\lambda^w_u$.

\subsection{Parametrization through symmetries} \label{sec:parametrization}

It follows from the discussion of the antibracket map involving
$H^{-2}$ after \eqref{eq:53A} that cohomologies of $U,W,V$-type in
ghost numbers $g\geqslant 0$ can be parametrized by symmetries of the
corresponding type with suitably constrained coefficients
\begin{equation}
  \label{gd}k\indices{^{u}_{K_1\dots K_{g+1}}},
  \quad k\indices{^{v}_{K_1\dots K_{g+1}}},\quad k\indices{^{w}_{K_1\dots K_{g+1}}}.
\end{equation}
In this way, for $g=0$, the problem of finding all infinitesimal
gaugings can be reformulated as the question of which of these
symmetries can be gauged.

In order to do this, it is useful to first rewrite the
$h_{I|JK_1\dots K_{g+1}}$ appearing in the cohomology classes of $U$
and $W$-types in the equivalent symmetric convention
\begin{equation}
    X_{IJ,K_1\dots K_{g+1}}:=h_{(I|J)K_1\dots K_{g+1}}\iff 
    h_{I|JK_1\dots K_{g+1}}=\frac{2(g+2)}{g+3}X_{I[J,K_1\dots K_{g+1}]}
\label{Xintermsofh}
\end{equation}
where \eqref{Y1} is now replaced by
\begin{equation}\label{YX}
  X_{IJ,K_1 \dots K_{g+1}} = X_{(IJ),[K_1 \dots K_{g+1}]}, \quad
  X_{(IJ,K_1) K_2 \dots K_{g+1}} = 0\;. \end{equation}
Note that for $g=-1$, $h_{IJ} = X_{IJ}$.

For cohomology classes of $U,W$-type, we can write
\begin{align}
  \label{eq:352}
  {f^I}_{JK_1\dots K_{g+1}} &= (f_{u})\indices{^I_{J}}
                              \, k\indices{^{u}_{K_1\dots K_{g+1}}}, \\
  {}^U X_{IJ,K_1\dots K_{g+1}} &=(h_{u})_{IJ} \, k\indices{^{u}_{K_1\dots K_{g+1}}}, \\
  X_{IJ,K_1\dots K_{g+1}} &= (h_{w})_{IJ} \, k\indices{^{w}_{K_1\dots K_{g+1}}},
\end{align}
where $(f_{u})\indices{^I_{J}}$, $(h_{u})_{IJ}$ and $(h_{w})_{IJ}$
appear in the basis elements $U_{u}$ and $W_{w}$. (One has similar
parametrizations for the quantities $g^I_{\mu K_1\dots K_{g+1}}$,
$\Phi^i_{K_1\dots K_{g+1}}$, $J_{\mu K_1\dots K_{g+1}}$ in the
cohomology classes of the various types.)  This guarantees that
condition \eqref{eq:47A} (or \eqref{eq:16b}) is automatically
satisfied.

However, the symmetry properties \eqref{sk} and \eqref{YX} imply the
following linear constraints on the parameters:
\begin{align}
(f_{u})\indices{^I_{(J}} \, k\indices{^{u}_{K_1)K_2\dots K_{g+1}}} &= 0, \label{linf}\\
(h_{u})_{(IJ} \, k\indices{^{u}_{K_1)K_2 \dots K_{g+1}}} &= 0, \label{358}\\
(h_{w})_{(IJ} \, k\indices{^{w}_{K_1) K_2 \dots K_{g+1}}} &= 0. \label{465}
\end{align}
From the discussion of the cohomology, it also follows that $V$-type
cohomology classes are entirely determined by $V$-type symmetries in
terms of $k\indices{^{v}_{K_1\dots K_{g+1}}}$ without any additional
constraints.

\subsection{2nd order constraints on deformations and gauge algebra}
\label{sec:parametrization-2nd-order}

The most general infinitesimal gauging is given by
$S^{(1)}= \int ( U^0+W^0+V^0+I^0 )$. We have
\begin{equation}
\half (S^{(1)}, S^{(1)}) = \int \left( U^1+W^1+V^1+I^1 \right) \label{eq:22a}. 
\end{equation}
The infinitesimal deformation $S^{(1)}$ can be extended to second
order whenever the right hand side vanishes in cohomology, resulting in quadratic
constraints on the constants $k\indices{^{u_1}_K}$,
$k\indices{^{w_1}_K}$, $k\indices{^{v_1}_K}$ and $\rho^{\cA}$.
Working all of them out explicitly requires computing all brackets
between $U^0$, $W^0$, $V^0$ and $I^0$.

However, it follows from the previous section that the only
contribution to $U^1$ comes from $\half (U^0,U^0)$. The vanishing of
the terms containing the antighosts $C^*_I$ requires
\begin{equation}
  \label{eq:23}
  f\indices{^I_{J[K_1}}f\indices{^J_{K_2K_3]}}=0,
\end{equation}
i.e., the Jacobi identity for the $f\indices{^I_{JK}}$. The
  associated $n_v$-dimensional Lie algebra is the gauge
  algebra and is denoted by $\mathfrak g_g$.

Using
$f\indices{^I_{JK}} = (f_{u_1})\indices{^I_J} k\indices{^{u_1}_K}$ and
equation \eqref{eq:commU}, the Jacobi identity reduces to the
following quadratic constraint on $k\indices{^{u_1}_K}$:
\begin{equation} \label{eq:quadu}
 k\indices{^{u_1}_I} k\indices{^{u_2}_J} C\indices{^{u_3}_{u_1u_2}} -
  (f_{u_4})\indices{^K_I} k\indices{^{u_4}_J} k\indices{^{u_3}_K} = 0 .
\end{equation}
Note that the antisymmetry in $IJ$ of the second term is guaranteed by
the linear constraint \eqref{linf}. 
The terms at antifield number $1$ give the constraints 
\begin{align}
 \delta_I g^K_J + f\indices{^K_{MJ}} g^M_I - (I\leftrightarrow J) &= f\indices{^L_{IJ}} g^K_L \\
 \delta_I \Phi^i_J  - (I\leftrightarrow J) &= f\indices{^L_{IJ}} \Phi^i_L.
\end{align}
Expressed with $k$'s, this gives
\begin{equation}
    k\indices{^{\Gamma}_I} k\indices{^{\Delta}_J} C\indices{^{\Sigma}_{\Gamma\Delta}} -
  (f_{u})\indices{^K_I} k\indices{^{u}_J} k\indices{^{\Sigma}_K} = 0,
\end{equation}
where the capital Greek indices take all values $u,w,v$.  This gives
three constraints, according to the type of the free index
$\Sigma$. When $\Sigma = u$, we get the constraint $\eqref{eq:quadu}$,
because the only non-vanishing structure constants with an upper $u$
index are the $C\indices{^{u_3}_{u_1u_2}}$. When $\Sigma = w$, the
possible structure constants are $C\indices{^{w_3}_{w_1w_2}}$,
$C\indices{^{w_3}_{u_1w_2}} = -C\indices{^{w_3}_{w_2u_1}}$ and
$C\indices{^{w_3}_{u_1u_2}}$, giving the constraint
\begin{align}
  k\indices{^{w_1}_I} k\indices{^{w_2}_J} C\indices{^{w_3}_{w_1w_2}}
  + 2 k\indices{^{u_1}_{[I}} k\indices{^{w_2}_{J]}}
  C\indices{^{w_3}_{u_1w_2}}
  &+ k\indices{^{u_1}_I} k\indices{^{u_2}_J} C\indices{^{w_3}_{u_1u_2}}\nonumber
  \\
  &- (f_{u_4})\indices{^K_I} k\indices{^{u_4}_J} k\indices{^{w_3}_K} = 0 . 
\end{align}
When the free index $\Sigma$ is of type $v$, one gets a similar
identity with all possible types of values in the lower indices of the
structure constants,
\begin{multline}
  k\indices{^{v_1}_I} k\indices{^{v_2}_J} C\indices{^{v_3}_{v_1v_2}} +
  2 k\indices{^{w_1}_{[I}} k\indices{^{v_2}_{J]}}
  C\indices{^{v_3}_{w_1v_2}} + k\indices{^{w_1}_I} k\indices{^{w_2}_J}
  C\indices{^{v_3}_{w_1w_2}} + 2 k\indices{^{u_1}_{[I}}
  k\indices{^{v_2}_{J]}} C\indices{^{v_3}_{u_1v_2}} \\ + 2
  k\indices{^{u_1}_{[I}} k\indices{^{w_2}_{J]}}
  C\indices{^{v_3}_{u_1w_2}} + k\indices{^{u_1}_I} k\indices{^{u_2}_J}
  C\indices{^{v_3}_{u_1u_2}} - (f_{u_4})\indices{^K_I}
  k\indices{^{u_4}_J} k\indices{^{v_3}_K} = 0 . 
\end{multline}

\section{Quadratic vector models}
\label{sec:second order}

\subsection{Description of the model}
\label{sec:lagrangian}

To go further, one needs to specialize the form of the Lagrangian,
which has been assumed to be quite general so far.  In this section,
we focus on second order Lagrangians \eqref{eq:lag-V+S}-\eqref{eq:lag} arising in the context of
supergravities that contain $n_s$ scalar fields and depend
quadratically on $n_v$ abelian vector fields, non-minimally coupled to
each other, in four space-time dimensions.

To recall, we consider explicitly $\cL = \cL_S + \cL_V$, where
\begin{equation} 
  \cL_V = - \frac{1}{4}\, \mathcal{I}_{IJ}(\phi) F^I_{\mu\nu}
          F^{J\mu\nu}
          + \frac{1}{8} \,\mathcal{R}_{IJ}(\phi)\,
          \varepsilon^{\mu\nu\rho\sigma}
          F^I_{\mu\nu} F^J_{\rho\sigma} \label{eq:lagnew}
\end{equation}
and the scalar Lagrangian is of the sigma model form
\begin{equation}
  \label{eq:22}
  \mathcal L_S=-\half g_{ij}(\phi)\d_\mu\phi^ i\d^\mu\phi^j-V(\phi)
\end{equation}
where $g_{ij}$ is symmetric and invertible. Both $g_{ij}$ and $V$ depend only on undifferentiated scalar fields. As we pointed out in Section \ref{sec:noncompact_duality} neglecting gravity, this is the generic bosonic sector of ungauged supergravity. The
symmetric matrices $\cI$ and $\cR$, with $\cI$ invertible, depend only on
undifferentiated scalar fields and encode the non-minimal
couplings between the scalars and the abelian vectors.
The Bianchi identities and equations of motion for the vector fields
are given by
\begin{equation} \label{eq:eom}
  \partial_\mu (\star F^I)^{\mu\nu} = 0, \qquad \partial_\nu (\star G_I)^{\mu\nu} \approx 0.
\end{equation}
The Lagrangian
\eqref{eq:lagnew} falls into the general class of models described
previously, with the gauge invariant two-form
$G_I=\mathcal I_{IJ}\star F^J+\mathcal R_{IJ} F^J$ and
$d^4x\, \mathcal{L}_V=\frac{1}{2} G_I F^I$.

We also assume \cite{Barnich:2017nty}
\begin{equation}
  \label{eq:33}
  \mathcal R_{IJ}(0)=0. 
\end{equation}
Note that a constant part in $\mathcal R_{IJ}$ can be put to zero
without loss of generality since the associated term in the Lagrangian
is a total derivative. In most cases, we also take 
$V=0$ or assume (writing $\d_i=\ddl{}{\phi^i}$) that
\begin{equation}
  \label{eq:34}
  (\d_iV)(0)=0. 
\end{equation}

\subsection{Constraints on $U$, $W$-type symmetries}
\label{sec:u-w-type}

We assume here and in the examples below that there is no explicit
$x^\mu$-dependence in the space of local functions in order to
constrain $U$ and $W$-type symmetries. For simplicity, we also assume that the potential vanishes, $V=0$. 
In Section \ref{sec:applications}, these constraints will allow us to determine all symmetries of $U$ and $W$-type for specific models.

We need the scalar field equations, which are encoded in
\begin{equation}
s\star \phi^*_i+d(g_{ij} \star d\phi^j)=-\star
  \d_i(\mathcal{L}_S+\mathcal{L}_V)\label{eq:24},
\end{equation}
where $\d_i=\ddl{}{\phi^i}$. For $g=-1$, equation \eqref{eq:16b}
becomes
\begin{equation}
  \label{eq:17a}
  G_I F^J{f^I}_J+F^IF^Jh_{IJ}+dI^{n-1}-dG_I\, g^{I}-[d(g_{ij} \star
  d\phi^j)+\star
  \d_i(\mathcal{L}_S+\mathcal{L}_V)]\Phi^{i}
=0.
\end{equation}
When putting all derivatives of $F^I_{\mu\nu},\phi^i$ to zero, one
remains with
\begin{equation}
  \label{eq:18a}
  G_I F^J{f^I}_J+F^IF^Jh_{IJ}
-\star \d_i \mathcal{L}_V\, 
\Phi^{i}|_{{\rm der}=0}=0.
\end{equation}
It is here that the assumption that there is no explicit $x^\mu$
dependence in the gauge invariant functions
$g^{I\alpha},\Phi^{i\alpha}$ is used.  Using
$-\d_i\star \mathcal{L}_V=\half\d_i G_I \, F^I$, and the decomposition
$\Phi^{i}|_{{\rm
    der}=0}=\Phi^{i}_0+\Phi^{i}_1+\dots$,
where the $\Phi^{i}_{n}$ depend on undifferentiated scalar
fields and are homogeneous of degree $n$ in $F^I_{\mu\nu}$,
the equation implies that 
\begin{equation}
  \label{eq:28}
  \half M_{IJ}(\phi)\star F^I F^J+\half N_{IJ}(\phi)F^IF^J=0, 
\end{equation}
where 
\begin{equation}
  \label{eq:29}
  M_{IJ}=2\mathcal I_{K(I}{f^K}_{J)}
+\d_i\mathcal
I_{IJ}\Phi^{i}_0, 
\end{equation}
\begin{equation}
  \label{eq:27}
  \quad  N_{IJ}=2\mathcal R_{K(I}{f^K}_{J)}+2h_{IJ}+\d_i\mathcal
  R_{IJ}\Phi^{i}_0,
\end{equation}
by using that $h_{IJ}=h_{JI}$ on account of \eqref{Y1}.
When taking an Euler-Lagrange derivative of \eqref{eq:28} with respect
to $A_\mu^I$, one concludes that both terms have to vanish separately, 
\begin{equation}
M_{IJ}=0,\quad N_{IJ}=0.\label{eq:19}
\end{equation}
Setting $\phi^i=0$ and using \eqref{eq:33} then gives
\begin{equation}
  \label{eq:67}
  f^{(\mathcal I(0))}_{IJ}+f^{(\mathcal I(0))}_{JI}=-(\d_i\mathcal
  I_{IJ})(0)\Phi^i_0(0),\quad 
 2h_{IJ}=-(\d_i\mathcal R_{IJ})(0)\Phi^i_0(0),
\end{equation}
where the abelian index is lowered and raised with $\mathcal I_{IJ}(0)$ and
its inverse. Note that completely skew-symmetric $f^{(\mathcal I(0))}_{IJ}$
solve the equations with $\Phi^i_0(0)=0$, $h_{IJ}=0$. More conditions
are obtained by expanding equations \eqref{eq:19} in terms of power
series in $\phi^i$.

In all examples considered below, the algebra $\mathfrak g_U$ and the
$W$-type symmetries can be entirely determined from the analysis of
this subsection.

\subsection{Electric symmetry algebra}

An important result of our general analysis is that the symmetries of
the action that can lead to consistent gaugings may have a term that
is not gauge invariant. This term is present only in the variation
of the vector potential and is restricted to be linear in the
undifferentiated vector potential, i.e.,
$\delta A^I_{\mu} = f\indices{^I_J} A^J_{\mu} + g^I_\mu$,
$\delta \phi^i = \Phi^i$.  Here $f\indices{^I_J}$ are constants, and
$g^I_\mu$ and $\Phi^i$ are gauge invariant functions.  The symbol
$\delta$ represents the variation of the fields and is of course not
the Koszul-Tate differential.  No confusion should arise as the
context is clear.

It is of interest to investigate a subalgebra of the gaugeable
symmetries, obtained by restricting oneself from the outset to
transformations of the gauge potentials that are linear and
homogeneous in the undifferentiated potentials and to transformations
of the scalars that depend on undifferentiated scalars alone,
\begin{equation}
  \label{eq:64}
\delta A^I_{\mu} = f\indices{^I_J} A^J_{\mu}, \quad
\delta\phi^i = \Phi^i(\phi). 
\end{equation}
This means that one takes $g_\mu^I = 0$ and that the functions
$\Phi^i$ only depend on the undifferentiated scalar fields.  These
symmetries form a sub-algebra $\mathfrak g_e$ that includes the
symmetries usually considered in the supergravity which
is called the ``electric symmetry algebra'' (in the
given duality frame), see Section \ref{sec:elec-group}. As we discussed there, the electric group is a subgroup of the duality group $G \subset Sp(2 n_v, \mathbb{R})$ \cite{Gaillard:1981rj} which is comprised of lower triangular symplectic transformations. Although our Lagrangians are not necessarily
connected with supergravity, we shall nevertheless call the
symmetries of the form (\ref{eq:64}) ``electric symmetries'' and the
subalgebra $\mathfrak g_e$ the ``electric algebra''.  It need
not to be a
subalgebra of $Sp(2 n_v, \mathbb{R})$.  It generically does not
exhaust all symmetries and does not contain for example the conformal
symmetries of free electromagnetism.

The transformations of the form (\ref{eq:64}) are symmetries of the
action \eqref{eq:lagnew} + \eqref{eq:22} if and only if the scalar
variations leave the scalar action invariant separately, and
$f\indices{^I_J},\Phi^i(\phi)$ satisfy
\begin{align}
  \frac{\d \cI}{\d \phi^i} \Phi^i &
= - f^T \mathcal{I} - \mathcal{I} f,\label{var-I}\\
  \frac{\d \cR}{\d \phi^i} \Phi^i &
= - f^T \mathcal{R} - \mathcal{R} f - 2 h, \label{var-R-theta}
\end{align}
where the $h$ are constant symmetric matrices.  In
particular, when the scalar Lagrangian is given by
$\cL_S = \frac{1}{2} g_{ij}(\phi) \d_\mu \phi^i \d^\mu \phi^j$, the first condition means that $\Phi^i$ must be a Killing vector of the
metric $g_{ij}$.    If $U$
and $W$-type symmetries are of electric type, the electric symmetry
algebra contains in addition only $V$-type symmetries of electric
type, i.e., transformations among the undifferentiated scalars alone
that leave invariant both the scalar action and the matrices
$\mathcal I,\mathcal R$ (i.e., that satisfy $\delta S_S=0$ and
\eqref{var-I}, \eqref{var-R-theta} with $0$'s on the right hand
sides). This will be the case in all examples below. In particular,
the $f$'s, and thus also the gauge algebra, will be the same for
$\mathfrak g_U$ and $\mathfrak g_e$. The $h$ matrix is determined by the transformation parameters $\Phi^i$ and the parity-odd term $\mathcal{R}$ of the action via (\ref{var-R-theta}).

We then suppose that we have a basis of
symmetries of the action of this form,
\begin{align}
\delta_\Gamma A^I_{\mu} &= (f_{\Gamma})\indices{^I_J}
                          A^J_{\mu}, \label{global-sym-A}\\ 
\delta_\Gamma \phi^i &= \Phi_{\Gamma}^i(\phi). \label{global-sym-phi}
\end{align}
When compared to the previous sections, the index $\Gamma$ can take
$u$, $v$ or $w$ values. Only the $f_u$ matrices are non-zero. The
$h_\Gamma$ matrices are non-vanishing only for $\Gamma = u$ or
$w$. When $h_\Gamma \neq 0$, the Lagrangian is only invariant up to a
total derivative.

Closure of symmetries of this form then implies
\begin{align}
[f_\Delta, f_\Gamma] &= -C\indices{^\Sigma_{\Delta\Gamma}} f_\Sigma , \label{X-alg}\\
f_\Gamma^T h_\Delta - f_\Delta^T h_\Gamma + h_\Delta f_\Gamma -
  h_\Gamma f_\Delta &
= - C\indices{^\Sigma_{\Delta\Gamma}} \,h_\Sigma , \label{5idGlobal} \\
  \frac{\partial \Phi^i_\Delta}{\partial \phi^j} \, \Phi^j_\Gamma
  - \frac{\partial \Phi^i_\Gamma}{\partial \phi^j} \, \Phi^j_\Delta &
= -C\indices{^\Sigma_{\Delta\Gamma}} \Phi^i_\Sigma , \label{F-alg}
\end{align}
where one can obtain the equation (\ref{5idGlobal}) by the action of commutator $\left[\delta_{\Gamma},\delta_{\Delta}\right]$ on $\mathcal{R}$.
Decomposing the indices into $U$, $W$ and $V$-type, this is consistent
with the relations of Section \ref{sec:globalsymmetries} with $\lambda_u^w=0$.
Let us note that (\ref{var-R-theta}) expresses the surface term in the variation of the action.

\subsection{Restricted first order deformations}

We now limit ourselves to first order deformations of the master
action with the condition that all infinitesimal gaugings come from
symmetries that belong to the electric symmetry algebra above. In
order to simplify formulas, we will no longer make the distinction
between $U$-, $W$- and $V$-type which can easily be recovered.

According to Section \ref{sec:parametrization}, the deformations are
parametrized through electric symmetries by a matrix $k^\Gamma_I$,
with
\begin{align}
f\indices{^I_{JK}} &= (f_\Gamma)\indices{^I_J} k^\Gamma_K , \label{eq:f-k} \\
\Phi^i_I(\phi) &= \Phi^i_\Gamma(\phi) k^\Gamma_I , \label{eq:phi-k} \\
X_{IJ,K} &= (h_\Gamma)_{IJ} k^\Gamma_K .
\end{align}
The linear constraints \eqref{linf} -- \eqref{465} on the matrix
$k^\Gamma_K$ become
\begin{align}
(f_\Gamma)\indices{^I_J} k^\Gamma_K + (f_\Gamma)\indices{^I_K}
  k^\Gamma_J &= 0 ,
\label{eq:constr-antisymmetry}\\
h_{\Gamma \, (IJ} k^\Gamma_{K)} &= 0. \label{eq:constr-chernsimons}
\end{align}
They guarantee that the first order deformation of the master action
is given by
\begin{equation}
S^{(1)} = \int \!d^4\!x\,\left( a_2 + a_1 + a_0 \right),
\end{equation}
where 
\begin{equation}
  a_2 = \frac{1}{2} C^*_I f\indices{^I_{JK}} C^J C^K
\end{equation}
encodes the first order deformation of the gauge algebra and
\begin{equation}
  a_1 = A^{*\mu}_I f\indices{^I_{JK}} A^J_\mu C^K + \phi^*_i \Phi^i_K C^K
\end{equation}
encodes the first order deformation of the gauge symmetries.  When
taking \eqref{eq:f-k} and \eqref{eq:phi-k} into account, this
deformation of the gauge symmetries corresponds to gauging the
underlying global symmetries by using local parameters
$\eta^\Gamma (x) = k^\Gamma_I \epsilon^I (x)$. The deformation $a_0$
of the Lagrangian is given by the sum of three terms:
\begin{align}
  a^\text{(YM)}_0 &= \frac{1}{2} (\star G_I)^{ \mu\nu}
 f\indices{^I_{JK}} A^J_\mu A^K_\nu , \\
a^\text{(CD)}_0 &= J^\mu_K A^K_\mu ,  \\
  a^\text{(CS)}_0 &= \frac{1}{3} X_{IJ,K} \epsilon^{\mu\nu\rho\sigma}
F^I_{\mu\nu} A^J_\rho A^K_\sigma.
\end{align}
The terms $a^\text{(YM)}$ and $a^\text{(CD)}$ are exactly those
necessary to complete the abelian field strengths and ordinary
derivatives of the scalars into covariant quantities. The term $a^\text{(CD)}$ is responsible for charging the matter fields. The Chern-Simons
term $a_0^\text{(CS)}$ appears when $h_\Gamma \neq 0$: its role is to
cancel the variation
\begin{equation} \label{eq:CSobstr}
  \delta \mathcal{L} = - \frac{1}{4} \eta^\Gamma h_{\Gamma \,IJ}\,
  \varepsilon^{\mu\nu\rho\sigma} F^I_{\mu\nu} F^J_{\rho\sigma}
\end{equation}
that is no longer a total derivative when $\eta^\Gamma = k^\Gamma_I
\epsilon^I (x)$ \cite{deWit:1984rvr,deWit:1987ph}.

\subsection{Complete restricted deformations}

The second order deformation $S^{(2)}$ to the master action is then
determined by the first order deformation through equation
\eqref{eq:43}. As discussed in Section
\ref{sec:parametrization-2nd-order}, the existence of $S^{(2)}$
imposes additional quadratic constraints on the matrix $k^\Gamma_I$,
\begin{equation} \label{eq:quadraticconstraint}
  k\indices{^{\Gamma}_I} k\indices{^{\Delta}_J}
  C\indices{^{\Sigma}_{\Gamma\Delta}} -
  (f_{\Gamma})\indices{^K_I} k\indices{^\Gamma_J}
  k\indices{^{\Sigma}_K} = 0.
\end{equation}

Explicit computation shows that $S^{(2)}$ can be chosen such that
there is no further deformation of the gauge symmetries or of their
algebra. The second order terms in the Lagrangian are exactly those
necessary to complete abelian field strengths
$F^I_{\mu\nu} = \partial_\mu A^I_\nu - \partial_\nu A^I_\mu$ and
ordinary derivatives of the scalars to non-abelian field strengths and
covariant derivatives,
\begin{align}
  \mathcal{F}^I_{\mu\nu} &= \partial_\mu A^I_\nu - \partial_\nu A^I_\mu
 + f\indices{^I_{JK}} A^J_\mu A^K_\nu  , \\
D_\mu \phi^i &= \partial_\mu \phi^i - \Phi^i_I(\phi) A^I_\mu .
\end{align}
One also finds a non-abelian completion of the Chern-Simons term $a_0^\text{(CS)}$. 
Putting everything together, the Lagrangian after adding the second order deformation is \cite{Barnich:2017nty}
\begin{align}\label{Full-Lagrangian}
  \mathcal{L} = \mathcal{L}_S(\phi^i, D_\mu \phi^i) &
  - \tfrac{1}{4}\,\mathcal{I}_{IJ}(\phi)
   \mathcal{F}^I_{\mu\nu} \mathcal{F}^{J\mu\nu} + \tfrac{1}{8}\,
   \mathcal{R}_{IJ}(\phi)\,
    \varepsilon^{\mu\nu\rho\sigma}
 \mathcal{F}^I_{\mu\nu} \mathcal{F}^J_{\rho\sigma} \nonumber \\
    &+ \tfrac{2}{3}\, X_{IJ,K}\, \varepsilon^{\mu\nu\rho\sigma} A^J_\mu
      A^K_\nu \left( \partial_\rho A^I_\sigma + \tfrac{3}{8}\,
      f\indices{^I_{LM}} A^L_\rho A^M_\sigma \right).
\end{align}
The associated action can be checked to be invariant under the gauge
transformations
\begin{align}
  \delta A^I_{\mu} &= \partial_\mu \epsilon^I + f\indices{^I_{JK}}
 A^J_{\mu} \epsilon^K , \label{gauge-sym-A}\\
  \delta \phi^i &= \epsilon^I \Phi_I^i(\phi). \label{gauge-sym-phi}
\end{align}
This is equivalent to the fact that the deformation stops at second
order, i.e, that $S=S^{(0)} + S^{(1)} + S^{(2)}$ gives a solution to
the master equation $(S,S)=0$.

Checking directly the invariance of this action under
\eqref{gauge-sym-A} -- \eqref{gauge-sym-phi} without first
parametrizing $f\indices{^I_{JK}}$, $\Phi_I^i(\phi)$ and $X_{IJK}$
through symmetries requires the use of the linear identities
\begin{equation}
f\indices{^I_{JK}} = f\indices{^I_{[JK]}}, \quad X_{(IJ,K)} = 0
\end{equation}
and of the quadratic ones
\begin{align}
f\indices{^I_{J[K_1}}f\indices{^J_{K_2K_3]}} &= 0, \\
  f\indices{^K_{I[L}} X_{M]J,K} + f\indices{^K_{J[L}} X_{M]I,K}
  - \tfrac{1}{2}\, X_{IJ,K} f\indices{^K_{LM}} &= 0, \\
  \frac{\d \Phi^i_I}{\d \phi^j} \Phi^j_J -
  \frac{\d \Phi^i_J}{\d \phi^j} \Phi^j_I + f\indices{^K_{IJ}} \Phi^i_K &= 0.
\end{align}
In terms of $k^\Gamma_I$, these three quadratic identities all come
from the single quadratic constraint \eqref{eq:quadraticconstraint}
once the algebra of global symmetries \eqref{X-alg} -- \eqref{F-alg}
is taken into account.

\subsection{Remarks on $GL(n_v)$ transformations}
\label{sec:frame-transf}

Consider a linear field redefinition of the
abelian vector potentials, $A^I_\mu= {M^I}_J{A'}^J_\mu$ with
$M\in GL(n_v)$. Such a transformation gives rise to a trivial
infinitesimal gauging which corresponds to the antifield
independent part of the trivial ghost number $0$ cocycle
\begin{equation}
S^{(1)}_{triv.} = (S^{(0)},\Xi_s),\quad \Xi_s={f_s^I}_J[d^4x\, C^*_I C^J+\star A^*_I
A^J],\label{eq:46}
\end{equation}
with $f_s\in \mathfrak{gl}(n_v,\mathbb R)$.

Two remarks are in order.

The first concerns the relation to the algebra $\mathfrak{g}_U$
defined in Section \ref{sec:globalsymmetries}. It can also be defined
as the largest sub-algebra of $\mathfrak{gl}(n_v,\mathbb R)$ that can
be turned into symmetries of the theory by adding suitable gauge
invariant transformations of the vector and scalar fields, or in other
words, for which there exists a gauge invariant $K_u$ of ghost number
$-1$ such that $(S^{(0)},\Xi_u+K_u)=0$.

In particular, \eqref{eq:352} and \eqref{linf} for $g=0$, as well as
\eqref{eq:quadu}, can be summarized as follows: non-trivial $U$-type
gaugings require the existence of a map (described by $k^u_K$) from
the defining representation of the symmetry algebra
$\mathfrak g_U\subset \mathfrak{gl}(n_v,\mathbb R)$ into the adjoint
representation of the $n_v$-dimensional gauge algebra $\mathfrak g_g$.
  
The second remark is about families of Lagrangians related by linear
transformations of the vector potentials among themselves. It is
sometimes useful not to work with fixed (canonical) values for various
$GL(n_v)$ tensors that appear in the action. Instead, one considers
the deformation problem for sets of Lagrangians parametrized by
arbitrary $GL(n_v)$ tensors, for instance generic non-degenerate
symmetric $\mathcal I_{MN}$ and symmetric $\mathcal R_{MN}$ that
vanish at the origin of the scalar field space.

If the tensors of two such Lagrangians are related by a $GL(n_v)$
transformation, they should be considered as equivalent. Indeed, the
local BRST cohomology for all members of such an equivalence class are
isomorphic and related by the above anti-canonical field
redefinition. In particular, all members of the same equivalence class
have isomorphic gaugings.

All general considerations and results on local BRST cohomology above
apply in a unified way to all equivalence classes. When one explicitly
solves the obstruction equation \eqref{eq:16b} (for instance at $g=-1$
in order to determine the symmetries), the results on local BRST
cohomologies do depend on the various equivalence classes.

\subsection{Comparison with the embedding tensor constraints}

In the embedding tensor formalism
\cite{deWit:2002vt,deWit:2005ub,deWit:2007kvg,Samtleben:2008pe,%
  Trigiante:2016mnt},
  the possible gaugings are described by the embedding tensor
$\Theta\indices{_M^\alpha}=(\Th\indices{_I^\a},\Th^{I\a})$ with
electric and magnetic components\footnote{We refer the reader to \cite{Coomans:2010xd}, where a relation between the embedding tensor formalism and the BV-BRST antifield formalism has been considered for a different reason.}, which satisfies a number of linear
and quadratic constraints \eqref{linear2}--\eqref{quadratic2}. Recalling from
Section \ref{sec:Emb_Tensor}, the index $I$ runs from
$1$ to $n_v$, while $\a$ runs from $1$ to the dimension of the group
$G$ of invariances of the equations of motion of the initial
Lagrangian \eqref{eq:lagnew}. More precisely, $G$ is defined only as
the group of transformations that act linearly on the field strengths
$F^I$ and their ``magnetic duals'' $G_I$, and whose action on the
scalars contains no derivatives. This coincides with the group of
symmetries of the first order Lagrangian discussed in
Section \ref{sec:symplecticchoice} and \cite{Henneaux:2017kbx}, which are of the restricted form
\eqref{eq:64}.

As explained in Section \ref{sec:Emb_Tensor}, one can always
go to a duality frame in which the magnetic components of the
embedding tensor vanishes, $\Th^{I\a} = 0$, see Section \ref{sec:embedding} for an explicit construction of such a symplectic matrix.
Moreover, only the
components $\Th\indices{_{\check{I}}^\Gamma}$ survive, where $\Gamma$ runs over
the generators of the electric subgroup $G_e \subset G$ that act as
local symmetries of the Lagrangian in that frame and as before $\check{I}=1,...,n_v$ is referring to the $n_v$ physical vector fields.
Then, the gauged Lagrangian \eqref{boslag2}-\eqref{GCS} in the electric frame is exactly the Lagrangian
\eqref{Full-Lagrangian}, where the matrix $k$ is identified with the
remaining electric components of the embedding tensor,
$k^\Gamma_I = \Th\indices{_{\check{I}}^\Gamma}$. The linear constraint \eqref{linear2} in the electric frame corresponds to \eqref{eq:constr-antisymmetry} and
\eqref{eq:constr-chernsimons}. The quadratic constraint \eqref{quadratic2} on the
embedding tensor in the electric frame corresponds to the constraint \eqref{eq:quadraticconstraint} on $k$ \cite{Barnich:2017nty}. As explained in Sections
\ref{sec:parametrization} and \ref{sec:parametrization-2nd-order}, the
constraints can be refined using the split corresponding to the
various ($U$, $W$, $V$) types of symmetries.

We showed in Chapter \ref{ch:Isom-Emb-VS-models} that the embedding tensor
formalism does not allow for more general deformations \cite{Henneaux:2017kbx} than those of
the Lagrangian \eqref{eq:lagnew} studied in this chapter. Indeed, we have seen in Section \ref{sec:BV-def-Embed} that their
BRST cohomologies are isomorphic even though the field content and
gauge transformations are different.  Conversely, as long as one
restricts the attention to the symmetries of \eqref{eq:lagnew} that are
of the electric type \eqref{eq:64}, we showed that the embedding
tensor formalism captures all consistent deformations that deform the
gauge transformations of the fields.

\section{Applications}\label{sec:applications}

\subsection{Abelian gauge fields: $U$-type gauging}\label{sec:481}

As a first example, let us consider the case where we have no scalars,
$\mathcal{I}_{IJ} = \delta_{IJ}$, $\mathcal{R}_{IJ} = 0$. The
Lagrangian is then simply
\begin{equation}
  \label{eq:YM} \mathcal{L} = -\frac{1}{4}
  \delta_{IJ} F^I_{\mu\nu} F^{J\mu\nu}.
\end{equation}

From \eqref{eq:67} and comparing with \eqref{var-I} and \eqref{var-R-theta}, it can be shown that $U$-type symmetries are of
electric form and moreover there are no $W$-type
symmetries. We have in this case
$\mathfrak{g}_U = \mathfrak g_e=\mathfrak{so}(n_v)$. 

The vector fields transform in the fundamental representation of
$\mathfrak{so}(n_v)$. A basis of the Lie algebra $\mathfrak{so}(n_v)$ may be labeled by an
antisymmetric pair of indices $[LM]$ that now plays the role of the
index $u$,
\begin{equation}
\delta_{[LM]} A^I_\mu = (f_{[LM]})\indices{^I_J} A^J_\mu,\quad 
\qquad (f_{[LM]})\indices{^I_J} = \half(\delta^I_L \delta_{JM} - \delta^I_M
                                \delta_{JL})\label{eq:so(n)}. 
\end{equation}

Concerning the associated gaugings, the matrices
${f_{[LM]}}_{IJ}=\delta_{II'}(f_{[LM]})\indices{^{I'}_J}$ are
antisymmetric in $I,J$; therefore, the structure constants of the
gauge group
\begin{equation}
f^{(\delta)}_{IJK} = (f_{[LM]})\indices{_{IJ}}  k^{LM}_K
\end{equation}
are automatically antisymmetric in their first two indices. Here the superscript $(\delta)$ means that $f$ is contracted by the kronecker $\delta$, more precisely $f^{(\delta)}_{IJK}=\delta_{IL}f\indices{^{(\delta)L}_{JK}}$. The
constraint \eqref{eq:constr-antisymmetry} on $k^{LM}_K$ ensures antisymmetry in the last two indices
which in turn implies total antisymmetry and the Jacobi
identity is obtained from \eqref{eq:quadraticconstraint}. Moreover, any set of totally antisymmetric structure
constants can be obtained in this way by taking
$k^{LM}_K =f\indices{^{LM}_K}$, as can be easily seen
using the expression for $f_{[LM]}$ given above.

We thereby recover the result of \cite{Barnich:1993pa} stating
that the most general deformation of the free Lagrangian \eqref{eq:YM}
that is not of $V$ or $I$-type is given by the Yang-Mills Lagrangian
with a compact gauge group of dimension equal to the number of vector fields.

{\bf Remark:} Note that Poincar\'e (conformal) symmetries (for $n=4$)
are of $V$-type if one allows for $x^\mu$-dependent local
functions. If such a dependence is allowed for $U,W$-type symmetries
and gaugings as well, results can be very different. For instance, as
shown by equation (13.21) of \cite{Barnich:2000zw}, if $n\neq 4$, there
are additional $U$-type symmetries described by the cohomology class
\begin{equation}
  \label{eq:36}
U^{-1}=  d^nx f_{(IJ)}\big[C^{*I}C^J+A^{*\mu I} A_\mu^J+\frac{2}{n-4}F^I_{\mu\nu}x^\mu
  A^{*\nu J}\big],
\end{equation}
where indices $I,J,\dots$ are raised and lowered with the Kronecker
delta. The associated Noether current can be obtained by working out
the descent equation following \eqref{eq:18b},
$sU^{-1}+d (f_{(IJ)}[\star A^{*I}C^J+\star F^IA^J+{J_U}^{IJ}] ) = 0$, where 
\begin{equation}
  \label{eq:58}
  {J_U}^{IJ}=\frac{2}{n-4}(T_{\mu\nu})^{IJ} x^\nu\star dx^\mu,\quad
  ({T^\mu}_\nu)^{IJ}=F^{(I\vert \mu\rho}F^{\vert J)}_{\rho\nu}+\frac{1}{4}
  F^{I\alpha\beta}F^J_{\alpha\beta}\delta^\mu_\nu.
\end{equation}

In other words, $\mathfrak g_U=\mathfrak{gl}(n_v)$. Note also
that these $U$-type symmetries involve a non-vanishing ${}^U
g^I_\mu$. It has furthermore been shown in section 13.2.2
of \cite{Barnich:2000zw} that there are associated $U$-type gaugings
and cohomology classes in higher ghost numbers. In the present
context, they are obtained as follows: the role of $u$ for the
additional symmetries is played by a symmetric pair of indices $(LM)$,
\begin{equation}
  \label{eq:37}
  \delta_{(LM)} A^I_\mu= (f_{(LM)})\indices{^I_J} A^J_\mu,\quad
  (f_{(LM)})\indices{^I_J}
  =\half(\delta^I_L \delta_{JM} + \delta^I_M
                                \delta_{JL}). 
\end{equation}
Once the linear constraints \eqref{linf} on
$k^{(LM)}_{K_1\dots K_{g+1}}$ are fulfilled, the associated $U$-type
gaugings and higher cohomology classes can be read off from equation
\eqref{eq:18b} when taking \eqref{eq:352}.

After multiplying \eqref{eq:36} by $n-4$, it represents for $n=4$ the
$V$-type symmetry associated with the dilatation of the conformal
group. The associated cubic and higher order vertices for the full
conformal group have been studied in detail in \cite{Brandt:2001hs}.

\subsection{Abelian gauge fields with uncoupled scalars:
  $U,V$-type gaugings}

We now take the case
\begin{equation}
  \mathcal{L} = \mathcal{L}_S(\phi^i, \partial_\mu \phi^i)
  - \frac{1}{4} \delta_{IJ} F^I_{\mu\nu} F^{J\mu\nu},
\end{equation}
where there is no interaction between the scalars and the vector
fields. The $\mathfrak{g}_U$ algebra is again $\mathfrak{so}(n_v)$
and there are no $W$-type symmetries.

The electric symmetry algebra is the direct sum of
$\mathfrak{so}(n_v)$ with the electric $V$-type symmetry algebra
$\mathfrak g_s$ of the scalar Lagrangian. The matrices $f_\Gamma$
split into two groups and are given by
\begin{equation}
  (f_\alpha)\indices{^I_J} = 0, \qquad (f_{[LM]})\indices{^I_J}
  = \half(\delta^I_L \delta_{JM} - \delta^I_M \delta_{JL})
\end{equation}
where $\alpha = 1, \dotsc, \dim \mathfrak g_s$ labels the $G_s$
generators and the antisymmetric pair $[LM]$ labels the $SO(n_v)$
generators as before. The matrix $k^\Gamma_I$ accordingly splits in
two components $k^{LM}_I$ and $k^\alpha_I$. The constraints on
$k^{LM}_I$ again amount to the fact that the quantities
$f\indices{^I_{JK}} = (f_{[LM]})\indices{^I_J} k^{LM}_K$ are the
structure constants of a compact Lie group. The constraint on
$k^\alpha_I$ tells us that the gauge variations
\begin{equation}
\delta \phi^i = \epsilon^I k^\alpha_I \Phi^i_\alpha(\phi)
\end{equation}
close according to the structure constants $f\indices{^I_{JK}}$. In
the case where these variations are linear,
$\Phi^i_\alpha(\phi) = (t_\alpha)\indices{^i_j} \phi^j$, the constraint
is that the matrices $T_I = - k^\alpha_I t_\alpha$ form a representation
of the gauge group, $[T_I, T_J] = f\indices{^K_{IJ}} T_K$.

\subsection{Bosonic sector of $\mathcal{N}=4$ supergravity}\label{sec:N=4}

Neglecting gravity, the bosonic sector of $\mathcal{N}=4$ supergravity
is given by two scalar fields parametrizing the coset $SL(2,\R)/SO(2)$
along with $n_v=6$ vector fields
{\cite{Cremmer:1977tc,Das:1977uy,Cremmer:1977tt,Cremmer:1979up}.  We study three formulations of this model, where we
determine the symmetry algebras $\mathfrak{g}_U$ and $\mathfrak{g}_e$
and the allowed gaugings. In all formulations, the scalar Lagrangian
is determined by
\begin{equation}
\phi^i=(\phi,\chi),\quad g_{ij}={\rm diag}(1,e^{2\phi}),\quad
V=0\label{eq:14a}. 
\end{equation}
They differ by the form of the matrices $\mathcal{I}$ and $\mathcal{R}$.

\subsubsection*{$SO(6)$ formulation}

The vector Lagrangian
  \eqref{eq:lagnew} is determined by  
\begin{equation}
  \mathcal{I}_{IJ} = e^{-\phi}\delta_{IJ},\quad \mathcal{R}_{IJ}=
  \chi\delta_{IJ}.\label{eq:10a} 
\end{equation}
When $f^{(\delta)}_{IJ}$ is antisymmetric, the transformations
$\delta A_\mu^I={f^I}_J A^J_\mu$ define an $\mathfrak{so}(6)$
sub-algebra of $U$-type symmetries on their own. Note also that we can
assume $h_{IJ}$ to be symmetric. Equations \eqref{eq:67} then imply
that the traceless parts of $f^{(\delta)}_{IJ},h_{IJ}$ have to vanish.

If $f^{(\delta)}_{IJ}=\delta_{IJ}\eta^0$, equations \eqref{eq:67} are
solved with $h_{IJ}=0$, $\Phi^\phi_0(0)=2\eta^0$, and $\Phi^\chi_0(0)=0$. It
then follows that \eqref{eq:19} are solved with
\begin{equation}
\Phi_0^\phi=2\eta^0,\quad
\Phi_0^\chi=-2\eta^0\chi\label{eq:41}.
\end{equation}
Equation \eqref{eq:17a} is
then also solved with $g^I=0$ and
$I^3=2\eta^0(e^{2\phi}\star d\chi\chi-\star d\phi)$. According to
equation \eqref{eq:18b}, the associated cohomology class is given by
\begin{equation}
    \label{eq:15a} \omega^{-1,4}= \eta^0[d^4x C^{*}_IC^I+\star A^{*}_IA^I
+2(\star\phi^*-\star\chi^*\chi)],
\end{equation}
with $s\omega^{-1,4}+d[\eta^0(\star A^{*}_IC^I+G_IA^I)+I^3]=0$. This
cohomology class encodes the symmetry
\begin{equation}
  \delta A^I_\mu = \eta^0 A^I_\mu, \quad \delta \phi=2 \eta^0, \quad \delta\chi=-2\eta^0 \chi,
\end{equation}
with
$\delta \mathcal L_0=0$. The associated Noether current is given by 
\begin{equation}
  j^\mu=[-(e^{-\phi}F^{\mu\lambda}_I-\half\chi
\epsilon^{\mu\lambda\rho\sigma}F_{I\rho\sigma})A_\lambda^I-2\d^\mu\phi+2\chi
e^{2\phi}\d^\mu\chi].
\end{equation}
It cannot be made gauge invariant through allowed redefinitions. 

For $f_{IJ}=0$, $h_{IJ}=\eta^+\delta_{IJ}$,
$\Phi^\chi_0=-2\eta^+,\Phi^\phi_0=0$ is a solution to the full problem
\eqref{eq:17a} since
$\half F^I F_I=s\star \chi^*+d(e^{2\phi}\star d\chi)$. This gives then
the only class of $W$-type, which is also of restricted type. More
explicitly, $W^{-1} = \star\chi^*$ with
$s\star \chi^*+d(-\half A^IF_I+ e^{2\phi}\star d\chi)=0$. The symmetry
it describes is $\delta \chi= \eta^+$ with the associated Noether
current given above that can not be made gauge invariant.

The algebra $\mathfrak{g}_U$ is therefore isomorphic to
$\mathfrak{so}(6)\oplus \mathfrak h$, where $\mathfrak h$ is the
sub-algebra of $\mathfrak{sl}(2,\mathbb R)$ generated by diagonal
traceless matrices. It is a sub-algebra of the electric symmetry
algebra $\mathfrak{g}_e = \mathfrak{so}(6)\oplus \mathfrak b^+$, where
$\mathfrak b^+$ corresponds to the sub-algebra of
$\mathfrak{sl}(2,\mathbb R)$ of upper triangular matrices, see Section \ref{sec:EM-dual-N4_sug}. 
The electric algebra acts as
\begin{equation}
\delta \phi = 2 \h^0, \quad
\delta \chi = - 2 \h^0 \chi + \h^+,\quad
  \delta A^I_\mu = \h^0 A^I_\mu + \h^{LM} (f_{[LM]})
            \indices{^I_J} A^J_\mu \label{eq:N=4deltaA}.
\end{equation}
Accordingly $f_\Gamma=(f_0,f_+,f_{[LM]})$ where
$f_{[LM]}$ are given in \eqref{eq:so(n)}, while  
\begin{equation}
(f_0)\indices{^I_J} = \delta^I_J, \qquad (f_+)\indices{^I_J} = 0.
\end{equation}
The matrix $\mathcal{R}_{IJ} = \chi \delta_{IJ}$ transforms as
\begin{equation}
\delta \mathcal{R}_{IJ} = - 2 \h^0 \mathcal{R}_{IJ} + \h^+ \delta_{IJ}.
\end{equation}
Therefore, contrary to the previous examples, the tensor
$h_{\Gamma IJ}$ has a non-vanishing component $h_{+IJ} = - \frac{1}{2}
\delta_{IJ}$.

The generalized structure constants
$f\indices{^I_{JK_1\dots K_{g+1}}} = (f_\Gamma)\indices{^I_J}
k^\Gamma_{K_1\dots K_{g+1}}$ are then
\begin{equation}
f\indices{^I_{JK_1\dots K_{g+1}}} = \delta^I_J k^0_{K_1\dots K_{g+1}} + \half
(f_{[LM]})\indices{^I_J}  k^{LM}_{K_1\dots K_{g+1}} .
\end{equation}
The linear constraint \eqref{linf} now implies that
$k^0_{K_1\dots K_{g+1}}=0$ as can be seen by taking the first three indices
equal and using the antisymmetry of the matrices $f_{[LM]}$, while
$k^{(\delta)}_{LM K_1\dots K_{g+1}}$ is restricted to be completely
skew-symmetric in all indices. In the same way, the linear constraint
\eqref{465} implies that $k^+_{K_1\dots K_{g+1}}=0$. Indeed, it reduces to
\begin{equation}
\delta_{(IJ} k^+_{K_1)\dots K_{g+1}} = 0,
\end{equation}
from which we deduce $k^+_{K_1\dots K_{g+1}} = 0$ by taking the first three indices
equal.

It follows that
there are no cohomology classes of $W$-type when $g\geqslant 0$ and that
the only cohomology classes of $U$-type when $g\geqslant 0$ are given by
\begin{equation}
    \label{eq:11} [d^4x C^{*I}\d_I+\star A^{*I}A^J\d_J\d_I+\half
    G^I A^J A^K\d_K\d_J\d_I]\Theta,
  \end{equation}
with $\Theta$ a polynomial in $C^I$ of ghost number $\geqslant 1$.

In particular, the symmetries of $\mathfrak b^+$ cannot be gauged and
the gauge algebra is given by a compact sub-algebra of
$\mathfrak{so}(6)$. The gauged Lagrangian is the original one,
except that the abelian field strengths and ordinary derivatives are replaced by non-abelian
field strength and covariant derivatives.

\subsubsection*{Dual $SO(6)$ formulation}
We now have
\begin{equation}
    \mathcal{I}_{IJ} = \frac{1}{e^{-\phi}+\chi^2e^\phi}\delta_{IJ},\quad 
    \mathcal{R}_{IJ}= -\frac{\chi
      e^\phi}{e^{-\phi}+\chi^2e^\phi}\delta_{IJ},\label{eq:18}
  \end{equation}
and the same analysis gives similar conclusions:
\begin{enumerate}

\item The cohomology classes \eqref{eq:11} are again present since $\mathcal I_{IJ}$ and $\mathcal R_{IJ}$ are still proportional to $\delta_{IJ}$.
  
\item There are no additional gaugings or cohomology classes in ghost
  number higher than $0$ of $U$ or $W$-type.

\item The only additional
  non-covariantizable characteristic cohomology comes from two
  additional solutions to \eqref{eq:67}.
\end{enumerate}
The first of these additional solutions is of electric
  $U$-type and comes from
  $f^{(\delta)}_{IJ}=\tilde \eta^0\delta_{IJ}$, $h_{IJ}=0$, with
  $\Phi^\phi_0(0)=-2\tilde \eta^0$, $\Phi^\chi_0(0)=0$. Equation \eqref{eq:19}
  reduces to
  \begin{equation}
    \label{eq:39}
    2 \tilde \eta^0 \mathcal I_{IJ}+\d_i \mathcal I_{IJ}\Phi^i_0=0,\quad
    2 \tilde \eta^0
    \mathcal R_{IJ}+2h_{IJ}\delta_{IJ}+\d_i  \mathcal
    R_{IJ}\Phi^i_0=0, 
  \end{equation}
  and is solved by
  \begin{equation}
    \label{eq:9}
    h_{IJ}=0\quad \Phi^\phi_0=-2\tilde \eta^0,\quad  
    \Phi^\chi_0=2\tilde \eta^0\chi.
  \end{equation}
This gives also a solution to the full problem since this
transformation leaves the scalar field Lagrangian invariant.
According to equation \eqref{eq:18b}, the associated cohomology class is given by
\begin{equation}
    \label{eq:15b} \omega^{-1,4}= \tilde \eta^0[d^4x C^{*}_IC^I+\star A^{*}_IA^I-2(\star\phi^*-
\star\chi^*\chi)].
\end{equation} 
The second solution is of restricted $W$-type and comes from
$f^{(\delta)}_{IJ}=0$ while $h_{IJ}=\tilde \eta^+\delta_{IJ}$ with
$\Phi^\phi_0(0)=0$, $\Phi^\chi_0(0)=2\tilde \eta^+$. Equation
\eqref{eq:19} reduces to
\begin{equation}
  \label{eq:40}
  \d_i \mathcal I_{IJ}\Phi^i_0=0,\quad 2\tilde \eta^+\delta_{IJ}+\d_i  \mathcal
    R_{IJ}\Phi^i_0=0,
\end{equation}
and is solved by
\begin{equation}
\Phi^\chi_0=2\tilde \eta^+(e^{-2\phi}-\chi^2),\quad 
\Phi^\phi_0=4\tilde \eta^+\chi\label{eq:38},
\end{equation}
This is also a solution to the full problem since these
transformations leave the scalar field Lagrangian invariant. 
The associated cohomology class is given by
\begin{equation}
  \label{eq:12a}
  2\tilde \eta^+ [\star \phi^*2\chi+\star\chi^*(e^{-2\phi}-\chi^2)]. 
\end{equation}
In this case, we therefore have
$\mathfrak{g}_U = \mathfrak{so}(6)\oplus \mathfrak h \subset
\mathfrak{g}_e = \mathfrak{so}(6)\oplus \mathfrak{b}^-$, where
$\mathfrak b^-$ is the sub-algebra of $\mathfrak{sl}(2,\mathbb R)$ of
lower triangular matrices.

Again, the symmetries of $\mathfrak b^+$ cannot be gauged and the gauge algebra is given by a compact sub-algebra of $\mathfrak{so}(6)$.

\subsubsection*{$SO(3)\times SO(3)$ formulation}
The indices split as $I=(A,A')$, where $A,A'=1,2,3$, and we have 
  \begin{equation}
\begin{split}
    \mathcal
    I_{IJ}={\rm diag}(\mathcal I_{AB},\mathcal I_{A'B'}),\quad
    \mathcal R_{IJ}={\rm diag}( \mathcal R_{AB},\mathcal R_{A'B'})\label{eq:17},\\
\mathcal I_{AB}=e^{-\phi}\delta_{AB},\quad \mathcal
I_{A'B'}=\frac{1}{e^{-\phi}+\chi^2e^\phi}\delta_{A'B'},\\
\mathcal R_{AB}=\chi\delta_{AB},\quad \mathcal R_{A'B'}=-\frac{\chi
      e^\phi}{e^{-\phi}+\chi^2e^\phi}\delta_{A'B'}.
  \end{split}
\end{equation}
Spelling out equation \eqref{eq:19} gives 
\begin{equation}
  \label{eq:21}
  \begin{split}
 & \mathcal I_{IL} {f^L}_{J}+\mathcal I_{LJ}{f^L}_{I}+\d_i\mathcal I_{IJ}{\Phi^i_0}=0,\\
& \mathcal R_{IL} {f^L}_{J}+\mathcal R_{JL}{f^L}_{I}+h_{IJ}+h_{JI}+\d_i\mathcal
    R_{IJ}{\Phi^i_0}=0. 
  \end{split}
\end{equation}
Choosing $I=A$, $J=A'$, the first equation reduces to 
\begin{equation}
  \label{eq:30}
  e^{-\phi}f^{(\delta)}_{AA'}+\frac{1}{e^{-\phi}+\chi^2e^\phi}f^{(\delta)}_{A'A}=0.
\end{equation}
Putting $\phi=0=\chi$ gives $f^{(\delta)}_{AA'}$ is symmetric under exchange of $A$, $A'$ while taking the derivative with respect to $\phi$
and then putting $\phi=0=\chi$ gives $f^{(\delta)}_{AA'}$ is antisymmetric under exchange of $A$, $A'$, thus 
$f^{(\delta)}_{AA'}=0=f^{(\delta)}_{A'A}$. When combined with the
linear constraint \eqref{linf}, this implies that the
${f^{I}}_{JK_1\dots K_{g+1}}$'s have to vanish
unless all indices are of $A$, or of $A'$, type respectively, which is
thus a necessary condition to have non trivial $U$-type solutions.

When the $f$'s vanish, the second part of \eqref{eq:21} for $I=A$, $J=A'$ gives
$h_{AA'}+h_{A'A}=0$. When combined with the linear constraint
\eqref{465}, this implies that for non-trivial solutions of $W$-type
associated to $X_{IJ,K_1\dots K_{g+1}}$ one again needs all indices to
be either of $A$ or of $A'$ type.

The discussion then reduces to the one we had before in each of the sectors. For $U$-type solutions, this gives in a first stage the
symmetries, gaugings and higher ghost cohomology classes associated with
each of the $SO(3)$ rotations separately. There are again no
additional solutions of $U$ or $W$ type when $g\geqslant 0$.

Only the remaining non-covariantizable symmetries, i.e., solutions of
type $U$ and $W$ at $g=-1$ that correspond to $\mathfrak{b}^\pm$, remain to be discussed. For the $U$ type
solutions, one finds in the first sector that
$f_{AB}=\eta^0\delta_{AB}$ with \eqref{eq:41} holding, while for the
second sector $f_{A'B'}=\tilde \eta^0\delta_{A'B'}$ with (\ref{eq:9})
holding. This gives a solution to the full problem if and only if
$\tilde \eta^0=-\eta^0$. Hence $\mathfrak
g_U=\mathfrak{so}(3)\oplus\mathfrak{so}(3)\oplus \mathfrak h$.
On the other hand the solutions of $W$ type for both sectors are
solutions to the full problem if and only if $\eta^+=\tilde \eta^+=0$ so
that there is no surviving $W$-type symmetry. In particular $\mathfrak
g_e=\mathfrak g_U$.
The symmetry of $\mathfrak{h}$ cannot be gauged, and the gauge algebra is a compact sub-algebra of $\mathfrak{so}(3)\oplus \mathfrak{so}(3)$.

This concludes the discussion with the expected results (see \cite{Das:1977pu,Freedman:1978ra,Gates:1982ct}).

\section{First order manifestly duality-invariant actions}\label{sec:first-order-actions}

\subsection{Non-minimal version with covariant gauge structure}
\label{sec:non-minimal}

We now investigate the first order formulation \cite{Bunster:2011aw,Henneaux:2017kbx}
of the models discussed previously. Those models are interesting
because they contain more symmetries and therefore potentially more
gaugings. In the original, minimal version, they are given by the
action\footnote{Note that in this section the convention for Levi-Civita is $\varepsilon^{01234...}=+1$.}
\begin{equation} \label{eq:symlag2} S = \int \!d^4x \left( \frac{1}{2}
    \Omega_{MN} {B}^{Mi} \dot{{A}}^N_i - \frac{1}{2}
    \mathcal{M}_{MN}(\phi) {B}^M_i {B}^{Ni} \right),
\end{equation}
where the potentials are packed into a vector
\begin{equation}
({A}^M) = (A^I, Z_I), \quad M= 1, \dots, 2n_v,
\end{equation}
and the magnetic fields are
\begin{equation}
{B}^{Mi} = \epsilon^{ijk} \partial_j {A}^M_k .
\end{equation}
The matrices $\Omega$ and $\mathcal{M}(\phi)$ are the
$2n_v \times 2n_v$ matrices
\begin{equation} \label{eq:OMdefa}
\Omega = \begin{pmatrix}
0 & I \\ -I & 0
\end{pmatrix}, \qquad
\mathcal{M} = \begin{pmatrix}
  \mathcal{I} + \mathcal{R}\mathcal{I}^{-1}\mathcal{R} &
  - \mathcal{R} \mathcal{I}^{-1} \\
- \mathcal{I}^{-1} \mathcal{R} & \mathcal{I}^{-1}
\end{pmatrix},
\end{equation}
each block being $n_v\times n_v$. The matrix
$\mathcal{N} = \Omega^{-1} \mathcal{M}$ is symplectic,
$\mathcal{N}^T \Omega \mathcal{N} = \Omega$.

Local BRST cohomology and gaugings for this class of models with
non-covariant gauge symmetries $\delta A^M_i=\partial_i\epsilon^M$
could then be discussed by generalizing the results of
\cite{Bekaert:2001wa} in the presence of coupled scalars.

However, in order to be able to directly use the discussion of local
BRST cohomology developed for the second order covariant Lagrangian in
the case of the first order manifestly duality invariant formulation,
we consider a modification of the non-minimal variant
\cite{Barnich:2007uu} with additional scalar potentials for the
longitudinal parts of electric and magnetic fields. More precisely, we
now take instead of \eqref{eq:symlag2} the action
\begin{equation}
  \label{eq:48}
  S[A_\mu^M,D^M,\pi_M,\phi^i]=S_S[\phi]+S_{DP},
\end{equation}
with
\begin{multline}
  \label{eq:49}
  \mathcal L_{DP}=\half [ \Omega_{MN}(\mathcal B^{Mi}+\d^iD^M)(\d_0
  A_i^N-\d_i A_0^N)-\mathcal B^{Mi} \mathcal M_{MN}(\phi)\mathcal
  B^{N}_i ]\\+\pi_M\d_0 D^M -\half \pi_M (\mathcal
  M^{-1})^{MN}\pi_N-\mathcal V(\phi,D).
\end{multline}
Here 
\begin{equation}
\mathcal B^{Mi}=\epsilon^{ijk}\d_j
  A_k^M+\d^i D^M,\label{eq:68}
\end{equation}
spatial indices $i,j,k,\dots$ are raised and lowered with
$\delta_{ij}$ and its inverse, with $\Omega_{MN}$ the symplectic
matrix, 
$\mathcal M_{MN}$ symmetric and invertible and
\begin{equation}
(\d_i\mathcal
V)(0,0)=0=(\d_M\mathcal V)(0,0), \quad \cV (\phi, 0) = 0 . \label{eq:51} 
\end{equation}
The modification with respect to \cite{Barnich:2007uu} consists in the
addition of the kinetic and potential terms for the longitudinal
electric and magnetic potentials in the last line of \eqref{eq:49}.
Defining
\begin{equation} \label{eq:53}
\begin{split}
  \mathcal F^M_{\mu\nu} &=\d_\mu A_\nu^M-\d_\nu A_\mu^M+\star
  S^M_{\mu\nu},\\
   \star  S^M_{0i} &=\Delta^{-1}(\Omega^{-1})^{MN}\d_0\d_i\pi_N,\\
   \star  S^M_{ij} &=\epsilon_{ijk}\d^k D^M,
\end{split}
\end{equation}
we have $\mathcal B^{Mi}=\half\epsilon^{ijk}\mathcal F_{jk}$ and can write
\begin{align}
  S_{DP}=\frac{1}{4} \int d^4x  [&\Omega_{MN}\epsilon^{ijk} (\mathcal
  F^M_{jk}+\star S^M_{jk})\mathcal F^N_{0i}-\mathcal F^M_{ij} \mathcal
  M_{MN}\mathcal F^{Nij}\nonumber\\
  &-2  \pi_M (\mathcal M^{-1})^{MN}\pi_N-4\mathcal V],    \label{eq:54}
\end{align}
where a total derivative has been dropped.

The gauge invariances are then doubled but still of the same covariant
form as in the second order Lagrangian case,
\begin{equation}
  \label{eq:52}
  \delta A^M_\mu=\d_\mu\epsilon^M,\quad \delta D^M=0,\quad
  \delta\pi_M=0,\quad \delta \phi^i=0.
\end{equation}
The equations of motion for the gauge and scalar potentials are
determined by the vanishing of
\begin{equation}
  \label{eq:50a}\begin{split}
 &   \vddl{\mathcal L_{DP}}{\pi_M}=-(\mathcal
 M^{-1})^{MN}(\pi_N-\mathcal M_{NL}\d_0 D^L),\\
 &   \vddl{\mathcal L_{DP}}{D^M}=\Omega_{MN}(\Delta A_0^N-\d_0 \d^i A_i^N)
    +\d^i(\mathcal M_{MN} \mathcal B^N_i)-\d_0\pi_M-\frac{\d \mathcal
      V}{\d D^M},\\
 &   \vddl{\mathcal L_{DP}}{A^M_0}=-\Omega_{MN}\Delta
    D^N=-\frac{1}{2}\Omega_{MN}\epsilon^{ijk} \d_i\mathcal
    F_{jk}^N,\\
 &   \vddl{\mathcal L_{DP}}{A^M_i}=\Omega_{MN}\d_0
    \mathcal B^{Ni}-\epsilon^{ijk}\d_j(\mathcal M_{MN} \mathcal B^N_k)
    =\half\Omega_{MN}\epsilon^{ijk} \d_0\mathcal F_{jk}^N-\d_j
    (\mathcal M_{MN}\mathcal F^{Nij}).
\end{split}
\end{equation}
The first set of equations then allows one to eliminate the momenta
$\pi_M$ by their own equations of motion. When $\Delta$ is invertible,
the second and third set of equations allow one to solve $D^M$ and
$A_0^M$ by their own equations of motion in the action, which yields
\eqref{eq:symlag2}. It is in this sense that these variants of the
double potential formalism are equivalent, but of course not locally
so. The third and fourth set of equations can be written as
\begin{equation}
  \label{eq:55}
 \vddl{\mathcal L_{DP}}{A^M_\mu}=\d_\nu\star G^{\mu\nu}_M,
\end{equation}
when defining
\begin{equation}
\star G^{i0}_{M}=\half\Omega_{MN}\epsilon^{ijk} \mathcal F_{jk}^N,
  \quad \star G_{M}^{ij}=-\mathcal M_{MN}\mathcal F^{Nij}.
\end{equation}

This definition implies that the components of
\be
G_M=\half G_{Mjk} dx^jdx^k+G_{Mi0}dx^idx^0
\ee
are explicitly given by
\begin{equation}
  \label{eq:57}
  G_{Mjk}=-\Omega_{MN}\mathcal F_{jk}^N,\quad G_{Mi0}=\half
  \epsilon_{ijk}\mathcal M_{MN}\mathcal F^{Njk}.
\end{equation}
After elimination of the $\pi_M$, the action of the theory can
then also be written as the integral of 
$\mathcal{L}_0=\mathcal{L}_{ES}+ \mathcal{L}_V$ with
\begin{equation}
  \mathcal{L}_{ES}=\mathcal L_S-\frac{1}{2}\partial_\mu
  D^M\mathcal{M}_{MN}\partial^\mu D^N-\mathcal V,
  \quad d^4x \mathcal{L}_V=\int_0^1\frac{dt}{t} [G_MF^M][tA^M,D^M,\phi^i].
\end{equation}
so that the scalar sector has been enlarged to $\phi^m=(\phi^i,D^M)$
and the scalar metric and potential are now
$(g_{ij},\mathcal M_{MN})$, respectively $(V,\mathcal V)$. It is thus
a particular case of the actions of the form
\eqref{4.4} studied in Section \ref{sec:gauging}.

\subsection{Local BRST cohomology}
\label{sec:local1st}

The master action is given by 
\begin{equation}
  S=\int d^4x\,[ \mathcal L_0+A^{*\mu}_M\partial_\mu C^M]
\end{equation}
with an antifield and ghost sector that is doubled as compared to the
second order covariant formulation.

We then can copy previous results:

(i) $H^g(s)=0$ for $g\leq -3$.

(ii) $H^{-2}(s)$ is doubled: $U^{-2}=\mu^M d^4x C^*_M$ with descent equation

\begin{equation}
s\, d^4x\, C^*_M+d \star A^*_M=0,\quad s \star A^*_M+d
G_M=0,\quad sG_M=0.\label{eq:14bis}
\end{equation}
Characteristic cohomology $H^{n-2}_W(d)$ is then represented by the
2-forms $\mu^MG_M$.

For $g\geq -1$, the discussion in terms of $U$, $W$, $V$ types is the
same as before with indices $I,J,K,\dots \to M,N,O,\dots$ on vector
potentials, ghosts and their antifields, and
$i,j,k\dots\to m,n,o\dots$ on scalar fields.

The obstruction equation for symmetries, equation \eqref{eq:17a},
becomes
\begin{multline}
  \label{eq:17c}
  G_M F^N {f^M}_N+F^MF^Nh_{MN}+dI^{n-1}\\
-(dG_Mg^{M}+[d(g_{ij} \star
  d\phi^j)+\star
  \d_i(\mathcal{L}_{ES}+\mathcal{L}_V)]\Phi^{i}
  +\vddl{\mathcal L_0}{D^M} \Phi^{M})
=0.
\end{multline}

\subsection{Constraints on $W,U$-type cohomology}
\label{sec:absence-w-type}

When there is no explicit $x^\mu$ dependence and $V=0=\mathcal V$,
putting all derivatives of $F^M_{\mu\nu},\phi^i,D^M$ to zero, one
remains with
\begin{equation}
  \label{eq:18c}
  {G_M}|_{{\rm der}=0} F^N{f^M}_N+F^MF^Nh_{MN}
-\star \d_m \mathcal{L}_V
\Phi^{m}|_{{\rm der}=0}=0,
\end{equation}
where ${G_M}|_{{\rm der}=0}$ amounts to replacing $\mathcal
F^M_{\mu\nu}$ by $F^M_{\mu\nu}$ in \eqref{eq:57}. 

Using
$-\d_m\star \mathcal{L}_V=\delta^i_m\half\d_i G_M|_{{\rm der}=0} F^M$,
and the decomposition 
$\Phi^{m}|_{{\rm der}=0}=\Phi^{m}_0+\Phi^{m}_1+\dots$, where the
$\Phi^{m}_{n}$ depend on undifferentiated scalar fields and are
homogeneous of degree $n$ in $F^M_{\mu\nu}$, the
equation implies
\begin{equation}
  \label{eq:26c}
G_M|_{{\rm der}=0}F^N {f^M}_N
+F^MF^Nh_{MN}+\half
  \d_iG_M|_{{\rm der}=0}F^M\Phi^{i}_0=0.
\end{equation}
When taking account that
\begin{equation}
  G_M|_{{\rm der}=0} F^N=d^4x\frac{1}{2} [\Omega_{OM}\epsilon^{ijk}
  F^O_{jk} F^N_{0i}-  \mathcal
  M_{MO} F^O_{jk} F^{Njk}],\label{eq:56}
\end{equation}
this gives an equation of the type
\begin{equation}
  \label{eq:28b}
  \frac{1}{4} d^4x[  \mathcal O_{MN}(\phi)\epsilon^{ijk}
  F^M_{jk} F^N_{0i}-\mathcal P_{MN}(\phi)F^M_{jk}F^{Njk}]=0, 
\end{equation}
where 
\begin{equation}
  \label{eq:29a}
  \mathcal P_{MN}=2\mathcal M_{O(M}{f^O}_{N)}
+ \d_i\mathcal
M_{MN}\Phi^{i}_0, 
\end{equation}
\begin{equation}
  \label{eq:27b}
  \quad  \mathcal
  O_{MN}=2\Omega_{MO}{f^O}_N+2h_{MN},
\end{equation}
and $h_{MN}=h_{NM}$ on account of \eqref{Y1}.  Note that there is one
less term as compared to \eqref{eq:27} since the kinetic term does not
depend on the scalars and also that $\mathcal O_{MN}$ is not
symmetric.

Now both terms have to vanish separately because they involve
different field strengths, 
\begin{equation}
\mathcal P_{MN}=0,\quad \mathcal O_{MN}=0.\label{eq:19a}
\end{equation}
Setting $\phi^i=0=D^M$ then gives
\begin{equation}
\begin{split}
  & f^{(\mathcal M(0))}_{MN}+f^{(\mathcal M(0))}_{NM}
  +(\d_i\mathcal M_{MN})(0)\Phi^{i}_0(0)=0,\\ 
& f^{(\Omega)}_{MN}=-h_{MN}, \label{eq:35a}
\end{split}
\end{equation}
with $f^{(\Omega)}_{MN}=\Omega_{MO}{f^O}_N$. 
Consider first symmetries of $W$-type, i.e., take the case when
the $f$'s vanish. The first equation is then satisfied with
$\Phi^i_0(0)=0$, while the second equation then requires
$h_{MN}$ to vanish. This implies:

{\it There are neither $W$-type symmetries nor $W$-type cohomology in ghost
  numbers $g\geq 0$ for the first order model.}

As a consequence, $h_{MN}={h_u}_{MN}$, and the second of equation \eqref{eq:35a} is
equivalent to 
\begin{equation}
{f_u}^{(\Omega)}_{[MN]}=0,\quad {f_u}^{(\Omega)}_{(MN)}=-{h_u}_{MN}\label{eq:35}.
\end{equation}
It follows that:

{\it The algebra $\mathfrak{g}_U$ is the largest sub-algebra of
  $\mathfrak{sp}(2n_v,\mathbb R)$ that can be turned into symmetries of
  the full theory. All non-trivial $U$-type symmetries require a non-vanishing
  ${h_u}_{MN}$ and thus involve a Chern-Simons term in their Noether
  currents.}

On its own, the first equation of \eqref{eq:35a} is solved for
skew-symmetric $f^{(\mathcal M(0))}_{MN}$ with vanishing
$\Phi^i(0)$. Symmetric $f^{(\mathcal M(0))}_{MN}$ needs a non trivial scalar
symmetry.

For $U$-type cohomologies in higher ghost number $g\geq 0$, the
$k^u_{O_1\dots O_{g+1}}$ tensor has to satisfy \eqref{linf},
which becomes
\begin{equation}
  \label{eq:31}
  {f_u}^{(\Omega)}_{M(N}k^u_{O_1)\dots O_{g+1}}=0.
\end{equation}
The object
$D_{MO_1NO_2\dots O_{g+1}}={f_u}^{(\Omega)}_{MN}k^u_{O_1\dots
  O_{g+1}}$ is then symmetric in the first and third indices because
${f_u}^{(\Omega)}_{MN}$ is symmetric, and antisymmetric in the second
and third indices on account of \eqref{eq:31}. It thus has to vanish,
\begin{multline}
  D_{MO_1NO_2\dots O_{g+1}} =D_{NO_1MO_2\dots O_{g+1}} =-D_{NMO_1O_2\dots
    O_{g+1}}   =-D_{O_1MNO_2\dots O_{g+1}}\\
  =D_{O_1NMO_2\dots
  O_{g+1}} =D_{MNO_1O_2\dots O_{g+1}} =-D_{MO_1NO_2\dots O_{g+1}}.
\end{multline}
It follows that $k^u_{O_1\dots O_{g+1}}=0$:

{\it There are no $U$-type cohomology classes in ghost number $g \geq 0$.}

In particular, there are no $U$-type gaugings even though there are
$U$-type symmetries. We thus recover the results on gaugings of
\cite{Bunster:2010wv} from the current perspective.

\subsection{Remarks on $GL(2n_v)$ transformations}

The two remarks on linear changes of variables from Section
\ref{sec:frame-transf} also apply in the first order case. More
precisely, the second remark can be rephrased as follows. 

The general discussion of the structure of the BRST cohomology of the
first order model in Sections \ref{sec:local1st} and
\ref{sec:absence-w-type} goes through unchanged for arbitrary
skew-symmetric non-degenerate $\Omega_{MN}$ and symmetric
non-degenerate $\mathcal M_{MN}$. The local BRST cohomology for sets
of $\Omega_{MN}$, $\mathcal M_{MN}$ related by $GL(2n_v, \mathbb R)$
transformations will be isomorphic, whereas explicit results for the
local BRST cohomology do depend on the equivalence classes. For
instance for the symmetries, this is the case when explicitly solving
the obstruction equation \eqref{eq:17c}.
As concerns $\Omega_{MN}$, there is just one equivalence class since
all such matrices are related to a canonical $\Omega_{MN}$, say
$\Omega_{MN}=\delta_{IJ}\epsilon_{ab}$, by a $GL(2n_v, \mathbb{R})$
transformation. Hence, one can restrict oneself to equivalence classes
of $\Omega_{MN}$, $\mathcal M_{MN}$ with canonical $\Omega_{MN}$, and
$\mathcal M_{MN}$'s related by $Sp(2n_v,\mathbb R)$ changes of variables.

The first remark of Section \ref{sec:frame-transf} then boils down to
the statement that the algebra $\mathfrak{g}_U$ is the largest
sub-algebra of $\mathfrak{sp}(2n_v,\mathbb R)$ that can be turned into
symmetries of the full theory, in agreement with the discussion of the
previous section. In addition we have recovered there the result that
the gauge algebra remains abelian.

\subsection{Application to the bosonic sector of $\mathcal N=4$
  supergravity}
\label{sec:application1}

For definiteness, let us again concentrate on the bosonic sector of
four dimensional supergravity, without gravity. As in Section
\ref{sec:N=4}, we use the standard second order formulation for the
$SL(2,\mathbb R)/SO(2)$ sigma model. Alternatively, one could use a
first order formulation in terms of fields parametrizing
$SL(2,\mathbb R)$, with a first class constraint eliminating the field
for the $SO(2)$ subgroup. It would provide a first order formulation
for all fields and make all global symmetries manifest.

To this scalar action, we first couple one vector field, i.e. add the
action associated to \eqref{eq:49} where $\mathcal V=0$, the indices
$M,N$ take two values $a,b$, $\Omega_{ab}=\epsilon_{ab}$, and
$\mathcal M_{ab}=M^{-1}_{ab}$. The matrix $M$ and its inverse are
given by
\begin{equation}
  \label{eq:2a}
  {M}=\begin{pmatrix} e^\phi & \chi e^\phi \\
\chi e^\phi & \chi^2 e^\phi +e^{-\phi}
\end{pmatrix},\quad 
{{M^{-1}}}=\begin{pmatrix} \chi^2 e^\phi +e^{-\phi} & -\chi e^\phi \\
-\chi e^\phi & e^{\phi}
\end{pmatrix}
\end{equation}
and are such that $M$ transforms as $M\to g^TMg $ under an
$SL(2,\mathbb R)$ transformation. The model is invariant under
$SL(2,\mathbb R)$ if the other fields transform as $A^a\to (g^T A)^a$,
$D^a\to (g^T D)^a$, $\pi_a\to (g^{-1}\pi)_a$ because $SL(2,\mathbb R)$
transformations are symplectic, $g\epsilon g^T=\epsilon$.

For the $U$-type symmetries, equation \eqref{eq:35} requires
$f^{(\epsilon)}_{ab}=\epsilon_{ac} {f^c}_b$ to be symmetric, so there
are at most 3 linearly independent solutions. According to the above
discussion, all of these give rise to symmetries, which need
${h_u}_{ab}$ and also $\Phi^i_u$. The $U$-type symmetries constitute
the $\mathfrak{sl}(2,\mathbb R)$
electric symmetry algebra. 

We now consider the coupling to six vector fields in the different
formulations of Section \ref{sec:N=4}.  For the $SO(6)$ invariant
model, $M=(I,a)$, $\Omega_{MN}=\delta_{IJ}\epsilon_{ab}$ and
$\mathcal M_{MN}= \delta_{IJ}{M^{-1}}_{ab}$, while the dual
formulation corresponds to $\mathcal M_{MN}=
\delta_{IJ}M_{ab}$. Finally, in the $SO(3)\times SO(3)$ formulation
$\mathcal M_{MN}=(M^{-1}_{ab}\delta_{AB},M_{ab}\delta_{A'B'})$.

It then follows from \eqref{eq:35a} that both in the $SO(6)$ invariant
formulation and in the dual formulation, the electric symmetry algebra
is $\mathfrak{sl}(2,\mathbb R)\oplus \mathfrak{so}(6)$, where the
$\mathfrak{sl}(2,\mathbb R)$ transformations on the vectors $A^{(I,a)}$
and on $D^{(I,a)},\pi_{(J,b)}$ in the dual formulation corresponds to
the infinitesimal version of the above transformations where
$g^T\to g^{-1}$.

Finally, in the $SO(3)\times SO(3)$ formulation, the electric symmetry
algebra is also $\mathfrak{sl}(2,\mathbb
R)\oplus\mathfrak{so}(6)$. This is so because the $SL(2,\mathbb R)$
element $\epsilon$ is such that $M=\epsilon^TM^{-1} \epsilon$.

\chapter{Conclusions and Prospects}\label{ch:conclusion}

In \textbf{Part} \ref{part1} of this thesis, we discussed some aspects of timelike duality in the context of exotic versions of maximal supergravity. We first described the extension of Adams' theorem, which states that the only spheres which admit global parallelism are
$S^0$, $S^1$, $S^3$ and $S^7$, to the pseudo-Riemannian geometry. There we argued that in three-dimensions, 
the pseudo-heperbolic manifold $H^{2,1}=SO(2,2)/SO(2,1)$, commonly known as AdS$_3$, and
the pseudo-sphere $S^{1,2}=SO(2,2)/SO(1,2)$ are parallelizable. These two manifolds are interchangeable under the 
metric parity ($g_{\mu\nu}\rightarrow -g_{\mu\nu}$) which changes the sign of the curvature. For that reason, sometimes in literature $S^{1,2}$ is called $-\textrm{AdS}_3$ or AAdS$_3$.
In seven-dimensions, the manifolds $S^{3,4}=SO(4,4)/SO(3,4)$ and $H^{4,3}=SO(4,4)/SO(4,3)$ are parallelizable. Once again, we can find one from one another by metric reversal. We discussed as well the construction of parallelizable seven-(pseudo-)spheres in the presence of torsion (or similarly with special holonomy) which are the non-symmetric coset spaces $Spin^+(3,4)/G^*_{2,2}$ and $Spin^+(4,3)/G^*_{2,2}$.

We discussed the $S^7$ manifold as it plays an important role in the compactification of eleven-dimensional supergravity, where one can find an $SO(8)$ gauged maximal supergravity in AdS$_4$ background via Freund-Rubin compactification on (round) $S^7$ or a $Spin(7)$ gauged non-supersymmetric theory in AdS$_4$ background via Englert compactification on a parallellized squashed $S^7=Spin(7)/G_2$. We discussed the effect of timelike T-duality which propounds, beside the standard $M$-theory, the presence of $M'$ and $M^*$-theories which are dual to the standard theory. We showed that the compactification on parallelizable $S^{3,4}$ is possible if one works in the picture described by $M'$ dual theory \cite{Henneaux:2017afd}. \textit{This solution similarly to the Englert one will result in a non-supersymmetric theory in AdS$_4$ background.} Besides the known Englert solution, the compactification on parallelizable 7-manifolds for the other supergravity theories in eleven-dimensions (the low energy limit of standard $M$- and $M^*$-theories) is not possible, however one can see that there are plenty of solutions as a result of Freund-Rubin compactification of all dual $M$ theories which results in four dimensional theories in non-AdS backgrounds.

In \textbf{Part} \ref{part2}, we investigated gaugings of models inspired by the vector sector of extended supergravities. We first explained the changes that were brought upon by considering scalar-vector couplings of particular non-minimal type. It enhances the electric-magnetic duality group to a subgroup of symplectic group $Sp(2n_v,\mathbb{R})$ for a theory with $n_v$ vectors coupled to scalars. Then we discussed the interplay between the first-order and second-order Lagrangians where the electric-magnetic duality is manifest in the former while only part of electric-magnetic duality acts locally in the latter and is manifest. However, the first-order formulation is more rigid; we presented a proof that there is no Yang-Mills deformations of first-order Lagrangians in the presence of scalars \cite{Henneaux:2017kbx}. A result that in fact is true regardless of the presence of scalars. 

Given the existence of the embedding tensor formalism to construct a gauged supergravity by minimally coupling scalars to vectors and promoting the vector fields to non-abelian connections on the fibre bundle (and adding the appropriate interactions to render the Lagrangian invariant), we showed that not only the standard second-order action with the appropriate choice of the corresponding duality frame is equivalent with the undeformed limit of the embedding tensor Lagrangian but also that the space of local deformations of both theories are isomorphic. In fact, all local cohomology groups $H^g(s\vert d)$ are isomorphic in both formulations. In particular the potential anomalies (cohomology at ghost number one) are also the same.

We then extended our analysis to examine the possibility of those deformations that change the gauge symmetry while keeping
the gauged algebra untouched. We have systematically analyzed gaugings of vector-scalar
models through a standard deformation theoretic approach. In the case
of gauge systems, this is most naturally done in the BV-BRST antifield
formalism. 

\textit{We have shown that different types of symmetries behave
differently when one tries to gauge them. The method allows one to
find all the infinitesimal gaugings and higher order cohomology
classes once all symmetries are known.}

We classified the symmetries into $U$, $W$ and $V$-types, where the
\textbf{$U$-type} corresponds to those deformations that deform the abelian
gauge algebra, {\textbf{$W$-type} symmetries contain the topological gaugings
of \cite{deWit:1987ph} and \textbf{$V$-type} symmetries contain gauge invariant couplings and admit gauge-invariant Noether currents. The Noether currents of $U$ and $W$-type symmetries cannot be made gauge invariant. 

For the models explicitly considered, we have found that the only possible gaugings
of $U$ and $W$-types are the ones previously considered in the literature, namely
Yang-Mills and topological couplings among the gauge fields, with minimal couplings of the scalars.

We have shown in Section \ref{sec:AntiMap} that given the graded structure of the antibracket
map, the leading obstruction to extend first order deformations of
$U$-type to second order, leading to the Jacobi identity for the
structure constants, cannot be eliminated by adding $V$-terms.
Furthermore, in some cases, for instance when one imposes Poincar\'e
invariance as relevant to relativistic theories, the $V$-type
symmetries can be shown to be absent \cite{Torre:1994kb}. It turns out
that the effect of coupling the models to Einstein gravity justifies
this assumption \cite{Barnich:1995ap} and simplifies the problem. 

We have analyzed the problem in the second order Lagrangian and in the
first order manifestly duality invariant formulation, both of which
are non-locally related in space (but not in time). The results are
very different: whereas the former formulation allows for standard
gaugings, the latter formulation allows for more (generalized)
symmetries of $U$-type, but none of those can be gauged. 

\subsection*{Prospects}

As the final remark, we list the future research directions inspired by results of the thesis:

1) The no-go theorem which states that there is no local deformation available in the first-order formulation is highly based on the locality present in the formulation. The inevitable result that in all examples presented in the thesis, one could not gauge any part of global symmetry but the U-type one corresponding to the Yang-Mills deformations, can also be related to the local framework of our BV-BRST analysis. In order to go past such no-go results, one should presumably try to work in a controlled way with deformations that are spatially non-local. This as well may open a window to discuss the gaugings in the presence of sources, most importantly in the presence of Dirac magnetic monopoles, since it requires going beyond locality, see the comment under the equation \eqref{eq:EM-dual_trans_current} in Chapter \ref{ch:EM_duality}.

2) It is interesting to extend the analysis of gauging of the bosonic sector of supergravity to include the fermion sector as well. Even though, one may in general expect no new constraints but it is worth to see explicitly that the supersymmetric extension does not require any extra constraints. There has been works along this way \cite{Boulanger:2018fei}, where the deformation of a free theory in the presence of both bosons and fermions were considered. It is then intriguing to consider scalar coupling from the beginning as it enhances the electric-magnetic duality symmetry group and to analyze the deformations of such a theory.  

3) In any dimensions, the positive level generators of the very extended Kac–Moody algebra $E_{11}$ with completely antisymmetric spacetime indices are associated with the form fields of the corresponding maximal supergravity. In \cite{Bergshoeff:2007vb}, a correspondence between the gaugings of maximal and half-maximal supergravity and certain generators of $E_{11}$ Kac-Moody algebra has been developed. It has been generalized in \cite{Riccioni:2010xx,Riccioni:2010fc} to include the gauging of local scaling symmetry (Trombone symmetry) \cite{Cremmer:1997xj} of maximal supergravity. Later in \cite{LeDiffon:2011wt}, it was shown how it can be realized in the context of embedding tensor formalism. Since we have established a relation between our result and the one of embedding tensor formalism, it would be interesting to consider how the above correspondence is realized in our framework.

4) In Chapter \ref{ch:exotic}, we considered the solutions to supergravity theories in $6+5$ dimensions with a non-degenerate metric. It may be of interest to consider a case where the eleven-dimensional metric is not of maximal rank. See for example \cite{Dragon:1991fn} where similar question has been addressed in the context of standard supergravity.

5) We discussed in Chapter \ref{ch:exotic} that there are Freund-Rubin solutions for each timelike T-dual $M$-theories. Given the standard $M$-theory in $10+1$ dimensions, the corresponding Freund-Rubin solution is a supergravity theory in $AdS_4$ background spacetime. There is a superconformal field theory in three-dimensions known as ABJM theory; an $\mathcal{N}=6$ supersymmetric Chern-Simons-matter theory \cite{Aharony:2008ug}. It was shown that the ABJM theory can be seen as the large N limit of a stack of N $M_2$-branes sitting at the orbifold singularity $\mathbb{Z}_k$ such that it can be realized as the CFT dual to maximal supergravity in the near horizon geometry $AdS_4 \times S^7/{\mathbb{Z}_k}$. Therefore one finds an example of AdS/CFT correspondence among the very few existing examples.

Given the existence of Freund-Rubin solution of $M'$- and $M^*$-theories on $AdS_4 \times S^{3,4}$ and $-AdS_4 \times -dS_7$ respectively, it would be interesting to build the corresponding three-dimensional superconformal field theories. This would then provide us with new examples of AdS/CFT correspondence.


\appendix
\chapter{Notation}\label{app-notation}

The following notation, beside the definition of Levi-Civita tensor, is specifically used in the Part \ref{part1}. We generally follow the notation of \cite{freedman2012supergravity}, when we discuss supergravity.

\begin{itemize}

\item
$(s,t)$: Signature of a spacetime with
\begin{itemize}
    \item 
    $s$: Number of space directions of the ambient space manifold
    \item
    $t$: Number of time directions of the ambient space manifold with $t=T+T'$
\end{itemize}

\item
$T$: Number of time directions of the internal manifold with 

\item
$T'$: Number of time directions of spacetime ``background" manifold

\item
$\eta_{ab}$: the metric on the tangent space $T\mathcal{M}$ of a manifold $\mathcal M$,  where $a,b=1,2,..., dim(T\mathcal{M})$.

\item
$g_{MN} =\left( \begin{smallmatrix} g_{\mu\nu} & 0 \\  0 & g_{mn} \end{smallmatrix} \right)$: the metric on $11$ dimensional ambient space, $M,N=0,...,10$.

\item
$g_{mn}$: the metric on seven dimensional internal space, $m,n=4,...,10$.

\item
$g_{\mu\nu}$: the metric on four dimensional spacetime, $\mu,\nu=0,...,3$.

\item
$F_{MNPQ}=4 \partial_{[M}A_{NPQ]}$.

\item
$\varepsilon_{M_1...M_{11}}$: the eleven dimensional Levi-Civita antisymmetric covariant tensor density of rank eleven and weight $-1$,  with $\varepsilon_{01234...}=+1$.

\item
$\varepsilon^{M_1...M_{11}}$: the eleven dimensional Levi-Civita antisymmetric contravariant tensor density of rank eleven and weight $+1$,  with $\varepsilon^{01234...}=(-1)^t$.  This choice is such that
$$\f{1}{\sqrt{\vert g^{11}\vert}}\varepsilon^{M_1 ... M_p}= \sqrt{\vert g^{11}\vert}  g^{M_1 N_1} g^{M_2 N_2}\cdots g^{M_{11} N_{11}} \varepsilon_{N_1 ... N_p},$$  where $g^{11}$ is the determinant of $g_{MN}$.  

\item
\begin{align}
\varepsilon_{M_1...M_k M_{k+1}...M_{k+p}}\varepsilon^{M_1...M_k N_{k+1}...N_{k+p}}&=(-1)^t\, k!\,\delta^{N_{k+1}...N_{k+p}}_{M_{k+1}...M_{k+p}}\nonumber\\
&=(-1)^t\, k!\, p!\, \delta^{N_{k+1}}_{[M_{k+1}}\delta^{N_{k+2}}_{M_{k+2}}...\delta^{N_{k+p}}_{M_{k+p}]}.\nonumber
\end{align}

\item
$SO(p,q)$: the group of all transformations which leaves invariant the bilinear form $\eta_{p,q}=\sum_{i=1}^{p} dx_i^2 - \sum_{j=1}^{q} dx_j^2$.

\item
The pseudo-sphere $S^{p-1,q}=SO(p,q)/SO(p-1,q)$ ($p \geq 1, q \geq 0$) is a manifold with induced metric of  signature $(p-1,q)$ and  positive curvature.  In particular, $S^{p-1,0}$ is the standard $(p-1)$-sphere $S^{p-1}$.

\item
The  pseudo-hyperbolic space $H^{p,q-1}=SO(p,q)/SO(p,q-1)$ ($p \geq 0, q \geq 1$) is a manifold with induced metric of  signature $(p,q-1)$ and negative curvature.  While one removes a spacelike direction from the ambient space to get $S^{p-1,q}$, one removes a timelike direction for $H^{p, q-1}$ (see Appendix \ref{App:pseudo} for an explicit example).  In particular, $H^{p,0}$ is the negative curvature standard hyperbolic space $H^p$ of dimension $p$. 

\end{itemize}

The following notations are specifically used in Part \ref{part2}:

\begin{itemize}
    \item 
    $H^{g,n}(s\vert d)$: BRST cohomology at ghost number $g$ in the space of local functional $n$-forms.
    
    \item
    $H^{g}(s\vert d)$: BRST cohomology at ghost number $g$ in the space of local functional of a top form-degree.
    
    \item
    $H^{n+g}_{\textrm{char}}(d)$: Characteristic cohomology
    
    \item
    $H^{g}(s|d)\simeq H^{n+g}_{\rm char}(d)$ for $g\leq -1\,$.
    
    \item
    The Section \ref{sec:first-order-actions}, has the opposite notation for Levi-Civita symbol compared to the rest of the thesis, i.e. $\varepsilon_{0123}=-1$. 
    
\end{itemize}

\chapter{Octonion and Split Octonion}\label{app:quaternion_octonion}

In this appendix, we give a quick review of Cayley-Dickson construction of division algebras. We discuss the algebra of octonions and split octonions and
we present the corresponding Fano planes.

\section{The Cayley-Dickson Construction of Divison Algebras}

The Cayley-Dickson construction of the normed division algebras is a powerful tool as it can nicely explain
the reason that why each of the normed division algebras can be obtained as a subalgebra of a larger normed division algebra. It also gives a manifestation of why as we go up in the dimension of division algebras, the algebras lose realness condition, commutativity and associativity gradually.

Let's consider $a,b\in \mathbb{R}$, then the complex number $a+bi$ can be thought of as a pair $(a, b)$. The addition and multiplication for any two complex numbers $(a,b)$ and $(c,d)$ is defined as
\begin{align}
(a,b) + (c,d) &= (a+c,b+d),\label{complex-add}\\
(a, b)  (c, d) &= (ac-db, ad+cb).\label{complex-mult}
\end{align}
The complex conjugate of a complex number thus is given by
\be
\Lbar{(a,b)} = (a,-b).
\ee
Having defined the complex numbers, we can define the quaternions in a similar way. A
quaternion can be thought of as a pair of complex numbers. Addition is again like \eqref{complex-add}, and
the multiplication is defined for $a,b \in \mathbb{C}$ as
\be\label{quater-mult}
(a, b)(c, d) = (ac-\bar{d}b , ad +b\bar{c}).
\ee
The quaternion conjugate is given by
\be\label{quater-conj}
\Lbar{(a, b)} = (\bar{a} , -b).
\ee
If one repeats the same procedure once again, the algebra of octonions is obtained.
This iterative method of building up normed division algebras is called the Cayley-Dickson construction. If one continues to build up algebras by this construction,
all algebras larger than octonions will be normed but not division algebras \cite{Baez:2002}.

\section{Octonions}

Let $a$ and $b$ be two quaternions, $a = a_0 + a_1 i + a_2 j + a_3 k$, $b = b_0 + b_1 i + b_2 j + b_3 k$ ($a_i, b_i \in \mathbb{R}$, $i^2 = j^2 = k^2 = -1$, $ij = k = - ji$ etc). As we explained in the previous section, the octonions are pairs of quaternions $a + b \ell $, where $\ell $ is a new element, for which one defines the multiplication as \eqref{quater-mult} or explicitly as
\be 
(a +  b \ell) (c +  d \ell ) = ac + (\ell)^2 \bar{d} b +  (d a + b \bar{c}) \ell, \label{oct-mult}
\ee
where
\be \ell^2 = -1.
\ee
In the equation (\ref{oct-mult}), the overbar denotes quaternionic conjugation defined in \eqref{quater-conj}, $\bar{d} = d_0 - d_1 i - d_2 j - d_3 k$.

Considering the basis $\{e_0, e_1, ... , e_7\}$ the algebra of unit octonion is defined as
\begin{align}
e_0 e_i &= e_i e_0 = e_i,\\
e_i e_j &= -\delta_{ij} e_0 + f_{ijk} e_k,
\end{align}
where $f_{ijk}$ is a totally antisymmetric tensor and
\be
f_{ijk}=+1 \quad \textrm{for} \quad (ijk) = \{(123), (145), (176), (246), (257), (347),(365)\}.
\ee
One should notice that the above definition of $f_{ijk}$ is not unique but fixing the unit element $e_0=1$, then all different definitions are isomorphic. 

There is a mnemonic representation of the multiplication of octonions which is called Fano plane and is the simplest example of the projective plane. The Fano plane for the octonions is depicted in Figure \ref{fig:Fano_plane_oct} and one can easily find the other isomorphic tables using the symmetries of Fano plane.
\begin{figure}[htbp]
\centering
\def\svgscale{0.6}
\input{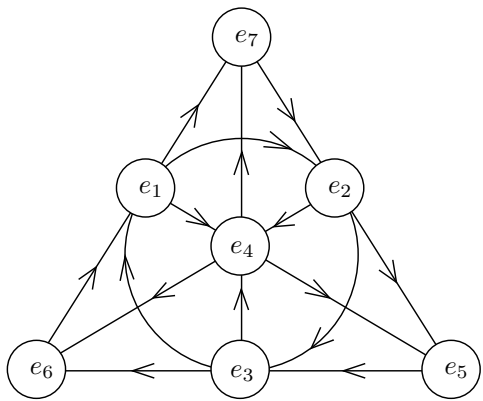}
    \caption{The Fano plane for the octonions multiplication table.}\label{fig:Fano_plane_oct}
\end{figure}

\section{Split Octonions}\label{app:split_oct}

Let $a$ and $b$ be two quaternions, $a = a_0 + a_1 i + a_2 j + a_3 k$, $b = b_0 + b_1 i + b_2 j + b_3 k$ ($a_i, b_i \in \mathbb{R}$, $i^2 = j^2 = k^2 = -1$, $ij = k = - ji$ etc). The split octonions are pairs of quaternions $a +   b \ell $, where $\ell $ is a new element, for which one defines the product as
\be 
(a +  b \ell) (c +  d \ell ) = ac + (\ell)^2 \bar{d} b +  (d a + b \bar{c}) \ell, \label{Product}
\ee
where
\be \ell^2 = +1,
\ee
instead of $-1$ as for the standard octonions.  In (\ref{Product}), the overbar denotes quaternionic conjugation, $\bar{d} = d_0 - d_1 i - d_2 j - d_3 k$.

Setting $e_0 = 1$, $e_1 = i$, $e_2 = j$, $e_3 = k$, $e_4 = \ell$, $e_5 = \ell i$, $e_6 = \ell j$ and $e_7 = \ell k$, one gets the same products $e_i e_j$ as for the standard octonions, except when two $\ell$'s are involved, in which case one gets an additional minus sign. For instance $e_5 e_4 = - e_1$ as for octonions, but $e_5 e_1 = e_4$ while it is $-e_4$ for standard octonions.  Similarly, $e_4^2 = e_5^2 = e_6^2 = e_7^2 = +1$ instead of $-1$.

One defines octonionic conjugation for a general split octonion
\be
x =  x^0 e_0 + x^1 e_1 + x^2 e_2 + x^3 e_3 + x^4 e_4 + x^5 e_5 + x^6 e_6 + x^7 e_7 , \; \; \; x^i \in \mathbb{R}
\ee
as
\be
\bar{x}  =  x^0 e_0 - x^1 e_1 - x^2 e_2 - x^3 e_3 - x^4 e_4 - x^5 e_5 - x^6 e_6 - x^7 e_7.
\ee
One has
\be
\overline{xy} = \bar{y} \bar{x},
\ee
and of course $\overline{\bar{x}} = x$.
A scalar product can then be introduced as
\be
(x,y) = \frac12 \left( x \bar{y}+ y \bar{x} \right) = \frac12 \left( \bar{x} y+ \bar{y} x \right).
\ee
Explicitly, one finds
\begin{eqnarray}
(x,y) &=& x^0 y^0 + x^1 y^1 + x^2 y^2 + x^3 y^3 - x^4 y^4 - x^5 y^5 - x^6 y^6 - x^7 y^7 \\
&=& x^0 y^0 + \sum_{i=1}^7 \eta_{ij} x^i y^j,
\end{eqnarray}
where $\eta_{ij}$ is the flat metric with signature $(3,4)$.

An octonion is pure imaginary if $\bar{x} = - x$.  The unit octonions $e_i$ are pure imaginary.  The pure imaginary condition is equivalent to  $(e_0, x) = 0$. 

One can rewrite the product of the unit split octonions $e_i$ as
\be
e_i e_j = - \eta_{ij} e_0+ f\indices{^k_{ij}} e_k,
\ee
where the structure constants $f\indices{^k_{ij}}$ are such that the $f_{kij} \equiv \eta_{km} f\indices{^m_{ij}}$ are completely antisymmetric in $(i,j,k)$. The tensor $f_{ijk}$ has value $+1$ for $(ijk) = (123)$, $(145)$, $(167)$, $(246)$, $(275)$, $(347)$ and $(356)$ (and cyclic permutations) and it vanishes otherwise. Similarly to octonions, there is a mnemonic picture of the multiplication table of split octonions. The corresponding Fano plane has been given in Figure \ref{fig:Fano_plane_sp_oct}, where once again one can read off all possible multiplication tables using the symmetries of the Fano plane. The ordering of $abc$ for the non-zero components of $f_{abc}$ for the split octonions can be obtained from its octonionic counterpart and reversing the diagonal arrows in the Fano plane.

\begin{figure}[htbp]
\centering
\def\svgscale{0.6}
\input{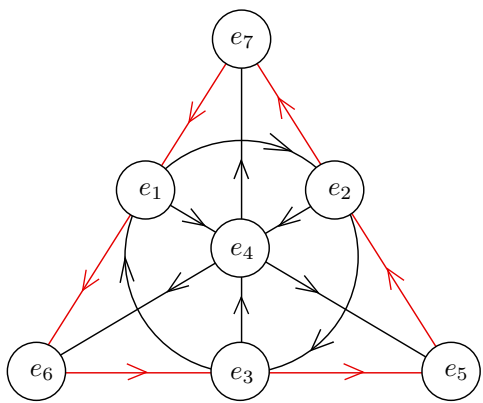}
    \caption{The Fano plane for split octonions multiplication table. This Fano plane can be obtained from its octonionic counterpart by reversing the surrounding arrows, shown in red, in the Fano plane.}\label{fig:Fano_plane_sp_oct}
\end{figure}

The squared norm $N(x) \equiv (x,x)$ of $x$ reads
\begin{eqnarray}
N(x)  &=& \left( x^0 \right)^2 +  \left( x^1 \right)^2 +  \left( x^2 \right)^2 +  \left( x^3 \right)^2 -  \left( x^4 \right)^2 - \left( x^5 \right)^2 -  \left( x^6 \right)^2 -  \left( x^7 \right)^2 \nonumber\\
&=& \eta_{\lambda \mu} x^\lambda x^\mu, \; \; \; \; (\lambda, \mu = 0, 1, \cdots, 7)
\end{eqnarray}
where $\eta_{\lambda \mu}$ is the flat metric with signature $(4,4)$.  The split octonions form a composition algebra, i.e., 
\be
N(xy) = N(x) N(y).
\ee

Just as the standard octonions, the split octonions do not form an associative algebra.  However, in the same way as for the standard octonions, the associator $[x,y,z]$ of three split octonions $x$, $y$, $z$, defined through
\be
(x,y,z) = (xy)z - x(yz),
\ee 
is an alternating function of $x$, $y$, $z$, i.e. 
\be
(x,y,z) = (y,z,x) = (z,x,y) = - (y,x,z) = -(x,z,y) = -(z,y,x).
\ee

More information on the split octonions can be found in \cite{Gunaydin:1973rs,Baez:2014}.

\section{Proof of Theorem \ref{Theorem:Horwitz} }

Using Lemma \ref{lemma-iden-left-mult} and writing $y=y+\lambda x$ for any $\lambda \in F, x\in A$, we get
\be
L_x L_{\bar{y}}+L_y L_{\bar{x}} = 2 (x|y) I.
\ee
Now let $N = \textrm{Dim}\, A$ be the dimension of the normed division algebra $A$. Let $A_0$ be a sub-vector space of $A$ defined by
\be
A_0 = \{x| (x,e) = 0,  x\in A\}.
\ee
Any $x \in A$ can be uniquely rewritten as $x = \lambda e+ y$ where $\lambda \in F,\, y \in A_0$, with $\lambda=(e,x)$ and $y=x-(e,x)e$. Then,
\be
A= Fe\oplus A_0, \qquad \textrm{Dim}\, A_0 = N-1.
\ee
For $x\in A_0$, we have $\bar{x}=-x$ and therefore
\be\label{Eq-pre-Clifford}
L_x L_{y}+L_y L_{x} = -2 (x,y) I,
\ee
for any $x,y \in A_0$. For the moment, let's assume the algebra $A$ is defined over the field of complex numbers. One can find an $N$ dimensional orthonormal basis $e_0(=
e),e_1,e_2,...,e_{N-1}$ of $A$ satisfying the orthonormal condition \cite{okubo1995}
\be
(e_\mu,e_\nu)= \delta_{\mu\nu},\qquad \mu,\nu=0,1,...,N-1.
\ee
In particular, $e_1,...,e_{N-1}$ realizes an orthonormal basis of $A_0$. Choosing $x = e_j$ and $y = e_k$ for $j,k \neq 0$ in \eqref{Eq-pre-Clifford}, we obtain
\be
L_j L_k + L_k L_j = -2 \delta_{jk},\qquad j,k=1,...,N-1.
\ee
which is a Clifford algebra in $N-1$ dimensional vector space $A_0$. Any matrix representation of a Clifford algebra is reducible and the dimension of irreducible representation
in $N-1$ dimensional space is\cite{dixon2013}
\be
d=2^n,\qquad \bigg\{ \begin{array}{l} n={\frac{N-1}{2}}\qquad\quad (N-1)\, \textrm{even}\\ ~  \\ n={\frac{N}{2}-1}  \qquad\, (N-1)\, \textrm{odd} \end{array}.
\ee
Suppose that the $N \times N$ matrix representation of $L_j$s contains $m$ irreducible components
with its total matrix dimension of $md$. This must be equal to $N$, so that we have
\be
N=\bigg\{ \begin{array}{l} 2^{\frac{N-1}{2}}m \qquad (N-1)\, \textrm{even}  \\ 2^{\frac{N}{2}-1}m  \qquad (N-1)\, \textrm{odd} \end{array}.
\ee
There are four solutions to this equation:
\begin{itemize}
    \item $N=1$, $m=1$, 
    \item $N=2$, $m=2$,
    \item $N=4$, $m=2$,
    \item $N=8$, $m=1$.
\end{itemize}
This proves the theorem. Even though we assumed that the algebra was defined on a complex field, the proof still holds for the real field \cite{okubo1995}.

\chapter{The Pseudo-Sphere $S^{3,4}$ and the Split Octonions}\label{App:pseudo}

\section{A Survey of $S^{3,4}$ as a Homogeneous Space}\label{app:homogeneous}

Let $\mathbb{R}^{4,4}$ be the $8$-dimensional real vector space endowed with the flat metric of mixed signature $(4,4)$,
\be
ds^{2}_{4,4} = \left(dx^1\right)^2 + \left(dx^2\right)^2 + \left(dx^3\right)^2 + \left(dx^4\right)^2 - \left(dy^1\right)^2 -\left(dy^2\right)^2 - \left(dy^3\right)^2 - \left(dy^4\right)^2. 
\ee
We consider the  hypersurfaces:
\begin{align}
S^{3,4}&: \left(x^1\right)^2 + \left(x^2\right)^2 + \left(x^3\right)^2 + \left(x^4\right)^2 - \left(y^1\right)^2 -\left(y^2\right)^2 - \left(y^3\right)^2 - \left(y^4\right)^2 = 1,\\
H^{4,3}&: \left(x^1\right)^2 + \left(x^2\right)^2 + \left(x^3\right)^2 + \left(x^4\right)^2 - \left(y^1\right)^2 -\left(y^2\right)^2 - \left(y^3\right)^2 - \left(y^4\right)^2 = -1.
\end{align}
These are connected. Because the signature of the embedding space is symmetric for the exchange of the space and time directions, these hypersurfaces are clearly isomorphic.  The ``pseudo-sphere" $S^{3,4}$ has an induced metric with signature $(3,4)$ (three $+$ signs and four $-$ signs), while the ``pseudo-hyperbolic space" $H^{4,3}$ has an induced metric with signature $(4,3)$.   One goes from one to the other by an overall sign change of the metric.  

The split group $O(4,4)$ and its subgroups $SO(4,4)$ and $SO^+(4,4)$, the connected component of $SO(4,4)$, act transitively on $S^{3,4}$ and $H^{4,3}$. Given a point on $S^{3,4}$, for example the north pole given by $(1,0,0,...,0)$, the subgroups of $O(4,4)$, $SO(4,4)$ and $SO^+(4,4)$ which leave the pointed fixed are $O(3,4)$, $SO(3,4)$ and $SO^+(3,4)$ respectively. This subgroups are called the stability subgroups. Similarly, the corresponding stability subgroups for a fixed point on $H^{4,3}$ are $O(4,3)$, $SO(4,3)$ and $SO^+(4,3)$. The groups $O(3,4)$ and $O(4,3)$ (as well as $SO(3,4)$ and $SO(4,3)$, or  $SO^+(3,4)$ and $SO^+(4,3)$) are isomorphic, but we prefer to adopt different notations in each case to keep track of the signature of the metric.  Thus, the pseudosphere $S^{3,4}$ and pseudo-hyperbolic space $H^{4,3}$ are the homogeneous spaces:
\be 
S^{3,4} = \frac{O(4,4)}{O(3,4)} = \frac{SO(4,4)}{SO(3,4)} =  \frac{SO^+(4,4)}{SO^+(3,4)},
\ee
and
\be
H^{4,3} = \frac{O(4,4)}{O(4,3)} = \frac{SO(4,4)}{SO(4,3)} =  \frac{SO^+(4,4)}{SO^+(4,3)}.
\ee

The pseudo-sphere $S^{3,4}$ and pseudo-hyperbolic space $H^{4,3}$ are not only homogeneous spaces, they are in fact maximally symmetric and hence spaces of constant curvature,
\be R_{mnpq} = K \left(g_{mp}g_{nq} - g_{mq}g_{np} \right),
\ee
with $K=1 >0$ for the pseudo-sphere $S^{3,4}$ and $K = -1 <0$ for the pseudo-hyperbolic space $H^{4,3}$,  in agreement with the observation that under an overall change of sign of the metric, the Riemann curvature changes sign but the product $\left(g_{mp}g_{nq} - g_{mq}g_{np} \right)$ does not, so that $K$ changes sign.

\section{Two Infinite Family of Parallelizations of the Pseudo-Sphere $S^{3,4}$}\label{app:parallel_pseudo_sphere}

There exist two infinite families of parallelizations of the pseudo-sphere $S^{3,4}$.  We first start by describing one of these parallelizations.

Consider the pseudo-sphere $S^{3,4}$ of unit split octonions, $x \in S^{3,4} \Leftrightarrow N(x) =1$.  The real number $1$ is the ``North pole".  The tangent space at the North pole can be identified with the seven-dimensional vector space of imaginary split octonions.  A Lorentz basis of this tangent space is given by the $e_i$'s.  Let  $x \in S^{3,4}$.   Right multiplication by $x$ maps  the North pole to $x$, and the tangent space at the North pole to the tangent space at $x$.  Indeed, the vectors $e_i$ at the North pole are mapped on $t^+_i  \equiv e_i x$.  One has:
\be
(t^+_i, x) = 0, \; \; \; (t^+_i, t^+_j) = \eta_{ij}.
\ee
The first equality expresses that the seven imaginary split octonions $t^+_i$, viewed as vector in $\mathbb{R}^{4,4}$ are orthogonal to $x$ and hence tangent to the pseudo-sphere $S^{3,4}$ at $x$.  The second equality expresses that the $t^+_i$'s form a Lorentz basis of that tangent space.  

We have thus defined at each point of the pseudo-sphere $S^{3,4}$ a Lorentz basis of the tangent space.  This provides an absolute parallelism for $S^{3,4}$. The tangent vector at $x$ parallel to the tangent vector $e_i$ at $1$ is the vector $e_i x$ obtained by right multiplication with $x$.  The corresponding parallel transport preserves the metric of $S^{3,4}$ and its geodesics can be verified to coincide with those defined by the metric. The parallelization defined by $\{t^+_i\}$ is thus consistent with the metric. 

Similarly, left multiplication also defines a parallelization of $S^{3,4}$ that maps the tangent basis $\{e_i\}$ at $1$ on the tangent basis $\{t_i^- \equiv x e_i\}$ at $x$.  The two parallelisms are inequivalent since the tangent vectors $t^+_i$ and $t_i^-$   coincide only at the North pole and at the ``South pole" $-1$. 

Yet other parallelizations can be defined by using a reference point on  $S^{3,4}$ different from unity.  More precisely, let $\alpha \in S^{3,4}$.  One goes from $\alpha$ to $x$ by multiplying $\alpha$ by $\bar{\alpha} x$, $\alpha (\bar{\alpha } x) = x$.  One defines the tangent vector $^{(\alpha)} t_i^+$ at $x$ parallel to the tangent vector $e_i \alpha$ at $\alpha$ through right multiplication by the octonion $\bar{\alpha} x$ that connects $\alpha$ to $x$, $^{(\alpha)} t_i^+ = (e_i \alpha) (\bar{\alpha} x)$.  The tangent vectors $^{(\alpha)} t_i^+$ and $t^+_i$ at $x$ do not coincide because octonionic multiplication is not associative.  A similiar construction yields the parallelism $^{(\alpha)} t_i^- = (x \bar{\alpha})(\alpha e_i)$.

Because the families of parallelisms given by the above construction are consistent with the metric, the corresponding torsion tensors obey the equations (\ref{Torsion101})-(\ref{Torsion103}) given above \cite{Cartan1,Cartan2,wolf1972I,wolf1972II}.  

The parallelisms of $S^{3,4}$ are related to the split octonions in the same way as the parallelisms of the seven-sphere are related to the standard octonions.  For that reason, the reader can find more information on the parallelisms of $S^{3,4}$ in the literature on the parallelisms of the seven-sphere.  A reference that we have found useful is \cite{Rooman:1984}.

\section{$Spin^+(3,4)$ and $G_{2,2}^*$}\label{app:sp-and-g2}

The complex Lie algebra $\mathfrak{g}_2$ possesses two real forms, the compact one and the split one.  To the compact real form corresponds the unique compact group $G_2$.  To the split real form correspond the simply connected non compact group $G_{2,2}$ with center $\mathbb{Z}_2$ and the quotient $G_{2,2}^* \equiv \frac{G_{2,2}}{\mathbb{Z}_2} $ (``adjoint real form") which has trivial center (and is not simply connected). The group  $G_{2,2}^*$ is the automorphism group of the split octonions.

The group $Spin(7)$ is well-known to have a transitive action on the seven-sphere $S^7$, with isotropy group $G_2$.  Similarly, the group $Spin^+(3,4)$ (connected component of $Spin(3,4)$) has a transitive action on the pseudo-sphere $S^{3,4}$ with isotropy group $G_{2,2}^*$ \cite{Kath:1996gm}.  We can thus also identify $S^{3,4}$ with the homogeneous space $Spin^+(3,4) /G_{2,2}^*$,
\be
S^{3,4} \simeq \frac{Spin^+(3,4)}{G_{2,2}^*}
\ee
\chapter{Derivation of Equation (\ref{eq:Utrasf})}
\label{app:derivation}

In this appendix, we derive formula \eqref{eq:Utrasf} for the variation $\delta_u G_I = - (U_u, G_I)$ of the two-form $G_I$ under a $U$-type symmetry. This is done in two steps:
\begin{enumerate}
\item First, we show that
\begin{equation}\label{eq:step1}
\delta_u G_I + (f_u)\indices{^J_I} G_J \approx c_{IJ} F^J + d(\text{invariant})
\end{equation}
for some constants $c_{IJ}$.
\item Then, we prove that the $c_{IJ}$ take the form
\begin{equation}\label{eq:step2}
c_{IJ} = - 2 (h_u)_{IJ} + \lambda^w_u (h_w)_{IJ},
\end{equation}
where the constants $(h_u)_{IJ}$ and $(h_w)_{IJ}$ are those appearing in the currents associated with $U_u$ and $W_w$ respectively.
\end{enumerate}
The proof is given in the case where the Lagrangian (or, equivalently, $G_I$) does not depend on the derivatives of $F^I_{\mu\nu}$.

\subsubsection*{A lemma}

The proof of the above
 steps uses the following result on the $W$-type cohomology classes (with $g=-1$):
\begin{equation}\label{eq:Wlemma}
t_{IJ} F^I F^J \approx d(\text{invariant}) \;\Rightarrow\; t_{IJ} = \sum_{w} \lambda^w (h_w)_{IJ} \,\text{ for some } \lambda^w .
\end{equation}
This is proven as follows: $t_{IJ} F^I F^J \approx d(\text{invariant})$ implies that
\begin{equation}\label{eq:appt}
t_{IJ} F^I F^J + dI + \delta k = 0
\end{equation}
for some gauge invariant $I$ and some $k$ of antifield number $1$, where $\delta$ is here the Koszul-Tate differential. Now, it is proven in \cite{Barnich:2000zw} that $k$ must be gauge invariant; hence, it can be written as
\begin{equation}\label{eq:appk}
k = \hat{K} + d R, \quad \hat{K} = d^4x [ A^{*\mu}_I g^I_\mu + \phi^*_i \Phi^i ]
\end{equation}
for some gauge invariant $R$, $g^I_\mu$ and $\Phi^i$. Indeed, derivatives acting on the antifields contained in $k$ are pushed to the term $dR$ by integration by parts, leaving the form \eqref{eq:appk} where $\hat{K}$ contains only the undifferentiated antifields. Putting this back in \eqref{eq:appt} and using the fact that $\delta \hat{K} = s \hat{K}$ because $\hat{K}$ is gauge invariant, we get
\begin{equation}
s \hat{K} + d\left( t_{IJ} A^I F^J + J \right) = 0
\end{equation}
for some gauge invariant $J = I - \delta R$. This shows that $\hat{K}$ is a $W$-type cohomology class: we can therefore expand $\hat{K}$ in the $W_w$ basis as $\hat{K} = \sum \lambda^w W_w$. In particular, this implies that $t_{IJ} = \sum \lambda^w (h_w)_{IJ}$, which proves the lemma.

\subsubsection*{First step}

We start from the chain of descent equations involving $G_I$,
\begin{equation}\label{eq:descCstar}
s\, d^4x\, C^*_I+d \star A^*_I=0,\quad s \star A^*_I+d
G_I=0,\quad sG_I=0.
\end{equation}
Applying $(U_u, \cdot)_\text{alt}$ to this chain, we get
\begin{align}
s \left[ \,d^4x\, (f_{u})\indices{^J_I} C^*_J \right] + d \left[ (f_{u})\indices{^J_I} \star A^*_J +
\frac{\delta K_u}{\delta A^I} \right] &= 0, \\
s \left[ (f_{u})\indices{^J_I} \star A^*_J +
\frac{\delta K_u}{\delta A^I} \right] + d \left[ - \delta_u G_I \right] &= 0, \\
s \left[ - \delta_u G_I \right] &= 0,
\end{align}
which can be simplified to
\begin{align}
d \left( \frac{\delta K_u}{\delta A^I} \right) &= 0, \label{eq:dK}\\
s \left( \frac{\delta K_u}{\delta A^I} \right) + d \left( - \delta_u G_I - (f_{u})\indices{^J_I} G_J \right) &= 0, \label{eq:sK}\\
s \left( - \delta_u G_I \right) &= 0,
\end{align}
using equations \eqref{eq:descCstar} again.
Equation \eqref{eq:dK} implies that
\begin{equation}
\frac{\delta K_u}{\delta A^I} = d \eta^{-1,2}
\end{equation}
for some $\eta^{-1,2}$ of ghost number $-1$ and form degree $2$. Because the left-hand side is gauge invariant and $\eta^{-1,2}$ is of form degree two, $\eta^{-1,2}$ must also be gauge invariant. This follows from theorems on the invariant cohomology of $d$ in form degree $2$ \cite{Brandt:1989gy,DuboisViolette:1992ye}. Equation \eqref{eq:sK} implies then
\begin{equation}
d\left( \delta_u G_I + (f_{u})\indices{^J_I} G_J + s\eta^{-1,2} \right) = 0 ,
\end{equation}
i.e.
\begin{equation}\label{eq:dGdeta}
\delta_u G_I + (f_{u})\indices{^J_I} G_J + s\eta^{-1,2} = d \eta^{0,1}
\end{equation}
for some $\eta^{0,1}$ of ghost number $0$ and form degree $1$. Again, the left-hand side of this equation is gauge invariant: results on the invariant cohomology of $d$ in form degree $1$ \cite{Brandt:1989gy,DuboisViolette:1992ye} now imply that the non-gauge invariant part of $\eta^{0,1}$ can only be a linear combination of the one-forms $A^I$,
\begin{equation}
\eta^{0,1} = c_{IJ} A^J + \text{(gauge invariant)} .
\end{equation}
Plugging this back in equation \eqref{eq:dGdeta} and using the fact that $s\eta^{-1,2} \approx 0$ (since $\eta^{-1,2}$ is gauge invariant), we recover equation \eqref{eq:step1}. This concludes the first step of the proof.

\subsubsection*{Second step}

For the second step, we introduce
\begin{equation}
N = - \int \!d^4x\,( C^*_I C^I + A^{*\mu}_I A^I_\mu), \quad \hat{N} = (N, \cdot)_\text{alt} .
\end{equation}
The operator $\hat{N}$ counts the number of $A^I$'s and $C^I$'s minus the number of $A^*_I$'s and $C^*_I$'s. Because it carries ghost number $-1$, it commutes with the exterior derivative, $\hat{N} d = d \hat{N}$.
Applying this operator to the equation
\begin{equation}
s U_u + d \left[ (f_u)\indices{^I_J} (\star A^*_I C^J + G_I A^J) + (h_u)_{IJ} F^I A^J  + J_u \right] = 0
\end{equation}
gives
\begin{equation}\label{eq:NsU}
(\int\! G_I F^I, U_u)_\text{alt} + d\left[ (f_u)\indices{^I_J} (\hat{N} + 1)(G_I) A^J + 2 (h_u)_{IJ} F^I A^J  + \hat{N}(J_u) \right] \approx 0 .
\end{equation}
The second term is evident. The first term is
\begin{align}
\hat{N}(sU_u) = (N, (S,U_u)_\text{alt})_\text{alt} &= ( (N, S) , U_u)_\text{alt} + (S, (N,U_u)_\text{alt})_\text{alt}
\end{align}
according to the graded Jacobi identity.
The counting operator $\hat{N}$ kills the $A^{*\mu}_I \partial_\mu C^I$ term in the master action $S$, which implies
\begin{equation}
(N, S) = \int\!d^4x\, A^I_\mu \frac{\delta \mathcal{L}_V}{\delta A^I_\mu} = \int\!d^4x\, A^I_\mu \partial_\nu(\star G_I)^{\mu\nu} = \int\! G_I F^I .
\end{equation}
Similarly, $\hat{N}$ kills the first two terms of $U_u$ given in \eqref{eq:5bis}, leaving $\hat{N} U_u = \hat{N} K_u$ which is gauge invariant. This implies $(S, (N,U_u)_\text{alt})_\text{alt} = s(\hat{N} U_u) \approx 0$.
Therefore, we have indeed
\begin{equation}
\hat{N}(sU_u) \approx (\int G_I F^I, U_u)_\text{alt}
\end{equation}
which proves equation \eqref{eq:NsU}.

We now compute $(\int G_I F^I, U_u)_\text{alt}$ using the result of the first step. We have
\begin{equation}
(\int G_I F^I, U_u)_\text{alt} = \frac{\delta (G_K F^K)}{\delta A^I_\mu} \, \delta_u A^I_\mu + \frac{\delta (G_K F^K)}{\delta \phi^i} \, \delta_u \phi^i .
\end{equation}
This looks like the $U$-variation $\delta_u (G_I F^I)$, but it is not because there are Euler-Lagrange derivatives. For a top form $\omega$, the general rule is \cite{Anderson92introductionto}
\begin{equation}
\delta_Q \omega = Q^a \frac{\delta \omega}{\delta z^a} + d \rho, \quad \rho = \partial_{(\nu)} \left[ Q^a \frac{\delta}{\delta z^a_{(\nu)\rho}} \frac{\partial \omega}{\partial dx^\rho} \right] .
\end{equation}
In our case, this becomes
\begin{align}
\delta_u (G_I F^I) &= \frac{\delta (G_K F^K)}{\delta A^I_\mu} \, \delta_u A^I_\mu + \frac{\delta (G_K F^K)}{\delta \phi^i} \, \delta_u \phi^i + d\rho_A + d\text{(inv)} ,\\
\rho_A &= \partial_{(\nu)} \left( (f_u)\indices{^I_J} A^J_\mu \frac{\delta}{\delta A^I_{\mu, (\nu)\rho}} \frac{\partial (G_K F^K)}{\partial dx^\rho} \right) .
\end{align}
Using property \eqref{eq:step1} and putting together the terms of the form $d\text{(invariant)}$, we get then from \eqref{eq:NsU}
\begin{equation}
(c_{IJ} + 2 (h_u)_{IJ} ) F^I F^J + d \left[ (f_u)\indices{^I_J} A^J (\hat{N} + 1)(G_I) - \rho_A \right] + d \text{(inv)} \approx 0 .
\end{equation}
Now, it is sufficient to prove that
\begin{equation}\label{eq:dxy}
 d \left[ (f_u)\indices{^I_J} A^J (\hat{N} + 1)(G_I) - \rho_A \right] \approx d \text{(inv)}.
\end{equation}
Indeed, this implies $(c_{IJ} + 2 (h_u)_{IJ} ) F^I F^J \approx d \text{(inv)}$, which in turn gives
\begin{equation}
c_{IJ} = - 2 (h_u)_{IJ} + \lambda^w_u (h_w)_{IJ}
\end{equation}
for some constants $\lambda^w_u$ using property \eqref{eq:Wlemma} of the $W$-type cohomology classes.

\subsubsection*{Proof of \eqref{eq:dxy}}

We will actually prove the stronger equation
\begin{equation}\label{eq:xy}
\rho_A = (f_u)\indices{^I_J} A^J (\hat{N} + 1)(G_I)
\end{equation}
in the case where $G_I$ depends on $F$ but not on its derivatives.

To do this, we can assume that $G_I$ a homogeneous function of degree $n$ in $A^I$, i.e. $\hat{N}(G_I) = n G_I$. If it is not, we can separate it into a sum of homogenous parts; the result then still holds because equation \eqref{eq:xy} is linear in $G_I$.

In components, equation \eqref{eq:xy} is
\begin{equation}
\frac{1}{2}\partial_{(\nu)} \left( (f_u)\indices{^I_J} A^J_\mu \frac{\delta}{\delta A^I_{\mu, (\nu)\rho}} G_{K\sigma\tau} F^K_{\lambda\gamma} \varepsilon^{\sigma\tau\lambda\gamma} \right) = (n+1) (f_u)\indices{^I_J} A^J_\lambda G_{I\sigma\tau} \varepsilon^{\rho\lambda\sigma\tau} .
\end{equation}
Under the homogeneity assumption $\hat{N}(G_I) = n G_I$, we have
\begin{equation}
G_{K\sigma\tau} F^K_{\lambda\gamma} \varepsilon^{\sigma\tau\lambda\gamma} = 4 (n+1) \mathcal{L}_V .
\end{equation}
Equation \eqref{eq:xy} now becomes
\begin{equation}\label{eq:xyL}
\frac{1}{2}\partial_{(\nu)} \left( (f_u)\indices{^I_J} A^J_\mu \frac{\delta \mathcal{L}_V}{\delta A^I_{\mu, (\nu)\rho}} \right) = \frac{1}{4} (f_u)\indices{^I_J} A^J_\lambda G_{I\sigma\tau} \varepsilon^{\rho\lambda\sigma\tau} .
\end{equation}
We now use the fact that $G_I$ does not depend on derivatives of $F$, which implies that the higher order derivatives $\partial_{(\nu)}$ are not present and that the Euler-Lagrange derivatives are only partial derivatives. We then have
\begin{equation}
\frac{1}{2} \frac{\delta \mathcal{L}_V}{\delta A^I_{\mu,\rho}} = \frac{\delta \mathcal{L}_V}{\delta F^I_{\rho\mu}} = \frac{1}{4} \varepsilon^{\rho\mu\sigma\tau} G_{I\sigma\tau}
\end{equation}
(see \eqref{eq:47}), which proves \eqref{eq:xyL} in this case.

\addcontentsline{toc}{chapter}{Bibliography}


\end{document}